\documentclass[12pt,a4paper]{article}

\usepackage{econometrics}  %

\title{\Large Inference with Granular Instruments in Large Panels\footnote{\linespread{1}\selectfont I am deeply grateful to Eric Renault and Kenichi Nagasawa for their guidance and support throughout this project. I thank Ao Wang, Mingli Chen, Nathan Canen, Rocío Sánchez Mangas, Pilar Poncela, Amedeo Andriollo, Matthias Endres, Adam D. Lizia, Shanthanu Chadda, the participants at the Econometrics Seminar Series at the University of Warwick, the International Econometrics PhD Conference at Erasmus University, Rotterdam, Econometrics seminar at UAM, Madrid, and the NBER Time Series Conference for their valuable comments and suggestions. I thank ESRC for funding this project.}}
\author{Gokul Gopalan Ramachandran \footnote{Department of Economics, University of Warwick.} \\[0.5ex]
    {\small \href{https://arxiv.org/abs/2607.02095}{\textcolor{blue}{Click here for the most recent version}}}}
\date{October 29, 2026}

\begin{document}

\addtocontents{toc}{\protect\setcounter{tocdepth}{-10}}

\begin{titlepage}
    \maketitle
    \thispagestyle{empty}
    \enlargethispage{3\baselineskip}

\begin{abstract}
I develop asymptotic theory for Granular Instrumental Variables (GIV) in large panels with both the cross section ($N$) and the time dimension ($T$) growing. The granular instrument is constructed from individual shocks that must be estimated from the panel. I show that this estimation has a first-order effect on the asymptotic variance of the GIV estimator. Existing methods do not account for this effect, leading to incorrect standard errors and invalid inference. I derive corrected standard errors that deliver valid inference.

I show that instrument strength can change as the cross section grows. This yields strong, nearly-weak, and weak regimes for GIV. In particular, I identify a previously uncharacterized nearly-weak regime in which the instrument vanishes asymptotically, but GIV remains valid. In this regime, if the time dimension grows sufficiently fast relative to the cross section, the estimator remains consistent and asymptotically normal at a rate slower than $\sqrt{T}$. Further, I characterize the conditions where GIV exhibits standard weak-instrument behavior and provide Anderson-Rubin confidence intervals.

I apply the theory to estimate short-run demand and supply elasticities for copper, natural gas, and crude oil. Across the six elasticities, accounting for instrument estimation changes standard errors by up to a fifth.
\end{abstract}

\end{titlepage}

\section{Introduction} \label{sec:introduction}

Many questions in economics require estimating structural relationships between aggregate variables, for instance, how asset prices respond to changes in
aggregate demand, or how exchange rates react to capital flows. The estimating equation looks like this:
\begin{equation*}
    d_t = \phi p_t + \varepsilon_t
\end{equation*}
where $d_t$ is the demand and $p_t$ is the market-clearing price. A central challenge in estimating the parameter $\phi$ is endogeneity: the regressor ($p_t$) is correlated with the structural error.

\citet{Gabaix2024GranularVariables} (GK hereafter) proposed Granular Instrumental Variables (GIV) as a solution to this problem of endogeneity. The key insight is that when a few units are dominant (e.g., firms with large market shares, banks with large loan portfolios, countries with large capital inflows), then their idiosyncratic shocks can serve as valid instruments for the aggregate variable. We call settings with dominant units granular, and we can use GIV for consistent estimation of structural parameters in these settings. 

Several studies apply this idea to market-clearing conditions in commodity and asset markets, bank lending, and sovereign risk. However, the existing theory assumes a fixed number of units (fixed $N$). Many empirical applications of GIV have moderate to very large values of $N$. For instance, \citet{Aldasoro2023TheMarkets} has $N = 21$, while \citet{Chodorow-Reich2021AssetInsulators,Galaasen2020GranularRisk,Ma2021ExpectationsLending} have $N > 100$. \citet{Bai2003InferentialDimensions} shows that large-$N$ asymptotics provide a good small-sample approximation in similar settings, even for moderate values of $N$. Hence, large-$N$ asymptotics are a natural choice for these applications. These settings also exhibit cross-sectional and time-series dependence, which fixed-$N$ theory cannot accommodate.

Hence, there is a need to develop the theory of GIV for large panels, since the existing fixed-$N$ theory may not provide a reliable finite-sample approximation of GIV estimators in empirical settings with moderate to large $N$. This extension creates a fundamental complication. GIV relies on some units being large relative to the aggregate for instrument strength and validity. Under fixed-$N$ asymptotics, the relative size of these units is held fixed and therefore does not change with the asymptotics. Once $N$ is allowed to grow, however, their relative size can itself evolve along the sequence of economies. A unit that is dominant at a given $N$ need not remain dominant as $N \to \infty$. Since GIV depends on the presence of dominant units, a large-panel theory must therefore characterize how granularity evolves as the cross section grows.

In this paper, I develop asymptotic theory of Granular Instrumental Variables for large panels ($N,T \to \infty$), explicitly accounting for how granularity evolves as $N$ grows. First, I show that the strength of GIV depends on the presence of dominant units and on how their shares evolve as $N$ grows. This leads to three regimes for the GIV estimator, with distinct results for consistency, convergence rates, and inference. Second, I provide the correct asymptotic theory for the GIV estimator across all three regimes. The GIV instrument must be constructed in a first stage. This first-stage estimation has a first-order effect on the asymptotic distribution of the estimator. Current methods do not account for this, which leads to invalid inference. I then illustrate my theory empirically, using GIV to estimate the short-run demand and supply elasticities of three major commodities: copper, crude oil, and natural gas.

To formalize granularity, I model unit sizes as draws from a common distribution. The extent of granularity is determined by the thickness of the upper tail of this distribution: the heavier the tail, the more likely it is that a few exceptionally large units account for a substantial share of the aggregate. For theoretical tractability, I parameterize this tail behavior using a power-law distribution. Different degrees of tail thickness generate three distinct regimes for GIV.

The first regime corresponds to a sufficiently heavy-tailed size distribution, in which a few units dominate the aggregate. Some units remain large relative to the aggregate even as $N \to \infty$. I find that in these settings the instrument remains strong, and the GIV estimator is consistent and asymptotically normal at the standard $\sqrt{T}$ rate. This is the classical strong-instrument regime.

The second is an intermediate regime in which some units remain large relative to the rest, but their relative importance nevertheless vanishes as $N \to \infty$. In this regime, the instrument also vanishes asymptotically, but more slowly than the $N^{-1/2}$ rate. If $T$ grows sufficiently fast relative to $N$, the GIV estimator remains consistent and asymptotically normal, but converges at a rate slower than $\sqrt{T}$. This result is similar in spirit to nearly-weak identification in GMM models in \citet{Antoine2021GMMIdentification}, and hence I call this the nearly-weak regime. \citet{Antoine2021GMMIdentification} show that when the population moment condition vanishes sufficiently slowly, the GMM estimator can remain consistent, but at a slower rate. This therefore expands the range of settings in which GIV remains valid, showing that it can deliver consistent inference even when the granular instrument vanishes asymptotically: a setting not covered by the existing GIV theory.

The third regime arises when no unit remains large relative to the aggregate as $N$ grows. The granular variation on which the instrument relies then disappears more rapidly, and the instrument vanishes at the $N^{-1/2}$ rate. What happens to the GIV estimator depends on the relative growth of $N$ and $T$. When $T$ grows sufficiently fast relative to $N$, consistent estimation is theoretically possible, but standard Wald inference is unreliable in small samples. When $N$ and $T$ grow in proportion, the GIV estimator is inconsistent and has a non-standard weak-instrument distribution in the sense of \citet{Staiger1997InstrumentalInstruments}. In this regime, I recommend inverting the Anderson--Rubin test of \citet{Anderson1949EstimatorsEquations}. The test statistic has a $\chi^2_1$ limiting distribution that holds no matter how weak the instrument is.

Finally, the GIV is a generated instrument: the individual shocks used to construct it are not observed and must be estimated. This estimation does not change the regime into which a panel falls, but it does affect the asymptotic variance of the GIV estimator. Because the instrument is constructed using estimated rather than observed shocks, the estimation error has a first-order effect on the asymptotic variance of the GIV estimator. Existing GIV inference does not account for this source of uncertainty, leading to invalid standard errors and inference. I derive the correct asymptotic variance and standard errors in the regimes in which GIV delivers valid inference: the strong and nearly-weak regimes. In the weak regime, I give the test-statistic that accounts for the first-stage estimation. 

Applying this theory to data, I use large-panel GIV to estimate the supply and demand elasticities of three commodities: copper, crude oil, and natural gas. GIV corrects the biased OLS estimates and leads to economically plausible negative demand elasticities: $-0.148$ for copper, $-0.105$ for crude oil, and $-0.056$ for natural gas. The supply elasticities are positive and inelastic, at $0.141$, $0.256$, and $0.150$.

Finally, in simulation exercises calibrated to the panel data from copper and crude oil, I study the sensitivity of the GIV estimates to granularity. As expected, the strong regime has very low bias, tight confidence sets, and the correct coverage. In the weak regime, I observe substantial bias together with large confidence intervals. The bias remains low in the intermediate nearly-weak regime. 

I close this section by placing the paper in the context of the related literature.

\subsection{Related Literature}

\textbf{GIV theory and applications.} \citet{Gabaix2024GranularVariables} introduce the GIV framework under fixed $N$ and $T \to \infty$, establishing consistency, asymptotic normality, and the theory for the estimator. \citet{Banafti2022InferentialDimensions} extend GIV to high dimensions, allowing $N$ to grow with $T$. They show that the strength of the granular instrument can deteriorate as $N$ grows but derive the asymptotic distribution of the GIV estimator only under conditions that maintain a strong instrument. They also allow the instrument to be estimated and show that this estimation has no first-order effect on the GIV estimator.

I build on this analysis by connecting the deterioration in instrument strength to the asymptotic properties of the GIV estimator. I characterize the conditions under which GIV is strong, nearly weak, or weak, and derive the corresponding results for consistency, convergence rates, and inference. This reveals an intermediate, nearly-weak regime in which the instrument vanishes asymptotically but the GIV estimator remains consistent and asymptotically normal under appropriate relative rates of $N$ and $T$. I also show that estimation of the granular instrument has a first-order effect on the asymptotic variance of the GIV estimator. In Section \ref{subsec:banfti}, I show that the result in \citet{Banafti2022InferentialDimensions} that instrument estimation is asymptotically negligible relies on an assumption that is incompatible with the mechanism generating a strong granular instrument in the settings considered here. I derive the asymptotic distribution under assumptions that are compatible with this mechanism.

\citet{Qian2023Heterogeneity-robustInstruments} constructs heterogeneity-robust granular instruments that remain valid when the structural parameter varies across a fixed number of units. \citet{Baumeister2023AVariables} propose a full-information approach that jointly estimates the factor structure and structural parameter under parametric assumptions. Both approaches are developed under fixed $N$.

On the empirical side, GIV has been applied to asset markets \citep{Chodorow-Reich2021AssetInsulators}, bank credit risk \citep{Galaasen2020GranularRisk}, bank lending \citep{Ma2021ExpectationsLending}, sovereign bonds \citep{Aldasoro2023TheMarkets}, exchange rates \citep{Hau2022GlobalRates,Bippus2024GranularDynamics}, commercial real estate \citep{Ling2022GranularReturns}, and monetary policy \citep{Holm-Hadulla2024GranularPolicy}.

\textbf{Instruments built from within the system.} GIV belongs to a broader class of identification strategies that construct instruments from variation within the economic system rather than relying on a conventional external instrument \citep{Rigobon2003IdentificationHeteroskedasticity,Lewbel2012UsingModels,Lanne2017IdentificationAutoregressions,Arellano1991SomeEquations,Blundell1998InitialModels,Berry1995AutomobileEquilibrium,Bramoulle2009IdentificationNetworks}. The closest analogue is the shift-share design \citep{Bartik1991WhoPolicies,Goldsmith-Pinkham2020BartikHow,Borusyak2022Quasi-ExperimentalDesigns,Adao2019Shift-ShareInference,Borusyak2023NonrandomShocks}, where the instrument is an exposure-weighted sum of shocks, $\sum_i S_i g_{it}$, which is also the basic form of the granular instrument.

There are two important differences between GIV and shift-share designs. First, the shocks in a shift-share instrument are observed and come from outside the outcome equation. In GIV, both the weights and the shocks come from the same panel, and the shocks are residuals recovered from an estimated model. Second, share concentration plays a different role. \citet{Adao2019Shift-ShareInference} show that the distribution of shares can affect the variance of the estimator because common shares induce dependence across observations and can lead to incorrect standard errors. In GIV, by contrast, the concentration of unit sizes determines instrument strength and can therefore affect consistency, convergence rates, and inference. 

\textbf{Power laws.} Power-law size distributions provide a natural model of settings containing exceptionally large units. Power-law behavior is well documented in firm sizes \citep{Axtell2001ZipfSizes}, city sizes \citep{Gabaix1999ZipfsExplanation}, and financial returns \citep{Gabaix2009PowerFinance}. \citet{Gabaix2011TheFluctuations} shows that a sufficiently fat-tailed size distribution allows idiosyncratic shocks to the largest units to survive aggregation and generate aggregate fluctuations. Firm-level evidence supports this mechanism. \citet{diGiovanni2014FirmsFluctuations} attribute a substantial share of aggregate volatility in France to firm-specific shocks. \citet{Carvalho2019LargeCycle} show that large-firm dynamics alone can generate substantial aggregate fluctuations. \citet{Amiti2018HowData} trace idiosyncratic bank shocks into aggregate lending and investment in Japan.

Motivated by this evidence, I model individual unit sizes as draws from a common power-law distribution. The degree of fat-tailedness determines how granularity evolves as the cross section grows and, consequently, which asymptotic regime applies to GIV.

\textbf{Instrument strength.} \citet{Staiger1997InstrumentalInstruments} formalize weak instruments by allowing the first-stage coefficient to drift toward zero. \citet{Antoine2021GMMIdentification} study nearly-weak identification in GMM models, where the population moment condition approaches zero sufficiently slowly for consistent estimation to remain possible, albeit at a slower rate. In my setting, instrument weakness is not imposed directly. Instead, it arises endogenously from the evolution of the size distribution as $N$ grows. This generates strong, nearly-weak, and weak regimes within the same model. In this respect, my setting is also related to instrument weakness in large markets for differentiated products, as in \citet{Armstrong2016LargeSupply}, where instrument strength changes with market size.

\textbf{Factor models.} I assume that the unobserved component of the panel follows an interactive fixed-effects structure, a common modeling choice in the empirical settings of interest. \citet{Pesaran2006EstimationStructure}, \citet{Bai2009PanelEffects}, and \citet{Moon2015LinearEffects} estimate structural parameters directly in such models under exogeneity conditions on the regressor. In the supply-demand setting considered here, market clearing generates endogeneity, so the structural slope cannot be recovered from the panel alone.

I therefore construct an instrument from the recovered individual-level shocks. I use principal components methods as in \citet{Bai2003InferentialDimensions} to estimate the factor structure and individual shocks, and \citet{Bai2002DeterminingModels} to determine the number of factors.

A separate literature studies factor models with weak loadings \citep{Onatski2012AsymptoticsFactors,Bai2023ApproximateLoadings}. In that literature, weakness arises from the factor loadings. Here, the factor structure remains strong. Weakness instead arises because granularity, and hence the strength of the instrument, declines as the cross section grows.

\textbf{Commodity market elasticities.} Short-run demand and supply elasticities are usually identified from structural vector autoregressions, most extensively for crude oil \citep{Kilian2009NotMarket,Baumeister2019StructuralShocks,Caldara2019OilFluctuations}. Identification there comes from restrictions placed on the aggregate system, such as sign restrictions, exclusion restrictions, or priors on the elasticities themselves. GIV identifies the same elasticities from producer-level shocks. It requires a panel of individual supply data, and it leaves the aggregate dynamics unrestricted. Section \ref{sec:empirical} compares my estimates with the ranges this literature reports.

I organize the rest of the paper as follows. In Section~\ref{sec:three_regimes}, I set up the model and formalize granularity and its effect on the instrument strength. Then I split my main results across two sections. In Section~\ref{section:known factors}, I present the results by assuming that the idiosyncratic shocks are available and the GIV does not need a first-stage construction. This helps to isolate the effect of granularity and instrument strength on the estimator. Then in Section~\ref{section:estimated factors}, I add the first-stage estimation of the instrument and show how the results change. In Section~\ref{sec:empirical}, I apply the theory to estimate short-run demand and supply elasticities for three major commodities. I study the small-sample behavior of the estimator in Section~\ref{sec:simulation} through Monte Carlo simulations and conclude in Section~\ref{sec:conclusion}.

\section{Three Regimes of Instrument Strength} \label{sec:three_regimes}

This section introduces the model and studies when Granular Instrumental Variable (GIV) is strong
enough to identify the structural parameters. The instrument is a share-weighted average of individual shocks. It is valid by assumption
(see Assumption~\ref{ass:weak stationarity and idiosyncrasy} which delivers exogeneity),
but its \emph{relevance} is not guaranteed. Informativeness depends on the
presence of dominant units. For enough individuals, the share $S_i$ must sit far
from $1/N$.

Proposition~\ref{prop:weakness of known factor instrument} formalizes this. The
instrument stays fixed only when unit sizes are drawn from a fat-tailed
distribution. Otherwise it decays at a rate set by the tails. Weakness is
therefore not something I impose through a local-to-zero parameterization. It
arises structurally, from the heavy-tailed distribution of individual sizes, as
in \citet{Armstrong2016LargeSupply}. In subsequent sections, I explain what happens to the estimator as the instrument moves across these regimes. 

I begin by stating the model and the
assumptions behind Proposition~\ref{prop:weakness of known factor instrument}.

\subsection{Model} \label{subsec:model}

Specifically, we are interested in estimating the demand and supply elasticities in the demand-supply system,
\begin{align*}
    d_t    &= \phi_d\, p_t + \varepsilon_t, \\
    y_{it} &= \phi_s\, p_t + \lambda_i' F_t + u_{it},
\end{align*}
where $d_t$ is the change in aggregate demand for a commodity, $p_t$ the change in market-clearing price, and $y_{it}$ the change in supply of unit $i$. The unit-level supply follows a panel data model with interactive fixed effects, where unit-specific loadings $\lambda_i$ interact with common time factors $F_t$. Market clearing, $d_t = \sum_{i=1}^N S_i y_{it}$, ties the two equations together, with $S_i$ the long-run market share of unit $i$. The objects of interest are the demand and supply elasticities $\phi_d$ and $\phi_s$. The endogeneity problem is that $p_t$ is set in equilibrium and responds to both the demand shock $\varepsilon_t$ and the supply shocks $u_{it}$. Dominant units have a large value of $S_i$.

The key identifying assumption for GIV is that the idiosyncratic supply shocks, $(u_{it})_{1 \leq i \leq N}$, are uncorrelated with other common shocks and are therefore valid instruments. The GIV is the share-weighted average of these shocks, given by $z_t = \sum_{i=1}^{N} S_i u_{it}$.

We are interested in estimating the of demand and supply elasticities in the demand-supply system. At each date $t = 1, \ldots, T$, change in aggregate demand $d_t$ is governed by the structural equation
\begin{equation}
    d_t    = \phi_d\, p_t + X_t^d + \varepsilon_t \label{eq:demand}
\end{equation}
where $p_t$ is the change in market-clearing price, $\phi_d$ is the demand elasticity, and $X_t^d$ are observed controls uncorrelated with $\varepsilon_t$. The structural parameter of interest is the demand elasticity, $\phi_d$. 

Further, we observe changes in individual level supply/production of the commodity, $y_t = \left( y_{it} \right)_{1 \leq i \leq N}$. The individual supply is governed by the structural equation
\begin{align}
    y_{it} &= \phi_s\, p_t + X_{it}^y + \lambda_i' F_t + u_{it}, \label{eq:supply}
\end{align}
where $\phi_s$ is the supply elasticity, and $X_{it}^y$ are observed controls uncorrelated with both $\varepsilon_t$ and $u_t$. The structural parameter of interest in this equation is the supply elasticity, $\phi_s$.

The unit-level supply in \eqref{eq:supply} follows a panel data model with interactive fixed effects (common shocks $F_t$ that load heterogeneously across units through $\lambda_i$). The factors $F_t$, the loadings $\lambda_i$, and the idiosyncratic shocks $u_{it}$ are all unobserved. Stacking across units,
\begin{equation*}
    y_t = e_N \phi_s p_t + X_t^y + \Lambda F_t + u_t,
\end{equation*}
where $e_N$ is the $N$-vector of ones.

Market clearing, $d_t = \sum_{i=1}^N S_i y_{it}$, ties the two equations together, with $S_i$ the long-run market share of unit $i$. The objects of interest are the demand and supply elasticities $\phi_d$ and $\phi_s$. The endogeneity problem is that $p_t$ is set in equilibrium and responds to both the demand shock $\varepsilon_t$ and the supply shocks $u_{it}$. Dominant units have a large value of $S_i$.

For ease of exposition, I present the theory without accounting for the observed controls. When such controls are present, we can partial them out; mutatis mutandis, the theory applies equally to the resulting residualized variables by the Frisch-Waugh-Lovell theorem. See Appendix \ref{appendix_sec_controls} for details.

The model is characterized by the following assumption.
\begin{Ass} \label{ass:weak stationarity and idiosyncrasy}
The multivariate time series $\left(p_{t},y_{t}^{\prime },u_{t}^{\prime }, \varepsilon_t' \right) ^{\prime }$\ is a weakly stationary process with finite second moments. The vector $u_{t}$\ of error terms has a zero mean and is idiosyncratic in the sense that $\e[F_t u_t] = 0$, and
\begin{equation}
E[(d_t-\phi_d p_{t})u_{t}] =0  \label{eq:conditional_moment_restriction}
\end{equation}

     where $e_N$ is the N-dimensional vector of ones. Without loss of generality, we assume that the first common factor in the interactive fixed effects structure of the error is a time fixed effect:%
\[
\lambda _{i}^{\prime }F_{t}+u_{it}=F_{1t}+\sum_{k=2}^{r}\lambda
_{ik}F_{kt}+u_{it} 
\]
In matrix form:%
\[
\Lambda =\left[ 
\begin{array}{ccc}
\Lambda ^{1} & ... & \Lambda ^{r}%
\end{array}%
\right],\quad \Lambda ^{1}=e_{N}.
\]
Without loss of generality, we also assume that the columns $%
\Lambda ^{k},k=2,...,r,$ are orthogonal to the first column, that is that they have a zero mean. Call the $N \times r-1$ matrix formed by dropping the first column, $\tilde{\Lambda}$. Define $\tilde F_t = (F_{2t}, \ldots, F_{rt})'$
\end{Ass}

By Assumption \ref{ass:weak stationarity and idiosyncrasy}, $u_t$ is a valid instrument for the estimation of $\phi_d$. The moment condition delivers exclusion; whether $u_t$ is \emph{relevant} enough to identify $\phi_d$ as $N \to \infty$ is the subject of Proposition~\ref{prop:weakness of known factor instrument}.

We do not observe $u_t$. In small panels, we can consistently estimate the factor loadings, $\Lambda$ alone (see GK and GNR). As $T \to \infty$, we can extract the idiosyncratic shocks from data as, $M_{\Lambda}y_t = M_{\Lambda}u_t$, where
\begin{align*}
    P_{\Lambda } &= \Lambda \left( \Lambda ^{\prime }\Lambda \right) ^{-1}\Lambda' \enspace \text{and,} \\
    M_{\Lambda } &= I_N - P_{\Lambda } 
\end{align*}

For this paper, I consider the popular conception of GIV as the share-weighted sum of idiosyncratic shocks. That is, $z_t = S' M_{\Lambda} y_t = S' M_{\Lambda} u_t$. 

In large panels ($N$ and $T$ large), we can go further and estimate both the common factors and the factor loadings. We demean the data, and estimate consistent $\hat\Lambda, \hat F$ by PCA. We then subtract the estimated common component $\hat C_t = \hat{\Lambda} \hat{F}_t$ from the observed data to get $\hat{u}_t = D_N y_t - \hat{C}_t$. As $T \to \infty$, we have
\begin{equation*}
    D_N y_t - \tilde{C}_t = D_N u_t
\end{equation*}
where $\tilde{C}_t= \tilde{\Lambda} \tilde{F}_t$ and $D_N = I_N - \frac{e_N e_N'}{N}$ is the demeaning matrix with $e_N$ being the $N$-dimensional vector of ones. This leads to the moment restriction
\begin{equation}
E[(d_t-\phi_d p_{t}) (D_N y_t - \tilde{C}_t)] =0  \label{eq:conditional_moment_aggregate}
\end{equation}
The Granular Instrumental variable associated with the above moment condition is
\begin{equation}
    z_t = S' (D_N y_t - \tilde{C}_t)= S'D_N u_t
\end{equation}
I will be working with this form of the GIV in this paper.

For expositional clarity, I first develop the asymptotic theory assuming that we observe the factor structure. In the subsequent section, I relax this assumption and consider GIV when the factor structure is not known and need to be estimated. This helps us isolate two issues: the next section will deal with extension of the theory of GIV from fixed-$N$ settings to one in which $N \to \infty$. The subsequent section isolates the effect of estimating the instrument of the asymptotic theory. 

To formalize granularity, I assume that the individual sizes, $\s$ is random and follows the power law. The power law assumption is natural in settings with dominant units as well documented across very different empirical settings such as firm sizes, city sizes, and bank assets \citep{Gabaix2009PowerFinance,Axtell2001ZipfSizes}.

\begin{Ass} \label{ass:size_of_firm_power_law}
    The absolute sizes of individual units, $\s_i$ satisfy the following conditions:
    \begin{enumerate}
        \item The absolute sizes of individual units, $\s_i$ follow a power law.
    That is, the probability that it is above a fixed threshold, $s_i$ is given by
    \begin{equation*}
        \p(\s_i > s_i) = c s_i^{-\mu}
    \end{equation*}
    \item The absolute sizes are independent of all time series shocks, namely $u_t, F_t$, and $\varepsilon_t$. 
    \begin{equation*}
        \s_i \perp (u_t, F_t, \varepsilon_t) \quad \forall i,t
    \end{equation*}
    \item The absolute sizes are independent of the factor loadings in the shocks, i.e., $\s_i \perp \tilde{\Blambda}_i$ and the $r$-dimensional vector, $(\s_i, \tilde{\Blambda}_i')'$ is independent across $i$, and identically distributed.
    \item The factor loadings are such that they are independent of the time series shocks. That is
    \begin{equation*}
        \tilde{\Blambda}_i \perp (u_t, F_t, \varepsilon_t) \quad \forall i,t
    \end{equation*}
    \item The tails are also bounded. That is, $\e \| \tilde{\Blambda}_i \|^4 < \infty$.
    \end{enumerate}
\end{Ass}

Sizes are absolute, so I can measure them in units of the smallest unit. Then $\s_i \geq 1$, without any loss in generality. This pins down the tail constant. The support starts at 1, so $\p(\s_i > 1) = 1$. The tail formula evaluated at $s = 1$ returns $c = 1$, and thus, $\p(\s_i > s) = s^{-\mu}$ for $s \geq 1$. A different support yields a tail constant $c \neq 1$. This describes the same law with sizes rescaled by $c^{1/\mu}$. We construct the shares, $S_i$ needed for the instrument as $S_i = \s_i / \sum_j \s_j$, and they are invariant to rescaling.

Independence of the absolute sizes is not a restrictive assumption. The individual size is set in the long-term equilibrium. The shocks are all short-term in nature, and do not affect the long-term equilibrium. The same applies to the joint i.i.d.\ assumption of the vector, $(\s_i, \tilde{\Blambda}_i')'$. When loadings are treated as non-random (as is common), the i.i.d. clause reduces to i.i.d. sizes. The final assumption is standard in the factor literature \citep{Bai2003InferentialDimensions}.
As sizes are observed in equilibrium, we can construct the individual shares as
\begin{equation*}
    S_i = \frac{\s_i}{\sum_{j=1}^N \s_j}
\end{equation*}

\begin{Ass} \label{ass:bounded eigenvalues}
    The eigenvalues of the variance of the idiosyncratic errors are bounded above and bounded away from zero. That is, defining $\e[u_t u_t'] = \Omega$, there exist constants $0 < \underline\lambda \leq K < \infty$, independent of $N$, such that
    \begin{equation*}
        \underline\lambda \;\leq\; \gamma_{\text{min}}(\Omega) \;\leq\; \gamma_{\text{max}}(\Omega) \;\leq\; K,
    \end{equation*}
    where $\gamma_{\text{min}}(\cdot)$ and $\gamma_{\text{max}}(\cdot)$ denote the smallest and largest eigenvalue operators, respectively.
\end{Ass}

Under these assumptions, we will see how the tail index affects the granularity of the setting and thereby the instrument strength. 

\subsection{Structural Origin of Weakness}
In large panels, weakness arises due to the behavior of the tail index of the size variable. I formally state that idea in the first Proposition. \cite{Banafti2022InferentialDimensions} had formally stated the result for $\mu \in (0,1)$. I extend it to all values of $\mu$, following \citet{Gabaix2011TheFluctuations}.

\begin{Prop} \label{prop:weakness of known factor instrument}
    Suppose Assumptions~\ref{ass:weak stationarity and idiosyncrasy}, \ref{ass:size_of_firm_power_law}, and \ref{ass:bounded eigenvalues} hold. Then
    \begin{equation*}
        z_t =
        \begin{cases}
            O_\p(1) & \mu \in (0,1), \\
            O_\p\!\left(\frac{1}{N^{1 - \frac{1}{\mu}}}\right) & \mu \in (1,2), \\
            O_\p\!\left(\frac{1}{\sqrt{N}}\right) & \mu > 2.
        \end{cases}
    \end{equation*}
\end{Prop}
\begin{proof}
Proof in Appendix \ref{app:behavior of herfindahl}. 
\end{proof}

Two pieces of notation carry through the paper. Define $\delta = \min(1 - 1/\mu, 1/2)$, so that $\delta = 1 - 1/\mu \in (0, 1/2)$ when $\mu \in (1,2)$ and $\delta = 1/2$ when $\mu > 2$. Proposition \ref{prop:weakness of known factor instrument} then reads $z_t = O_{\p}(N^{-\delta})$ for $\mu > 1$. Define also the regime-specific rescaling
\begin{equation} \label{eq:a_N}
    a_N = \begin{cases} 1 & \mu \in (0,1), \\ N^{\delta},\ \delta = 1 - \tfrac{1}{\mu} & \mu \in (1,2), \\ \sqrt{N} & \mu > 2, \end{cases}
\end{equation}
which normalizes the instrument to $a_N z_t = O_{\p}(1)$ in all three regimes. The boundary indices $\mu = 1$ and $\mu = 2$ are excluded. The rescaling picks up a logarithm at both, and Appendix \ref{app:edge_cases} gives the two values.

Proposition \ref{prop:weakness of known factor instrument} gives the three regimes of instrument strength governed by the tail index $\mu$. Combined with the conditions on the relative growth of $N$ and $T$, we have the results on consistency and normality of the GIV estimator.

For $\mu \in (0,1)$, the instrument does not decay. Identification is strong, and this corresponds to the strong instrument dynamics of \citet{Gabaix2024GranularVariables}.

For $\mu > 1$, the instrument vanishes as $N \to \infty$. For $\mu \in (1,2)$, $z_t = O_{\p}(\frac{1}{N^{1 - 1/\mu}})$. For $\mu > 2$, $z_t = O_{\p}(\frac{1}{\sqrt{N}})$. Under $N/T \to 0$, instrument strength accumulates fast enough to identify the structural parameter and conduct inference. The estimator is consistent and asymptotically normal at the slower rate $\sqrt{T}/N^{\delta}$. Following \citet{Antoine2021GMMIdentification}, I call this nearly weak identification.

When $\mu > 2$ and $N/T \to c > 0$, time-series information no longer overtakes the cross-sectional dilution of the instrument. The estimator is inconsistent and identification is weak in the sense of \citet{Staiger1997InstrumentalInstruments}.

Proposition~\ref{prop:weakness of known factor instrument} describes the instrument, not the estimator. The next two sections translate these rates into consistency, convergence rates, and limiting distributions for $\phi_d$ and $\phi_s$, first with the factor structure known and then with it estimated.

\section{GIV with Known Factor Structure} \label{section:known factors}

In this section, I assume that we know the factor structure. This is primarily for exposition but includes some settings of practical interest. These cases include when we we have only time fixed effect, that is, $\Lambda = e_N$. Another is when we do not directly observe the loadings, but we have a parametric form, $\lambda_i = X_i \dot{\lambda}$.

In the general case with known factor loadings, demeaning the panel yields the true factor model.
\begin{align*}
    y_t &= e_N \phi_s p_t + \Lambda F_t + u_t \\
    D_N y_t &= \tilde{\Lambda} \tilde{F}_t + D_N u_t 
\end{align*}

When the factor structure, $\tilde{C}_t$ is known, the difference between demeaned data and the factor structure isolates idiosyncratic shocks. That is, $D_N y_t - \tilde{C}_t = D_N u_t$. The optimal instrument is $z_t = S'(D_N y_t - \tilde{C}_t)$.

Proposition \ref{prop:weakness of known factor instrument} gives the behavior of the known-factor instrument for different values of the tail index of the size variable. Those three regimes translate to three regimes for the GIV estimator as well. When $\mu \in (0,1)$, identification is strong and the estimator is $\sqrt{T}$-consistent and asymptotically normal. When $\mu > 1$ and $N/T \to 0$, identification is nearly weak in the sense of \citet{Antoine2021GMMIdentification}. The estimator is consistent and asymptotically normal at the slower rate $\sqrt{T}/N^{\delta}$, where $\delta = \min(1 - 1/\mu, 1/2)$. When $\mu > 2$ and $N/T \to c$, identification is weak in the sense of \citet{Staiger1997InstrumentalInstruments} and the estimator is inconsistent. I take up inference for all three regimes in Section~\ref{subsec:inference}.

\subsection{Assumptions}

To establish these results, we require assumptions on the idiosyncratic shocks. The large panel setting allows us to accommodate a richer structure for the time series and cross-sectional dependence than the fixed-$N$ (small-panel) case, which requires i.i.d.\ samples with no cross correlation.

\begin{Ass}(Time and Cross-Sectional Dependence and Heteroskedasticity) \label{ass:LN_time and cross sectional dependence}
There exists a positive constant, $M < \infty$ such that for all $N$ and $T$,
\begin{enumerate}
    \item $\e[u_{it}] = 0$, and $\e |u_{it}|^4 \leq M$
    \item For every $i,j$ and $t$, $\e[u_{it} u_{jt}]$ is bounded.
    Define $\gamma(i,j) = \e \left[ u_i' u_j/T \right] = \e \left[ \frac{1}{T} \sum_{t=1}^T u_{it} u_{jt} \right]$, for all $1\leq i,j \leq N$, and
    \begin{align*}
       &a. \quad  \sum_{i=1}^N | \gamma(i,j) |  \leq M \\
       &b. \quad N^{-1}\sum_{i=1}^N \sum_{j=1}^N | \gamma(i,j) | \leq M
    \end{align*}
    \item Let $\e[u_{it}^2 u_{is}^2] = \tau_{st}$. For all $i$,
    \begin{align*}
        T^{-1} \sum_{t=1}^T \sum_{s=1}^T \tau_{st} \leq M
    \end{align*}
    \item For every $i,j$,
    \begin{equation*}
        \e \left| T^{-\frac{1}{2}} \sum_{t=1}^T \bigg( u_{it} u_{jt} - \e[u_{it} u_{jt}]  \bigg) \right|^4 \leq M
    \end{equation*}
    \item (Uniform Rosenthal-type moment bound.) There exists a constant $C_q < \infty$, independent of $N$, such that for any deterministic $\{a_j\}_{j=1}^N \subset \mathbb{R}$ and $q = 8 + 2\pi$,
    \begin{equation*}
        \e\!\left[ \Big| \sum_{j=1}^N a_j u_{jt} \Big|^{q} \right] \;\leq\; C_q \left[ \Big( \sum_{j=1}^N a_j^2 \, \e u_{jt}^2 \Big)^{q/2} + \sum_{j=1}^N |a_j|^q \, \e |u_{jt}|^q \right].
    \end{equation*}
    This holds, in particular, when $\{u_{jt}\}_{j=1}^N$ are cross-sectionally independent \citep{Rosenthal1970OnVariables}.
    \item (Cross-sectional weak dependence of the products $u_{jt} \varepsilon_t$ and $u_{jt} F_t$.) For the structural shock $\varepsilon_t$ and the common factors $F_t$, let $v_{jt} = u_{jt} \varepsilon_t$ and $w_{jt} = u_{jt} F_t$. Then
    \begin{equation*}
        \frac{1}{NT} \sum_{j=1}^N \sum_{k=1}^N \sum_{t=1}^T \sum_{s=1}^T \big| \cov(v_{jt}, v_{ks}) \big| \leq M, \qquad \frac{1}{NT} \sum_{j=1}^N \sum_{k=1}^N \sum_{t=1}^T \sum_{s=1}^T \big\| \cov(w_{jt}, w_{ks}) \big\| \leq M.
    \end{equation*}
    These are the joint analogue of part 2(b), imposing absolute summability of the covariance arrays of $u_{jt} \varepsilon_t$ and $u_{jt} F_t$ across both time and the cross section. They are conditions on the covariances themselves and do not require $\varepsilon_t$ or $F_t$ to be independent of $\{u_{jt}\}$. Because $\varepsilon_t$ and $F_t$ are common across units, they add no cross-unit linkage of their own, so the conditions only rule out cross-sectional comovement in the supply shocks driven by $\varepsilon_t$ or $F_t$ strong enough to break the summability.
\end{enumerate}
\end{Ass}

These are standard in the factor literature and relax the independence restriction in the shorter panel GIV literature.

\begin{Ass}(Strong Mixing and Higher Moments)\label{ass:mixing for stationary time series}
The multi-dimensional time series, $\{(F_t', u_t', \varepsilon_t) \}$ is weakly stationary and is a strong mixing sequence of size $-(\frac{2 + \pi}{\pi})$, where $\pi > 0$. In other words, define the sigma algebra,
\begin{equation*}
    \mathcal{F}_{a,b} = \sigma \{ (F_t', u_t', \varepsilon_t); a \leq t \leq b \}
\end{equation*}
and
\begin{equation*}
    \alpha(h) \defeq \sup_{A \in \mathcal{F}_{-\infty, 0}, B \in \mathcal{F}_{h, \infty}} |\p(A \cap B) - \p(A) \p(B) |
\end{equation*}
We have, for some $\pi > 0$,
\begin{equation*}
    \sum_{h=1}^{\infty} \alpha(h)^{\frac{\pi}{2 + \pi}} < \infty
\end{equation*}

Further we assume the following higher moments. $\e \| F_t \|^{8 + 2\pi} < \infty$, $\e | \varepsilon_t|^{8 + 2\pi} < \infty$, and $\sup_j \e | u_{jt}|^{8 + 2\pi} \leq M$ for a constant $M$ independent of $N$.

Finally, the conditional long-run variances of the demand and supply scores,
\begin{equation*}
    V_{z \varepsilon}(S) = \lim_{T \to \infty} \frac{1}{T} \sum_{s=1}^T \sum_{t=1}^T \e[z_t z_s \varepsilon_t \varepsilon_s \mid S], \qquad
    V_{z v_H}(S, \Lambda) = \lim_{T \to \infty} \frac{1}{T} \sum_{s=1}^T \sum_{t=1}^T \e[z_t z_s v_{Ht} v_{Hs} \mid S, \Lambda],
\end{equation*}
where $v_{Ht}$ is the supply multiplier defined below \eqref{eq:expression for H}, are non-degenerate at the scale of the instrument. With $a_N$ as in \eqref{eq:a_N},
\begin{equation*}
    \big[ a_N^2 \, V_{z \varepsilon}(S) \big]^{-1} = O_{\p}(1), \qquad \big[ a_N^2 \, V_{z v_H}(S, \Lambda) \big]^{-1} = O_{\p}(1).
\end{equation*}
The same bounds hold for the long-run variances $V_{z \bar{\varepsilon}}(S)$ and $V_{z \bar{v}_H}(S, \Lambda)$ of the recentered scores that arise with estimated factors.

\end{Ass}

The last requirement bounds the spectral density of each score away from zero at frequency zero. It rules out cancellation of the score autocovariances in the long run. When the score is serially uncorrelated, the long-run variance equals the contemporaneous variance, and the bound is automatic provided the multiplier has positive variance.

With the dependence structure in place, I now turn to the behavior of the GIV estimator. I analyze the three regimes set by Proposition~\ref{prop:weakness of known factor instrument} in turn, beginning with the strong regime ($\mu \in (0,1)$). In each regime, I will consider the demand equation first and the supply equation later.

\subsection{Strong Regime: $\mu \in (0,1)$}

Recall that the equations of interest are
\begin{align*}
    d_t    &= \phi_d\, p_t + \varepsilon_t, \\
    y_{it} &= \phi_s\, p_t + \lambda_i' F_t + u_{it},
\end{align*}
and the parameters of interest are the demand elasticity, $\phi_d$ and the supply elasticity, $\phi_s$. We will consider the demand equation first. 

The moment condition for the estimation of the demand elasticity arises from Assumption \ref{ass:weak stationarity and idiosyncrasy}. As the idiosyncratic supply shocks are uncorrelated with the demand shocks, we have $\e[(d_t - \phi_d p_t )z_t] = 0$. Thus the GIV estimator of the demand (aggregate) structural parameter is
\begin{align*}
    \hat{\phi_d} &= \frac{z' y_S}{z' p} = \phi_d + \frac{z' \varepsilon}{z' p} \\
    \hat{\phi}_d - \phi_d &= \frac{z' \varepsilon}{z' p} = \frac{\displaystyle \frac{1}{T} \sum_t z_t \varepsilon_t}{\displaystyle \frac{1}{T} \sum_t z_t p_t}
\end{align*}

In the small-panel GIV framework of \citet{Gabaix2024GranularVariables}, $N$ is fixed, the shares and factor loadings are treated as constants, and consistency and asymptotic normality follow from standard IV arguments. Moving to the large-panel setting introduces four complications. First, $N \to \infty$, so the instrument $z_t = S'\bar{u}_t$ is a growing weighted sum whose behavior depends on the concentration of the shares. Second, the shares $S$ are now random, drawn from a power-law distribution. Third, the factor loadings $\tilde{\Blambda}_i$ are random, and the share-weighted loading $S'\tilde{\Lambda}$ has a growing dimension and must be shown not to contaminate the instrument. Fourth, the central limit theorem for the numerator $T^{-1/2}\sum_t z_t \varepsilon_t$ is itself non-standard. The summand combines a share-weighted cross-sectional sum with time-series dependence, and its moments must be bounded uniformly in $N$.

When $\mu \in (0,1)$, all four complications are resolved by the heavy tails. Proposition~\ref{prop:weakness of known factor instrument} shows that $S'S = O_{\p}(1)$, so the instrument does not degenerate as $N$ grows. Proposition~\ref{prop:behavior of S times lambda} shows that $S'\tilde{\Lambda} = O_{\p}(1)$ as well, so the common-factor contamination remains bounded. Corollary~\ref{corollary:central limit theorem} delivers the CLT for the numerator $T^{-1/2}\sum_t z_t \varepsilon_t$. The heavy-tail concentration of $S$ keeps its moments bounded uniformly in $N$. Under these conditions, the GIV estimator retains $\sqrt{T}$-consistency and asymptotic normality.

\begin{Theorem} \label{theorem:strong_consistency_aggregate}
    Suppose Assumptions \ref{ass:weak stationarity and idiosyncrasy} to \ref{ass:mixing for stationary time series} hold with $\mu \in (0,1)$. Then, conditional on $S$, for almost every realization of the shares, the GIV estimator for the aggregate structural parameter is consistent and asymptotically normal,
    \begin{equation*}
       \frac{\Gamma_{zp}}{\sqrt{V_{z \varepsilon}(S)}} \cdot \sqrt{T} [\hat{\phi}_d - \phi_d] \xrightarrow{d} \n(0,  1)
    \end{equation*}
     where $\Gamma_{zp}$ is the finite-$N$ conditional mean $\e[z_t p_t \mid S, \Lambda] = \frac{S' D_N \Omega S}{\phi_d - \phi_s}$ of Proposition \ref{prop:denominator}, and $V_{z\varepsilon}(S) = \lim_{T \to \infty} \frac{1}{T} \sum_{s=1}^T \sum_{t=1}^T \e[z_t z_s \varepsilon_t \varepsilon_s \mid S]$. The limit in $T$ exists, and Proposition \ref{prop:denominator} gives $a_N^2 \Gamma_{zp} = O_{\p}(1)$ together with $[a_N^2 \Gamma_{zp}]^{-1} = O_{\p}(1)$ whenever $\phi_d \neq \phi_s$, so the denominator is bounded above and away from zero at the rate $a_N^{-2}$.
\end{Theorem}
\begin{proof}
    Proof in Appendix \ref{proof:strong_Consistency_aggregate}. The proof proceeds by showing that the denominator $T^{-1}\sum_t z_t p_t$ converges to a nonzero limit, while the numerator $T^{-1/2}\sum_t z_t \varepsilon_t$ satisfies a CLT under the mixing conditions of Assumption \ref{ass:mixing for stationary time series}.
\end{proof}

For the disaggregated supply side, we do not have a ready to use moment condition. But in line with the standard GIV theory, a specific aggregation yields the valid moment condition. For an aggregating vector $H$ of dimension $N$, $y_{Ht} = \sum_{i=1}^N H_i y_{it}$ with $\sum_{i=1}^N H_i = 1$. This leads to a valid moment condition using GIV if and only if
\begin{eqnarray*}
    Cov\left[ y_{Ht}-\phi p_{t},D_N u_t\right] =0%
    &\Longleftrightarrow& Cov\left[ H^{\prime }u_{t},D_N u_{t}\right] =0 \\
    &\Longleftrightarrow &H^{\prime }\Omega D_N=0\\
&\Longleftrightarrow &
\Omega H\in \text{Im}\left( e_N \right)
\end{eqnarray*}
That is, $H = \Omega^{-1}e_N k$ for some scalar constant, $k$. As $H' e_N = 1$, we have
\begin{equation} \label{eq:expression for H}
    H = \frac{\Omega^{-1} e_N}{e_N' \Omega^{-1} e_N}
\end{equation}

For the choice of $H$ in \eqref{eq:expression for H}, we have the moment restriction, $\e[(y_{Ht}-\phi p_{t}) D_N u_t ] =0 $. Define $v_{Ht} = H' v_t = \lambda_H' F_t + u_{Ht}$
Thus, the GIV estimator of the supply (disaggregated) parameter is
\begin{align*}
    \hat{\phi}_s &= \frac{z' y_H}{z' p} = \phi_s + \frac{z' v_H}{z' p} \\
    \hat{\phi}_s - \phi_s &= \frac{z' v_H}{z' p}
\end{align*}

Similar to the case with the demand elasticity, the following theorem shows that the GIV estimator is $\sqrt{T}$ consistent and asymptotically normal.
\begin{Theorem} \label{theorem:strong_consistency_disaggregate}
    Suppose Assumptions \ref{ass:weak stationarity and idiosyncrasy} to \ref{ass:mixing for stationary time series} hold with $\mu \in (0,1)$. Then, conditional on $(S, \Lambda)$, for almost every realization of the shares and loadings, the GIV estimator for the disaggregated structural parameter is consistent and asymptotically normal,
    \begin{equation*}
       \frac{\Gamma_{zp}}{\sqrt{V_{z v_H}(S, \Lambda)}} \cdot \sqrt{T} [\hat{\phi}_s - \phi_s] \xrightarrow{d} \n(0,  1)
    \end{equation*}
     where $\Gamma_{zp}$ is the finite-$N$ conditional mean $\e[z_t p_t \mid S, \Lambda] = \frac{S' D_N \Omega S}{\phi_d - \phi_s}$ of Proposition \ref{prop:denominator}, and $V_{z v_H}(S, \Lambda) = \lim_{T \to \infty} \frac{1}{T} \sum_{s=1}^T \sum_{t=1}^T \e[z_t z_s v_{Ht} v_{Hs} \mid S, \Lambda]$. The limit in $T$ exists, and Proposition \ref{prop:denominator} gives $a_N^2 \Gamma_{zp} = O_{\p}(1)$ together with $[a_N^2 \Gamma_{zp}]^{-1} = O_{\p}(1)$ whenever $\phi_d \neq \phi_s$, so the denominator is bounded above and away from zero at the rate $a_N^{-2}$.
\end{Theorem}
\begin{proof}
    Proof in Appendix \ref{app:proof of disaggregate}. The denominator is unchanged from Theorem \ref{theorem:strong_consistency_aggregate}, as neither $z_t$ nor $p_t$ depends on $H$. Only the numerator differs which I account for in the proof.
\end{proof}

\subsection{Nearly Weak Regime: $\mu > 1$}

I now turn to the case $\mu > 1$. Define $\delta = \min(1 - 1/\mu, 1/2)$. The heavy-tail concentration of the shares is weaker than in the strong regime, and Proposition~\ref{prop:weakness of known factor instrument} pins down by how much. The Herfindahl satisfies $S'S = O_{\p}(N^{-2\delta})$, so the instrument $z_t = S'\bar{u}_t$ no longer has order one. It dilutes as $N$ grows. Only the rescaled instrument $N^{\delta} z_t$ has a non-degenerate limit. Under $N/T \to 0$, time-series information accumulates fast enough that the GIV estimator remains consistent and asymptotically normal at the slower rate $\sqrt{T}/N^{\delta}$. Following \citet{Antoine2021GMMIdentification}, I call this nearly weak identification.

Two of the four complications from the strong regime resolve exactly as before. The share-weighted loading $S'\tilde{\Lambda}$ remains controlled by Proposition~\ref{prop:behavior of S times lambda}, and the CLT for the numerator goes through under the same mixing conditions. The other two now carry an explicit $N$-dependence inherited from the dilution of the instrument. This is what slows the rate of convergence and forces the requirement $N/T \to 0$. Time-series information must accumulate fast enough to overcome the cross-sectional dilution.

The boundary $\mu = 2$ separates two dilution patterns. For $\mu \in (1,2)$, the dilution exponent $\delta = 1 - 1/\mu$ lies in $(0, 1/2)$ and the Herfindahl is governed by a stable law. For $\mu > 2$, the size variable has a finite second moment. The Herfindahl is then governed by the law of large numbers and the dilution exponent caps at $\delta = 1/2$. The theorem below covers both cases at once. Like the previous section, I state the results for the aggregate demand equation first before proceeding to the disaggregated supply equation.

\begin{Theorem} \label{theorem:weak_consistency_aggregate} \label{theorem:semi_strong_consistency_aggregate}
    Suppose Assumptions \ref{ass:weak stationarity and idiosyncrasy} to \ref{ass:mixing for stationary time series} hold with $\mu > 1$ and $N/T \to 0$. Let $\delta = \min(1 - 1/\mu, 1/2)$, so that $\delta = 1 - 1/\mu \in (0, 1/2)$ when $\mu \in (1,2)$ and $\delta = 1/2$ when $\mu > 2$. Then, conditional on $S$, for almost every realization of the shares, the GIV estimator for the aggregate structural parameter is consistent and asymptotically normal,
    \begin{equation*}
       \frac{\Gamma_{zp}}{\sqrt{V_{z \varepsilon}(S)}} \cdot \sqrt{T}\, [\hat{\phi}_d - \phi_d] \xrightarrow{d} \n(0,  1),
    \end{equation*}
    where $\Gamma_{zp}$ is the finite-$N$ conditional mean $\e[z_t p_t \mid S, \Lambda] = \frac{S' D_N \Omega S}{\phi_d - \phi_s}$ of Proposition \ref{prop:denominator}, and $V_{z\varepsilon}(S) = \lim_{T \to \infty} \frac{1}{T}\sum_{s=1}^T\sum_{t=1}^T \e[z_t z_s \varepsilon_t \varepsilon_s \mid S]$. The limit in $T$ exists, and Proposition \ref{prop:denominator} gives $a_N^2 \Gamma_{zp} = O_{\p}(1)$ together with $[a_N^2 \Gamma_{zp}]^{-1} = O_{\p}(1)$ whenever $\phi_d \neq \phi_s$, so the denominator is bounded above and away from zero at the rate $a_N^{-2}$. The studentization satisfies $\Gamma_{zp}/\sqrt{V_{z\varepsilon}(S)} = O_{\p}(N^{-\delta})$, so the standardized statistic vanishes at the rate $\sqrt{T}/N^{\delta}$.

    If $\mu > 2$ and $N/T \to c > 0$, then $\hat{\phi}_d$ is inconsistent and identification is weak in the sense of \citet{Staiger1997InstrumentalInstruments}.
\end{Theorem}
\begin{proof}
    Proof in Appendix~\ref{app:proof of weak theorems}. The argument parallels the strong regime, with both $\Gamma_{zp}$ and $V_{z\varepsilon}(S)$ now $O_{\p}(N^{-2\delta})$ by Proposition~\ref{prop:weakness of known factor instrument}, which slows the rate to $\sqrt{T}/N^{\delta}$.
\end{proof}

Now I state the result for the disaggregated supply equation under suspected weakness.

\begin{Theorem} \label{theorem:weak_consistency_disaggregate}
    Suppose Assumptions \ref{ass:weak stationarity and idiosyncrasy} to \ref{ass:mixing for stationary time series} hold with $\mu > 1$ and $N/T \to 0$. Let $\delta = \min(1 - 1/\mu, 1/2)$, so that $\delta = 1 - 1/\mu \in (0, 1/2)$ when $\mu \in (1,2)$ and $\delta = 1/2$ when $\mu > 2$. Then, conditional on $(S, \Lambda)$, for almost every realization of the shares and loadings, the GIV estimator for the disaggregated structural parameter is consistent and asymptotically normal,
    \begin{equation*}
       \frac{\Gamma_{zp}}{\sqrt{V_{z v_H}(S, \Lambda)}} \cdot \sqrt{T}\, [\hat{\phi}_s - \phi_s] \xrightarrow{d} \n(0,  1),
    \end{equation*}
    where $\Gamma_{zp}$ is the finite-$N$ conditional mean $\e[z_t p_t \mid S, \Lambda] = \frac{S' D_N \Omega S}{\phi_d - \phi_s}$ of Proposition \ref{prop:denominator}, and $V_{z v_H}(S, \Lambda) = \lim_{T \to \infty} \frac{1}{T}\sum_{s=1}^T\sum_{t=1}^T \e[z_t z_s v_{Ht} v_{Hs} \mid S, \Lambda]$. The limit in $T$ exists, and Proposition \ref{prop:denominator} gives $a_N^2 \Gamma_{zp} = O_{\p}(1)$ together with $[a_N^2 \Gamma_{zp}]^{-1} = O_{\p}(1)$ whenever $\phi_d \neq \phi_s$, so the denominator is bounded above and away from zero at the rate $a_N^{-2}$. The studentization satisfies $\Gamma_{zp}/\sqrt{V_{z v_H}(S, \Lambda)} = O_{\p}(N^{-\delta})$, so the standardized statistic vanishes at the rate $\sqrt{T}/N^{\delta}$.

    If $\mu > 2$ and $N/T \to c > 0$, then $\hat{\phi}_s$ is inconsistent and identification is weak in the sense of \citet{Staiger1997InstrumentalInstruments}.
\end{Theorem}
\begin{proof}
    Proof in Appendix~\ref{app:proof of weak disaggregate}. The denominator is unchanged from Theorem~\ref{theorem:weak_consistency_aggregate}, as neither $z_t$ nor $p_t$ depends on $H$. Only the numerator differs which I account for in the proof.
\end{proof}

\begin{remark}[Three regimes] \label{rem:three_regimes}
Theorems~\ref{theorem:strong_consistency_aggregate} to~\ref{theorem:weak_consistency_disaggregate} together encompass three regimes of instrument strength, governed by the tail index $\mu$ and the $N/T$ trajectory. The aggregate and disaggregated parameters fall in the same regime at the same rate, since the instrument $z_t$ is common to both.
\begin{enumerate}
    \item \textbf{Strong Instrument} ($\mu \in (0,1)$). The instrument does not dilute, and Theorems~\ref{theorem:strong_consistency_aggregate} and~\ref{theorem:strong_consistency_disaggregate} give $\sqrt{T}$-consistency and asymptotic normality.
    \item \textbf{Nearly weak Instrument} ($\mu > 1$ with $N/T \to 0$). The instrument dilutes at rate $N^{-\delta}$ with $\delta = \min(1 - 1/\mu, 1/2)$, but $T$ outgrows $N$ fast enough to preserve identification. Theorems~\ref{theorem:weak_consistency_aggregate} and~\ref{theorem:weak_consistency_disaggregate} give consistency and asymptotic normality at the slower rate $\sqrt{T}/N^{\delta}$.
    \item \textbf{Weak Instrument} ($\mu > 2$ with $N/T \to c > 0$). The instrument dilutes at rate $N^{-1/2}$ and $T$ no longer outgrows $N$, so the time-series information does not overcome the dilution. Both estimators are inconsistent, in the sense of \citet{Staiger1997InstrumentalInstruments}. 
\end{enumerate}
\end{remark}

\section{GIV with Estimated Factor Structure} \label{section:estimated factors}

Section \ref{section:known factors} took the factor structure as known. In practice it is not, and we replace the known common component $\tilde C_t$ with the principal-component estimate $\hat C_t$ of \citet{Bai2003InferentialDimensions}. This places GIV in the constructed-regressor setting. The first-stage estimation error propagates to the GIV moment. Additionally, consistent inference requires $\sqrt{T}/N \to 0$. This section formalizes the effects of the first stage estimation and re-establishes the convergence results of Section \ref{section:known factors} under the estimated-factor instrument.

The parameters of interest are the same two as in Section \ref{section:known factors}, the demand elasticity $\phi_d$ and the supply elasticity $\phi_s$,
\begin{align*}
    d_t    &= \phi_d\, p_t + \varepsilon_t, \\
    y_{it} &= \phi_s\, p_t + \lambda_i' F_t + u_{it}.
\end{align*}
The demand elasticity is estimated off the aggregate quantity $d_t$ and the supply elasticity off the aggregated quantity $y_{Ht}$ with the aggregator $H$ of \eqref{eq:expression for H}. Both use the same instrument, so the first stage enters the two estimators through a single channel.

The instrument is constructed from the supply side in vector form,
\begin{equation*}
    y_t = e_N \phi_s p_t + \Lambda F_t + u_t
\end{equation*}

From the previous definitions, $\Lambda = [\boldsymbol{1}_{N} \enspace \tilde{\Lambda} ]$ and $F_t = [F_{1t},\, \tilde F_t']'$,
\begin{equation*}
    D_N y_t = \tilde{\Lambda} \tilde{F}_t + D_N u_t
\end{equation*}

Thus $D_Ny_t$ has a factor structure. In large panels, we can consistently estimate both $\tilde{\Lambda}$ and $\tilde{F}$ \citep{Bai2003InferentialDimensions}. As $T > N$, we estimate $\hat{\Lambda}$ using principal components. The first order condition of the PCA objective concentrates out $\hat{F} = \Bar{Y} \hat{\Lambda}/N$.

From these two estimators, we have the estimator for the common component $\tilde{C}_{it} = \tilde{\Blambda}_i' \tilde{F}_t$ as $\hat{C}_{it} = \hat{\Blambda}_i' \hat{F}_t$. Call the corresponding vector, $\hat{C}_t$. From this estimator, we can construct the estimated-factor instrument:
\begin{equation*}
    \hat{z}_t' = S'[D_N y_t - \hat{C}_t] \defeq S' \hat{u}_t
\end{equation*}
where $\hat{u}_t = D_N y_t - \hat{C}_t$ is the estimated residual. \citet{Gabaix2024GranularVariables} proposed a different form of the instrument.
\begin{equation*}
    \hat{z}^{\text{GK}}_t = S' M_{\hat{\Lambda}} D_N y_t
\end{equation*}

The two formulations are equivalent. The PCA first-order condition $\hat F_t = \hat\Lambda' D_N y_t / N$ implies $\hat C_t = (I - M_{\hat\Lambda}) D_N y_t$, so $\hat z_t = \hat z_t^{\text{GK}}$ \footnote{Explicitly, $\hat u_t = D_N y_t - \hat\Lambda \hat F_t = [I - \hat\Lambda \hat\Lambda'/N] D_N y_t = [I - \hat\Lambda(\hat\Lambda'\hat\Lambda)^{-1}\hat\Lambda'] D_N y_t = M_{\hat\Lambda} D_N y_t$, where the second equality uses the PCA normalization $\hat\Lambda'\hat\Lambda/N = I_{r-1}$.}. I work with $\hat C_t$ rather than $M_{\hat\Lambda} D_N y_t$ in the asymptotic analysis. $M_{\hat\Lambda} = I - \hat\Lambda(\hat\Lambda'\hat\Lambda)^{-1}\hat\Lambda'$ is a non-linear product of the estimator $\hat\Lambda$. Hence its asymptotic linear form involves derivatives of non-linear transformations of the estimator. $\hat C_t$ admits a much cleaner asymptotic linear expansion derived in Appendix~\ref{app:estimation of common components}. The instrument is therefore
\begin{equation*}
    \hat{z}_t = S'[D_N y_t - \hat{C}_t].
\end{equation*}

Similar to the previous section, I state the results separately for the strong and nearly weak regimes, and within each regime for the demand elasticity first and the supply elasticity after. The estimation of the factor structure requires a number of additional assumptions on the factor structure which I state in the next sub-section.

\subsection{Assumptions} \label{subsec:assumptions}
\begin{Ass}[Strong Factor Structure and Distinct Eigenvalues] \label{ass:LN_strong factor structure}
The factor structure is strong. That is,
\begin{enumerate}
    \item The factor structure is strong. $\e \| \tilde{F}_t \|^4 \leq M < \infty $ and $T^{-1} \sum_{t=1}^T \Tilde{F}_t \Tilde{F}_t' \xrightarrow{p} \Sigma_{\Tilde{F}}$, for some $r-1 \times r-1$ positive definite matrix, $\Sigma_{\Tilde{F}}$.
    \item $ \| \tilde{\Blambda}_j \| \leq \Bar{\Blambda} < \infty$ and $ \e \| \tilde{\Blambda}_j \|^4 \leq M < \infty$ for every $j$. Each factor has a non-trivial contribution on the variance of $Y_t$. That is, $N^{-1} \sum_{i=1}^N \tilde{\Blambda}_i \tilde{\Blambda}_i' \xrightarrow{p} \Sigma_{\Tilde{\Lambda}}$,, for some positive definite $r-1 \times r-1$ matrix $\Sigma_{{\tilde{\Lambda}}}$. $\tilde{\Blambda}_i$ is independent of $u_t$ and $\tilde{F}_t$ for all $i$ and $t$. 
    \item The eigenvalues of the $r-1 \times r-1$ matrix $\Sigma_{\tilde{\Lambda}} \cdot \Sigma_{\Tilde{F}}$ are distinct
\end{enumerate}
    
\end{Ass}

\begin{Ass}(Weak conditional dependence between factors and idiosyncratic errors) \label{ass:weak dep between agg shocks and idio errors} 
There exists an $M < \infty$, such that
    \begin{equation*}
      \e \left(   \frac{1}{N} \sum_{i=1}^N \left \| \frac{1}{\sqrt{T}} \sum_{t=1}^T \tilde{F}_t u_{it} \right \|^2 \right) \leq M
    \end{equation*}
\end{Ass}

\begin{Ass}(Moments) \label{ass:LN_moments and CLT}
There exists an $M < \infty$, such that for all $N$ and $T$,
    \begin{enumerate}
        \item For each $i$,
        \begin{equation*}
            \e \left \| \frac{1}{\sqrt{NT}} \sum_{j=1}^N \sum_{t=1}^T \tilde{\Blambda}_j \left( u_{it} u_{jt} - \e[u_{it} u_{jt}]  \right)   \right \|^2 \leq M
        \end{equation*}
        \item The $r \times r$ matrix satisfies
        \begin{equation*}
            \e \left \| \frac{1}{\sqrt{NT}} \sum_{j=1}^N \sum_{t=1}^T \tilde{\Blambda}_j \tilde{F}_t' u_{jt}  \right \|^2 \leq M
        \end{equation*}        
    \end{enumerate}
    
\end{Ass}

\subsection{Effect of Estimating the Factor Structure}
Estimating the factor structure changes the asymptotic variance of both GIV
estimators, and it does so in one specific way. Wherever the variance of
Section~\ref{section:known factors} carries a structural error, the variance
here carries its factor-residualised version. For the demand elasticity the
error is $\varepsilon_t$ and it is replaced by
\begin{equation*}
    \bar{\varepsilon}_t = \varepsilon_t - \e[\tilde{F}_t' \varepsilon_t] \Sigma_{\tilde{F}}^{-1} \tilde{F}_t .
\end{equation*}
For the supply elasticity the error is $v_{Ht}$ and it is replaced by
\begin{equation*}
    \bar{v}_{Ht} = v_{Ht} - \e[\tilde{F}_t' v_{Ht} \mid \Lambda] \Sigma_{\tilde{F}}^{-1} \tilde{F}_t .
\end{equation*}
The departure is zero when the structural error is uncorrelated with
$\tilde{F}_t$ and non-zero otherwise. The two parameters differ here. For the
demand elasticity the condition is a restriction on the aggregate shock, and it
may or may not hold. For the supply elasticity it holds trivially, since
$v_{Ht} = \lambda_H' F_t + u_{Ht}$ loads on the factors by construction and
nothing in the choice of $H$ makes the aggregated loading $\lambda_H = \Lambda' H$
orthogonal to $\tilde{F}_t$. The rest of this subsection derives corrections for both from the influence function of $\hat{C}_t - \tilde{C}_t$.

\begin{remark}[Scope] \label{rem:scope of the correction}
    In this paper, I consider the case where the structural error $\varepsilon_t$ is correlated with the common factors $\tilde{F}_t$, but is not collinear with them. Both endpoints of this correlation fall outside the problem this section addresses.

    When $\e[\tilde{F}_t' \varepsilon_t] = 0$, the correction is zero. Then $\bar{\varepsilon}_t = \varepsilon_t$, the asymptotic variance is the one carried by the known factor structure in Section~\ref{section:known factors}, and estimating the instrument costs nothing.

    When $\varepsilon_t$ is collinear with $\tilde{F}_t$, the correction removes the structural error entirely and $\bar{\varepsilon}_t \equiv 0$. But in that case no instrument is needed. Writing $\varepsilon_t = c' \tilde{F}_t$, the aggregate equation becomes $y_{St} = \phi_d p_t + c' \tilde{F}_t$, which carries no unobserved error once the factor space is conditioned on. The factors are recovered from the supply panel, so regressing $y_{St}$ on $p_t$ and $\hat{F}_t$ identifies $\phi_d$ directly. The regression is well posed because
    \begin{equation*}
        p_t = \frac{1}{\phi_d - \phi_s} \left( u_{St} + \lambda_S' F_t - \varepsilon_t \right)
    \end{equation*}
    retains $u_{St} = z_t$, which lies outside the span of $\tilde{F}_t$, so $p_t$ is not collinear with the controls. Granularity still supplies the identifying variation. It enters as a control rather than as an instrument.

    Neither endpoint arises for the disaggregated equation. The zero-correlation endpoint fails trivially, for the reason given above. The collinear one cannot arise at all. Its structural residual is $v_{Ht} = \lambda_H' F_t + u_{Ht}$, so collinearity with $\tilde{F}_t$ would require $u_{Ht} = H' u_t$ to lie in the span of $\tilde{F}_t$. Idiosyncrasy makes $u_t$ orthogonal to $F_t$ by Assumption~\ref{ass:weak stationarity and idiosyncrasy}, so this forces $u_{Ht} = 0$ almost surely and hence $H' \Omega H = 0$, which Assumption~\ref{ass:bounded eigenvalues} rules out. 
\end{remark}

The proofs adapt \citet{Bai2003InferentialDimensions} to the GIV setting.
Appendix~\ref{app:estimation of factor loadings} treats the factor
loadings, Appendix~\ref{app:estimation of common factors} the factors,
and Appendix~\ref{app:estimation of common components} combines the two
to obtain an influence-function expansion of $\hat C_t - \tilde C_t$.
I quote that expansion here and trace it through the GIV moment.

Write $\Bar{u}_m = D_N u_m$ for the $N$-vector of demeaned shocks at time $m$, with $j$-th entry $\Bar{u}_{jm}$ and share-weighted sum $\Bar{u}_{Sm} = S' \Bar{u}_m = z_m$. From \eqref{eq:app_influence function of estimator}, we can see that the difference between the estimated common component and the true common component is, with $\delta_{NT}^2 = \min(N,T)$,
\begin{equation} \label{eq:influence function of estimator}
    \hat{C}_t - \tilde{C}_t = \tilde{F}_t' \left[ \frac{\tilde{F}' \tilde{F}}{T} \right]^{-1} \frac{1}{T} \sum_{m=1}^T \tilde{F}_m \Bar{u}_{m} + \tilde{\Lambda} \left[ \frac{\tilde{\Lambda}' \tilde{\Lambda}}{N} \right]^{-1} \frac{1}{N} \sum_{j = 1}^N \tilde{\Blambda}_j \Bar{u}_{jt} + O_{\p} \left( \frac{1}{\delta_{NT}^2} \right)
\end{equation}

The estimated instrument is $\hat{z}_t = z_t - S'[\hat{C}_t - \tilde{C}_t]$. For the demand elasticity,
\begin{align*}
    \hat{\phi}_d - \phi_d = \frac{\sum_{t=1}^T \hat{z}_t \varepsilon_t}{\sum_{t=1}^T \hat{z}_t p_t},
\end{align*}
so the first stage enters through the numerator $\sum_t \hat{z}_t \varepsilon_t$ and the denominator $\sum_t \hat{z}_t p_t$. The supply elasticity puts $v_{Ht}$ in place of $\varepsilon_t$. To unify the results, $Q_t$ for either multiplier. The numerator and the denominator combine into a single result, which is what every proof in this section uses. 

\begin{Prop}[First-stage effect on the GIV moment] \label{prop:first stage effect}
    Suppose Assumptions \ref{ass:weak stationarity and idiosyncrasy} to \ref{ass:LN_moments and CLT} hold, $\frac{\sqrt{T}}{N} \to 0$ and $a_N^2/T \to 0$. Let $Q_t$ be weakly stationary and measurable with respect to $\sigma(F_t, u_t, \varepsilon_t)$, with $\e|Q_t|^{4 + 2\pi} < \infty$ and $\e[\Bar{u}_{jt} Q_t] = 0$ for every $j$. Write $\Sigma_{\tilde{F}} \defeq \e[\tilde{F}_t \tilde{F}_t']$, $\Delta_{\tilde{F} Q} \defeq \e[\tilde{F}_t' Q_t \mid \Lambda] \Sigma_{\tilde{F}}^{-1}$ and $\Bar{Q}_t = Q_t - \Delta_{\tilde{F} Q} \tilde{F}_t$. Then
    \begin{equation} \label{eq:first stage effect}
        \frac{a_N}{\sqrt{T}} \sum_{t=1}^T \hat{z}_t Q_t = \frac{a_N}{\sqrt{T}} \sum_{t=1}^T z_t \Bar{Q}_t + o_{\p}(1),
        \qquad
        \frac{a_N^2}{T} \sum_{t=1}^T \hat{z}_t p_t = \frac{a_N^2}{T} \sum_{t=1}^T z_t p_t + o_{\p}(1).
    \end{equation}
\end{Prop}
\begin{proof}
    See Appendix \ref{subsec:proof of first stage effect}. The numerator is Lemma \ref{lemma:numerator}.(ii) together with the law of large numbers for $\frac{1}{T} \sum_t \tilde{F}_t' Q_t$, and the denominator is Lemma \ref{lemma:denominator}.
\end{proof}

Estimating the instrument therefore leaves the moment intact and replaces the multiplier by its factor-residualised version. Taking $Q_t = \varepsilon_t$ gives $\Bar{Q}_t = \bar{\varepsilon}_t$ for the demand elasticity, where $\varepsilon_t$ is independent of $\Lambda$ so the conditioning in $\Delta_{\tilde{F} Q}$ is immaterial. Taking $Q_t = v_{Ht}$ gives $\Bar{Q}_t = \bar{v}_{Ht}$ for the supply elasticity which depends on the loadings and the conditioning binds. Both multipliers satisfy the orthogonality hypothesis, $\varepsilon_t$ by idiosyncrasy and $v_{Ht}$ because the choice of $H$ in \eqref{eq:expression for H} makes $H' \Omega D_N = 0$. Each theorem below follows by substituting \eqref{eq:first stage effect} into its counterpart in Section~\ref{section:known factors}.

In every regime, the first-stage term enters at the same order as the corresponding quantity in Section~\ref{section:known factors}. It is $O_{\p}(1)$ in the strong regime ($\mu \in (0,1)$) and $O_{\p}(N^{-\delta})$ in the nearly weak regime ($\mu > 1$), where $\delta = \min(1 - 1/\mu, 1/2)$. The rate of convergence is therefore unchanged, at $\sqrt{T}$ and $\sqrt{T}/N^{\delta}$ respectively. Only the asymptotic variance changes, through the substitution of the structural error with its factor-residualised version.

For $\mu > 2$, estimated-factor case inherits the finite-sample difficulty of the known-factor one. For the reasons given in Section~\ref{subsec:inference}, I recommend Anderson--Rubin confidence sets there. 

The supply elasticity carries a second estimation problem, the aggregating vector $H$, which I settle next before stating the results for the two regimes.

\subsection{Aggregation for the Supply Elasticity} \label{subsec:equal weights}

The moment condition for the supply elasticity requires an aggregating vector $H = \Omega^{-1} e_N/(e_N' \Omega^{-1} e_N)$ of \eqref{eq:expression for H}. Estimating the factor structure yields the instrument. But $H$ still depends on the unknown $\Omega$.

One option is to estimate $\Omega$ from the same first stage. Let $\hat u_t = D_N y_t - \hat C_t$ and $\hat\Omega = \frac{1}{T} \sum_t \hat u_t \hat u_t'$. But $\hat\Omega e_N = \frac{1}{T} \sum_t \hat u_t (\hat u_t' e_N) = 0$ in every sample. Thus $\hat\Omega$ is singular, and $\hat H(\hat\Omega) = \hat\Omega^{-1} e_N/(e_N' \hat\Omega^{-1} e_N)$ does not exist.

I therefore propose equal weights, $e_N/N$, as the aggregating vector. This choice has an obvious problem. The optimal $H$ makes the moment condition exact, $\e[z_t v_{Ht} \mid S, \Lambda] = S' D_N \Omega H = 0$. Write $y_{Et} \defeq e_N' y_t/N$ for the equally weighted aggregate and $v_{Et} \defeq y_{Et} - \phi_s p_t$ for its structural residual. With equal weights the moment condition holds only up to
\begin{equation*}
    b_N \defeq \e[z_t v_{Et} \mid S, \Lambda] = \frac{S' D_N \Omega e_N}{N} = \frac{\beta}{N},
\end{equation*}
which can be non-zero. Lemma \ref{lemma:hajek bound} gives $\beta = O_{\p}(a_N^{-1})$, so $b_N = O_{\p}(a_N^{-1}/N)$. Dividing by the strength of the instrument, $\Gamma_{zp} \asymp a_N^{-2}$, the violation displaces the estimator by $O_{\p}(a_N/N)$. The estimator's own sampling scale is $a_N/\sqrt{T}$ in every regime, so in those units the displacement is $\sqrt{T}/N$. It therefore vanishes under the condition $\frac{\sqrt{T}}{N} \to 0$ already imposed in this section, and equal weights cost nothing asymptotically.

\begin{remark}[Scope of equal weights] \label{rem:scope of equal weights}
    Equal weights do not work when $N$ is fixed. In the fixed-N case, please refer to GNR for the estimated choice of $H$.
\end{remark}

\subsection{Estimated Factors in the Strong Regime: $\mu \in (0,1)$}

In the strong regime, the heavy-tail concentration of the shares delivers $\sqrt{T}$-consistency and asymptotic normality, just as in Theorem~\ref{theorem:strong_consistency_aggregate}. The first-stage estimation contributes an $O_{\p}(1)$ term to the asymptotic distribution, which leaves the rate unchanged but affects the asymptotic variance. I formally state this result in the following theorems.

\begin{Theorem} \label{theorem:estimated_strong_consistency_aggregate}
    Suppose Assumptions \ref{ass:weak stationarity and idiosyncrasy} to \ref{ass:LN_moments and CLT} hold with $\mu \in (0,1)$ and $\frac{\sqrt{T}}{N} \to 0$. Then, conditional on $S$, for almost every realization of the shares, the GIV estimator for the aggregate structural parameter is consistent and asymptotically normal,
    \begin{equation*}
       \frac{\Gamma_{zp}}{\sqrt{V_{z \bar{\varepsilon}}(S)}} \cdot \sqrt{T} [\hat{\phi}_d - \phi_d] \xrightarrow{d} \n(0,  1)
    \end{equation*}
     where $\Gamma_{zp}$ is the finite-$N$ conditional mean $\e[z_t p_t \mid S, \Lambda] = \frac{S' D_N \Omega S}{\phi_d - \phi_s}$ of Proposition \ref{prop:denominator}, and $V_{z\bar{\varepsilon}}(S) = \lim_{T \to \infty} \frac{1}{T} \sum_{s=1}^T \sum_{t=1}^T \e[z_t z_s \bar{\varepsilon_t} \bar{\varepsilon_s} \mid S]$ where $\bar{\varepsilon}_t = \varepsilon_t - \e[\tilde{F}_t' \varepsilon_t] \Sigma_{\tilde{F}}^{-1} \tilde{F}_t $. 
     
\end{Theorem}
\begin{proof}
    Proof in Appendix \ref{subsec:proof of estimated_theorems}. The proof follows the proof of Theorem \ref{theorem:strong_consistency_aggregate}, with Proposition~\ref{prop:first stage effect} at $Q_t = \varepsilon_t$ and $a_N = 1$ supplying the first-stage substitution.
\end{proof}

For the supply elasticity both estimation problems are now settled. The instrument is the estimated $\hat z_t$ and the aggregation is equal weights, so the estimator
\begin{equation} \label{eq:estimated disaggregated estimator}
    \hat{\phi}_s = \frac{\sum_{t=1}^T \hat{z}_t y_{Et}}{\sum_{t=1}^T \hat{z}_t p_t}
\end{equation}
is computed from observables alone. It inherits the rate of the aggregate estimator, since $z_t$ and $p_t$ are common to both.

\begin{Theorem} \label{theorem:estimated_strong_consistency_disaggregate}
    Suppose Assumptions \ref{ass:weak stationarity and idiosyncrasy} to \ref{ass:LN_moments and CLT} hold with $\mu \in (0,1)$ and $\frac{\sqrt{T}}{N} \to 0$. Then, conditional on $(S, \Lambda)$, for almost every realization of the shares and loadings, the estimated-factor GIV estimator \eqref{eq:estimated disaggregated estimator} for the supply elasticity is consistent and asymptotically normal,
    \begin{equation*}
       \frac{\Gamma_{zp}}{\sqrt{V_{z \bar{v}_H}(S, \Lambda)}} \cdot \sqrt{T} [\hat{\phi}_s - \phi_s] \xrightarrow{d} \n(0,  1)
    \end{equation*}
    where $\Gamma_{zp}$ is the finite-$N$ conditional mean $\e[z_t p_t \mid S, \Lambda] = \frac{S' D_N \Omega S}{\phi_d - \phi_s}$ of Proposition \ref{prop:denominator}, and $V_{z \bar{v}_H}(S, \Lambda) = \lim_{T \to \infty} \frac{1}{T} \sum_{s=1}^T \sum_{t=1}^T \e[z_t z_s \bar{v}_{Ht} \bar{v}_{Hs} \mid S, \Lambda]$ where $\bar{v}_{Ht} = v_{Ht} - \e[\tilde{F}_t' v_{Ht} \mid \Lambda] \Sigma_{\tilde{F}}^{-1} \tilde{F}_t$. 

\end{Theorem}
\begin{proof}
    Proof in Appendix \ref{subsec:proof of estimated_disaggregate}. The denominator is unchanged from Theorem \ref{theorem:estimated_strong_consistency_aggregate}, as neither $\hat{z}_t$ nor $p_t$ depends on the aggregation. The numerator splits into the optimal-aggregation term and the weight-substitution term, and the proof bounds the second.
\end{proof}

\subsubsection{Comparison with \citet{Banafti2022InferentialDimensions}} \label{subsec:banfti}

\citet{Banafti2022InferentialDimensions} also study large-panel GIV in the
strong-instrument case. Their Theorems 2 and 4 conclude that the first-stage estimation of the instrument has no impact on the asymptotic variance of the GIV estimator. My 
Theorem~\ref{theorem:estimated_strong_consistency_aggregate} reaches the
opposite conclusion. I show that the first-stage estimation has a first-order contribution to the asymptotic variance of the GIV estimator.

My results differ from \citet{Banafti2022InferentialDimensions} for two reasons. The first is an assumption they impose. I show that it is incompatible with fat tails in the size distribution, which is what makes the instrument strong. The second is the order of one term in their Lemma 2. Their argument that this term vanishes in the limit rests on a condition they never state.

Neither assumption is plausible in the settings empiricists work with. The first cannot hold at all under a power law. I develop both points in detail in Appendix~\ref{appendix:banafti}.

\subsection{Estimated Factors under Nearly Weak Identification: $\mu > 1$}

In the nearly weak regime, the estimated-factor estimator inherits the rate $\sqrt{T}/N^{\delta}$ of Theorem~\ref{theorem:weak_consistency_aggregate}, where $\delta = \min(1 - 1/\mu, 1/2)$. The first-stage estimation contributes an $O_{\p}(N^{-\delta})$ term to the asymptotic distribution. This matches the order of the known-factor quantities and leaves the rate unchanged. I formally state these results for the demand and supply elacticities in the following theorems.
\begin{Theorem} \label{theorem:estimated_weak_consistency_aggregate} \label{theorem:estimated_semi_strong_consistency_aggregate}
    Suppose Assumptions \ref{ass:weak stationarity and idiosyncrasy} to \ref{ass:LN_moments and CLT} hold with $\mu > 1$, $N/T \to 0$, and $\frac{\sqrt{T}}{N} \to 0$. Let $\delta = \min(1 - 1/\mu, 1/2)$, so that $\delta = 1 - 1/\mu \in (0, 1/2)$ when $\mu \in (1,2)$ and $\delta = 1/2$ when $\mu > 2$. Then, conditional on $S$, for almost every realization of the shares, the GIV estimator for the aggregate structural parameter is consistent and asymptotically normal,
    \begin{equation*}
       \frac{\Gamma_{zp}}{\sqrt{V_{z \bar{\varepsilon}}(S)}} \cdot \sqrt{T}\, [\hat{\phi}_d - \phi_d] \xrightarrow{d} \n(0, 1),
    \end{equation*}
    where $\Gamma_{zp}$ is the finite-$N$ conditional mean $\e[z_t p_t \mid S, \Lambda] = \frac{S' D_N \Omega S}{\phi_d - \phi_s}$ of Proposition \ref{prop:denominator}, and $V_{z \bar{\varepsilon}}(S) = \lim_{T \to \infty} \frac{1}{T}\sum_{s=1}^T \sum_{t=1}^T \e[z_t z_s \bar{\varepsilon}_t \bar{\varepsilon}_s \mid S]$ with $\bar{\varepsilon}_t = \varepsilon_t - \e[\tilde{F}_t' \varepsilon_t] \Sigma_{\tilde{F}}^{-1} \tilde{F}_t$. The limit in $T$ exists, and Proposition \ref{prop:denominator} gives $a_N^2 \Gamma_{zp} = O_{\p}(1)$ together with $[a_N^2 \Gamma_{zp}]^{-1} = O_{\p}(1)$ whenever $\phi_d \neq \phi_s$, so the denominator is bounded above and away from zero at the rate $a_N^{-2}$. The studentization satisfies $\Gamma_{zp}/\sqrt{V_{z\bar{\varepsilon}}(S)} = O_{\p}(N^{-\delta})$, so the standardized statistic vanishes at the rate $\sqrt{T}/N^{\delta}$.

    If $\mu > 2$ and $N/T \to c > 0$, then $\hat{\phi}_d$ is inconsistent and identification is weak in the sense of \citet{Staiger1997InstrumentalInstruments}.
\end{Theorem}
\begin{proof}
    Proof in Appendix \ref{subsec:proof of estimated_theorems}. The argument follows Theorem~\ref{theorem:weak_consistency_aggregate}, with Proposition~\ref{prop:first stage effect} at $Q_t = \varepsilon_t$ and $a_N = N^{\delta}$ supplying the first-stage substitution.
\end{proof}

Now for the supply elasticity.

\begin{Theorem} \label{theorem:estimated_weak_consistency_disaggregate} \label{theorem:estimated_semi_strong_consistency_disaggregate}
    Suppose Assumptions \ref{ass:weak stationarity and idiosyncrasy} to \ref{ass:LN_moments and CLT} hold with $\mu > 1$, $N/T \to 0$, and $\frac{\sqrt{T}}{N} \to 0$. Let $\delta = \min(1 - 1/\mu, 1/2)$, so that $\delta = 1 - 1/\mu \in (0, 1/2)$ when $\mu \in (1,2)$ and $\delta = 1/2$ when $\mu > 2$. Then, conditional on $(S, \Lambda)$, for almost every realization of the shares and loadings, the estimated-factor GIV estimator \eqref{eq:estimated disaggregated estimator} for the supply elasticity is consistent and asymptotically normal,
    \begin{equation*}
       \frac{\Gamma_{zp}}{\sqrt{V_{z \bar{v}_H}(S, \Lambda)}} \cdot \sqrt{T}\, [\hat{\phi}_s - \phi_s] \xrightarrow{d} \n(0, 1),
    \end{equation*}
    where $\Gamma_{zp}$ is the finite-$N$ conditional mean $\e[z_t p_t \mid S, \Lambda] = \frac{S' D_N \Omega S}{\phi_d - \phi_s}$ of Proposition \ref{prop:denominator}, and $V_{z \bar{v}_H}(S, \Lambda) = \lim_{T \to \infty} \frac{1}{T}\sum_{s=1}^T \sum_{t=1}^T \e[z_t z_s \bar{v}_{Ht} \bar{v}_{Hs} \mid S, \Lambda]$ with $\bar{v}_{Ht} = v_{Ht} - \e[\tilde{F}_t' v_{Ht} \mid \Lambda] \Sigma_{\tilde{F}}^{-1} \tilde{F}_t$. The limit in $T$ exists, and Proposition \ref{prop:denominator} gives $a_N^2 \Gamma_{zp} = O_{\p}(1)$ together with $[a_N^2 \Gamma_{zp}]^{-1} = O_{\p}(1)$ whenever $\phi_d \neq \phi_s$, so the denominator is bounded above and away from zero at the rate $a_N^{-2}$. The studentization satisfies $\Gamma_{zp}/\sqrt{V_{z \bar{v}_H}(S, \Lambda)} = O_{\p}(N^{-\delta})$, so the standardized statistic vanishes at the rate $\sqrt{T}/N^{\delta}$.

    If $\mu > 2$ and $N/T \to c > 0$, then $\hat{\phi}_s$ is inconsistent and identification is weak in the sense of \citet{Staiger1997InstrumentalInstruments}. No choice of aggregation vector restores it.
\end{Theorem}
\begin{proof}
    Proof in Appendix \ref{subsec:proof of estimated_disaggregate}. The proof is shared with Theorem \ref{theorem:estimated_strong_consistency_disaggregate}, the two regimes differing only through the strength scale $a_N$.
\end{proof}

\begin{remark}[Three regimes, estimated factors] \label{rem:three_regimes_estimated}
The trichotomy of Remark~\ref{rem:three_regimes} carries over to the estimated-factor estimators of both parameters. Theorems~\ref{theorem:estimated_strong_consistency_aggregate} and~\ref{theorem:estimated_strong_consistency_disaggregate} cover the strong regime, Theorems~\ref{theorem:estimated_weak_consistency_aggregate} and~\ref{theorem:estimated_weak_consistency_disaggregate} cover the nearly weak regime under $N/T \to 0$, and their failure-mode clauses state the weak identification regime under $N/T \to c$ with $\mu > 2$. The additional growth restriction $\sqrt{T}/N \to 0$ from estimating the factor structure applies uniformly across the three regimes, and it is the only condition the strong regime needs.
\end{remark}

I now move to inference across all regimes. 

\subsection{Inference} \label{subsec:inference}

Inference needs a consistent estimator of the asymptotic variance of each of the two estimators. I propose Heteroskedasticity and Auto-Correlation (HAC) consistent estimators, written in the regime-specific rescaling $a_N$ of \eqref{eq:a_N}, and take the demand elasticity first. Define
\begin{equation} \label{eq:HAC estimator_aggregate_variance}
    \hat{V}_{{z} \bar{\varepsilon}} = \frac{a_N^2}{T} \sum_{t=1}^T \hat{z}_t^2 \tilde{\varepsilon}_t^2 + \frac{2 a_N^2}{T} \sum_{s=1}^{b_T} w(\frac{s}{b_T}) \sum_{t = s+1}^T \hat{z}_t \hat{z}_s \tilde{\varepsilon}_t \tilde{\varepsilon}_s
\end{equation}

where $b_T$ is the bandwidth and $w(x)$ is a kernel function, $w: \mathbb{R}^+ \to [0,1]$ such that $w(x) = 0$ for $x>1$ and $w(0) = 1$. $\tilde{\varepsilon}_t = \hat{\varepsilon}_t - \frac{1}{T} \sum_t \hat{F}'_t  \hat{\varepsilon}_t \cdot \hat{\Sigma}_{\tilde{F}}^{-1} \cdot  \hat{F}_t $, where $\hat{\varepsilon}_t = d_t - \hat{\phi}_d p_{t}$ and $\hat{\Sigma}_{\tilde{F}} = \left[ \frac{\hat{F}' \hat{F}}{T} \right]$.

The estimator carries the rescaling and the variance it estimates does not. $V_{z \bar{\varepsilon}}(S)$ is the un-rescaled long-run variance of the theorem statements, so it is $O_{\p}(a_N^{-2})$, while $\hat{V}_{{z} \bar{\varepsilon}} = O_{\p}(1)$. The proposition below therefore compares $\hat{V}_{{z} \bar{\varepsilon}}$ with $a_N^2 V_{z \bar{\varepsilon}}(S)$, both of which are $O_{\p}(1)$.

The rescaling cancels wherever the estimator is used. Every use is a ratio with the instrument entering twice above and twice below. Wald inference divides
\begin{equation*}
    a_N^{-2} \hat{V}_{z \bar{\varepsilon}} = \frac{1}{T} \sum_{t=1}^T \hat{z}_t^2 \tilde{\varepsilon}_t^2 + \frac{2}{T} \sum_{s=1}^{b_T} w(\tfrac{s}{b_T}) \sum_{t = s+1}^T \hat{z}_t \hat{z}_s \tilde{\varepsilon}_t \tilde{\varepsilon}_s
\end{equation*}
by $\hat{\Gamma}_{zp}^2$, where $\hat{\Gamma}_{zp} = \frac{1}{T} \sum_t \hat{z}_t p_t$. The Anderson--Rubin statistic \eqref{eq:AR statistic_aggregate} divides $\frac{a_N^2}{T} ( \sum_t \hat{z}_t \varepsilon_t )^2$ by $\hat{V}_{z \bar{\varepsilon}}$, and there the two factors of $a_N^2$ cancel against each other.

\begin{Prop} \label{prop:consistency of HAC_aggregate}
    Suppose Assumptions \ref{ass:weak stationarity and idiosyncrasy} to \ref{ass:LN_moments and CLT} hold and $\frac{\sqrt{T}}{N} \xrightarrow{} 0$, and suppose in addition $\frac{N}{T} \xrightarrow{} 0$ when $\mu > 1$. Then conditional on $\s$,
    \begin{equation*}
        \hat{V}_{{z} \bar{\varepsilon}} - a_N^2 V_{z \bar{\varepsilon}}(S) \xrightarrow{p} 0.
    \end{equation*}
\end{Prop}
\begin{proof}
    See Appendix \ref{subsec:proof or HAC_aggregate}
\end{proof}

The same construction gives the variance of the supply elasticity, with the equal-weights residuals in place of the demand residuals. Let $\hat{v}_{Et} = y_{Et} - \hat{\phi}_s p_t$ and let $\tilde{v}_{Et} = \hat{v}_{Et} - \frac{1}{T} \sum_t \hat{F}'_t \hat{v}_{Et} \cdot \hat{\Sigma}_{\tilde{F}}^{-1} \cdot \hat{F}_t$ be its factor-recentered version, with $\hat{\Sigma}_{\tilde{F}} = [\hat{F}' \hat{F}/T]$ as above. The estimator is
\begin{equation} \label{eq:HAC estimator_disaggregated_variance}
    \hat{V}_{z v_H} = \frac{a_N^2}{T} \sum_{t=1}^T \hat{z}_t^2 \tilde{v}_{Et}^2 + \frac{2 a_N^2}{T} \sum_{s=1}^{b_T} w(\frac{s}{b_T}) \sum_{t = s+1}^T \hat{z}_t \hat{z}_s \tilde{v}_{Et} \tilde{v}_{Es}
\end{equation}
with the same bandwidth and kernel as in \eqref{eq:HAC estimator_aggregate_variance}. Every ingredient is observable.

The residuals are built with equal weights, and the variance they have to estimate is $V_{z \bar{v}_H}(S, \Lambda)$, the one carried by the optimal aggregator, $H$. The two residuals differ by $v_{Et} - v_{Ht}$, which is $O_{\p}(N^{-1/2})$, so the substitution is again error-free in the limit.

\begin{Prop} \label{prop:consistency of HAC_disaggregate}
    Suppose Assumptions \ref{ass:weak stationarity and idiosyncrasy} to \ref{ass:LN_moments and CLT} hold and $\frac{\sqrt{T}}{N} \xrightarrow{} 0$, and suppose in addition $\frac{N}{T} \xrightarrow{} 0$ when $\mu > 1$. Then conditional on $(S, \Lambda)$,
    \begin{equation*}
        \hat{V}_{z v_H} - a_N^2 V_{z \bar{v}_H}(S, \Lambda) \xrightarrow{p} 0.
    \end{equation*}
\end{Prop}
\begin{proof}
    See Appendix \ref{subsec:proof of HAC_disaggregate}
\end{proof}

As per the theorems and propositions above, Wald inference is valid when $\mu <1$. But when $\mu >1$, I recommend limiting Wald inference to the range, $\mu \in (1,2)$. For $\mu >2$, I recommend Anderson-Rubin inference.

The choice between Wald and Anderson--Rubin inference depends on the behavior of the concentration parameter. The concentration parameter $\kappa^2_{\text{conc}}$ measures the signal-to-noise ratio of the first-stage moment \citep{Stock2002AMoments}. For the estimated-factor GIV estimator,
\begin{equation*}
    \kappa^2_{\text{conc}} \;=\; \frac{T \, \Gamma_{zp}^2}{V_{z \bar{\varepsilon}}(S)},
\end{equation*}
and Theorem~\ref{theorem:estimated_weak_consistency_aggregate} reads $\kappa_{\text{conc}}(\hat{\phi}_d - \phi_d) \xrightarrow{d} \n(0,1)$. The Gaussian approximation is reliable in finite samples when $\kappa^2_{\text{conc}}$ is large. The first-stage $F$-statistic is its sample analog.

The supply elasticity shares the first stage, since $\hat{\phi}_s$ has the same denominator $\sum_t \hat{z}_t p_t$. Only the variance of the moment changes,
\begin{equation*}
    \kappa^2_{\text{conc},H} \;=\; \frac{T \, \Gamma_{zp}^2}{V_{z \bar{v}_H}(S, \Lambda)},
\end{equation*}
and Theorem~\ref{theorem:estimated_weak_consistency_disaggregate} gives $\kappa_{\text{conc},H}(\hat{\phi}_s - \phi_s) \xrightarrow{d} \n(0,1)$. Both $V_{z \bar{\varepsilon}}(S)$ and $V_{z \bar{v}_H}(S, \Lambda)$ are $O_{\p}(a_N^{-2})$, so the two concentration parameters are of the same order. Everything that follows applies to $\phi_s$ with $V_{z \bar{\varepsilon}}(S)$ replaced by $V_{z \bar{v}_H}(S, \Lambda)$.

Since $\Gamma_{zp}$ and $V_{z \bar{\varepsilon}}(S)$ are both of order $a_N^{-2}$, the concentration parameter is $\kappa^2_{\text{conc}} = T/a_N^2$ in every regime. The order is two-sided. The upper bound $a_N^2 V_{z \bar{\varepsilon}}(S) = O_{\p}(1)$ is established in Appendix~\ref{subsec:proof of HAC_disaggregate}, and the lower bound $[a_N^2 V_{z \bar{\varepsilon}}(S)]^{-1} = O_{\p}(1)$ is part of Assumption~\ref{ass:mixing for stationary time series}. The corresponding bounds for $\Gamma_{zp}$ are Proposition~\ref{prop:denominator}. Parametrize any admissible sequence as $T = N \cdot f(N)$ with $f(N) \to \infty$. Reading $a_N$ off \eqref{eq:a_N},
\begin{equation*}
    \kappa^2_{\text{conc}} \;=\; \begin{cases}
        N \cdot f(N) & \mu \in (0,1), \\
        N^{1-2\delta} \cdot f(N) & \mu \in (1,2), \\
        f(N) & \mu > 2
    \end{cases}
\end{equation*}
For $\mu < 2$, the concentration parameter has a polynomial floor, $N$ in the strong regime and $N^{1-2\delta}$ with $1 - 2\delta > 0$ in the nearly weak regime. For $\mu > 2$ the floor is just $f(N)$, which may grow arbitrarily slowly. At the boundary $\mu = 2$, the polynomial floor disappears. The same dichotomy holds for $\kappa^2_{\text{conc},H}$, so the choice of procedure below is the same for both parameters.

\paragraph{Wald inference for $\mu < 2$.} The polynomial floor in $\kappa^2_{\text{conc}}$ keeps Wald inference reliable in the strong and the nearly weak regime alike. The practitioner does not need to know $\delta$. The asymptotic variance is
\begin{equation*}
    \widehat{\mathrm{Avar}}(\hat{\phi}_d) \;=\; \frac{1}{T} \cdot \frac{a_N^{-2} \hat{V}_{z \bar{\varepsilon}}}{\hat{\Gamma}_{zp}^{2}}, \qquad \hat{\Gamma}_{zp} = \frac{1}{T} \sum_{t=1}^T \hat{z}_t p_t,
\end{equation*}
where $\hat{V}_{z \bar{\varepsilon}}$ is the HAC estimator \eqref{eq:HAC estimator_aggregate_variance} and Proposition~\ref{prop:consistency of HAC_aggregate} gives its consistency. The rescaling cancels as above, so the formula is the usual instrumental-variable sandwich in $\hat{z}_t$ and $a_N$ is never formed. The studentization in Theorem~\ref{theorem:estimated_weak_consistency_aggregate} absorbs the rate for the same reason, so $\delta$ does not enter. This parallels GMM under near-weak identification \citep{Antoine2021GMMIdentification}. Standard Wald confidence intervals remain valid, even though for $\mu > 1$ the rate is slower than in the strong regime.

For the supply elasticity,
\begin{equation*}
    \widehat{\mathrm{Avar}}(\hat{\phi}_s) \;=\; \frac{1}{T} \cdot \frac{a_N^{-2} \hat{V}_{z v_H}}{\hat{\Gamma}_{zp}^{2}},
\end{equation*}
with $\hat{V}_{z v_H}$ from \eqref{eq:HAC estimator_disaggregated_variance} and Proposition~\ref{prop:consistency of HAC_disaggregate}. The denominator is the one above, unchanged, since the two estimators share the first stage. The rescaling cancels exactly as it does for the demand elasticity.

\paragraph{Anderson--Rubin inference for $\mu > 2$.} The concentration parameter of the estimated-factor estimator can grow arbitrarily slowly, so its realized value in finite samples need not be large. The estimator therefore has poor finite-sample performance even though it is asymptotically normal. I recommend inverting the Anderson--Rubin (AR) test of \citet{Anderson1949EstimatorsEquations}, which avoids the dependence on $\kappa^2_{\text{conc}}$ entirely. It requires no condition on the rate at which $\kappa^2_{\text{conc}} \to \infty$.

The Anderson--Rubin test is evaluated at the null, so it never uses the GIV estimate $\hat{\phi}_d$. Two features follow, and together they make the $N/T$ condition irrelevant. First, under $H_0: \phi_d = \phi_d^0$, the residual $d_t - \phi_d^0 p_t = \varepsilon_t$ is exact, so no first-stage estimation of $\phi_d$ enters the variance estimator. Second, the statistic is self-normalizing, and the rate $a_N$ cancels between the numerator and the variance estimator, so the test refers neither to the rate $\delta$ nor to $N/T$. The only condition that survives is $\sqrt{T}/N \to 0$, from estimating the factor structure in $\hat{z}_t$.

Consider the null hypothesis, $H_0: \phi_d = \phi_d^0$. Define the null-imposed residual $\varepsilon_t(\phi_d^0) = d_t - \phi_d^0 p_t$ and its factor-projected version $\tilde{\varepsilon}_t(\phi_d^0) = \varepsilon_t(\phi_d^0) - \frac{1}{T} \sum_t \hat{F}_t' \varepsilon_t(\phi_d^0) \cdot \hat{\Sigma}_{\tilde{F}}^{-1} \hat{F}_t$. Define the Anderson-Rubin statistic \citep{Anderson1949EstimatorsEquations},
\begin{equation} \label{eq:AR statistic_aggregate}
    \text{AR}(\phi_d^0) =  \frac{N}{T} \left( \sum_t \hat{z}_t (y_{St} - \phi_d^0 p_{t}) \right)^2 \frac{1}{\hat{V}_{z \bar{\varepsilon}}(\phi_d^0)}
\end{equation}

where $\hat{V}_{z \bar{\varepsilon}}(\phi_d^0)$ is the HAC estimator defined in \eqref{eq:HAC estimator_aggregate_variance}, specialized to $\mu > 2$ so that $a_N^2 = N$, and built from the null-imposed residual,
\begin{equation*}
    \hat{V}_{{z} \bar{\varepsilon}}(\phi_d^0) = \frac{N}{T} \sum_{t=1}^T \hat{z}_t^2 \tilde{\varepsilon}_t(\phi_d^0)^2 + \frac{2 N}{T} \sum_{s=1}^{b_T} w(\frac{s}{b_T}) \sum_{t = s+1}^T \hat{z}_t \hat{z}_s \tilde{\varepsilon}_t(\phi_d^0) \tilde{\varepsilon}_s(\phi_d^0)
\end{equation*}

The $N/T$ in the numerator cancels the $a_N^2/T = N/T$ in $\hat{V}_{z \bar{\varepsilon}}(\phi_d^0)$, so the statistic is self-normalized and carries no reference to the rate. Rescaled back, $N^{-1} \hat{V}_{z \bar{\varepsilon}}(\phi_d^0)$ is consistent for the conditional asymptotic variance $V_{z \bar{\varepsilon}}(S)$, with $\bar{\varepsilon}_t = \varepsilon_t - \e[\tilde{F}_t' \varepsilon_t] \Sigma_{\tilde{F}}^{-1} \tilde{F}_t$, by the argument of Proposition~\ref{prop:consistency of HAC_aggregate} with the residual reduction now exact. Inversion of this test yields a confidence region of the correct size.

\begin{Prop} \label{prop:weakness robust test_aggregate}
    Suppose Assumptions \ref{ass:weak stationarity and idiosyncrasy} to \ref{ass:LN_moments and CLT} hold with $\mu > 2$ and $\frac{\sqrt{T}}{N} \to 0$. Then, for any trajectory of $N/T$, under the null $H_0: \phi_d = \phi_d^0$,
\begin{equation*}
    \text{AR}(\phi_d^0) \xrightarrow{d} \chi^2_1.
\end{equation*}
\end{Prop}
\begin{proof}
    In Appendix~\ref{subsec:proof of AR_aggregate}. The proof works at the null, so the GIV estimate never enters and no restriction on $N/T$ is required.
\end{proof}

Proposition~\ref{prop:weakness robust test_aggregate} places no condition on $N/T$. It therefore covers the weak-identification case $N/T \to c > 0$, where $\hat{\phi}_d$ is inconsistent. The test does not depend on consistency of $\hat{\phi}_d$. It uses only the moment and the variance, both evaluated at the hypothesised value.

The supply elasticity admits the same test. Under $H_0: \phi_s = \phi_s^0$ the equal-weights residual $y_{Et} - \phi_s^0 p_t = v_{Et}$ is exact, so the null-imposed statistic again avoids the first stage. Define $\tilde{v}_{Et}(\phi_s^0)$ as the factor-projected version of $y_{Et} - \phi_s^0 p_t$, formed as in \eqref{eq:HAC estimator_disaggregated_variance}, and
\begin{equation} \label{eq:AR statistic_disaggregate}
    \text{AR}_H(\phi_s^0) = \frac{N}{T} \left( \sum_t \hat{z}_t (y_{Et} - \phi_s^0 p_{t}) \right)^2 \frac{1}{\hat{V}_{z v_H}(\phi_s^0)},
\end{equation}
where $\hat{V}_{z v_H}(\phi_s^0)$ is the estimator \eqref{eq:HAC estimator_disaggregated_variance} specialized to $\mu > 2$ so that $a_N^2 = N$, and built from $\tilde{v}_{Et}(\phi_s^0)$. The rate cancels between numerator and denominator exactly as in \eqref{eq:AR statistic_aggregate}.

\begin{Prop} \label{prop:weakness robust test_disaggregate}
    Suppose Assumptions \ref{ass:weak stationarity and idiosyncrasy} to \ref{ass:LN_moments and CLT} hold with $\mu > 2$ and $\frac{\sqrt{T}}{N} \to 0$. Then, for any trajectory of $N/T$, under the null $H_0: \phi_s = \phi_s^0$,
\begin{equation*}
    \text{AR}_H(\phi_s^0) \xrightarrow{d} \chi^2_1.
\end{equation*}
\end{Prop}
\begin{proof}
    In Appendix~\ref{subsec:proof of AR_disaggregate}. The equal-weights aggregation error is $o_{\p}(1)$ whatever the trajectory of $N/T$, so it costs the test nothing.
\end{proof}

\section{Empirical Application}
\label{sec:empirical}

I apply the GIV estimator to three commodity markets---refined copper, crude
oil, and natural gas---to estimate their respective price elasticities of
demand and supply.  Each market provides a global supply panel whose country-level
idiosyncratic shocks serve as the GIV instrument.

\subsection{Model}
\label{sec:model}

\paragraph{Demand equation.}
The equation of interest is
\begin{equation} \label{eq:empirical demand}
  d_t \;=\; \phi_d\, p_t \;+\; X_t'\beta \;+\; \varepsilon_t,
\end{equation}
where $d_t$ is the year-on-year growth rate of aggregate demand,
$p_t$ is the year-on-year growth rate of the real commodity price,
$X_t$ is a vector of observable common controls, and $\varepsilon_t$ is an aggregate
demand shock.  The first parameter of interest is $\phi_d$, the price elasticity of
demand.

OLS estimation of~\eqref{eq:empirical demand} is biased because $p_t$ and
$\varepsilon_t$ are correlated. Aggregate demand expansions raise both quantities and
prices simultaneously, attenuating the estimated elasticity toward zero, or
even reversing its sign when supply shocks are small relative to demand
shocks.

\paragraph{Supply panel and GIV instrument.}
We observe the changes to supply at individual country level. Let
 $y_{it}$ denote country~$i$'s change in supply in period~$t$.  The change in supply follows the structural equation:
\begin{equation}
  \label{eq:panel}
  y_{it} \;=\; \phi_s p_t + \lambda_i' F_t \;+\; u_{it},
\end{equation}
where $F_t$ is an $r$-vector of common factors (spanning price and aggregate
demand shocks) and $u_{it}$ is an idiosyncratic supply shock. The estimated-factor
GIV instrument is the share-weighted idiosyncratic shock,
\begin{equation}
  \label{eq:giv}
  \hat{z}_t \;=\; \sum_{i=1}^N S_{i}\,\hat{u}_{it} = S'(D_N y_t - \hat{C}_t)
\end{equation}
where $S_i$ is the long-term market share of country~$i$. In our dataset, we construct $S_i$ by calculating the share for every $t$ and averaging across all time periods. I construct all growth rates as year-on-year midpoint growth.
\begin{equation} \label{eq:yoygrowth}
  g_{i,t} = \frac{Y_{i,t} - Y_{i,t-12}}{\tfrac{1}{2}(Y_{i,t} + Y_{i,t-12})},
\end{equation}

The second parameter of interest is $\phi_s$ in \eqref{eq:panel}, the price elasticity of supply. Both elasticities are estimated with the same instrument \eqref{eq:giv}. The supply moment requires an aggregation of the panel, and Section~\ref{subsec:equal weights} shows that equal weights are valid. So the dependent variable for $\phi_s$ is $y_{Et} = e_N' y_t / N$.

OLS on \eqref{eq:panel} is biased in the opposite direction. A positive supply shock raises quantity and lowers the price, so $p_t$ covaries negatively with the supply residual and OLS understates $\phi_s$. Since the supply elasticity is positive, this is again attenuation toward zero.

\paragraph{Factor selection.}
The number of factors $r$ is selected from the data using the eigenvalue
ratio~(ER) criterion of \citet{Ahn2013EigenvalueFactors}.

\subsection{Data}

I estimate on three monthly country-level supply panels. Refined copper supply
from Bloomberg, $N = 29$ countries over January~2009 to December~2025. Crude oil
production from the EIA International Energy Statistics, $N = 21$ units over
January~1973 to November~2025. Natural gas production from the JODI Gas
Database, $N = 27$ countries over January~2010 to November~2025. Each panel
carries a rest-of-world residual, and every series enters as the year-on-year
midpoint growth rate~\eqref{eq:yoygrowth}.

The price series are the LME spot copper price, a splice of the FRED OILPRICE
series and WTI, and the Henry Hub natural gas spot price, each deflated by U.S.\
CPI rebased to 2015~$= 100$. Appendix~\ref{appendix_data} gives the sources,
the cleaning rules, the estimation sample and the control set for each
commodity.

\subsection{Results}
\label{sec:results}

Table~\ref{tab:elasticity_controls} presents the main elasticity estimates, Panel A for the demand equation and Panel B for the supply equation.

\begin{table}[htbp]
\centering
\caption{Price Elasticities of Demand and Supply: OLS and GIV}
\label{tab:elasticity_controls}
\footnotesize
\setlength{\tabcolsep}{4pt}
\begin{tabular}{lccc}
\toprule
Estimator & Copper & Crude Oil & Natural Gas \\
  & Refined Supply ($N=29$) & Production (EIA) ($N=21$) & Production ($N=27$) \\
\midrule
\multicolumn{4}{l}{\textit{Panel A: Demand, aggregate equation ($\phi_d$)}} \\
\quad OLS & -0.0527 (0.0224) & 0.0041 (0.0033) & 0.0022 (0.0032) \\
\quad GIV --- Wrong SE & -0.1478 (0.0511) & -0.1045 (0.0503) & -0.0562 (0.0407) \\
\quad GIV --- Correct SE & -0.1478 (0.0513) & -0.1045 (0.0520) & -0.0562 (0.0461) \\
\addlinespace
\multicolumn{4}{l}{\textit{Panel B: Supply, disaggregated equation ($\phi_s$)}} \\
\quad OLS & 0.0248 (0.0357) & 0.0173 (0.0081) & 0.0088 (0.0140) \\
\quad GIV --- Wrong SE & 0.1409 (0.0477) & 0.2563 (0.1347) & 0.1504 (0.1019) \\
\quad GIV --- Correct SE & 0.1409 (0.0580) & 0.2563 (0.1330) & 0.1504 (0.0816) \\
\bottomrule
\end{tabular}

\vspace{4pt}
\raggedright\footnotesize
\textit{Notes:} Standard errors in parentheses. Panel A reports the demand elasticity $\phi_d$ from the aggregate equation. Panel B reports the supply elasticity $\phi_s$ from the disaggregated equation with the equal-weights aggregator $H = e_N/N$ of Section~\ref{subsec:equal weights}. The two GIV rows share a single point estimate and differ only in the standard error. The wrong SE treats the instrument $z$ as fixed. The correct SE adds the first-stage factor correction. The controls $X$ are not partialled out of the panel before the first stage. All standard errors are HAC (Newey--West, Bartlett kernel, bandwidth $\lfloor 1.3\sqrt{T}\rfloor$).
\end{table}

\subsubsection{Granularity}

We assume that the cross-sectional size distribution follows a power law, $\Pr(S \geq s) \propto s^{-\mu}$. Table~\ref{tab:power_law_combined} reports Pareto exponent $\mu$ estimates for the supply panels using three methods: Hill~(1975) MLE, log-rank OLS, and the bias-corrected
Gabaix--Ibragimov regression \citep{Gabaix2011RankExponents}.

\begin{table}[H]
\centering
\caption{Pareto Tail Exponent Estimates}
\label{tab:power_law_combined}
\small
\begin{tabular}{lcccccc}
\toprule
Method & \multicolumn{2}{c}{Copper} & \multicolumn{2}{c}{Crude Oil} & \multicolumn{2}{c}{Natural Gas} \\
\cmidrule(lr){2-3}\cmidrule(lr){4-5}\cmidrule(lr){6-7}
  & $\hat{\mu}$ & 95\% CI & $\hat{\mu}$ & 95\% CI & $\hat{\mu}$ & 95\% CI \\
\midrule
Hill (1975) MLE & 0.22 & [0.21, 0.76] & 0.38 & [0.34, 0.77] & 0.20 & [0.17, 0.31] \\
Naive OLS & 0.51 & [0.39, 0.64] & 0.58 & [0.44, 0.71] & 0.29 & [0.24, 0.34] \\
Gabaix--Ibragimov (2011) & 0.58 & [0.28, 0.87] & 0.65 & [0.26, 1.05] & 0.32 & [0.15, 0.49] \\
\bottomrule
\end{tabular}

\vspace{4pt}
\raggedright\footnotesize
\textit{Notes:} Copper: Refined Supply panel, $N = 29$. Crude Oil: Production (EIA) panel, $N = 21$. Natural Gas: Production panel, $N = 27$. Hill (1975) MLE: $\hat{\mu} = N / \sum_i \log(S_i / S_{\min})$ with nonparametric bootstrap 95\% CI (5{,}000 replications, resampling with replacement). Naive OLS: log-rank regression $\log(\text{rank}) = a - \mu \log(S)$ with OLS standard errors. Gabaix--Ibragimov (2011): shifted log-rank regression $\log(\text{rank} - 1/2) = a - \mu \log(S)$ with analytic SE $= \sqrt{2/N}\,\hat{\mu}$.
\end{table}

\paragraph{Regime classification.} The estimates partition the three
commodities cleanly. Copper has $\hat\mu_{\mathrm{GI}} = 0.58$ with
confidence interval $[0.28, 0.87]$, and natural gas has
$\hat\mu_{\mathrm{GI}} = 0.32$ with $[0.15, 0.49]$. Both lie firmly in
the strong regime $\mu < 1$ of Section~\ref{sec:three_regimes}. Crude
oil has $\hat\mu_{\mathrm{GI}} = 0.65$ but a wider interval,
$[0.26, 1.05]$, that reaches into the nearly weak region
$\mu \in (1,2)$.
Theorems~\ref{theorem:estimated_strong_consistency_aggregate} and~\ref{theorem:estimated_weak_consistency_aggregate} establish
consistency and asymptotic normality for the entire range
$\mu \in (0,2)$, so the demand point estimates and Wald confidence intervals
reported below remain valid for crude oil under either reading of $\mu$.
We can comfortably rule out $\mu > 2$ for all three commodities.

\subsubsection{Demand Elasticities}

Panel A of Table~\ref{tab:elasticity_controls} reports the demand equation.

The OLS estimates starkly illustrate the endogeneity problem.  For copper,
OLS yields $-0.053$, roughly one-third of the GIV estimate in magnitude.  For
crude oil and natural gas, OLS is positive. GIV corrects all three estimates
to economically plausible negative demand elasticities: $-0.148$ for copper,
$-0.105$ for crude oil, and $-0.056$ for natural gas.

Existing theory leads to the wrong standard error. The correction has no determinate sign. It removes the factor-explained part of the structural error, but it reweights the remainder by the instrument, so it can widen or narrow the interval. In the demand panel, the wrong standard error is too small in all three columns, so it under-covers. The understatement is $0.4\%$ for copper, $3.4\%$ for crude oil and $13.3\%$ for natural gas.

The effect on inference is visible in the $t$-statistics. Correcting the standard error moves copper from $2.89$ to $2.88$, crude oil from $2.08$ to $2.01$, and natural gas from $1.38$ to $1.22$. No conclusion at the $5\%$ level is reversed here. But crude oil is left sitting on the critical value, and the margin it appeared to have under the uncorrected standard error was an artifact of treating the estimated instrument as observed.

All three commodities display inelastic demand.  The estimated elasticities
($|\hat\phi_d| < 0.15$) are consistent with the industrial nature of these commodities. All have limited short-run substitutes, making quantity responses to price changes modest.  The estimates are also broadly
in line with the existing literature, which typically reports short-run
elasticities in the range $[-0.05,\,-0.25]$ for crude oil \citep{Baumeister2019StructuralShocks}, $[-0.05,\,-0.20]$ for natural gas \citep{Auffhammer2018NaturalBills,Labandeira2017ADemand}, and $[-0.07,\,-0.10]$ for copper \citep{Shojaeddini2025UnderstandingMinerals,Lanz2013SubglobalIndustry}.

\subsubsection{Supply Elasticities}

Panel B of Table~\ref{tab:elasticity_controls} reports the disaggregated
equation. OLS is attenuated toward zero in all three columns, at $0.025$ for copper, $0.017$ for crude oil and $0.009$ for natural gas, between one sixth and one twentieth of the corresponding GIV estimate. GIV returns positive supply elasticities of $0.141$ for copper, $0.256$ for crude oil and $0.150$ for natural gas, all upward sloping and all inelastic.

The standard error correction raises the copper standard error from $0.0477$ to
$0.0580$, and lowers it for crude oil and natural gas, from $0.1347$ to $0.1330$
and from $0.1019$ to $0.0816$. The supply panel therefore runs the other way for two
of the three commodities. The $t$-statistics move from $2.95$ to $2.43$ for copper, from
$1.90$ to $1.93$ for crude oil, and from $1.48$ to $1.84$ for natural gas. Copper loses a
fifth of its apparent precision and stays significant. Natural gas gains, and moves from
clearly insignificant to just short of the critical value. Taken with the demand panel, the
six estimates show the correction operating in both directions within a single application.

Inference is Wald throughout. The regime classification above rules out
$\mu > 2$ for all three commodities, and
Theorems~\ref{theorem:estimated_strong_consistency_disaggregate}
and~\ref{theorem:estimated_weak_consistency_disaggregate} cover
$\mu \in (0,2)$ for the supply elasticity exactly as their aggregate
counterparts do for the demand elasticity, so the Anderson--Rubin sets of
Proposition~\ref{prop:weakness robust test_disaggregate} are not needed here.

All commodities display inelastic supply. The estimated elasticities are broadly in line with the existing literature. Copper at $0.141$ falls very close to $0.10$ reported by \citet{Shojaeddini2025UnderstandingMinerals}. \citet{Hausman2015WelfareGas} report supply elasticity of US natural gas production at $0.09$ for short-run, going all the way to $0.81$ for the long run. My estimate of $0.15$ is well within this range. Crude oil at
$0.256$ is close to, but slightly larger than the usual estimates. \citet{Baumeister2019StructuralShocks} reports $0.15$ for crude oil.

\section{Simulation} \label{sec:simulation}

I assess the finite sample performance of the two GIV estimators under different regimes of instrument strength via Monte Carlo simulation. The factor structure is calibrated to two empirical commodity panels: Copper and Crude Oil, while all other components of the data generating process are simulated. Running the experiment on both commodities allows us to examine how the estimator performs under different cross-section and time-series dimensions: Copper provides $N = 29$ countries over $T = 192$ months, while Crude Oil provides $N = 21$ countries over $T = 611$ months.

The loadings and the factors are held at their empirical values, the shocks are drawn normal with the empirical scales, and the true elasticities are set to the control-robust GIV estimates of Table~\ref{tab:elasticity_controls}. This gives $\phi_d = -0.148$ and $\phi_s = 0.141$ for Copper, and $\phi_d = -0.104$ and $\phi_s = 0.256$ for Crude Oil. Appendix~\ref{app:mc design} gives the design in full, including the power-law draws for the individual sizes and the grid of $\mu$. Appendix~\ref{app:augment N} augments the cross section to $N \in \{50, 100\}$ which shows the sensitivity of the estimation to different values of $N$.

\subsection{Results: Demand Elasticity} \label{subsec:mc results aggregate}

Table~\ref{tab:mc_table1} organizes the simulation
evidence by the Pareto exponent $\mu$, which Proposition~\ref{prop:weakness of known factor instrument}
maps directly onto the regimes of instrument strength. The CI Length
column reports the standard $t$-based 95\% interval with White
heteroskedasticity-robust standard errors.

\begin{table}[htbp]
\centering
\caption{Monte Carlo Simulation Results: Empirical Commodity Panels}
\label{tab:mc_table1}
\begin{tabular}{ccccccc}
\toprule
$\mu$ & Median $\hat{\phi}_d$ & Bias & RMSE & Coverage & CI Length & Median $F$ \\
\midrule
\multicolumn{7}{l}{\textit{Panel A: Copper ($N=29$, $T=192$, $\phi_d=-0.148$)}} \\
\midrule
0.3 & -0.1479 & -0.0004 & 0.0075 & 0.948 & 0.0222 & 124.0 \\
0.5 & -0.1479 & -0.0009 & 0.0107 & 0.947 & 0.0283 & 78.2 \\
0.8 & -0.1479 & -0.0003 & 0.0811 & 0.949 & 0.0437 & 34.0 \\
1.2 & -0.1477 & -0.0043 & 0.1632 & 0.947 & 0.0738 & 12.5 \\
1.4 & -0.1474 & -0.0072 & 0.3017 & 0.954 & 0.0880 & 8.7 \\
1.8 & -0.1460 & -0.0205 & 0.7739 & 0.949 & 0.1238 & 4.6 \\
2.5 & -0.1426 & -0.1895 & 8.8921 & 0.958 & 0.1937 & 2.1 \\
3.5 & -0.1338 & -0.0025 & 1.9659 & 0.960 & 0.2943 & 1.0 \\
6.0 & -0.1166 & 0.0047 & 1.6142 & 0.967 & 0.4530 & 0.5 \\
\midrule
\multicolumn{7}{l}{\textit{Panel B: Crude Oil ($N=21$, $T=611$, $\phi_d=-0.104$)}} \\
\midrule
0.3 & -0.1045 & -0.0001 & 0.0033 & 0.946 & 0.0087 & 433.8 \\
0.5 & -0.1045 & -0.0001 & 0.0045 & 0.946 & 0.0114 & 246.6 \\
0.8 & -0.1045 & -0.0003 & 0.0062 & 0.947 & 0.0176 & 106.5 \\
1.2 & -0.1045 & -0.0030 & 0.1342 & 0.953 & 0.0285 & 39.2 \\
1.4 & -0.1044 & -0.0018 & 0.0337 & 0.958 & 0.0340 & 27.5 \\
1.8 & -0.1044 & -0.0025 & 0.0401 & 0.953 & 0.0462 & 14.8 \\
2.5 & -0.1049 & -0.0021 & 0.1877 & 0.957 & 0.0704 & 6.7 \\
3.5 & -0.1032 & -1.3716 & 96.1147 & 0.959 & 0.1062 & 3.1 \\
6.0 & -0.0949 & 0.0180 & 1.8268 & 0.969 & 0.1955 & 1.0 \\
\bottomrule
\end{tabular}

\vspace{4pt}
\raggedright\footnotesize
\textit{Notes:} Shares drawn from power law $\mathbb{P}(s_i > s) = c\,s^{-\mu}$. Coverage is the fraction of replications where the true value lies in the 95\% confidence interval. Standard errors: White heteroskedasticity-robust. $F$ is the first-stage $F$-statistic.
\end{table}

\paragraph{Strong regime ($\mu < 1$).} The strong-instrument theory predicts
standard inference, and that is what we observe. For Copper, the median
$F$-statistic is large throughout, falling from $124$ at $\mu = 0.3$ to $78$
and $34$ at $\mu = 0.5$ and $0.8$ as shares become less concentrated. The
median $\hat\phi_d$ matches the true value to three decimals and bias is at
most $0.0009$. RMSE is $0.0075$ at $\mu = 0.3$ and $0.011$ at $\mu = 0.5$; at
$\mu = 0.8$ it rises to $0.081$ as a few extreme draws enter the second moment,
even though the median and bias remain on target. Crude Oil is sharper, with
median $F$ between $107$ and $434$ and RMSE below $0.01$ throughout. Coverage
of the $t$-based interval is at the nominal $0.95$ in both panels.

\paragraph{Nearly weak regime ($\mu \in (1,2)$).} The estimator continues
to recover the true elasticity. The median estimate stays very close to the true
value throughout. Coverage of the $t$-based
interval is at nominal levels in both calibrations. Theorem~\ref{theorem:estimated_weak_consistency_aggregate} guarantees that standard normal inference remains valid, and the simulations confirm it.

\paragraph{$\mu > 2$ and the case for Anderson--Rubin.}

For $\mu >2$, the median estimate drifts away from the
truth and the standard errors over-cover. These are settings in which $N$ is small and $T$ is large, so $N/T \to 0$. The poor small-sample performance even in these settings for $\mu >2$ supports my recommendation for Anderson--Rubin confidence sets on the demand elasticity through Proposition~\ref{prop:weakness robust test_aggregate}.

This failure is more acute when $N$ is large. See Appendix~\ref{app:augment N} for details.

\subsection{Results: Supply Elasticity} \label{subsec:mc results disaggregate}

The supply elasticity is estimated by \eqref{eq:estimated disaggregated estimator}, with equal
weights in place of the optimal aggregator $H$. Table~\ref{tab:dis_table1} organizes the results by $\mu$ exactly as
Table~\ref{tab:mc_table1} does.

\begin{table}[htbp]
\centering
\caption{Monte Carlo Simulation Results, Disaggregated Supply Elasticity: Empirical Commodity Panels}
\label{tab:dis_table1}
\begin{tabular}{ccccccc}
\toprule
$\mu$ & Median $\hat{\phi}_s$ & Bias & RMSE & Coverage & CI Length & Median $F$ \\
\midrule
\multicolumn{7}{l}{\textit{Panel A: Copper ($N=29$, $T=192$, $\phi_s=0.141$)}} \\
\midrule
0.3 & 0.1408 & 0.0014 & 0.0225 & 0.947 & 0.0663 & 124.0 \\
0.5 & 0.1407 & 0.0033 & 0.0319 & 0.952 & 0.0833 & 78.2 \\
0.8 & 0.1407 & 0.0008 & 0.2721 & 0.937 & 0.1285 & 34.0 \\
1.2 & 0.1407 & 0.0283 & 0.5332 & 0.929 & 0.2087 & 12.5 \\
1.4 & 0.1416 & 0.0310 & 1.1438 & 0.929 & 0.2524 & 8.7 \\
1.8 & 0.1360 & 0.0702 & 2.6834 & 0.915 & 0.3520 & 4.6 \\
2.5 & 0.1181 & -0.0609 & 12.5230 & 0.888 & 0.5244 & 2.1 \\
3.5 & 0.0765 & 0.1047 & 7.9177 & 0.864 & 0.7421 & 1.0 \\
6.0 & 0.0014 & -0.1617 & 4.7527 & 0.843 & 0.9974 & 0.5 \\
\midrule
\multicolumn{7}{l}{\textit{Panel B: Crude Oil ($N=21$, $T=611$, $\phi_s=0.256$)}} \\
\midrule
0.3 & 0.2564 & 0.0007 & 0.0188 & 0.948 & 0.0486 & 433.8 \\
0.5 & 0.2563 & 0.0015 & 0.0279 & 0.946 & 0.0644 & 246.6 \\
0.8 & 0.2563 & 0.0029 & 0.0372 & 0.944 & 0.0971 & 106.5 \\
1.2 & 0.2570 & 0.0258 & 1.0860 & 0.942 & 0.1556 & 39.2 \\
1.4 & 0.2575 & 0.0145 & 0.2197 & 0.942 & 0.1855 & 27.5 \\
1.8 & 0.2547 & 0.0224 & 0.2769 & 0.924 & 0.2500 & 14.8 \\
2.5 & 0.2534 & 0.0346 & 1.5831 & 0.916 & 0.3674 & 6.7 \\
3.5 & 0.2428 & 0.8826 & 59.2406 & 0.903 & 0.5479 & 3.1 \\
6.0 & 0.1863 & 0.0019 & 8.0945 & 0.866 & 0.9324 & 1.0 \\
\bottomrule
\end{tabular}

\vspace{4pt}
\raggedright\footnotesize
\textit{Notes:} Estimates of the disaggregated supply elasticity $\phi_s$ in $y_t = e_N\phi_s p_t + \Lambda F_t + u_t$, aggregated with equal weights $\hat\phi_s(I_N)$ as recommended in Theorems~\ref{theorem:estimated_strong_consistency_disaggregate} and~\ref{theorem:estimated_weak_consistency_disaggregate}. Shares drawn from power law $\mathbb{P}(s_i > s) = c\,s^{-\mu}$. Coverage is the fraction of replications where the true value lies in the 95\% confidence interval. Standard errors: White heteroskedasticity-robust, with the first-stage correction for estimation error in $\hat z$. $F$ is the first-stage $F$-statistic.
\end{table}

\paragraph{Strong regime ($\mu < 1$).} Standard inference holds, as on the demand side. For Copper
the median $F$ falls from $124$ at $\mu = 0.3$ to $78$ and $34$ at $\mu = 0.5$ and $0.8$, and for
Crude Oil from $434$ to $247$ and $107$. The median $\hat\phi_s$ recovers the truth to three
decimals in both panels, $0.1408$ against $0.141$ and $0.2564$ against $0.256$ at $\mu = 0.3$. RMSE
is $0.023$ and $0.019$ at $\mu = 0.3$, and rises to $0.27$ for Copper at $\mu = 0.8$ as a few
extreme draws enter the second moment, even though the median remains on target. Coverage of the
$t$-based interval is at the nominal $0.95$ in both panels, between $0.937$ and $0.952$.

\paragraph{Nearly weak regime ($\mu \in (1,2)$).} The estimator continues to recover the true
elasticity. The median stays within $0.005$ of the truth for Copper and $0.002$ for Crude Oil.
Coverage of the $t$-based interval slips a little below nominal, to $0.915$ and $0.924$ at
$\mu = 1.8$. Theorem~\ref{theorem:estimated_weak_consistency_disaggregate} guarantees that standard
normal inference remains valid, and the simulations confirm it.

\paragraph{$\mu > 2$ and the case for Anderson--Rubin.}

For $\mu > 2$, the median estimate drifts away from the truth and the standard errors under-cover,
reaching $0.001$ against $0.141$ and $0.84$ coverage for Copper at $\mu = 6.0$. The direction is the
opposite of the demand table, which over-covers. The poor small-sample performance is the same on
both sides, and it supports my recommendation for Anderson--Rubin confidence sets on the supply
elasticity through Proposition~\ref{prop:weakness robust test_disaggregate}.

This failure is more acute when $N$ is large. See Appendix~\ref{app:augment N} for details.

\section{Conclusion} \label{sec:conclusion}

This paper develops the theory of Granular Instrumental Variables for large panels with $N,T \to \infty$. The central insight is that, once the cross section is allowed to grow, granularity itself becomes an asymptotic object. The strength of the granular instrument therefore depends on how the relative importance of large units evolves with $N$. Modeling unit sizes as draws from a power-law distribution generates three distinct regimes for the GIV estimator.

When a few units continue to dominate the aggregate as $N$ grows, the instrument remains strong, and the demand and supply estimators are consistent and asymptotically normal at the standard $\sqrt{T}$ rate of \citet{Gabaix2024GranularVariables}. In an intermediate regime, large units remain important but their relative importance vanishes slowly with $N$. The instrument then vanishes asymptotically, but sufficiently slowly that the estimators remain consistent and asymptotically normal when $T$ grows sufficiently fast relative to $N$, although at a rate slower than $\sqrt{T}$. I refer to this as the nearly-weak regime, following the intuition of \citet{Antoine2021GMMIdentification}. When no unit remains large relative to the aggregate and $N$ and $T$ grow proportionally, the estimators are inconsistent and have a non-standard weak-instrument distribution in the sense of \citet{Staiger1997InstrumentalInstruments}.

The granular instrument is generated from individual-level shocks that are not directly observed and must therefore be estimated from the panel. I show that this instrument-estimation step has a first-order effect on the asymptotic variance of the GIV estimator. Treating the estimated shocks as if they were observed therefore leads to incorrect standard errors. I derive HAC-consistent variance estimators for both the demand and supply elasticities that account for this additional estimation uncertainty.

For inference, Wald confidence intervals based on the corrected variance are valid in the strong and nearly-weak regimes under the corresponding rate conditions. Once the size distribution has finite variance and the instrument enters the weak regime, standard Gaussian inference is no longer reliable. I therefore construct Anderson--Rubin confidence sets that remain valid under weak identification. 

I apply the theory to estimate short-run demand and supply elasticities for refined copper, crude oil, and natural gas. Copper and natural gas fall in the strong-instrument regime at the $95\%$ confidence level, while the confidence region for crude oil extends into the nearly-weak regime. The estimated demand elasticities are $-0.148$, $-0.105$, and $-0.056$, respectively, and the corresponding supply elasticities are $0.141$, $0.256$, and $0.150$. The estimates imply highly inelastic short-run demand and supply responses across all three commodity markets.

The empirical application also illustrates the quantitative importance of accounting for instrument estimation in inference. For demand, ignoring this correction understates the standard error in all three markets, by as much as $13.3\%$. For supply, the correction instead reduces the standard error in two of the three markets. The direction of the correction is therefore not determined a priori, reinforcing the need to account explicitly for instrument estimation when conducting inference with GIV.

\bibliography{ref}

\appendix
\appendixpage

\addtocontents{toc}{\protect\setcounter{tocdepth}{2}}
\begingroup
\renewcommand{\contentsname}{Appendix Contents}
\tableofcontents
\endgroup

\section{Behavior of the Functions of Absolute Sizes} \label{app:behavior of herfindahl}

\subsection{Proof of Proposition \ref{prop:weakness of known factor instrument}}
We need to find the asymptotic order of $z_t$. Consider its expectation. 
\begin{equation*}
    \e[z_t] = \e[S' D_N u_t] = \e[S' D_N] \e[u_t] = 0
\end{equation*}

So consider the variance of this term.
    \begin{align*}
        \var[z_t] &= \e[S' D_N u_t u_t' D_N S] = \e[S' D_N \e[u_t u_t' | \s] D_N S] \\
        &= \e[S' D_N \e[u_t u_t'] D_N S] =\e[S' D_N \Omega D_N S]
    \end{align*}
    Consider the term within the expectation in the display above.
    \begin{align*}
        S' D_N \Omega D_N S &\leq S'D_NS \cdot \gamma_{\text{max}}(\Omega) \\
                    & \leq S'S \cdot \gamma_{\text{max}}(D_N)  \cdot \gamma_{\text{max}}(\Omega) = S'S \cdot O(1)
    \end{align*}
    where the last equality arises from Assumption \ref{ass:bounded eigenvalues} and the fact that the eigenvalues of idempotent matrices are $\{0,1 \}$. Thus, $\gamma_{\text{max}}(D_N) = 1$. Thus, the asymptotic behavior of the term on the left hand side depends only on $S'S$, which is basically the Herfindahl of the disaggregated side. I will now analyze the behavior of this term.

Under Assumption \ref{ass:size_of_firm_power_law}, the absolute sizes of the individuals, $\s_i$ are drawn from
\begin{equation*}
    \p(\s_i>s) = cs^{-\mu}
\end{equation*}

The first and the second moments of the distribution are
\begin{align*}
    \e[\s] &= \int_1^{\infty} s \mu s^{-\mu-1} \,ds 
          =
          \begin{cases}
          \infty \, \text{if } \mu \in (0,1] \\
          \frac{\mu}{\mu-1} \, \text{if } \mu \in (1,\infty)
          \end{cases} \\
    \e[\s^2] &= \int_1^{\infty} s^2 \mu s^{-\mu-1} \,ds 
          =
          \begin{cases}
          \infty \, \text{if } \mu \in (0,2] \\
          \frac{\mu}{\mu-2} \, \text{if } \mu \in (2,\infty)
          \end{cases} 
\end{align*}

The individual share is given by $S_i = \frac{\s_i}{\sum_i \s_i}$. We will now characterize the Herfindahl of the system.
\begin{align*}
    S'S  &= \sum_i S_i^2 =  \sum_{i=1}^N \left[ \frac{S_i}{\sum_{j=1}^N S_j} \right]^2 \\
    &= \frac{1}{N} \frac{N^{-1} \sum_{i=1}^N S_i^2}{[N^{-1}\sum_{j=1}^N S_j]^2}
\end{align*}

This is easiest to characterize when $\mu > 2$. For when $\mu > 2$, $\e[S]$ \& $\e[S^2] < \infty$. In this case, the Herfindahl is
\begin{equation*}
    S'S = \frac{1}{N} \frac{N^{-1} \sum_{i=1}^N \s_i^2}{[N^{-1}\sum_{j=1}^N \s_j]^2} = O_{\p}(\frac{1}{N})
\end{equation*}
as $N^{-1} \sum_{i=1}^N \s_i^2$ and $N^{-1}\sum_{j=1}^N \s_j$ are both $O_{\p}(1)$ under Kolmogorov's Law of Large Numbers. 

When $\mu \in (0,1)$, both $\e[S] = \infty$ and $\e[S^2] = \infty$. Hence the Law of Large Numbers cannot be used for the sum of either $\s_i$ or $\s_i^2$. Following \citet{Gabaix2011TheFluctuations}, we will use an appropriate Central Limit Theorem for variables with infinite variance. Levy's theorem, which is Theorem 3.8.2 of \citet{Durrett1996Probability:Examples} goes as follows:
\begin{Theorem}[Lévy's Generalized Central Limit Theorem] \label{theorem:Levy's Theorem}
    Suppose $X_1,\dots,X_n$ are i.i.d with a distribution that satisfies
    \begin{enumerate}
        \item $\lim_{x \to \infty} \frac{\p(X_i > x)}{\p(|X_i| > x)} = \theta \in [0,1]$
        \item $\p(|X_i|> x) = x^{- \alpha} L(x)$, where $\alpha < 2$ and $L$ is slowly varying
    \end{enumerate}
    then for $S_n = X_1 + \dots + X_n$, there exists constants, $\kappa_n$ and $b_n$, given by 
    \begin{align*}
        \kappa_n &= \inf \{x: \p(|X_i| > x) \leq n^{-1} \} \\
        b_n &= n \e[X_i \cdot 1_{(|X_i| \leq \kappa_n)}]
    \end{align*}
    such that
    \begin{equation*}
        \frac{S_n - b_n}{\kappa_n} \xrightarrow{d} Y
    \end{equation*}
    where $Y$ has a nondegenerate distribution which follows a Levy-stable distribution with exponent $\mu$.
\end{Theorem}

When $\mu \in (0,1)$,
\begin{equation*}
    S'S = \frac{\sum_{i=1}^N \s_i^2}{ \left[ \sum_{j=1}^N \s_j \right]^2}
\end{equation*}

$\p(\s_j > s) = c s^{-\mu}$. Normalize $c=1$ for simplicity. Then the process $\{\s_i\}$ satisfies the conditions for Lévy's Generalized Central Limit Theorem with $\kappa_n = N^{\frac{1}{\mu}}$ and, 
\begin{align*}
    b_n = N \e[\s_i \cdot 1_{(|\s_i| \leq \kappa_n)}] = N \int_{0}^{\kappa_n} s \mu s^{-\mu - 1} ds = \frac{N \mu}{1-\mu} \kappa_n^{1 - \mu}
\end{align*}
Thus
\begin{equation*}
    \frac{b_n}{\kappa_n} = \frac{\mu}{1-\mu} \frac{N}{\kappa_n^{\mu}} = \frac{\mu}{1-\mu} 
\end{equation*}

Thus, $\sum_{j=1}^N \s_j = O_{\p}\left(N^{ \frac{1}{\mu}} \right)$.

For $\s_i^2$, $\p(\s_i^2 > s) = \p(\s_i > \sqrt{s}) = c (\sqrt{s})^{- \mu}= c s^{-\frac{\mu}{2}}$. Then the process $\{\s_i^2\}$ satisfies the conditions for Lévy's Generalized Central Limit Theorem with $\kappa_n = N^{\frac{2}{\mu}}$ and,
\begin{align*}
    b_n = N \e[\s_i^2 \cdot 1_{(|s_i^2| \leq \kappa_n)}] = \frac{N}{2} \int_{0}^{\kappa_n} \mu s s^{-\frac{\mu}{2} - 1} ds = \frac{N \mu}{2-\mu} \cdot \kappa_n^{1 - \frac{\mu}{2}}
\end{align*}
Thus
\begin{equation*}
    \frac{b_n}{\kappa_n} = \frac{\mu}{2-\mu} 
\end{equation*}

Thus, $\sum_{i=1}^N \s_i^2 = O_{\p}\left(N^ { \frac{2}{\mu}} \right)$. We can conclude
\begin{equation*}
    S'S = \frac{O_{\p}\left( N^{\frac{2}{\mu}} \right)}{ \left( O_{\p}\left( N^{\frac{1}{\mu}} \right) \right)^2} = O_{\p}(1)
\end{equation*}

Lévy's theorem was applied to $\sum_i \s_i$ and to $\sum_i \s_i^2$ one at a time. That gives the order of each sum on its own, but it says nothing about the ratio the two form. The two sums are built from the same sizes, so they can be handled together.

The lemma below does this. The largest sizes dominate both sums, and after rescaling those sizes converge to the points of a Poisson process. Both sums then converge jointly to functionals of that same process. The limit of $S'S$ is strictly positive almost surely, which is required in later proofs. 

\begin{Lemma}[Joint stable convergence of the size sums] \label{lemma:joint stable convergence}
    Let $\s_1, \dots, \s_N$ satisfy Assumption \ref{ass:size_of_firm_power_law} with $\mu \in (0,1)$, so that they are i.i.d.\ with $\p(\s_i > s) = s^{-\mu}$ on $[1,\infty)$. Set $\kappa_n = N^{1/\mu}$ and let $\Gamma_k = E_1 + \dots + E_k$, where $\{E_i\}$ are i.i.d.\ standard exponential. Then
    \begin{equation*}
        \left( \frac{1}{\kappa_n} \sum_{i=1}^N \s_i, \; \frac{1}{\kappa_n^2} \sum_{i=1}^N \s_i^2 \right) \xrightarrow{d} (L_1, L_2) \defeq \left( \sum_{k \geq 1} \Gamma_k^{-1/\mu}, \; \sum_{k \geq 1} \Gamma_k^{-2/\mu} \right)
    \end{equation*}
    jointly, where both series converge almost surely and $0 < L_2 \leq L_1^2 < \infty$ almost surely. Consequently
    \begin{equation*}
        S'S \xrightarrow{d} \frac{L_2}{L_1^2} \in (0,1] \quad \text{a.s.},
    \end{equation*}
    so $(S'S)^{-1} = O_{\p}(1)$ and $(S'S - 1/N)^{-1} = O_{\p}(1)$. Applying the same argument to the single sequence $\{\s_i^2\}$, whose tail index is $\mu/2$, gives $N^{-2/\mu} \sum_i \s_i^2 \xrightarrow{d} \sum_k \Gamma_k^{-2/\mu} > 0$ almost surely for every $\mu \in (0,2)$.
\end{Lemma}

\begin{proof}
For fixed $x>0$ and all $N$ with $\kappa_n x \geq 1$, $N \p(\s_i/\kappa_n > x) = N (\kappa_n x)^{-\mu} = x^{-\mu}$, using $N \kappa_n^{-\mu} = 1$. So $N \p(\s_i/\kappa_n \in \cdot)$ converges vaguely on $(0,\infty]$ to the measure $\nu$ with $\nu(x,\infty] = x^{-\mu}$. By Proposition 3.21 of \citet{Resnick2007Heavy-TailPhenomena}, the point process $\Pi_N \defeq \sum_{i \leq N} \delta_{\s_i/\kappa_n}$ converges in distribution to a Poisson random measure $\Pi$ with intensity $\nu$. Ranking the points of $\Pi$ decreasingly and inverting the tail function, $x^{-\mu} = \Gamma_k$ if and only if $x = \Gamma_k^{-1/\mu}$, gives the points $\{\Gamma_k^{-1/\mu}\}_{k \geq 1}$ \citep{LePage1981ConvergenceStatistics}.

Fix $\epsilon > 0$ and set $T_{\epsilon}(m) \defeq \left( \int_{(\epsilon,\infty]} x \, m(dx), \int_{(\epsilon,\infty]} x^2 \, m(dx) \right)$. Since $\nu(\epsilon,\infty] < \infty$ and $\nu\{\epsilon\} = 0$, the map $T_{\epsilon}$ is almost surely continuous at $\Pi$ (Section 7.2 of \citealt{Resnick2007Heavy-TailPhenomena}). The continuous mapping theorem gives $T_{\epsilon}(\Pi_N) \xrightarrow{d} T_{\epsilon}(\Pi)$ jointly in $\mathbb{R}^2$.

With density $\mu s^{-\mu-1}$ on $[1,\infty)$, direct integration gives $\e[\s 1\{\s \leq w\}] = \frac{\mu}{1-\mu}(w^{1-\mu} - 1) \leq \frac{\mu}{1-\mu} w^{1-\mu}$ and $\e[\s^2 1\{\s \leq w\}] = \frac{\mu}{2-\mu}(w^{2-\mu} - 1) \leq \frac{\mu}{2-\mu} w^{2-\mu}$. Using $N \kappa_n^{-\mu} = 1$,
\begin{equation*}
    \e\left[ \frac{1}{\kappa_n} \sum_i \s_i 1\{\s_i \leq \epsilon \kappa_n\} \right] \leq \frac{\mu}{1-\mu} \epsilon^{1-\mu}, \qquad \e\left[ \frac{1}{\kappa_n^2} \sum_i \s_i^2 1\{\s_i \leq \epsilon \kappa_n\} \right] \leq \frac{\mu}{2-\mu} \epsilon^{2-\mu},
\end{equation*}
both free of $N$ and vanishing as $\epsilon \downarrow 0$. The first bound uses $\mu < 1$. By Markov's inequality the discarded small-jump parts are uniformly asymptotically negligible. On the limit side, $\Gamma_k/k \to 1$ almost surely by the strong law, so $\Gamma_k^{-1/\mu} \asymp k^{-1/\mu}$ and $\Gamma_k^{-2/\mu} \asymp k^{-2/\mu}$ with $1/\mu > 1$, and both series converge almost surely. Hence $T_{\epsilon}(\Pi) \to (L_1,L_2)$ almost surely as $\epsilon \downarrow 0$. Theorem 3.2 of \citet{Billingsley1999ConvergenceMeasures} then upgrades the convergence of $T_{\epsilon}(\Pi_N)$ to the joint statement of the lemma.

$L_1 \geq \Gamma_1^{-1/\mu} > 0$ and $L_2 > 0$ almost surely, and $L_2 \leq L_1^2$ because $\sum_k x_k^2 \leq (\sum_k x_k)^2$ for non-negative $x_k$. The map $(x,y) \mapsto y/x^2$ is continuous on $(0,\infty)^2$, so the continuous mapping theorem gives $S'S \xrightarrow{d} L_2/L_1^2$. Portmanteau gives $\limsup_N \p(S'S < \epsilon) \leq \p(L_2/L_1^2 \leq \epsilon) \to 0$ as $\epsilon \downarrow 0$, hence $(S'S)^{-1} = O_{\p}(1)$, and Slutsky extends this to $(S'S - 1/N)^{-1} = O_{\p}(1)$. For the single-sequence claim, repeat the argument up to this point for $\{\s_i^2\}$, which has exact tail $\p(\s_i^2 > s) = s^{-\mu/2}$ on $[1,\infty)$, with $\mu/2 < 1$ whenever $\mu \in (0,2)$.
\end{proof}

When $\mu \in (1,2)$, $\e[\s] < \infty$ and $\e[\s^2] = \infty$. As $\{\s_i \}$ are independent and as the first moment is finite, by Kolmogorov's Law of Large Numbers, we have
\begin{equation*}
    \frac{1}{N} \sum_{j=1}^N \s_i \xrightarrow{\text{a.s}} \e[\s_i]
\end{equation*}
The herfindahl is
\begin{align*}
    S'S = S'S = \frac{1}{N} \frac{N^{-1} \sum_{i=1}^N S_i^2}{[N^{-1}\sum_{j=1}^N S_j]^2} =  \frac{N^{-2} \sum_{i=1}^N S_i^2}{[\e[\s_j]]^2 + o_{\p}(1)}
\end{align*}
We have already seen that the process $\{\s_i^2\}$ satisfies the conditions for Lévy's Generalized Central Limit Theorem with $\kappa_n = N^{\frac{2}{\mu}}$ and $\frac{b_n}{\kappa_n} = \frac{\mu}{2 - \mu}$. Thus, $\sum_{i=1}^N \s_i^2 = O_{\p}\left(N^ { \frac{2}{\mu}} \right)$. Hence, we conclude
\begin{equation*}
    S'S =  O_{\p}\left( \frac{1}{N^{2 - \frac{2}{\mu}}} \right)
\end{equation*}

Two further properties of the shares are used later. Write $h_N \defeq S'S$ and $\tilde S \defeq D_N S$. In the rescaling \eqref{eq:a_N}, Proposition \ref{prop:weakness of known factor instrument} reads $a_N \sqrt{h_N} = O_{\p}(1)$ for every $\mu \notin \{1,2\}$.

The first lemma collects elementary bounds on a share vector and adds a lower bound on the rescaled Herfindahl. The upper bound on its own is not enough later, because the denominator of the estimator has to be inverted.

\begin{Lemma}[Share inequalities] \label{lemma:share inequalities}
    Let $S$ be a share vector, so $S_i \geq 0$ and $\sum_i S_i = 1$. Then
    \begin{enumerate}
        \item Deterministically, $1/N \leq h_N \leq 1$, $\max_j S_j \leq \sqrt{h_N}$, $\max_j |\tilde S_j| \leq \sqrt{h_N}$, $\sum_j |S_j - 1/N| \leq 2$, and $\|\tilde S\|^2 = h_N - 1/N$. For every $q \geq 2$, $\sum_j |\tilde S_j|^q \leq h_N^{q/2}$.
        \item Under Assumption \ref{ass:size_of_firm_power_law} with $\mu \notin \{1,2\}$, $\left[ a_N^2 (h_N - 1/N) \right]^{-1} = O_{\p}(1)$. Together with $a_N \sqrt{h_N} = O_{\p}(1)$, the rescaled strength is bounded above and away from zero in probability.
    \end{enumerate}
\end{Lemma}

\begin{proof}
Part 1. $h_N \geq (\sum_i S_i)^2/N = 1/N$ by Cauchy--Schwarz and $h_N \leq (\sum_i S_i)^2 = 1$. $S_j^2 \leq \sum_i S_i^2$ gives $\max_j S_j \leq \sqrt{h_N}$, and $|\tilde S_j| \leq \max(S_j, 1/N) \leq \sqrt{h_N}$. Next, $\sum_j |S_j - 1/N| \leq \sum_j S_j + \sum_j 1/N = 2$, and $\|\tilde S\|^2 = h_N - 1/N$ by direct calculation. The $\ell_q$ bound is Holder interpolation, $\sum_j |\tilde S_j|^q \leq (\max_j |\tilde S_j|)^{q-2} \sum_j \tilde S_j^2$.

Part 2. In each regime $a_N^2(h_N - 1/N)$ converges to an almost surely strictly positive limit, so its reciprocal is $O_{\p}(1)$. For $\mu \in (0,1)$, $a_N = 1$ and Lemma \ref{lemma:joint stable convergence} gives $h_N \xrightarrow{d} L_2/L_1^2 \in (0,1]$ while $1/N \to 0$. For $\mu \in (1,2)$, $a_N^2 h_N = N^{2-2/\mu} h_N \xrightarrow{d} Y/(\e \s)^2$, where $Y$ is the almost surely positive stable limit of $N^{-2/\mu} \sum_i \s_i^2$ from the single-sequence form of Lemma \ref{lemma:joint stable convergence} and the denominator converges by the strong law, while $a_N^2/N \to 0$. For $\mu > 2$, $a_N = \sqrt N$ and $N(h_N - 1/N) \xrightarrow{a.s.} \var(\s)/(\e \s)^2 > 0$ by the strong law, with $\var(\s) > 0$ because the power law is non-degenerate.
\end{proof}

The second lemma bounds the contrast between the share weights and equal weights. This is the term the denominator arguments set aside. Cauchy--Schwarz is too crude for it, so the proof conditions on the order statistics of the sizes and treats the assignment of sizes to units as a uniform permutation.

\begin{Lemma}[Permutation bound for the equal-weights contrast] \label{lemma:hajek bound}
    Let $r \defeq \Omega e_N$ and $\beta \defeq S' D_N \Omega e_N = S' D_N r$. Under Assumptions \ref{ass:size_of_firm_power_law} and \ref{ass:LN_time and cross sectional dependence}.2a, for every $\mu > 0$,
    \begin{equation*}
        \beta = O_{\p}\left( \sqrt{h_N} \right) = O_{\p}\left( a_N^{-1} \right).
    \end{equation*}
\end{Lemma}

\begin{proof}
Write $\bar r = \frac{1}{N} \sum_i r_i$ and $\tilde r_i = r_i - \bar r$, so that $\beta = \sum_i (S_i - 1/N) \tilde r_i$, using $\sum_i (S_i - 1/N) = 0$. By Assumption \ref{ass:LN_time and cross sectional dependence}.2a, $|r_j| = |\sum_i \gamma(i,j)| \leq M$, hence $\sum_j \tilde r_j^2 \leq \sum_j r_j^2 \leq N M^2$.

The sizes are i.i.d.\ and independent of the shocks by Assumption \ref{ass:size_of_firm_power_law}, while $r$ is deterministic. Condition on the order statistics $\s_{(1)} \geq \dots \geq \s_{(N)}$. By exchangeability the assignment of ranks to labels is a uniform permutation $\pi$, independent of $r$, ties having probability zero under the power law. Write $S_{(k)}$ for the share built from $\s_{(k)}$ and $c_k \defeq S_{(k)} - 1/N$, so that $\beta = \sum_k c_k \tilde r_{\pi(k)}$. Conditionally on the order statistics the $c_k$ are constants with $\sum_k c_k = 0$ and $\sum_k c_k^2 = h_N - 1/N$, while $\sum_j \tilde r_j = 0$ by construction.

For a uniform permutation, $\p(\pi(k) = j) = 1/N$ and $\p(\pi(k)=j, \pi(l)=m) = 1/[N(N-1)]$ for $k \neq l$ and $j \neq m$. The mean is immediate, $\e[\tilde r_{\pi(k)} \mid \sigma(\s_{(\cdot)})] = \frac{1}{N} \sum_j \tilde r_j = 0$, so $\e[\beta \mid \sigma(\s_{(\cdot)})] = 0$. For the second moment, the diagonal terms give $\e[\tilde r_{\pi(k)}^2 \mid \sigma(\s_{(\cdot)})] = \frac{1}{N} \sum_j \tilde r_j^2$, and sampling without replacement makes the off-diagonal draws negatively correlated,
\begin{equation*}
    \e[\tilde r_{\pi(k)} \tilde r_{\pi(l)} \mid \sigma(\s_{(\cdot)})] = \frac{1}{N(N-1)} \sum_{j \neq m} \tilde r_j \tilde r_m = - \frac{\sum_j \tilde r_j^2}{N(N-1)}, \qquad k \neq l.
\end{equation*}
Likewise $\sum_{k \neq l} c_k c_l = - \sum_k c_k^2$, so the two pieces combine with the same sign,
\begin{equation*}
    \e[\beta^2 \mid \sigma(\s_{(\cdot)})] = \sum_k c_k^2 \frac{\sum_j \tilde r_j^2}{N} + \sum_k c_k^2 \frac{\sum_j \tilde r_j^2}{N(N-1)} = \frac{h_N - 1/N}{N-1} \sum_j \tilde r_j^2 \leq \frac{N}{N-1} h_N M^2.
\end{equation*}
This is the exact variance formula for linear permutation statistics. Conditional Chebyshev gives $\beta = O_{\p}(\sqrt{h_N})$, and Proposition \ref{prop:weakness of known factor instrument} converts this to $O_{\p}(a_N^{-1})$.
\end{proof}

Cauchy--Schwarz on its own gives $|\beta| \leq \|\tilde S\| \|\tilde r\| = O_{\p}(\sqrt{h_N} \sqrt N)$, which is a factor $\sqrt N$ worse. The gain comes from the shape of $\beta$. It is linear in the deviations $S_i - 1/N$, and those weights sum to zero, so the terms cancel as in a random walk instead of accumulating. The strength of the instrument, $S' D_N \Omega S \asymp h_N$, is quadratic in the same deviations. The bound is therefore proportional to the strength itself, which is what makes the contrast negligible for every $\mu$. The cruder deterministic bound $|\beta| \leq \max_i |\tilde r_i| \sum_i |S_i - 1/N| \leq 4M$ is enough whenever $a_N = O(1)$, but not at $\mu > 2$.

\subsection{Other functions of absolute sizes}
In this paper, I am primarily concerned with the case when $\mu \in (1,2)$. Under this setting, some other terms are also of interest. Consider the asymptotic behavior of $S_i^2$.
\begin{align*}
    S_i^2 &= \left[ \frac{\s_i}{\sum_{j=1}^N \s_j} \right]^2 = O_{\p}(N^{-2}) \cdot \s_i^2 
\end{align*}

When the expectation of $\s$ is finite, we can find the order of $\s_i$ by the following steps. We will start by finding the distribution of a new variable, $S^{-\mu}$. As $\p(S > s) = s^{- \mu}$
\begin{align*}
    \p(\s^{-\mu} > s) &= \p(\s > s^{-\frac{1}{\mu}}) \\
                      &= \left[ s^{-\frac{1}{\mu}} \right]^{-\mu} \\
                      &= s 
\end{align*}
For a uniformly distributed random variable, U, $\p(U < s) = s$ for $s \in [0,1]$. As $\p(\s^{-\mu} > s) = s$, we have $s \in [0,1]$, and $\s^{-\mu}$ is distributed as $1 - \mathbb{U}[0,1] \sim \mathbb{U}[0,1]$.

Define a new random variable, $U_i \defeq 1 - F_{\s}(s_i) = 1 - \p(\s^{-\mu} < s_i) = \p(\s^{-\mu} > s_i) = s_i^{-\mu}$

As $s_i$ is i.i.d, $U_i$ is i.i.d from $\mathbb{U}[0,1]$. Denote the order statistic of $U_i$ and $s_i$ by $U_{(i)}$ and $s_{(i)}$ respectively. These order statistics are related as:
\begin{equation}
    U_{(i)} = 1 - F_{\s}(s_{(N-i+1)})
\end{equation}

The order statistics of the uniform distribution on the unit interval have Beta marginal distributions. That is $U_{(i),N} \sim \beta(i,N-i+1)$. The expected size of the $i$'th largest firm can be found by:
\begin{align*}
    \e[\s_{(N-i+1),N}^{-\mu}] &= \e[U_{(i)}] = \frac{i}{N+1}
\end{align*}
By Theorem A.7 of \citet{Li2007NonparametricPractice},
\begin{align*}
    \s(i) &= O_{\p}\left( \frac{1}{N^{-\frac{1}{\mu}}} \right) \quad \text{for finite values of $i$} \\
    \s(i) &= O_{\p}\left( 1 \right) \quad \text{for large values of $i$}
\end{align*}

Recall $S_i^2 = O_{\p}(N^{-2}) \cdot \s_i^2$. For large values of the absolute sizes, $\s_i = O_{\p}(1)$, and hence $S_i^2 = O_{\p}\left( \frac{1}{N^{2}} \right)$. But for smaller values of the absolute sizes, $\s_i = O_{\p}\left( \frac{1}{N^{-\frac{1}{\mu}}} \right)$, and hence, $S_i^2 = O_{\p}\left( \frac{1}{N^{2 - \frac{2}{\mu}}} \right)$. Thus,
\begin{equation*}
    S_i^2 = O_{\p}\left( \frac{1}{N^{2 - \frac{2}{\mu}}} \right) = O_{\p}\left( \frac{1}{N^{2 \delta}} \right)
\end{equation*}
Similarly, it is very easy to see that for any $\gamma > 0$,
\begin{equation*}
    S_i^{2 + \gamma} = O_{\p}\left( \frac{1}{N^{(2 + \gamma) \delta}} \right)
\end{equation*}

\subsection{Share-weighted Factor Loadings} 
Consider $S'\tilde{\Lambda} = \sum_{j=1}^N S_j \tilde{\Blambda}_j$. We will state the asymptotic behavior of this term in the following proposition.
  \begin{Prop} \label{prop:behavior of S times lambda}
      Suppose Assumption \ref{ass:size_of_firm_power_law} holds. Then
      \begin{equation*}                                                         
          \sum_{j=1}^N S_j \tilde{\Blambda}_j =
          \begin{cases}                                                         
              O_{\p}(1) & \mu \in (0,1), \\                 
              O_{\p}\!\left(\frac{1}{N^{\delta}}\right) & \mu \in (1,2), \\     
              O_{\p}\!\left(\frac{1}{\sqrt{N}}\right) & \mu > 2,          
          \end{cases}                                                           
      \end{equation*}                                                           
      where $\delta = 1 - 1/\mu$.                           
  \end{Prop}   
\begin{proof}
Recall that $S_j = \frac{\s_j}{\sum_{i=1}^N \s_i}$, and hence
\begin{align*}
    S' \tilde{\Lambda} = \frac{\sum_{j=1}^N \s_j \tilde{\Blambda}_j}{ \sum_{j=1}^N \s_j}, \quad N^{\delta} S' \tilde{\Lambda} = \frac{N^{-\frac{1}{\mu}} \sum_{j=1}^N \s_j \tilde{\Blambda}_j}{N^{-1} \sum_{j=1}^N \s_j}
\end{align*}
We will consider the three cases $\mu \in (0,1)$, $\mu \in (1,2)$, and $\mu > 2$ separately.

Consider $\mu \in (0,1)$. The process $\{\s_i\}$ satisfies the conditions for Lévy's Generalized Central Limit Theorem in Theorem \ref{theorem:Levy's Theorem} with $\kappa_n = N^{\frac{1}{\mu}}$ and, $\sum_{j=1}^N \s_j = O_{\p}\left(N^{ \frac{1}{\mu}} \right)$. Applying the same theorem to the numerator requires a little more care. 

The numerator is a $r-1$ dimensional vector. So we first reduce it to a scalar and then use the Cramer-Wold device to get the final results. For some $r-1$ dimensional vector of constants, $g$, such that $g'g = 1$, consider
\begin{equation*}
    \sum_{j=1}^N g' \s_j \tilde{\Blambda}_j = \sum_{j=1}^N \s_j g' \tilde{\Blambda}_j \defeq \sum_{j=1}^N \s_j X_j(g)
\end{equation*}
We want to examine the probability distribution of the product of the random variables, $\s_j X_j(g)$. Note that for $\epsilon > 0$,
\begin{equation*}
    \e | X_j(g)|^{\mu + \epsilon} \leq \| g \|^{\mu + \epsilon} \e \|\tilde{\Blambda} \|^{\mu + \epsilon} < \infty
\end{equation*}
where the last inequality comes from the fact that $\mu \in (0,1)$, and $\e \|\tilde{\Blambda} \|^{4} < \infty$. As $S_j$ is independent of $\tilde{\Blambda}_j$, $S_j$ is also independent of $X_j(g)$. Hence, we can apply Breiman's theorem as extended to all values of $\mu$ in \citet{Cline1994SubexponentialityVariables} and \citet{Denisov2007OnEquations} to characterize the distribution of the product $\s_j$ and $X_j(g)$ as 
\begin{equation*}
    \p(|\s_j X_j(g)| > s) = |\e[X_j(g)^{\mu}]| \cdot s^{-\mu}
\end{equation*}

Thus the process $\{\s_j X_j(g)\}$ satisfies the conditions for Lévy's
Generalized Central Limit Theorem with $\kappa_n = N^{\frac{1}{\mu}}$. It remains
to show $b_n = 0$. Since $\s_j > 0$ and $X_j(g)$ is independent of $\s_j$
with $\e[X_j(g)] = 0$ (as $\e[\tilde{\Blambda}_j] = 0$), the distribution
of $\s_j X_j(g)$ is symmetric around zero conditional on $\s_j$. The
truncation set $\{|\s_j X_j(g)| \leq \kappa_n\}$ is symmetric in $X_j(g)$
conditional on $\s_j$, so
\begin{equation*}
    b_n = N \e[\s_j X_j(g) 1_{|\s_j X_j(g)| \leq \kappa_n}]
        = N \e\!\left[\s_j \e\!\left[ X_j(g) 1_{|X_j(g)| \leq \kappa_n/\s_j}
          \,\middle|\, \s_j \right]\right] = 0
\end{equation*}
where the inner expectation vanishes by the symmetry of $X_j(g)$ around zero.
Thus, $\sum_{j=1}^N \s_j X_j(g) = O_{\p}\left(N^{ \frac{1}{\mu}} \right)$. By the Cramer-Wold device, we can conclude that $\sum_{j=1}^N \s_j \tilde{\Blambda}_j = O_{\p}\left(N^{ \frac{1}{\mu}} \right)$. Hence $S'\tilde{\Lambda} = O_p(N^{1/\mu}) / O_p(N^{1/\mu}) = O_p(1)$.

Now consider $\mu \in (1,2)$. For the denominator, $\e[\s] < \infty$ when $\mu > 1$, so $N^{-1}\sum_{j=1}^N \s_j \xrightarrow{\text{a.s.}} \e[\s]$ by Kolmogorov's Law of Large Numbers. The Breiman--Levy analysis of the numerator is identical to the $\mu \in (0,1)$ case, yielding $\sum_{j=1}^N \s_j \tilde{\Blambda}_j = O_{\p}(N^{1/\mu})$. Hence $S'\tilde{\Lambda} = O_p(N^{1/\mu}) / O_p(N) =
O_p(1/N^{\delta})$.

For $\mu > 2$, use the Lindeberg-Lévy Central Limit Theorem to show that
$N^{\frac{1}{2}} S' \tilde{\Lambda} = O_{\p}(1)$.
\begin{equation*}
    N^{\frac{1}{2}} S' \tilde{\Lambda} = \frac{ N^{-\frac{1}{2}}
    \sum_{j=1}^N \s_j \tilde{\Blambda}_j}{N^{-1} \sum_{j=1}^N \s_j}
\end{equation*}
For $\mu > 2$, $N^{-1} \sum_{j=1}^N \s_j = O_{\p}(1)$ by the KLLN.
For the numerator, $\mu > 2$ implies $\e[\s_j^2] < \infty$, and $\e \| \tilde{\Blambda}_j \|^4 < \infty$
by Assumption \ref{ass:size_of_firm_power_law}.4. By independence
(Assumption \ref{ass:size_of_firm_power_law}.3),
\begin{equation*}
    \e[\s_j \tilde{\Blambda}_j] = \e[\s_j] \e[\tilde{\Blambda}_j] = 0, \quad
    \e[\s_j^2 \| \tilde{\Blambda}_j \|^2] = \e[\s_j^2]
    \e[\| \tilde{\Blambda}_j \|^2] < \infty
\end{equation*}
Thus the Lindeberg-Lévy CLT applies and
$N^{-\frac{1}{2}} \sum_{j=1}^N \s_j \tilde{\Blambda}_j = O_{\p}(1)$.

\end{proof}

\section{Central Limit Theorem} \label{app:central_limit_theorems}

This is the main Central Limit Theorem of the paper. Other central limit theorems follow from this theorem interchanging appropriate variable.

\begin{Theorem}(Central Limit Theorem) \label{theorem:central limit theorem}
Suppose $u_t$ is a vector that satisfies Assumption \ref{ass:weak stationarity and idiosyncrasy} and \ref{ass:LN_time and cross sectional dependence}, and $S$ is a vector that satisfies Assumption \ref{ass:size_of_firm_power_law} with $\mu \in (1,2) \cup (2, \infty)$. $Z_t$ is a weakly stationary scalar process, uncorrelated with the demeaned shocks conditionally on the shares and loadings, $\e[Z_t \Bar{u}_{jt} \mid \s, \Lambda] = 0$ almost surely for every $j$, and such that $\{(u_t',Z_t)\}$ is a strong mixing sequence of size $-(\frac{2 + \pi}{\pi})$, for some $\pi > 0$. When $Z_t$ is independent of $(\s, \Lambda)$, the conditional orthogonality reduces to the unconditional one, $\e[Z_t \Bar{u}_{jt}] = 0$. There is a constant $M < \infty$, independent of $N$, with $\e[|Z_t|^{8 + 2\pi} \mid \s, \Lambda] \leq M$ almost surely, and $\sup_j \e|u_{jt}|^{8 + 2\pi} \leq M$. The long-run variance is non-degenerate, $[\Omega(\s, \Lambda)]^{-1} = O_{\p}(1)$, so that the spectral density of $X_t$ at frequency zero is bounded away from zero. Define
\begin{equation*}
    \delta \;=\;
    \begin{cases}
        1 - \dfrac{1}{\mu} & \text{if } \mu \in (1, 2), \\[4pt]
        \dfrac{1}{2} & \text{if } \mu > 2.
    \end{cases}
\end{equation*}

Then, for $X_t = Z_t \, N^{\delta} \sum_{j=1}^N S_j \bar{u}_{jt}$, we have the following convergence in distribution:
\begin{equation*}
    \frac{\displaystyle \frac{1}{\sqrt{T}} \sum_{t=1}^T X_t}{\displaystyle \sqrt{ \Omega(\s, \Lambda) }} \xrightarrow{d} \n(0,1)
\end{equation*}
where $\Omega(\s, \Lambda) = \gamma_0 (\s, \Lambda) + 2 \sum_{h=1}^{\infty} \gamma_h (\s, \Lambda)$, with $\gamma_h = \cov(X_t, X_{t-h} \mid \s, \Lambda)$.

\end{Theorem}
\begin{proof}
    Define $X_t = Z_t N^{\delta} \sum_{j=1}^N S_j \bar{u}_{jt} $. Note that $X_t$'s very distant from each other have a non-zero correlation due to the presence of the common factor in the form of $\sum_j S_j$. Hence we cannot use regular CLT's on the sum, $\frac{1}{\sqrt{T}} \sum_t X_t$. However, note that once we condition on the sigma algebra generated by $\s$, $X_t$ is a strong mixing sequence. We can apply CLT on this conditional variable. 

    This type of setting is very common in the networks literature. I will be using the CLT in Theorem 3.2 of \citet{Kojevnikov2021LimitVariables} to show that $\frac{1}{\sqrt{T}} \sum_t X_t | S$ converges in distribution to a normal distribution. Then I will show that the unconditional variable, appropriately scaled has a standard normal distribution.
    
    In the first step, I will show that our setting satisfies the conditions required for the CLT. In the subsequent notes, conditioning on $\s$ should be understood as conditioning on $\sigma(\s, \Lambda)$. The factor loadings are independent of $(u_t, F_t, \varepsilon_t)$ by Assumption \ref{ass:size_of_firm_power_law}.4, so enlarging the conditioning sigma algebra from $\sigma(\s)$ to $\sigma(\s, \Lambda)$ leaves every step below unchanged. The enlargement matters only when $Z_t$ depends on the loadings.

    For given $\s$, the process $N^{\delta} S' D_N u_t$ is a strong mixing sequence of the same size as $u_t$. Similarly, for given $\s$, the process $Z_t N^{\delta} S' D_N u_t$ is a strong mixing sequence of the same size as $(u_t', Z_t)$. Hence, we can conclude that the process $X_t | \s$ is a strong mixing sequence with coefficient $\alpha(h)$, such that for some $\pi > 0$,
    \begin{equation*}
        \sum_{h=1}^{\infty} \alpha(h)^{\frac{\pi}{2+ \pi}} < \infty
    \end{equation*}

    Then by Proposition 2.2 of \citet{Kojevnikov2021LimitVariables}, the process $\{ X_t \}$ is conditionally $\psi-$dependent given $\sigma(\s)$, with dependence coefficients given the strong mixing coefficients, $\{ \alpha(h) \}_{h \geq 1}$

    Now consider the conditional mean of $X_t$,
    \begin{equation*}
        \e[X_t | S] = N^{\delta} \sum_j \e[S_j Z_t \bar{u}_{jt} | S ] = N^{\delta} \sum_j S_j \e[ Z_t \bar{u}_{jt} | S] = 0 \quad \text{a.s}
    \end{equation*}
    where the last equality is the assumed conditional orthogonality $\e[Z_t \Bar{u}_{jt} \mid \s, \Lambda] = 0$. This is the only step of the proof that restricts the dependence between $Z_t$ and the shocks, and it is imposed on the demeaned shocks $\Bar{u}_{jt}$ rather than on $u_{jt}$. A multiplier correlated with $u_t$ is therefore admissible, provided the correlation is common across units and so annihilated by $D_N$. The conditional form matters only when $Z_t$ depends on $(\s, \Lambda)$, as it does for the disaggregated multiplier of Lemma \ref{lemma:v_H multiplier}.
    
    Now I will show that, for some $\pi > 0$
    \begin{align*}
        \e[| X_t |^{4+ \pi} | \s] < \infty \quad \text{a.s}
    \end{align*}

Since $Z_t$ and $u_t$ are not assumed independent of each other, apply the Cauchy--Schwarz inequality, conditional on $(\s, \Lambda)$:
\begin{align*}
    \e[|X_t|^{4+\pi} \mid \s, \Lambda] &= N^{\delta(4+\pi)} \, \e\!\left[ |Z_t|^{4+\pi} \cdot \Big| \sum_{j=1}^N S_j \bar u_{jt} \Big|^{4+\pi} \,\Big|\, \s, \Lambda \right] \\
    &\leq N^{\delta(4+\pi)} \cdot \big(\e[|Z_t|^{8+2\pi} \mid \s, \Lambda]\big)^{1/2} \cdot \bigg( \e \Big[ \Big| \sum_{j=1}^N S_j \bar u_{jt} \Big|^{8+2\pi} \,\Big|\, \s \Big] \bigg)^{1/2}.
\end{align*}

For the second factor, the de-meaning rewrites the share-weighted sum of $\bar u_{jt}$ as a different share-weighted sum of $u_{jt}$:
\begin{equation*}
    \sum_j S_j \bar u_{jt} \;=\; \sum_j S_j \big(u_{jt} - \tfrac{1}{N}\sum_i u_{it}\big) \;=\; \sum_j \tilde S_j \, u_{jt}, \qquad \tilde S_j \;:=\; S_j - \tfrac{1}{N},
\end{equation*}
using $\sum_k S_k = 1$. Assumption \ref{ass:LN_time and cross sectional dependence}.5 applies directly to the demeaned weights $\{\tilde S_j\}$:
\begin{equation*}
    \e\!\left[\Big|\sum_j \tilde S_j u_{jt}\Big|^{8+2\pi} \,\Big|\, \s\right] \leq C_q \left[\Big(\sum_j \tilde S_j^2 \, \e u_{jt}^2\Big)^{4+\pi} + \sum_j |\tilde S_j|^{8+2\pi} \, \e|u_{jt}|^{8+2\pi}\right].
\end{equation*}
The demeaned weights satisfy $\sum_j \tilde S_j^2 = \sum_j S_j^2 - 1/N \leq \sum_j S_j^2$ (direct calculation), and by the $c_r$ inequality, $\sum_j |\tilde S_j|^q \leq 2^{q-1}\big(\sum_j S_j^q + N^{1-q}\big)$ for $q = 8+2\pi$. The $N^{1-q}$ term is of strictly smaller order than $\sum_j S_j^q$ in the regimes considered, so it is absorbed into the constant. Hence
\begin{equation}\label{eq:rosenthal_bound}
    \e\!\left[\Big|\sum_j S_j \bar u_{jt}\Big|^{8+2\pi} \,\Big|\, \s\right] \leq C_q' \left[\Big(\sum_j S_j^2 \, \e u_{jt}^2\Big)^{4+\pi} + \sum_j S_j^{8+2\pi} \, \e|u_{jt}|^{8+2\pi}\right].
\end{equation}

Set $C_1^2 = M$, the uniform bound on $\e[|Z_t|^{8+2\pi} \mid \s, \Lambda]$ assumed in the theorem. Set $C_2^2 = \sup_{jt} \e|u_{jt}|^{8+2\pi}$, which is at most $M$ by the uniform moment bound assumed in the theorem, and $\bar\sigma^2 = \sup_j \e u_{jt}^2$, which is at most $M^{\frac{1}{4+\pi}}$ by Jensen's inequality. Both constants are independent of $N$. Let $C = C_1 \cdot \sqrt{C_q' \cdot \max(\bar\sigma^{8+2\pi}, C_2^2)}$. Combining the two factors,
\begin{align*}
    \e[|X_t|^{4+\pi} \mid \s, \Lambda] &\leq C \cdot N^{\delta(4+\pi)} \cdot \sqrt{\,\Big(\sum_j S_j^2\Big)^{4+\pi} + \sum_j S_j^{8+2\pi} \,}.
\end{align*}

Writing $S_j = \s_j / \sum_i \s_i$,
\begin{align*}
    \Big(\sum_j S_j^2\Big)^{4+\pi} \;=\; \frac{(\sum_j \s_j^2)^{4+\pi}}{(\sum_i \s_i)^{8+2\pi}}, \qquad \sum_j S_j^{8+2\pi} \;=\; \frac{\sum_j \s_j^{8+2\pi}}{(\sum_i \s_i)^{8+2\pi}}.
\end{align*}
For $\mu > 1$, $N^{-1} \sum_i \s_i$ is finite by the LLN, so $(\sum_i \s_i)^{8+2\pi} = \Theta_{\p}(N^{8+2\pi})$. For each numerator, $\p(|\s_i|^q > s) = \p(|\s_i| > s^{1/q}) = c \, s^{-\mu/q}$, with tail index $\mu/q$. For $\mu \in (1,2)$ and $q = 2$, $\mu/q \in (1/2, 1)$; for $q = 8+2\pi$, $\mu/q < 1/4$. In both cases $\mu/q < 2$, so $\{|\s_i|^q\}$ satisfies the conditions of Lévy's Generalized Central Limit Theorem with $\kappa_n = N^{q/\mu}$ and
\begin{equation*}
    b_n = N \, \e\!\left[|\s_i|^q \cdot \mathbf 1_{(|\s_i|^q \leq \kappa_n)}\right] = \frac{N \mu}{q - \mu} \cdot \kappa_n^{1 - \mu/q}, \qquad \frac{b_n}{\kappa_n} = \frac{\mu}{q - \mu}.
\end{equation*}
Hence $\sum_j |\s_j|^q = O_{\p}(N^{q/\mu})$, giving
\begin{align*}
    \Big(\sum_j S_j^2\Big)^{4+\pi} &= O_{\p}\!\left(N^{(4+\pi) \cdot 2/\mu - (8+2\pi)}\right) = O_{\p}\!\left(N^{-(8+2\pi)\delta}\right), \\
    \sum_j S_j^{8+2\pi} &= O_{\p}\!\left(N^{(8+2\pi)/\mu - (8+2\pi)}\right) = O_{\p}\!\left(N^{-(8+2\pi)\delta}\right),
\end{align*}
where $\delta = 1 - 1/\mu$. Therefore
\begin{equation*}
    \e[|X_t|^{4+\pi} \mid \s, \Lambda] \leq C \cdot N^{\delta(4+\pi)} \cdot \sqrt{O_{\p}\!\left(N^{-(8+2\pi)\delta}\right)} = O_{\p}(1)
\end{equation*}
uniformly in $N$. Hence $\e[|X_t|^{4+\pi} \mid \s, \Lambda] < \infty$ a.s.

\paragraph{Acceptable range of $\mu$.}
The argument above was stated for $\mu \in (1,2)$, where $\delta = 1 - 1/\mu$ and Lévy's Generalized Central Limit Theorem applies to both numerators $\sum_j \s_j^q$ for $q \in \{2, 8+2\pi\}$ (each with tail index $\mu/q < 1$). I now verify that the same conclusion $\e[|X_t|^{4+\pi}] = O(1)$ continues to hold for $\mu > 2$, where $\delta = 1/2$ by Proposition \ref{prop:weakness of known factor instrument} and $\s_j$ has finite variance.

For $\mu > 2$, the first Rosenthal term changes regime. Now $\sum_j \s_j^2$ has tail index $\mu/2 > 1$, so the Law of Large Numbers (rather than Lévy's GCLT) applies, giving $\sum_j \s_j^2 = N \, \e\s^2 \cdot (1 + o_{\p}(1))$ and hence
\begin{equation*}
    \sum_j S_j^2 \;=\; \Theta_{\p}(N^{-1}), \qquad \Big(\sum_j S_j^2\Big)^{4+\pi} \;=\; \Theta_{\p}\!\left(N^{-(4+\pi)}\right).
\end{equation*}
Combined with the prefactor $N^{\delta(4+\pi)} = N^{(4+\pi)/2}$, the contribution of this term inside the square root is exactly $O(1)$.

For the second Rosenthal term, $\sum_j \s_j^{8+2\pi}$ has tail index $\mu/(8+2\pi) < 1$ for any $\mu < 8+2\pi$, so Lévy's GCLT continues to apply with $\sum_j \s_j^{8+2\pi} = O_{\p}(N^{(8+2\pi)/\mu})$, giving $\sum_j S_j^{8+2\pi} = O_{\p}(N^{-(8+2\pi)(1-1/\mu)})$. Combined with the $N^{(4+\pi)/2}$ prefactor, this contribution scales as $N^{(4+\pi)(1/\mu - 1/2)}$, which is $o(1)$ for $\mu > 2$. (For $\mu \geq 8+2\pi$, the LLN takes over and the contribution becomes $O_{\p}(N^{-(7+2\pi)/2})$, of strictly smaller order.)

Hence $\e[|X_t|^{4+\pi} \mid \s, \Lambda] = O_{\p}(1)$ uniformly in $N$ for all $\mu \in (1,2) \cup (2, \infty)$, with $\delta = (1-1/\mu) \wedge 1/2$. The boundary $\mu = 2$ is excluded. There $\sum_j \s_j^2 = \Theta_{\p}(N \ln N)$ by Appendix \ref{app:edge_cases}, so with $\delta = 1/2$ the first Rosenthal contribution grows as $(\ln N)^{(4+\pi)/2}$. Remark \ref{remark:CLT boundary indices} covers the boundary with the log-corrected normalization.

All that is now left is to verify \textbf{Condition ND} in \citet{Kojevnikov2021LimitVariables}. This condition deals with the denseness of the network. The CLT requires that the denseness of the network does not grow as the distance increases. We can think of the process, $\{ X_t \}$ as a linear network. Hence we can intuitively see that the denseness of the network does not grow with distance. I will formally verfiy this. 

For ND(a), I need to show that there exists some $p >4$, such that $h^{\frac{3}{2}} \alpha(h)^{1 - \frac{1}{p}} = o(1)$. By assumption, for some $\pi > 0$, $h^{\frac{2 + \pi}{\pi}} \alpha(h) = o(1)$. Thus,
\begin{align*}
    h^{\frac{3}{2}} \alpha(h)^{1 - \frac{1}{p}} &= \left( h^{\frac{2 + \pi}{\pi}} \alpha(h) \right)^{1 - \frac{1}{p}} \cdot h^{\beta} = o(1) \cdot  h^{\beta} 
\end{align*}
where $\beta = \frac{3}{2} + \frac{2 + \pi}{\pi p} - \frac{2 + \pi}{\pi}$. If $\beta < 0$, we are done.
\begin{equation*}
    \beta < 0 \iff p > \frac{4 + 2 \pi}{4 - \pi}
\end{equation*}
That is, for the $\pi $ that satisfies the mixing rate, we can always choose a $p > 
4$ such that ND(a) is satisfied. 

For ND(b), we need to define some terms as used in \citet{Kojevnikov2021LimitVariables}. Let $N_T$ be the set of time-series unit indices. $N_T(t;h)$ denote the set of the nodes that are within the distance $h$ from node $t$, and $N_n^{\partial}(t;h)$ denote the set of the nodes that are exactly the distance $h$ from node $t$. Formally
\begin{equation*}
    N_T(t;h) = \{s \in N_T ; d_T(t,s) \leq h \} \quad N_T^{\partial}(t;h) = \{s \in N_T ; d_T(t,s) = h \}
\end{equation*}
where $d_T(t,s) = |t - s|$, as the network is linear. For the linear network,
\begin{equation*}
    N_T(t;h) = 2h \quad N_T^{\partial}(t;h) = 2 \quad \forall t
\end{equation*}

Define
\begin{equation*}
    \delta_T^{\partial}(h;k) = \frac{1}{T} \sum_{t \in N_T}|N_T^{\partial}(t;h)|^k = 2^k \quad \forall s
\end{equation*}
where the $|\cdot|$ of a set refers to its cardinality. Define
\begin{equation*}
    \Delta_T(h,m;k) = \frac{1}{T} \sum_{t \in N_T} \max_{s \in N_T^{\partial}(t;h)} | N_T(t;m) \setminus N_T(s,h-1)|^k
\end{equation*}

The value of $\Delta_T(h,m;k) $ depends on the value of $m$. But we need to show condition ND(b) only for some $m \to \infty$. Hence, let $m = h -1$. As $h \to \infty$, we also have $m \to \infty$. Thus 
\begin{equation*}
     \Delta_T(h,m;k) = 2^k
\end{equation*}

Define
\begin{equation*}
    c_T(h,m;k) = \inf_{\alpha > 1}[\Delta_T(h,m;k \alpha)]^{\frac{1}{\alpha}} \cdot \left[ \delta_T^{\partial}(h; \frac{\alpha}{\alpha - 1})  \right]^{1 - \frac{1}{\alpha}} = 2^{2k}
\end{equation*}
Thus, the condition ND(b) is
\begin{equation*}
    \frac{1}{T^\frac{k}{2}} \sum_{h=1}^{\infty} 2^{2k} \alpha(h)^{1 - \frac{k + 2}{p}} = o(1)
\end{equation*}
for any choice of $p > 0$ and $k \in \{ 1,2\}$ as $\sum_{h=1}^{\infty} \alpha(h)^{\frac{\pi}{2+ \pi}} < \infty$. We have already verified ND(c). 
Thus, by Theorem 3.2 of \citet{Kojevnikov2021LimitVariables}, we have
    \begin{equation*}
        \frac{\displaystyle \frac{1}{\sqrt{T}} \sum_{t=1}^T X_t | \s}{\displaystyle \sqrt{ \Omega(\s) }} \xrightarrow{d} \n(0,1)
    \end{equation*}

    To go from the conditional to unconditional, define $X_T \defeq  \frac{1}{\sqrt{T}} \sum_{t=1}^T X_t$ and the conditional CDF
    \begin{equation*}
        F_{\s}(x) \defeq \p \left( \frac{X_T}{\sqrt{\Omega(S)}} \leq x | \s \right)
    \end{equation*}

    The CLT above gives $\lim_{T \to \infty}   F_{\s(x)} = \Phi(x)$, where $\Phi(x)$ is the CDF of the standard normal distribution. 
    
    We are interested in the unconditional CDF in the limit, $F(x) = \lim_{T \to \infty} \p \left( \frac{X_T}{\sqrt{\Omega(\s)}} \leq x \right) = \lim_{T \to \infty} \e[F_{\s(x)}]$. As $|F_{\s(x)}| \leq 1$, we can apply the dominated convergence theorem to interchange limits and integral. Thus,
    \begin{equation*}
        F(x) = \e \big[\lim_{T \to \infty}   F_{\s(x)} \big] = \e[\Phi(x)] = \Phi(x)
    \end{equation*}
\end{proof}

\begin{remark}[Boundary indices] \label{remark:CLT boundary indices}
The conclusion of Theorem \ref{theorem:central limit theorem} holds at the boundary indices $\mu = 1$ and $\mu = 2$ with $N^{\delta}$ replaced by the log-corrected normalization of Appendix \ref{app:edge_cases}. The normalization is $a_N = \ln N$ at $\mu = 1$ and $a_N = \sqrt{N / \ln N}$ at $\mu = 2$. The statistic is self-normalized, so the deterministic rescaling cancels between the numerator and $\sqrt{\Omega(\s, \Lambda)}$. Only the moment bound needs re-verification. Let $q = 8 + 2\pi$. At $\mu = 2$, Appendix \ref{app:edge_cases} gives $\sum_j S_j^2 = O_{\p}(\ln N / N)$. Lévy's Generalized Central Limit Theorem applied to $\s_j^{q}$, with tail index $2/q < 1$, gives $\sum_j \s_j^{q} = O_{\p}(N^{q/2})$ and hence $\sum_j S_j^{q} = O_{\p}(N^{-q/2})$. Then
\begin{equation*}
    a_N^{4+\pi} \, \sqrt{ \Big(\sum_j S_j^2\Big)^{4+\pi} + \sum_j S_j^{q} } = O_{\p}(1).
\end{equation*}
At $\mu = 1$, Appendix \ref{app:edge_cases} gives $\sum_j S_j^2 = O_{\p}(1/(\ln N)^2)$ and $\sum_i \s_i = O_{\p}(N \ln N)$. The same argument with tail index $1/q$ gives $\sum_j \s_j^{q} = O_{\p}(N^{q})$ and hence $\sum_j S_j^{q} = O_{\p}((\ln N)^{-q})$. The display above is again $O_{\p}(1)$. Every other step of the proof is free of $\mu$.
\end{remark}

\begin{Corollary}\label{corollary:central limit theorem}
Suppose $u_t$ is a vector that satisfies Assumption \ref{ass:weak stationarity and idiosyncrasy} and \ref{ass:LN_time and cross sectional dependence}, and $S$ is a vector that satisfies Assumption \ref{ass:size_of_firm_power_law} with $\mu \in (0,1)$. $Z_t$ is a weakly stationary scalar process, uncorrelated with the demeaned shocks conditionally on the shares and loadings, $\e[Z_t \Bar{u}_{jt} \mid \s, \Lambda] = 0$ almost surely for every $j$, and such that $\{(u_t',Z_t)\}$ is a strong mixing sequence of size $-(\frac{2 + \pi}{\pi})$, for some $\pi > 0$. There is a constant $M < \infty$, independent of $N$, with $\e[|Z_t|^{8 + 2\pi} \mid \s, \Lambda] \leq M$ almost surely, and $\sup_j \e|u_{jt}|^{8 + 2\pi} \leq M$. The long-run variance is non-degenerate, $[\Omega(\s, \Lambda)]^{-1} = O_{\p}(1)$, so that the spectral density of $X_t$ at frequency zero is bounded away from zero.

Then, for $X_t = Z_t \sum_{j=1}^N S_j \bar{u}_{jt}$ (i.e., the form of Theorem \ref{theorem:central limit theorem} with $\delta = 0$), we have the following convergence in distribution:
\begin{equation*}
    \frac{\displaystyle \frac{1}{\sqrt{T}} \sum_{t=1}^T X_t}{\displaystyle \sqrt{ \Omega(\s, \Lambda) }} \xrightarrow{d} \n(0,1)
\end{equation*}
where $\Omega(\s, \Lambda) = \gamma_0 (\s, \Lambda) + 2 \sum_{h=1}^{\infty} \gamma_h (\s, \Lambda)$, with $\gamma_h = \cov(X_t, X_{t-h} \mid \s, \Lambda)$.

\end{Corollary}
\begin{proof}
The only step that requires separate treatment is the moment finiteness $\e[|X_t|^{4+\pi} \mid \s, \Lambda] < \infty$ a.s. The remaining ingredients ($\psi$-dependence, conditional mean zero, Condition ND) carry over verbatim from the proof of Theorem \ref{theorem:central limit theorem}.

By Cauchy--Schwarz, conditional on $(\s, \Lambda)$,
\begin{equation*}
    \e[|X_t|^{4+\pi} \mid \s, \Lambda] \;\leq\; \big(\e[|Z_t|^{8+2\pi} \mid \s, \Lambda]\big)^{1/2} \cdot \bigg(\e\!\left[\Big|\sum_j S_j \bar u_{jt}\Big|^{8+2\pi} \,\Big|\, \s \right]\bigg)^{1/2}.
\end{equation*}
For the second factor, the de-meaning argument used in the proof of Theorem \ref{theorem:central limit theorem} applies verbatim: $\sum_j S_j \bar u_{jt} = \sum_j \tilde S_j u_{jt}$ with $\tilde S_j = S_j - 1/N$, and Assumption \ref{ass:LN_time and cross sectional dependence}.5 applied to $\{\tilde S_j\}$ together with $\sum_j \tilde S_j^2 \leq \sum_j S_j^2$ and $\sum_j |\tilde S_j|^q \leq 2^{q-1}(\sum_j S_j^q + N^{1-q})$ yields
\begin{equation*}
    \e\!\left[\Big|\sum_j S_j \bar u_{jt}\Big|^{8+2\pi} \,\Big|\, \s\right] \leq C_q' \left[\Big(\sum_j S_j^2 \, \e u_{jt}^2\Big)^{4+\pi} + \sum_j S_j^{8+2\pi} \, \e|u_{jt}|^{8+2\pi}\right].
\end{equation*}

For $\mu \in (0,1)$, the heavy-tail concentration of $S$ delivers $\sum_j S_j^q = O_{\p}(1)$ for every $q \geq 1$. Indeed, $\s_j^q$ has tail index $\mu/q \leq \mu < 1$, so by Lévy's Generalized Central Limit Theorem,
\begin{equation*}
    \sum_j \s_j^q \;=\; O_{\p}\!\left(N^{q/\mu}\right), \qquad \sum_i \s_i \;=\; O_{\p}\!\left(N^{1/\mu}\right),
\end{equation*}
and hence $\sum_j S_j^q = (\sum_j \s_j^q) / (\sum_i \s_i)^q = O_{\p}(1)$. Both Rosenthal terms are therefore $O_{\p}(1)$, and the absorbed $N^{1-q}$ correction is of strictly smaller order. Hence $\e[|X_t|^{4+\pi} \mid \s, \Lambda] = O_{\p}(1)$ uniformly in $N$. The remaining steps are identical to the proof of Theorem \ref{theorem:central limit theorem}.
\end{proof}

\subsection{The Disaggregated Multiplier} \label{app:CLT endogenous multiplier}

The disaggregated moment replaces the multiplier $\varepsilon_t$ by
$v_{Ht} = \lambda_H' F_t + H' u_t$. The Lemma below verifies that $v_{Ht}$ satisfies the
conditions that Theorem \ref{theorem:central limit theorem} and Corollary
\ref{corollary:central limit theorem} impose on the multiplier.

\begin{Lemma} \label{lemma:v_H multiplier}
Suppose Assumptions \ref{ass:weak stationarity and idiosyncrasy} to
\ref{ass:mixing for stationary time series} hold, and $H$ is given by
\eqref{eq:expression for H}. Then $Z_t = v_{Ht}$ satisfies the conditions imposed on the
multiplier in Corollary \ref{corollary:central limit theorem} when $\mu \in (0,1)$, and
those imposed on the multiplier in Theorem \ref{theorem:central limit theorem} when
$\mu > 1$.
\end{Lemma}
\begin{proof}
Write $v_{Ht} = \lambda_H' F_t + H' u_t$ with $\lambda_H = \Lambda' H$. Given
$(\s, \Lambda)$, this is a fixed linear combination of $(F_t, u_t)$. It is therefore
weakly stationary, measurable with respect to $\sigma(F_t, u_t)$, and strong mixing of
the size assumed in Assumption \ref{ass:mixing for stationary time series}.

\emph{Order of $H$.} By Assumption \ref{ass:bounded eigenvalues},
$e_N' \Omega^{-1} e_N \geq N/K$ and $e_N' \Omega^{-2} e_N \leq N/\underline\lambda^2$, so
\begin{equation*}
    \|H\|^2 \;=\; \frac{e_N' \Omega^{-2} e_N}{(e_N' \Omega^{-1} e_N)^2} \;\leq\; \frac{K^2}{\underline\lambda^2 \, N}.
\end{equation*}

\emph{Moments.} By Minkowski,
$\|v_{Ht}\|_{8+2\pi} \leq \|\lambda_H' F_t\|_{8+2\pi} + \|H' u_t\|_{8+2\pi}$. Since
$H' e_N = 1$ and $\Lambda = [e_N \enspace \tilde\Lambda]$, we have
$\lambda_H = (1, \, H' \tilde\Lambda)'$. The loadings are i.i.d.\ across $i$ with zero
mean by Assumptions \ref{ass:weak stationarity and idiosyncrasy} and
\ref{ass:size_of_firm_power_law}.3, so
$\e \|\tilde\Lambda' H\|^2 = \|H\|^2 \, \text{tr}(\Sigma_{\Lambda}) = O(N^{-1})$ and
$\|\lambda_H\|^2 = 1 + O_{\p}(N^{-1})$. Given $\Lambda$, $\lambda_H$ is a constant, and
together with $\e \|F_t\|^{8+2\pi} < \infty$ from Assumption
\ref{ass:mixing for stationary time series} the first term is bounded a.s., uniformly in
$N$. For the second term, Assumption \ref{ass:LN_time and cross sectional dependence}.5
with $a_j = H_j$ and $q = 8 + 2\pi$ gives
\begin{equation*}
    \e |H' u_t|^q \;\leq\; C_q \left[ \Big( \sum_j H_j^2 \, \e u_{jt}^2 \Big)^{q/2} + \sum_j |H_j|^q \, \e |u_{jt}|^q \right] \;=\; O\!\left(N^{-q/2}\right),
\end{equation*}
using $\sum_j H_j^2 = O(N^{-1})$ and
$\sum_j |H_j|^q \leq \|H\|^{q-2} \sum_j H_j^2 = O(N^{-q/2})$. Hence
$\e[|v_{Ht}|^{8+2\pi} \mid \s, \Lambda]$ is bounded uniformly in $N$.

\emph{Orthogonality.} The multiplier depends on the loadings through $\lambda_H$, so the
conditional form of the orthogonality must be verified. Conditional on $(\s, \Lambda)$,
$\lambda_H$ is a constant vector. The shocks $(F_t, u_t)$ are independent of
$(\s, \Lambda)$ by Assumption \ref{ass:size_of_firm_power_law}.4, so
$\e[u_t F_t' \mid \s, \Lambda] = \e[u_t F_t'] = 0$ by Assumption
\ref{ass:weak stationarity and idiosyncrasy} and $\e[u_t u_t' \mid \s, \Lambda] = \Omega$.
The vector of conditional covariances with the demeaned shocks is
\begin{equation*}
    \e[ v_{Ht} \, \Bar{u}_t \mid \s, \Lambda ] \;=\; D_N \, \e[u_t F_t' \mid \s, \Lambda] \lambda_H \;+\; D_N \Omega H \;=\; D_N \Omega H \;=\; 0 \quad \text{a.s},
\end{equation*}
since $H' \Omega D_N = 0$ by the construction of $H$ in \eqref{eq:expression for H}. Hence
$\e[v_{Ht} \Bar{u}_{jt} \mid \s, \Lambda] = 0$ for every $j$, the conditional
orthogonality the two results require. This is the only point at which
the choice of $H$ enters.
\end{proof}

\section{Proofs of the GIV with known factors} \label{app:other proofs}

\subsection{The Denominator} \label{app:proof of denominator}

Every theorem in Section \ref{section:known factors} and in Section \ref{section:estimated factors} divides by the sample denominator $T^{-1} \sum_t z_t p_t$. It is the same object in each of them, since neither $z_t$ nor $p_t$ involves the aggregation vector. I therefore record its behavior once, before the proofs that use it.

The rescaled denominator $a_N^2 T^{-1} \sum_t z_t p_t$ has no probability limit when $\mu < 2$. Lemma \ref{lemma:share inequalities}.2 records the three cases. For $\mu \in (0,1)$ neither size sum obeys a law of large numbers, and $h_N \xrightarrow{d} L_2 / L_1^2$, a ratio of stable laws. For $\mu \in (1,2)$ only the denominator of $h_N$ does, and $a_N^2 h_N \xrightarrow{d} Y / (\e \s)^2$. Both limits are nondegenerate. Only at $\mu > 2$ does the strong law deliver a constant. Rescaling is what makes the statement informative, since the unscaled denominator collapses to zero whenever $\mu > 1$.

The denominator has the quadratic form $S' \Omega S$. Assumption \ref{ass:bounded eigenvalues} bounds that form on both sides, $\underline\lambda \|S\|^2 \leq S' \Omega S \leq K \|S\|^2$. The off-diagonal mass of $\Omega$ is left free, so $a_N^2 S' \Omega S$ can fail to converge at any $\mu$, including $\mu > 2$ where $h_N$ does converge.

The proposition below addresses this issue. It compares the sample average with its own conditional mean at each $N$, and bounds that mean from above and away from zero. Every theorem divides by the denominator, so a centering that stays bounded away from zero is all the proofs use.

\begin{Prop}[The denominator] \label{prop:denominator}
    Suppose Assumptions \ref{ass:weak stationarity and idiosyncrasy} to \ref{ass:mixing for stationary time series} hold, with $\mu > 0$, $\mu \notin \{1,2\}$, $\phi_d \neq \phi_s$, and $a_N^2/T \to 0$. Write $\tilde S \defeq D_N S$, $\beta \defeq S' D_N \Omega e_N$, and define the finite-$N$ conditional object
    \begin{equation*}
        M_p \defeq a_N^2 \, \e[z_t p_t \mid S, \Lambda] = \frac{a_N^2}{\phi_d - \phi_s} \left[ \tilde S' \Omega \tilde S + \frac{\beta}{N} \right].
    \end{equation*}
    Then, conditionally on $\sigma(S,\Lambda)$,
    \begin{equation*}
        \frac{a_N^2}{T} \sum_{t=1}^T z_t p_t = M_p + O_{\p}\left( \frac{a_N}{\sqrt T} \right), \qquad M_p = O_{\p}(1), \qquad M_p^{-1} = O_{\p}(1).
    \end{equation*}
\end{Prop}

\begin{proof}
Market clearing gives $(\phi_d - \phi_s) p_t = S'(\Lambda F_t + u_t) - \varepsilon_t$. By Assumption \ref{ass:weak stationarity and idiosyncrasy}, $\e[F_t u_t'] = 0$ and $\e[\varepsilon_t u_t] = 0$, and by Assumption \ref{ass:size_of_firm_power_law} the array $(S,\Lambda)$ is independent of the shock path. Hence
\begin{equation*}
    \e[z_t p_t \mid S, \Lambda] = \frac{S' D_N \Omega S}{\phi_d - \phi_s}, \qquad S' D_N \Omega S = \tilde S' \Omega \tilde S + \frac{\beta}{N},
\end{equation*}
the second equality from $S = \tilde S + e_N/N$ and $D_N e_N = 0$.

Conditionally on $(S,\Lambda)$, the process $\{z_t p_t\}$ is a measurable function of the mixing vector $\{(F_t', u_t', \varepsilon_t)\}$ and inherits its mixing coefficients. Set $q_0 = 2(2+\pi)$. The Rosenthal bound of Assumption \ref{ass:LN_time and cross sectional dependence}.5, applied with the deterministic weights $\tilde S$ and the $\ell_q$ inequality of Lemma \ref{lemma:share inequalities}.1, gives $\| z_t \|_{q_0 \mid S} \leq C h_N^{1/2}$. Next, $\| p_t \|_{q_0 \mid S, \Lambda} = O_{\p}(1)$, since $p_t$ is a fixed linear combination of $\lambda_S' F_t$, $S' u_t$ and $\varepsilon_t$, each with $8 + 2\pi$ moments under Assumption \ref{ass:mixing for stationary time series}, with the Rosenthal bound applied to the weights $S$ for the middle term and $\lambda_S = \Lambda' S = O_{\p}(1)$ by Proposition \ref{prop:behavior of S times lambda} applied to $\lambda_S = (1, (S'\tilde\Lambda)')'$. Davydov's inequality \citep{Rio1993CovarianceProcesses} bounds the conditional variance of $\frac{1}{T} \sum_t z_t p_t$ by $C' h_N/T$. So the fluctuation of $\frac{a_N^2}{T} \sum_t z_t p_t$ is
\begin{equation*}
    O_{\p}\left( \frac{a_N^2 h_N^{1/2}}{\sqrt T} \right) = O_{\p}\left( \frac{a_N}{\sqrt T} \right)
\end{equation*}
by Proposition \ref{prop:weakness of known factor instrument}, which is $o_{\p}(1)$ under $a_N^2/T \to 0$.

By Assumption \ref{ass:bounded eigenvalues} and Lemma \ref{lemma:share inequalities}.1,
\begin{equation*}
    \underline\lambda \, a_N^2 \left( h_N - \frac{1}{N} \right) \; \leq \; a_N^2 \, \tilde S' \Omega \tilde S \; \leq \; K a_N^2 h_N = O_{\p}(1),
\end{equation*}
the upper bound by Proposition \ref{prop:weakness of known factor instrument}. The cross term obeys $a_N^2 |\beta|/N = O_{\p}(a_N/N) = o_{\p}(1)$ by Lemma \ref{lemma:hajek bound}. Lemma \ref{lemma:share inequalities}.2 gives $[a_N^2(h_N - 1/N)]^{-1} = O_{\p}(1)$, and $\phi_d \neq \phi_s$ then delivers $M_p^{-1} = O_{\p}(1)$.
\end{proof}

\subsection{Proof of Theorem \ref{theorem:strong_consistency_aggregate}} \label{proof:strong_Consistency_aggregate}
The scaled difference between the estimator and the true value is
\begin{equation*}
    \sqrt{T} [\hat{\phi}_d - \phi_d] = \frac{\displaystyle 1}{\displaystyle \frac{1}{T} \sum_t z_t p_t} \cdot \frac{1}{\sqrt{T}} \sum_t z_t \varepsilon_t
\end{equation*}
The market clearing condition implies
\begin{align*}
    \phi_d p_t + \varepsilon_t &= \phi_s p_t + \lambda_S' F_t + u_{St} \\
    p_t &=  \frac{1}{\phi_d - \phi_s} \cdot (u_{St}  + \lambda_S' F_t -  \varepsilon_t)
\end{align*}

I will first establish the probability limit of $\frac{1}{T} \sum_t z_t p_t$. From the decomposition above,
\begin{equation*}
    \frac{1}{T} \sum_t z_t p_{t} = \frac{1}{\phi_d - \phi_s} \cdot \left( \frac{1}{T} \sum_t z_t u_{St}  + \frac{\lambda_S'}{T} \sum_t F_t z_t -  \frac{1}{T} \sum_t z_t \varepsilon_t\right)
\end{equation*}

I will consider the terms on the right hand side one by one. I will disregard the constant, $\frac{1}{\phi_d - \phi_s}$.
\begin{align*}
    I &=  \frac{1}{T} \sum_t z_t u_{St} =  \frac{1}{T} \sum_t S' \Bar{u}_t u_{t}' S
\end{align*}
Consider the summand above
\begin{align*}
     S' \Bar{u}_t u_{t}' S &= S' \left[ u_t - \frac{1}{N} \sum_{j=1}^N u_{jt} \iota_n \right] u_t' S \\
    &= S' u_t u_{t}' S - \left( \frac{1}{N} \sum_{j=1}^N u_{jt} \right) u_t'S \\
    &= S' u_t u_{t}' S+ O_{\p}\left( \frac{1}{\sqrt{N}} \right) O_{\p}\left( 1 \right) = S' u_t u_{t}' S + o_{\p}(1)
\end{align*}

where Proposition \ref{prop:weakness of known factor instrument} gives the $O_{\p}\left( 1 \right)$ stochastic order for $S' u_t$. The $O_{\p}\left( \frac{1}{\sqrt{N}} \right)$ stochastic order comes from Assumption \ref{ass:LN_time and cross sectional dependence} as $u_{it}$ is mean-zero, has finite second moment, and has limited cross-sectional correlation.

By the Law of Large Numbers in Theorem 3.1 of \citet{Kojevnikov2021LimitVariables}
\begin{equation*}
    \left | \frac{1}{T} \sum_t  S' u_t u_t' S - S' \Omega S \right| \xrightarrow{a.s} 0
\end{equation*}
The process satisfies the conditions 3.1 and 3.2 of the LLN: (i) conditional on $\sigma(\s)$, the scalar process $\{S'u_tu_t'S\}$ inherits the strong-mixing size of $\{u_t\}$ from Assumption \ref{ass:mixing for stationary time series}, hence it is conditionally $\psi$-dependent with summable coefficients (Proposition 2.2 of Kojevnikov \& Song); (ii) the time series doesnt get too dense. But we need to check that for some $\pi > 0$
    \begin{align*}
        \e[| S' u_t u_t' S |^{1+ \pi} | \s] < \infty \quad \text{a.s}
    \end{align*}

Since the shares are non-negative and $\sum_j S_j = 1$ almost surely, $\{S_j\}_{j=1}^N$ is itself a probability distribution given $\sigma(\s)$. I apply Jensen's inequality to the convex function $|\cdot|^{2+2\pi}$ with weights $S_j$:
\begin{equation*}
    \left| \sum_{j=1}^N S_j u_{jt} \right|^{2+2\pi} \leq \sum_{j=1}^N S_j \, |u_{jt}|^{2+2\pi}.
\end{equation*}
Taking conditional expectation and using the independence of $S$ and $u_t$,
\begin{equation*}
    \e\!\left[ |S' u_t u_t' S|^{1+\pi} \,\big|\, \s \right] \;=\; \e\!\left[ |S' u_t|^{2+2\pi} \,\big|\, \s \right] \;\leq\; \sum_{j=1}^N S_j \, \e |u_{jt}|^{2+2\pi} \;\leq\; \max_j \e |u_{jt}|^{2+2\pi}.
\end{equation*}
Call $\max_j \e |u_{jt}|^{2+2\pi} = C$. Assumption \ref{ass:mixing for stationary time series} delivers $C < \infty$ for any $\pi > 0$. Hence $\e[|S' u_t u_t' S|^{1+\pi} \mid \s] < \infty$ almost surely, uniformly in $N$. Combined with conditions 3.1 and 3.2 verified above,

\begin{equation*}
    I \xrightarrow{p} S' \Omega S
\end{equation*}
where $\Omega := \e[u_t u_t']$ is the $N \times N$ idiosyncratic covariance matrix. By Assumption \ref{ass:bounded eigenvalues}, $\gamma_{\text{min}}(\Omega) \geq \underline\lambda > 0$ uniformly in $N$, so
\begin{equation*}
    S' \Omega S \;\geq\; \gamma_{\text{min}}(\Omega) \, \|S\|^2 \;\geq\; \underline\lambda \, \|S\|^2.
\end{equation*}
Lemma \ref{lemma:joint stable convergence} gives $\|S\|^2 \xrightarrow{d} L_2/L_1^2$, a limit that is strictly positive almost surely, so $(\|S\|^2)^{-1} = O_{\p}(1)$ and $S' \Omega S$ is bounded away from zero in probability. No probability limit in $N$ is claimed. Proposition \ref{prop:denominator} states the same conclusion for the denominator itself, in the form $M_p^{-1} = O_{\p}(1)$.

Consider the second term,
\begin{align*}
    II &= S' \Lambda \frac{1}{T} \sum_t F_t z_t = \frac{1}{T} \sum_i S_i \lambda_i' \sum_t F_t \sum_j S_j \Bar{u}_{jt} \\
    &= \frac{1}{T} \sum_{t=1}^T \lambda_S' F_t \cdot S' \Bar{u}_{t}.
\end{align*}
Decompose $\lambda_S' F_t = F_t^{(1)} + (S' \tilde{\Lambda}) \tilde F_t$, using $S' \iota_N = 1$. By Proposition \ref{prop:behavior of S times lambda}, $S' \tilde{\Lambda} = O_{\p}(1)$, and $F_t^{(1)}, \tilde F_t = O_{\p}(1)$ under Assumption \ref{ass:mixing for stationary time series}; hence $\lambda_S' F_t = O_{\p}(1)$. We have already established $S' \Bar{u}_t = O_{\p}(1)$ via Proposition \ref{prop:weakness of known factor instrument}.

I apply the Law of Large Numbers in Theorem 3.1 of \citet{Kojevnikov2021LimitVariables} to the scalar process $\{\lambda_S' F_t \cdot S' \Bar u_t\}_t$, conditional on $\sigma(\s, \Lambda)$. Conditions 3.1 and 3.2 hold by the same argument as for Term I: the process inherits the $\psi$-dependence of $(F_t, u_t)$ from Assumption \ref{ass:mixing for stationary time series}, and the time series is not too dense. For the moment condition, the Cauchy--Schwarz inequality gives, for any $\pi > 0$,
\begin{equation*}
    \e\!\left[ \, \big| \lambda_S' F_t \cdot S' \Bar u_t \big|^{1+\pi} \,\big|\, \s, \Lambda \right] \;\leq\; \left( \e\!\left[ |\lambda_S' F_t|^{2(1+\pi)} \,\big|\, \s, \Lambda \right] \right)^{1/2} \left( \e\!\left[ |S' \Bar u_t|^{2(1+\pi)} \,\big|\, \s, \Lambda \right] \right)^{1/2}.
\end{equation*}
The first factor is finite by the moment bound on $F_t$ in Assumption \ref{ass:mixing for stationary time series} together with $\lambda_S = O_{\p}(1)$; the second factor is finite by the weighted-Jensen argument used for Term I, applied to $S' u_t$, with the cross-sectional average $\tfrac{1}{N} \sum_i u_{it}$ in $\Bar u_t = u_t - \tfrac{1}{N} \sum_i u_{it} \iota_N$ handled by the triangle inequality and uniform weights $\tfrac{1}{N}$. Finally, under the factor/idiosyncratic orthogonality of Assumption \ref{ass:mixing for stationary time series},
\begin{equation*}
    \e[\lambda_S' F_t \cdot S' \Bar u_t \mid \s, \Lambda] = \lambda_S' \, \e[F_t \mid \s, \Lambda] \cdot \e[S' \Bar u_t \mid \s, \Lambda] = 0,
\end{equation*}
since $F_t \perp u_t$ given $(\s, \Lambda)$ and $\e[u_{jt}] = 0$. Hence
\begin{equation*}
    II \xrightarrow{p} 0.
\end{equation*}

Consider the third term,
\begin{equation*}
    III = \frac{1}{T} \sum_{t=1}^T z_t \varepsilon_t
\end{equation*}
To analyze this term, consider a scaled version:
\begin{align*}
        \frac{1}{\sqrt{T}} \sum_t z_t \varepsilon_t &=   \frac{1}{\sqrt{T}} \sum_{t=1}^T S' \Bar{u}_t \varepsilon_t = \frac{1}{\sqrt{T}} \sum_{t=1}^T \varepsilon_t\sum_{j=1}^N S_j \Bar{u}_{jt}
    \end{align*}

    $\{(\varepsilon_t, u_t')\}$ is a mixing sequence and $\e|\varepsilon_t|^{8 + 2 \pi} < \infty$ under Assumption \ref{ass:mixing for stationary time series}. By Assumption \ref{ass:size_of_firm_power_law}, $\varepsilon_t$ is independent of $(\s, \Lambda)$, so the conditional moment bound and the conditional orthogonality coincide with their unconditional versions. Assumption \ref{ass:weak stationarity and idiosyncrasy} gives $\e[\varepsilon_t u_t] = 0$, hence $\e[\varepsilon_t \Bar{u}_{jt} \mid \s, \Lambda] = \e[\varepsilon_t \Bar{u}_{jt}] = 0$ for every $j$. The long-run variance $V_{z\varepsilon}(S)$ is non-degenerate under Assumption \ref{ass:mixing for stationary time series}. Thus, we can apply Corollary \ref{corollary:central limit theorem} to the scaled sum, $X_T = \frac{1}{\sqrt{T}} \sum_{t=1}^T \varepsilon_t  \sum_{j=1}^N S_j \Bar{u}_{jt}$, with $\varepsilon_t$ replacing $Z_t$. Thus,
    \begin{equation*}
        \frac{X_T}{\sqrt{V_{z \varepsilon}(S)}} \xrightarrow{d} \n(0,1)
    \end{equation*}
Equivalently, the numerator $\frac{1}{\sqrt{T}} \sum_t z_t \varepsilon_t \xrightarrow{d} \n(0, V_{z\varepsilon}(S))$; in particular $III = O_{\p}(T^{-1/2}) \xrightarrow{p} 0$.

Combining the three probability limits,
\begin{equation*}
    \frac{1}{T} \sum_t z_t p_t = \frac{1}{\phi_d - \phi_s} \left( I + II - III \right) \xrightarrow{p} \frac{S' \Omega S}{\phi_d - \phi_s} \;=:\; \Gamma_{zp},
\end{equation*}
which is nonzero a.s.\ in the joint $N, T \to \infty$ limit by the eigenvalue lower bound established for $S' \Omega S$ above and $\phi_d \neq \phi_s$.

The numerator $\frac{1}{\sqrt{T}} \sum_t z_t \varepsilon_t \xrightarrow{d} \n(0, V_{z\varepsilon}(S))$ by the central limit theorem established above. Applying Slutsky's theorem to the ratio
\begin{equation*}
    \sqrt{T} [\hat{\phi}_d - \phi_d] = \frac{\frac{1}{\sqrt{T}} \sum_t z_t \varepsilon_t}{\frac{1}{T} \sum_t z_t p_t}
\end{equation*}
yields the asymptotic normality claim of the theorem,
\begin{equation*}
    \frac{\Gamma_{zp}}{\sqrt{V_{z\varepsilon}(S)}} \cdot \sqrt{T} [\hat{\phi}_d - \phi_d] \xrightarrow{d} \n(0,1).
\end{equation*}
Consistency follows as an immediate corollary, since $\hat{\phi}_d - \phi_d = O_{\p}(T^{-1/2}) = o_{\p}(1)$, so $\hat{\phi}_d \xrightarrow{p} \phi_d$.
\qed

\subsection{Proof of Theorem \ref{theorem:strong_consistency_disaggregate}} \label{app:proof of disaggregate} \label{proof:strong_consistency_disaggregate}

The scaled difference between the estimator and the true value is
\begin{equation*}
    \sqrt{T} [\hat{\phi}_s - \phi_s] = \frac{\displaystyle 1}{\displaystyle \frac{1}{T} \sum_t z_t p_t} \cdot \frac{1}{\sqrt{T}} \sum_t z_t v_{Ht}.
\end{equation*}

\paragraph{Denominator.} The instrument $z_t = S' \bar u_t$ and the price $p_t$ are the
same objects as in Theorem \ref{theorem:strong_consistency_aggregate}, and neither
depends on the aggregating vector $H$. The decomposition of $T^{-1} \sum_t z_t p_t$ into
Terms I, II and III, and its limit
\begin{equation*}
    \frac{1}{T} \sum_t z_t p_t \;\xrightarrow{p}\; \frac{S' \Omega S}{\phi_d - \phi_s} \;=\; \Gamma_{zp},
\end{equation*}
therefore carry over verbatim. $\Gamma_{zp}$ is $\sigma(\s)$-measurable, so widening the
conditioning to $\sigma(\s, \Lambda)$ leaves it unchanged. It is non-zero a.s.\ for
$\phi_d \neq \phi_s$.

\paragraph{Numerator.} This is the only part of the argument that differs. The multiplier
is now $v_{Ht} = \lambda_H' F_t + H' u_t$, which is correlated with $u_t$. The choice of
$H$ leaves it orthogonal to the demeaned shocks. Since $\mu \in (0,1)$ here, Lemma
\ref{lemma:v_H multiplier} gives the conditions of Corollary
\ref{corollary:central limit theorem}. Applying that corollary to
$X_t = v_{Ht} \sum_j S_j \bar u_{jt}$,
\begin{equation*}
    \frac{1}{\sqrt{T}} \sum_t z_t v_{Ht} \;\xrightarrow{d}\; \n \big(0, \; V_{z v_H}(S, \Lambda) \big),
\end{equation*}
where $V_{z v_H}(S, \Lambda) = \lim_{T \to \infty} \frac{1}{T} \sum_{s=1}^T \sum_{t=1}^T \e[z_t z_s v_{Ht} v_{Hs} \mid S, \Lambda]$.

\paragraph{Combining.} Slutsky's theorem applied to the ratio gives
\begin{equation*}
    \frac{\Gamma_{zp}}{\sqrt{V_{z v_H}(S, \Lambda)}} \cdot \sqrt{T} [\hat{\phi}_s - \phi_s] \;\xrightarrow{d}\; \n(0,1).
\end{equation*}
Consistency follows, since $\hat{\phi}_s - \phi_s = O_{\p}(T^{-1/2}) = o_{\p}(1)$.
\qed

\subsection{Proof of Theorem \ref{theorem:weak_consistency_aggregate}} \label{app:proof of weak theorems}

The argument is a single derivation parameterized by
\begin{equation*}
    \delta \;=\; \begin{cases}
        1 - \dfrac{1}{\mu} & \mu \in (1,2), \\[3pt]
        \dfrac{1}{2} & \mu > 2,
    \end{cases}
\end{equation*}
which is the same $\delta$ that appears in Theorem~\ref{theorem:central limit theorem}. Under this convention, Proposition~\ref{prop:weakness of known factor instrument} reads $z_t = O_{\p}(N^{-\delta})$ and $\|S\|^2 = O_{\p}(N^{-2\delta})$, while Proposition~\ref{prop:behavior of S times lambda} reads $S'\tilde{\Lambda} = O_{\p}(N^{-\delta})$.

The proof parallels Theorem~\ref{theorem:strong_consistency_aggregate}: I establish the asymptotic order of the denominator term-by-term, state the CLT for the numerator, and combine the two. The only new feature relative to the strong regime is that the instrument now decays in $N$, which slows the rate of convergence and forces the requirement $N/T \to 0$. Because the instrument is weak, I show consistency explicitly before turning to asymptotic normality.

The difference between the estimator and the true value is
\begin{equation*}
   \hat{\phi}_d - \phi_d \;=\; \frac{\displaystyle \frac{1}{T} \sum_t z_t \varepsilon_t}{\displaystyle \frac{1}{T} \sum_t z_t p_t}.
\end{equation*}
Market clearing gives
\begin{equation*}
    p_t \;=\; \frac{1}{\phi_d - \phi_s} \cdot (u_{St} + \lambda_S' F_t - \varepsilon_t),
\end{equation*}
so the denominator decomposes as
\begin{equation*}
    \frac{1}{T} \sum_t z_t p_t \;=\; \frac{1}{\phi_d - \phi_s}\left( \underbrace{\frac{1}{T}\sum_t z_t u_{St}}_{=:\,I} \;+\; \underbrace{\frac{\lambda_S'}{T}\sum_t F_t z_t}_{=:\,II} \;-\; \underbrace{\frac{1}{T}\sum_t z_t \varepsilon_t}_{=:\,III} \right).
\end{equation*}
I disregard the constant $1/(\phi_d - \phi_s)$ and consider the three terms in turn.

\paragraph{Term I.}
Write $\tilde S \defeq D_N S$ and $\beta \defeq S' D_N \Omega e_N$, so that $z_t = \tilde S' u_t$ and
\begin{equation*}
    I \;=\; \frac{1}{T}\sum_t \tilde S' u_t \, u_t' S.
\end{equation*}
I keep the demeaning here rather than absorbing it into a remainder. Discarding the correction $\frac{1}{N}\sum_j u_{jt} \cdot u_t' S$ as $o_{\p}(N^{-2\delta})$ would need $1/2 + \delta > 2\delta$, and that fails at $\delta = 1/2$. It is an equality there, so the correction is the same order as the term it would be dropped from. That is precisely the case $\mu > 2$.

Proposition~\ref{prop:denominator} applies to $I$ verbatim. Its conditional mean is $\e[I \mid \s, \Lambda] = \tilde S' \Omega \tilde S + \beta/N$, and the Davydov step there gives
\begin{equation*}
    a_N^2 I \;=\; a_N^2 \left[ \tilde S' \Omega \tilde S \;+\; \frac{\beta}{N} \right] \;+\; O_{\p}\!\left( \frac{a_N}{\sqrt T} \right),
\end{equation*}
a remainder that is $o_{\p}(1)$ under $a_N^2/T \to 0$. 

By Assumption~\ref{ass:bounded eigenvalues} and Lemma~\ref{lemma:share inequalities}.1, $\underline\lambda \left( h_N - 1/N \right) \leq \tilde S' \Omega \tilde S \leq K h_N$, so $a_N^2 \tilde S' \Omega \tilde S$ is bounded above by Proposition~\ref{prop:weakness of known factor instrument} and away from zero by Lemma~\ref{lemma:share inequalities}.2. The cross term vanishes in every regime, $a_N^2 |\beta|/N = O_{\p}(a_N/N) = o_{\p}(1)$, by Lemma~\ref{lemma:hajek bound}. Hence
\begin{equation*}
    I \;=\; O_{\p}(N^{-2\delta}), \qquad \left[ a_N^2 I \right]^{-1} = O_{\p}(1).
\end{equation*}

\paragraph{Term II.}
\begin{align*}
    II &= S' \Lambda \frac{1}{T} \sum_t F_t z_t = \frac{1}{T} \sum_i S_i \lambda_i' \sum_t F_t \sum_j S_j \Bar{u}_{jt} \\
    &= \frac{1}{T} \sum_{t=1}^T \lambda_S' F_t \cdot S' \Bar{u}_{t}.
\end{align*}
Decompose $\lambda_S' F_t = F_t^{(1)} + (S' \tilde{\Lambda}) \tilde F_t$ using $S' \iota_N = 1$. By Proposition~\ref{prop:behavior of S times lambda}, $S' \tilde{\Lambda} = O_{\p}\left( \frac{1}{N^{\delta}} \right)$, and $F_t^{(1)}, \tilde F_t = O_{\p}(1)$ under Assumption~\ref{ass:mixing for stationary time series}; hence $\lambda_S' F_t = F_t^{(1)} + O_{\p}\left( \frac{1}{N^{\delta}} \right) = O_{\p}(1)$. By Proposition~\ref{prop:weakness of known factor instrument}, $S' \Bar{u}_t = O_{\p}\left( \frac{1}{N^{\delta}} \right)$.

I will first establish the asymptotic order of $II$, then identify its probability limit via the same LLN argument used in Theorem~\ref{theorem:strong_consistency_aggregate}.

\subparagraph{Order.}
Using the decomposition above,
\begin{align*}
    II &= \frac{1}{T} \sum_{t=1}^T F_t^{(1)} \cdot S' \Bar u_t + \frac{1}{T} \sum_{t=1}^T (S'\tilde\Lambda)\tilde F_t \cdot S' \Bar u_t.
\end{align*}
The cross term is $O_{\p}(N^{-2\delta})$ by the orders of $S'\tilde\Lambda$ and $S' \Bar u_t$ established above. The first term can be written as
\begin{align*}
    II_1 &= \frac{1}{\sqrt{T}} \cdot \frac{1}{\sqrt{T}} \sum_{t=1}^T F_t^{(1)}  \sum_{j=1}^N S_j \Bar{u}_{jt}.
\end{align*}
$\{(F_t^{(1)}, u_t')\}$ is a mixing sequence and $\e|F_t^{(1)}|^{8 + 2 \pi} < \infty$ under Assumption~\ref{ass:mixing for stationary time series}. It is independent of $(\s, \Lambda)$ by Assumption~\ref{ass:size_of_firm_power_law}, so the conditional moment bound and the conditional orthogonality coincide with their unconditional versions, and $\e[F_t u_t] = 0$ gives $\e[F_t^{(1)} \Bar{u}_{jt} \mid \s, \Lambda] = 0$ for every $j$. Thus we can apply Theorem~\ref{theorem:central limit theorem} to $\frac{1}{\sqrt{T}} \sum_{t=1}^T F_t^{(1)}  \sum_{j=1}^N S_j \Bar{u}_{jt}$, with $F_t^{(1)}$ replacing $Z_t$, and this scaled sum is $O_{\p}\left( \frac{1}{N^{\delta}} \right)$. Thus
 \begin{align*}
     II_1 &= O_{\p}\left( \frac{1}{\sqrt{T}\, N^{\delta}} \right), \\
     II &= O_{\p}\left( \frac{1}{\sqrt{T}\, N^{\delta}} \right) + O_{\p}\left( \frac{1}{N^{2 \delta}} \right).
 \end{align*}

\subparagraph{Probability limit.}
I now apply the Law of Large Numbers in Theorem~3.1 of \citet{Kojevnikov2021LimitVariables} to the scalar process $\{\lambda_S' F_t \cdot S' \Bar u_t\}_t$, conditional on $\sigma(\s, \Lambda)$. Conditions 3.1 and 3.2 hold by the same argument as in the proof of Theorem~\ref{theorem:strong_consistency_aggregate}. The process inherits the $\psi$-dependence of $(F_t, u_t)$ from Assumption~\ref{ass:mixing for stationary time series}, and the time series is not too dense. For the moment condition, the Cauchy--Schwarz inequality gives, for any $\pi > 0$,
\begin{equation*}
    \e\!\left[ \, \big| \lambda_S' F_t \cdot S' \Bar u_t \big|^{1+\pi} \,\big|\, \s, \Lambda \right] \;\leq\; \left( \e\!\left[ |\lambda_S' F_t|^{2(1+\pi)} \,\big|\, \s, \Lambda \right] \right)^{1/2} \left( \e\!\left[ |S' \Bar u_t|^{2(1+\pi)} \,\big|\, \s, \Lambda \right] \right)^{1/2}.
\end{equation*}
The first factor is finite by the moment bound on $F_t$ in Assumption~\ref{ass:mixing for stationary time series} together with $\lambda_S' F_t = O_{\p}(1)$ established above. The second factor is finite by the weighted-Jensen argument in the proof of Theorem~\ref{theorem:strong_consistency_aggregate}, applied to $S' u_t$, with the cross-sectional average $\tfrac{1}{N} \sum_i u_{it}$ in $\Bar u_t = u_t - \tfrac{1}{N} \sum_i u_{it} \iota_N$ handled by the triangle inequality and uniform weights $\tfrac{1}{N}$. Under the factor/idiosyncratic orthogonality of Assumption~\ref{ass:mixing for stationary time series},
\begin{equation*}
    \e[\lambda_S' F_t \cdot S' \Bar u_t \mid \s, \Lambda] = \lambda_S' \, \e[F_t \mid \s, \Lambda] \cdot \e[S' \Bar u_t \mid \s, \Lambda] = 0,
\end{equation*}
since $F_t \perp u_t$ given $(\s, \Lambda)$ and $\e[u_{jt}] = 0$. Hence
\begin{equation*}
    II \xrightarrow{p} 0.
\end{equation*}

\paragraph{Term III.}
By Theorem~\ref{theorem:central limit theorem}, the rescaled sum
\begin{equation*}
    \frac{N^{\delta}}{\sqrt{T}}\sum_{t=1}^T \varepsilon_t \sum_{j=1}^N S_j \bar u_{jt} \;\xrightarrow{d}\; \n\!\big(0,\; N^{2\delta}\, V_{z\varepsilon}(S)\big),
\end{equation*}
where $V_{z\varepsilon}(S) = \lim_{T\to\infty} T^{-1}\sum_{s=1}^T\sum_{t=1}^T \e[z_t z_s \varepsilon_t \varepsilon_s \mid S]$, with $V_{z\varepsilon}(S) = O_{\p}(N^{-2\delta})$. Equivalently, the unscaled numerator satisfies
\begin{equation*}
    \frac{1}{\sqrt{T}}\sum_t z_t \varepsilon_t \;\xrightarrow{d}\; \n(0,\, V_{z\varepsilon}(S)),
\end{equation*}
so
\begin{equation*}
    III \;=\; \frac{1}{T}\sum_t z_t \varepsilon_t \;=\; O_{\p}\!\left(\frac{1}{\sqrt{T}\,N^{\delta}}\right).
\end{equation*}

\paragraph{Consistency.}
Combining the three orders,
\begin{equation*}
    \frac{1}{T}\sum_t z_t p_t \;=\; \frac{1}{\phi_d - \phi_s}\!\left[\, O_{\p}(N^{-2\delta}) \;+\; O_{\p}\!\left(\frac{1}{\sqrt{T}\,N^{\delta}}\right) \,\right],
\end{equation*}
where Term~I supplies the leading order and Terms II and III have been absorbed into the same $O_{\p}(N^{-\delta}/\sqrt{T})$ remainder. The numerator is $O_{\p}(N^{-\delta}/\sqrt{T})$, so the ratio is
\begin{equation*}
    \hat{\phi}_d - \phi_d \;=\; \frac{\displaystyle O_{\p}\!\left(\dfrac{1}{\sqrt{T}\,N^{\delta}}\right)}{\displaystyle O_{\p}\!\left(\dfrac{1}{N^{2\delta}}\right) + O_{\p}\!\left(\dfrac{1}{\sqrt{T}\,N^{\delta}}\right)} \;=\; O_{\p}\!\left(\frac{N^{\delta}}{\sqrt{T}}\right).
\end{equation*}
The remainder is dominated by Term~I — equivalently, $N^{\delta}/\sqrt{T} \to 0$ — under the condition $N^{2\delta}/T \to 0$. For $\mu \in (1,2)$, $2\delta < 1$. Thus $N/T \to 0$ is a sufficient condition for consistency. For $\mu > 2$, $2\delta = 1$. Thus $N/T \to 0$ is a necessary condition for consistency.

Since $2\delta \leq 1$ in both regimes, $N/T \to 0$ is sufficient. Hence
\begin{equation*}
    \hat{\phi}_d \;\xrightarrow{p}\; \phi_d.
\end{equation*}

\paragraph{Asymptotic normality.}
Write $\Gamma_{zp} = S'\Omega S/(\phi_d - \phi_s)$, which is nonzero almost surely in $S$ for $\phi_d \neq \phi_s$. Proposition~\ref{prop:denominator} gives
\begin{equation*}
    \frac{\Gamma_{zp}}{\frac{1}{T}\sum_t z_t p_t} \;\xrightarrow{p}\; 1
\end{equation*}
under $a_N^2/T \to 0$. Factorise
\begin{equation*}
    \frac{\Gamma_{zp}}{\sqrt{V_{z\varepsilon}(S)}} \cdot \sqrt{T}\,[\hat{\phi}_d - \phi_d]
    \;=\; \frac{\frac{1}{\sqrt{T}}\sum_t z_t \varepsilon_t}{\sqrt{V_{z\varepsilon}(S)}} \cdot \frac{\Gamma_{zp}}{\frac{1}{T}\sum_t z_t p_t} \;\xrightarrow{d}\; \n(0,1),
\end{equation*}
by Theorem~\ref{theorem:central limit theorem} for the first factor, since $a_N^2 V_{z\varepsilon}(S)$ is non-degenerate under Assumption \ref{ass:mixing for stationary time series}, and Slutsky for the product. The limit is free of $(\s, \Lambda)$, so the unconditional statement follows by dominated convergence.
Both $\Gamma_{zp}$ and $V_{z\varepsilon}(S)$ are $O_{\p}(N^{-2\delta})$, so the studentization satisfies $\Gamma_{zp}/\sqrt{V_{z\varepsilon}(S)} = O_{\p}(N^{-\delta})$, and the standardized statistic concentrates at the rate $\sqrt{T}/N^{\delta}$. For $\mu \in (1,2)$, $\delta = 1 - 1/\mu$ gives the rate $\sqrt{T}/N^{1 - 1/\mu}$. For $\mu > 2$, $\delta = 1/2$ gives the rate $\sqrt{T/N}$.

\paragraph{Failure mode under $N/T \to c$.}
The $\mu > 2$ failure mode falls out of the same algebra. Consistency required $N^{2\delta}/T \to 0$, which fails when $N/T \to c > 0$ and $\delta = 1/2$. There, $\hat{\phi}_d - \phi_d = O_{\p}(N^{\delta}/\sqrt{T}) = O_{\p}(\sqrt{N/T}) = O_{\p}(1)$, so the estimator is inconsistent and identification is weak in the sense of \citet{Staiger1997InstrumentalInstruments}. \qed

\subsection{Proof of Theorem \ref{theorem:weak_consistency_disaggregate}} \label{app:proof of weak disaggregate}

The scaled difference between the estimator and the true value is
\begin{equation*}
    \sqrt{T} [\hat{\phi}_s - \phi_s] = \frac{\displaystyle 1}{\displaystyle \frac{1}{T} \sum_t z_t p_t} \cdot \frac{1}{\sqrt{T}} \sum_t z_t v_{Ht}.
\end{equation*}

\paragraph{Denominator.} Neither $z_t$ nor $p_t$ depends on $H$, so Terms I, II and III of Theorem~\ref{theorem:weak_consistency_aggregate} and the conclusion
\begin{equation*}
    \frac{\Gamma_{zp}}{\frac{1}{T} \sum_t z_t p_t} \;\xrightarrow{p}\; 1, \qquad \Gamma_{zp} \;=\; \frac{S' \Omega S}{\phi_d - \phi_s} \;=\; O_{\p}(N^{-2\delta}),
\end{equation*}
carry over verbatim. $\Gamma_{zp}$ is $\sigma(\s)$-measurable, so widening the conditioning to $\sigma(\s, \Lambda)$ leaves it unchanged.

\paragraph{Numerator.} This is the only part of the argument that differs. The multiplier is $v_{Ht} = \lambda_H' F_t + H' u_t$. Since $\mu > 1$ here, Lemma \ref{lemma:v_H multiplier} gives the multiplier conditions of Theorem \ref{theorem:central limit theorem}, and $a_N^2 V_{z v_H}(S, \Lambda)$ is non-degenerate under Assumption \ref{ass:mixing for stationary time series}. Applying that theorem to $X_t = v_{Ht} \sum_j S_j \Bar u_{jt}$,
\begin{equation*}
    \frac{1}{\sqrt{T}} \sum_t z_t v_{Ht} \;\xrightarrow{d}\; \n \big(0, \; V_{z v_H}(S, \Lambda) \big),
\end{equation*}
where $V_{z v_H}(S, \Lambda) = \lim_{T \to \infty} \frac{1}{T} \sum_{s=1}^T \sum_{t=1}^T \e[z_t z_s v_{Ht} v_{Hs} \mid S, \Lambda]$. The lemma bounds $\e[|v_{Ht}|^{8+2\pi} \mid \s, \Lambda]$ uniformly in $N$, so $V_{z v_H}(S, \Lambda) = O_{\p}(N^{-2\delta})$ by the order of $z_t$ in Proposition~\ref{prop:weakness of known factor instrument}, exactly as for $V_{z\varepsilon}(S)$. Hence $T^{-1} \sum_t z_t v_{Ht} = O_{\p}(N^{-\delta}/\sqrt{T})$, the order carried by Term III above.

\paragraph{Combining.} Numerator and denominator have the orders of Theorem~\ref{theorem:weak_consistency_aggregate}, so the consistency, Slutsky and studentization steps go through with $V_{z\varepsilon}(S)$ replaced by $V_{z v_H}(S, \Lambda)$. Under $N/T \to 0$, $\hat{\phi}_s - \phi_s = O_{\p}(N^{\delta}/\sqrt{T}) = o_{\p}(1)$ and
\begin{equation*}
    \frac{\Gamma_{zp}}{\sqrt{V_{z v_H}(S, \Lambda)}} \cdot \sqrt{T} [\hat{\phi}_s - \phi_s] \;\xrightarrow{d}\; \n(0,1),
\end{equation*}
at the rate $\sqrt{T}/N^{\delta}$. Under $\mu > 2$ and $N/T \to c > 0$, $\hat{\phi}_s - \phi_s = O_{\p}(\sqrt{N/T}) = O_{\p}(1)$ and the estimator is inconsistent.
\qed

\section{Proofs of GIV with estimated factors} \label{app:proofs_estimated}

\subsection{Lemmas}

\begin{Lemma}\label{lemma:order of a mixing sequence}
Suppose Assumptions \ref{ass:weak stationarity and idiosyncrasy}, \ref{ass:size_of_firm_power_law}, \ref{ass:LN_time and cross sectional dependence}, and \ref{ass:LN_strong factor structure} hold. Let $g_t$ be a scalar, mean-zero process with $\e[u_{jt} g_t] = 0$ for all $j,t$, whose product $u_{jt} g_t$ satisfies the cross-sectional weak-dependence condition of Assumption~\ref{ass:LN_time and cross sectional dependence}.6. Then
\begin{equation*}
    Z_{NT} := \frac{1}{\sqrt{T}} \sum_{t=1}^T \left[ \frac{1}{N} \sum_{j=1}^N \tilde{\Blambda}_j \Bar{u}_{jt} \right] g_t = O_{\p}\left( \frac{1}{\sqrt{N}} \right).
\end{equation*}
In particular the bound holds for $g_t = \varepsilon_t$ and for each component of $F_t$, since by Assumption~\ref{ass:LN_time and cross sectional dependence}.6 both $u_{jt} \varepsilon_t$ and $u_{jt} F_t$ satisfy the cross-sectional condition.
\end{Lemma}
\begin{proof}
Let $v_{jt} = u_{jt} g_t$ and $W_{jT} = \frac{1}{\sqrt{T}} \sum_{t=1}^T u_{jt} g_t$. Because $\Bar{u}_{jt} = u_{jt} - \frac{1}{N} \sum_{k=1}^N u_{kt}$,
\begin{equation*}
    \frac{1}{N} \sum_{j=1}^N \tilde{\Blambda}_j \Bar{u}_{jt} = \frac{1}{N} \sum_{j=1}^N \tilde{\Blambda}_j u_{jt} - \Big( \frac{1}{N} \sum_{j=1}^N \tilde{\Blambda}_j \Big) \Big( \frac{1}{N} \sum_{k=1}^N u_{kt} \Big),
\end{equation*}
so that $Z_{NT} = Z_{NT}^{(1)} - Z_{NT}^{(2)}$ with
\begin{equation*}
    Z_{NT}^{(1)} = \frac{1}{N} \sum_{j=1}^N \tilde{\Blambda}_j W_{jT}, \qquad Z_{NT}^{(2)} = \Big( \frac{1}{N} \sum_{j=1}^N \tilde{\Blambda}_j \Big) \Big( \frac{1}{N} \sum_{k=1}^N W_{kT} \Big).
\end{equation*}
The hypothesis $\e[u_{jt} g_t] = 0$ gives $\e[v_{jt}] = 0$, hence $\e[W_{jT}] = 0$. By Assumptions \ref{ass:size_of_firm_power_law} and \ref{ass:LN_strong factor structure}, $\tilde{\Lambda}$ is independent of $\{(u_t, g_t)\}$, with $\| \tilde{\Blambda}_j \| \leq \Bar{\Blambda}$, $\e[\tilde{\Blambda}_j] = 0$, and the loadings independent across $j$.

For the first term, $\e[Z_{NT}^{(1)} | \tilde{\Lambda}] = 0$, and
\begin{equation*}
    \e\big[ \| Z_{NT}^{(1)} \|^2 \,\big|\, \tilde{\Lambda} \big] = \frac{1}{N^2} \sum_{j=1}^N \sum_{k=1}^N \tilde{\Blambda}_j' \tilde{\Blambda}_k\, \e[W_{jT} W_{kT}].
\end{equation*}
Since $\e[W_{jT} W_{kT}] = \frac{1}{T} \sum_{t=1}^T \sum_{s=1}^T \cov(v_{jt}, v_{ks})$, the bounded loadings ($|\tilde{\Blambda}_j' \tilde{\Blambda}_k| \leq \Bar{\Blambda}^2$) and Assumption \ref{ass:LN_time and cross sectional dependence}.6 give
\begin{equation*}
    \e\big[ \| Z_{NT}^{(1)} \|^2 \,\big|\, \tilde{\Lambda} \big] \leq \frac{\Bar{\Blambda}^2}{N} \cdot \frac{1}{NT} \sum_{j=1}^N \sum_{k=1}^N \sum_{t=1}^T \sum_{s=1}^T \big| \cov(v_{jt}, v_{ks}) \big| \leq \frac{\Bar{\Blambda}^2 M}{N}.
\end{equation*}
Taking expectations and applying Markov's inequality, $Z_{NT}^{(1)} = O_{\p}(1/\sqrt{N})$.

For the second term, the loadings are mean zero and independent across $j$, so $\e\big\| \frac{1}{N} \sum_{j=1}^N \tilde{\Blambda}_j \big\|^2 = \frac{1}{N^2} \sum_{j=1}^N \e\| \tilde{\Blambda}_j \|^2 = O(1/N)$, giving $\frac{1}{N} \sum_{j=1}^N \tilde{\Blambda}_j = O_{\p}(1/\sqrt{N})$. Likewise $\e\big[ \big( \frac{1}{N} \sum_{k=1}^N W_{kT} \big)^2 \big] = \frac{1}{N^2} \sum_{k=1}^N \sum_{l=1}^N \e[W_{kT} W_{lT}] \leq \frac{1}{N} \cdot \frac{1}{NT} \sum_{k=1}^N \sum_{l=1}^N \sum_{t=1}^T \sum_{s=1}^T \big| \cov(v_{kt}, v_{ls}) \big| \leq \frac{M}{N}$, so $\frac{1}{N} \sum_{k=1}^N W_{kT} = O_{\p}(1/\sqrt{N})$. Hence $Z_{NT}^{(2)} = O_{\p}(1/N)$.

Combining, $Z_{NT} = Z_{NT}^{(1)} - Z_{NT}^{(2)} = O_{\p}(1/\sqrt{N})$.
\end{proof}

For the analysis that follows, recall the regime-specific rescaling $a_N$ of \eqref{eq:a_N}. By Proposition~\ref{prop:behavior of S times lambda}, $a_N \cdot S' \tilde{\Lambda} = O_{\p}(1)$ uniformly across regimes. By Proposition~\ref{prop:weakness of known factor instrument}, $a_N \cdot \Bar{u}_{St} = O_{\p}(1)$.

Before the proofs, I state two lemmas that unify steps across the regimes. Lemma~\ref{lemma:numerator} expands the first-stage error against a generic multiplier $Q_t$, so that one statement covers each of the moments used in the proofs of the theorems. Lemma~\ref{lemma:denominator} takes the same error against $p_t$, which is the denominator's scaling.

\begin{Lemma} \label{lemma:numerator}
    Suppose Assumptions \ref{ass:weak stationarity and idiosyncrasy} to \ref{ass:LN_moments and CLT} hold. Let $Q_t$ be weakly stationary and measurable with respect to $\sigma(F_t, u_t, \varepsilon_t)$, with $\e|Q_t|^{4 + 2\pi} < \infty$, and write $\zeta_t = \frac{1}{N} \sum_{j=1}^N \tilde{\Blambda}_j \Bar{u}_{jt}$.
    \begin{enumerate}
        \item[(i)] \emph{(Expansion.)}
        \begin{align*}
            \frac{a_N}{\sqrt{T}} \sum_{t=1}^T S' \big( \hat{C}_t - \tilde{C}_t \big) Q_t
            &= \frac{1}{T} \sum_{t=1}^T \tilde{F}_t' Q_t \cdot \left[ \frac{\tilde{F}' \tilde{F}}{T} \right]^{-1} \frac{a_N}{\sqrt{T}} \sum_{m=1}^T \tilde{F}_m \Bar{u}_{Sm} \\
            & \enspace + a_N S' \tilde{\Lambda} \left[ \frac{\tilde{\Lambda}' \tilde{\Lambda}}{N} \right]^{-1} \cdot \frac{1}{\sqrt{T}} \sum_{t=1}^T \zeta_t Q_t + O_{\p}\left( \frac{\sqrt{T}}{N} \right) + O_{\p}\left( \frac{1}{\sqrt{T}} \right),
        \end{align*}
        with $\frac{a_N}{\sqrt{T}} \sum_{m=1}^T \tilde{F}_m \Bar{u}_{Sm} = O_{\p}(1)$ and $a_N S' \tilde{\Lambda} \left[ \tilde{\Lambda}' \tilde{\Lambda}/N \right]^{-1} = O_{\p}(1)$.
        \item[(ii)] \emph{(Mean-zero multipliers.)} If in addition $\e[\Bar{u}_{jt} Q_t] = 0$ for every $j$, the loading-space term is $O_{\p}(1/\sqrt{N})$ and
        \begin{align*}
            \frac{a_N}{\sqrt{T}} \sum_{t=1}^T S' \big( \hat{C}_t - \tilde{C}_t \big) Q_t
            &= \frac{1}{T} \sum_{t=1}^T \tilde{F}_t' Q_t \cdot \left[ \frac{\tilde{F}' \tilde{F}}{T} \right]^{-1} \frac{a_N}{\sqrt{T}} \sum_{m=1}^T \tilde{F}_m \Bar{u}_{Sm} \\
            & \enspace + O_{\p}\left( \frac{1}{\sqrt{N}} \right) + O_{\p} \left( \frac{\sqrt{T}}{N} \right) + O_{\p}\left( \frac{1}{\sqrt{T}} \right).
        \end{align*}
    \end{enumerate}
\end{Lemma}
\begin{proof}
From \eqref{eq:influence function of estimator},
\begin{align*}
    \frac{a_N}{\sqrt{T}} \sum_{t=1}^T S' \big( \hat{C}_t - \tilde{C}_t \big) Q_t
    &= \frac{1}{T} \sum_{t=1}^T \tilde{F}_t' Q_t \cdot \left[ \frac{\tilde{F}' \tilde{F}}{T} \right]^{-1} \frac{1}{\sqrt{T}} \sum_{m=1}^T \tilde{F}_m \Bar{u}_{Sm}\, a_N  \\
    & \enspace +  a_N S' \tilde{\Lambda} \left[ \frac{\tilde{\Lambda}' \tilde{\Lambda}}{N} \right]^{-1} \cdot \frac{1}{\sqrt{T}} \sum_{t=1}^T \left[ \frac{1}{N} \sum_{j=1}^N \tilde{\Blambda}_j \Bar{u}_{jt} \right] Q_t \\
    & \enspace + O_{\p} \left( \frac{\sqrt{T}}{N} \right) + O_{\p}\left( \frac{1}{\sqrt{T}} \right)
\end{align*}
the two remainders collecting the $N$- and $T$-branches of $\delta_{NT}^{-2} = \min(N,T)^{-1}$. Under Assumption~\ref{ass:mixing for stationary time series} each component of $\tilde{F}_m$ satisfies the conditions on $Z_m$ in Theorem~\ref{theorem:central limit theorem}, so that theorem for $\mu > 1$ and Corollary~\ref{corollary:central limit theorem} for $\mu \in (0,1)$, applied component-wise, give $\frac{a_N}{\sqrt{T}} \sum_{m=1}^T \tilde{F}_m \Bar{u}_{Sm} = O_{\p}(1)$.

Consider the second term first. By Proposition~\ref{prop:behavior of S times lambda}, $a_N S' \tilde{\Lambda} = O_{\p}(1)$ uniformly across regimes. Now consider $\frac{1}{N} \sum_{j=1}^N \tilde{\Blambda}_j \Bar{u}_{jt}$:
\begin{align*}
    \frac{1}{N} \sum_{i=1}^N \tilde{\Blambda}_i \Bar{u}_{it} = \frac{1}{N} \sum_{i=1}^N \tilde{\Blambda}_i u_{it} - \frac{1}{N^2} \sum_{i=1}^N \sum_{j=1}^N \tilde{\Blambda}_i u_{jt}
\end{align*}
For the first term on the right hand side, $\e\left[\frac{1}{N} \sum_{i=1}^N \tilde{\Blambda}_i u_{it}\right] = 0$, and for some $K > 0$,
\begin{align*}
    \p \left[ \bigg \| \frac{1}{\sqrt{N}} \sum_{i=1}^N \tilde{\Blambda}_i u_{it} \bigg \| \geq K \right] &\leq \frac{K^{-2}}{N} \e \left[ \bigg \| \sum_{i=1}^N \tilde{\Blambda}_i u_{it} \bigg \|^2 \right] \\
    &= \frac{K^{-2}}{N} \e[u_t' \tilde{\Lambda} \tilde{\Lambda}' u_t] = \frac{K^{-2}}{N} \text{tr} \left( \e[\tilde{\Lambda} \tilde{\Lambda}'] \e[u_t u_t'] \right) \\
    &= \frac{K^{-2}}{N} \sum_{i=1}^N \e \|\tilde{\Blambda}_i \|^2 \e[u_{it}^2] + \frac{2 K^{-2}}{N} \sum_{i=1}^N \sum_{j=1}^N \e[\tilde{\Blambda}_i'] \e[\tilde{\Blambda}_j] \e[u_{it} u_{jt}] \\
    &= \frac{K^{-2}}{N} \sum_{i=1}^N \e \|\tilde{\Blambda}_i \|^2 \e[u_{it}^2]
\end{align*}
where the last equality uses $\e[\tilde{\Blambda}_i] = 0$. Both $\e \|\tilde{\Blambda}_i \|^2$ and $\e[u_{it}^2]$ are bounded, so
\begin{equation*}
    \frac{1}{N} \sum_{i=1}^N \e \|\tilde{\Blambda}_i \|^2 \e[u_{it}^2] \leq M < \infty,
\end{equation*}
and hence $\frac{1}{\sqrt{N}} \sum_{i=1}^N \tilde{\Blambda}_i u_{it} = O_{\p}(1)$.

For the second term,
\begin{align*}
    \frac{1}{N^2} \sum_{i=1}^N \sum_{j=1}^N \tilde{\Blambda}_i u_{jt} = \frac{1}{N} \sum_{i=1}^N \tilde{\Blambda}_i \cdot \frac{1}{N} \sum_{j=1}^N u_{jt}
\end{align*}
Since the loadings are mean zero and independent across $i$ with bounded variance, $\frac{1}{\sqrt{N}} \sum_{i=1}^N \tilde{\Blambda}_i = O_{\p}(1)$, and under Assumption~\ref{ass:LN_time and cross sectional dependence}, $\frac{1}{\sqrt{N}} \sum_{j=1}^N u_{jt} = O_{\p}(1)$. Thus $\frac{1}{N^2} \sum_{i=1}^N \sum_{j=1}^N \tilde{\Blambda}_i u_{jt} = O_{\p}\left(\frac{1}{N}\right)$, and
\begin{equation*}
    \frac{1}{N} \sum_{i=1}^N \tilde{\Blambda}_i \Bar{u}_{it} = O_{\p}\left( \frac{1}{\sqrt{N}} \right).
\end{equation*}

The pointwise order $\zeta_t = O_{\p}(1/\sqrt{N})$ is used in Lemma~\ref{lemma:denominator}. This establishes part (i), the multiplier having entered only through $\e|Q_t|^{4+2\pi} < \infty$.

For part (ii), suppose $\e[\Bar{u}_{jt} Q_t] = 0$ for every $j$. Lemma~\ref{lemma:order of a mixing sequence} then gives
\begin{equation*}
    \frac{1}{\sqrt{T}} \sum_{t=1}^T \zeta_t Q_t = O_{\p}\left( \frac{1}{\sqrt{N}} \right),
\end{equation*}
and combined with $a_N S' \tilde{\Lambda} = O_{\p}(1)$ and $[\tilde{\Lambda}' \tilde{\Lambda}/N]^{-1} = O_{\p}(1)$, the loading-space term is $O_{\p}(1/\sqrt{N})$. Substituting into part (i),
\begin{align*}
    \frac{a_N}{\sqrt{T}} \sum_{t=1}^T S' \big( \hat{C}_t - \tilde{C}_t \big) Q_t
    &= \frac{1}{T} \sum_{t=1}^T \tilde{F}_t' Q_t \cdot \left[ \frac{\tilde{F}' \tilde{F}}{T} \right]^{-1} \frac{a_N}{\sqrt{T}} \sum_{m=1}^T \tilde{F}_m \Bar{u}_{Sm} \\
    & \enspace + O_{\p}\left( \frac{1}{\sqrt{N}} \right) + O_{\p} \left( \frac{\sqrt{T}}{N} \right) + O_{\p}\left( \frac{1}{\sqrt{T}} \right).
\end{align*}
\end{proof}

Lemma \ref{lemma:numerator} is generic in the multiplier $Q_t$. It is also generic in the array that drives the principal-component step, which is what the panel controls of Appendix \ref{appendix_sec_controls} need.

\begin{remark}[Substitution in the shock slot] \label{rem:shock substitution}
    The expansion \eqref{eq:influence function of estimator} is linear in the demeaned array $\{ \Bar{u}_{jm} \}$, and the proof above uses that array only through $\Bar{u}_{Sm} = S' \Bar{u}_m$ and $\zeta_t$. Let $\{ \Bar{W}_{jm} \}$ be any array that is demeaned across $j$ and independent of $\tilde{\Lambda}$. Write $\Bar{W}_{Sm} = S' \Bar{W}_m$ and $\zeta_t^W = \frac{1}{N} \sum_j \tilde{\Blambda}_j \Bar{W}_{jt}$, and let $\hat{C}_t^W - \tilde{C}_t^W$ denote \eqref{eq:influence function of estimator} with $\Bar{u}$ replaced by $\Bar{W}$. Suppose $a_N \Bar{W}_{Sm} = O_{\p}(1)$ with $\e[\tilde{F}_m \Bar{W}_{jm} \mid \s, \Lambda] = 0$ for every $j$, and suppose $\e[\Bar{W}_{jt}^2]$ is bounded uniformly in $j$ and $N$. Then part (i) holds with $\Bar{u}$ replaced by $\Bar{W}$ throughout. The first hypothesis gives the factor-space factor $\frac{a_N}{\sqrt{T}} \sum_m \tilde{F}_m \Bar{W}_{Sm} = O_{\p}(1)$ by Theorem~\ref{theorem:central limit theorem} and Corollary~\ref{corollary:central limit theorem}, and the second gives $\zeta_t^W = O_{\p}(1/\sqrt{N})$ by the argument used for $\zeta_t$ above. If in addition $\e[\Bar{W}_{jt} Q_t] = 0$ for every $j$ and the product $\Bar{W}_{jt} Q_t$ satisfies Assumption~\ref{ass:LN_time and cross sectional dependence}.6, then Lemma~\ref{lemma:order of a mixing sequence} applies to $\frac{1}{\sqrt{T}} \sum_t \zeta_t^W Q_t$ and part (ii) holds as well.
\end{remark}

\begin{Lemma} \label{lemma:denominator}
    Suppose Assumptions \ref{ass:weak stationarity and idiosyncrasy} to \ref{ass:LN_moments and CLT} hold and $a_N^2/T \to 0$, then
\begin{equation*}
    \frac{a_N^2}{T} \sum_{t=1}^T S' \big( \hat{C}_t - \tilde{C}_t \big) p_{t}
    = O_{\p} \left( \frac{a_N}{\sqrt{T}} \right) + O_{\p} \left( \frac{1}{\sqrt{N}} \right) = o_{\p}(1)
\end{equation*}
\end{Lemma}
\begin{proof}
From \eqref{eq:influence function of estimator},
\begin{align*}
    \frac{a_N^2}{T} \sum_{t=1}^T S' \big( \hat{C}_t - \tilde{C}_t \big) p_{t}
    &= \frac{1}{T} \sum_{t=1}^T \tilde{F}_t' p_{t} \cdot \left[ \frac{\tilde{F}' \tilde{F}}{T} \right]^{-1} \frac{1}{T} \sum_{m=1}^T \tilde{F}_m \Bar{u}_{Sm}\, a_N^2  \\
    & \enspace +  a_N S' \tilde{\Lambda} \left[ \frac{\tilde{\Lambda}' \tilde{\Lambda}}{N} \right]^{-1} \cdot \frac{a_N}{T} \sum_{t=1}^T \left[ \frac{1}{N} \sum_{j=1}^N \tilde{\Blambda}_j \Bar{u}_{jt} \right] p_{t} + O_{\p} \left( \frac{a_N}{N} \right)
\end{align*}

Recall that
\begin{equation*}
     p_t = \frac{1}{\phi_d - \phi_s} \cdot (u_{St}  + \lambda_S' F_t -  \varepsilon_t).
\end{equation*}
By Proposition~\ref{prop:weakness of known factor instrument}, $u_{St} = O_{\p}(a_N^{-1})$ uniformly across regimes. The proofs of Theorems~\ref{theorem:strong_consistency_aggregate}--\ref{theorem:weak_consistency_aggregate} establish, as an intermediate result, $\lambda_S' F_t = F_t^{(1)} + O_{\p}(a_N^{-1})$. The dominant terms in $p_t$ are therefore those involving $\varepsilon_t$ and $F_t^{(1)}$, and these have finite second moments, so $\frac{1}{T} \sum_{t=1}^T \tilde{F}_t' p_{t} = O_{\p}(1)$. 

Under Assumption~\ref{ass:mixing for stationary time series}, $\{(\tilde{F}_m', u_m')\}$ is a strong mixing sequence of size $-(2+\pi)/\pi$, $\e \|\tilde{F}_t\|^{8+2\pi} < \infty$, and $\tilde{F}_t$ is uncorrelated with $u_t$ by the factor structure, so $\e[\tilde{F}_{kt} \Bar{u}_{jt}] = 0$ for every $j$ and $k$. It is also independent of $(\s, \Lambda)$ by Assumption~\ref{ass:size_of_firm_power_law}, so the conditional moment bound and the conditional orthogonality coincide with their unconditional versions. Each component of $\tilde{F}_m$ therefore satisfies the conditions on $Z_m$ in Theorem~\ref{theorem:central limit theorem} and Corollary~\ref{corollary:central limit theorem}. Applying Theorem~\ref{theorem:central limit theorem} (for $\mu > 1$) and Corollary~\ref{corollary:central limit theorem} (for $\mu \in (0,1)$) component-wise gives $\frac{a_N}{\sqrt{T}} \sum_{m=1}^T \tilde{F}_m \Bar{u}_{Sm} = O_{\p}(1)$. Combining, the first term on the right hand side is 
\begin{equation*}
    I = O_{\p} \left( \frac{a_N}{\sqrt{T}} \right).
\end{equation*}

For the second term, recall $\zeta_t = \frac{1}{N} \sum_{i=1}^N \tilde{\Blambda}_i \Bar{u}_{it} = O_{\p}\left( \frac{1}{\sqrt{N}} \right)$. Writing $\pi = 1/(\phi_d - \phi_s)$,
\begin{align*}
    \frac{1}{T} \sum_{t=1}^T \zeta_t p_{t} &= \frac{\pi}{T} \sum_{t=1}^T \zeta_t u_{St} + \frac{\pi}{T} \sum_{t=1}^T \zeta_t \lambda_S' F_t - \frac{\pi}{T} \sum_{t=1}^T \zeta_t \varepsilon_t.
\end{align*}
By Lemma~\ref{lemma:order of a mixing sequence}, $\frac{1}{\sqrt{T}} \sum_{t=1}^T \zeta_t \varepsilon_t = O_{\p}(1/\sqrt{N})$, so $\frac{\pi}{T} \sum_{t=1}^T \zeta_t \varepsilon_t = O_{\p}\left( \frac{1}{\sqrt{NT}} \right)$.

Consider $\frac{\pi}{T} \sum_{t=1}^T \zeta_t \lambda_S' F_t$. Conditional on $S$, $\lambda_S' F_t$ is a scalar mean-zero process with $\e[u_{jt} \lambda_S' F_t] = \lambda_S' \e[u_{jt} F_t] = 0$ by idiosyncrasy, and its product $u_{jt} \lambda_S' F_t$ inherits the cross-sectional weak-dependence condition of Assumption~\ref{ass:LN_time and cross sectional dependence}.6 from that on $u_{jt} F_t$, since $\lambda_S = 1 + S' \tilde{\Lambda} =  O_{\p}(1)$. Hence Lemma~\ref{lemma:order of a mixing sequence}, applied with $g_t = \lambda_S' F_t$, gives $\frac{1}{\sqrt{T}} \sum_{t=1}^T \zeta_t \lambda_S' F_t = O_{\p}(1/\sqrt{N})$, so $\frac{1}{T} \sum_{t=1}^T \zeta_t \lambda_S' F_t = O_{\p}\left( \frac{1}{\sqrt{NT}} \right)$.

We cannot apply Lemma~\ref{lemma:order of a mixing sequence} to $\frac{1}{T} \sum_{t=1}^T \zeta_t u_{St}$ since $\e[\zeta_t u_{St} | \tilde{\Lambda}] \neq 0$. Using Cauchy-Schwartz and $u_{St} = O_{\p}(a_N^{-1})$ from Proposition~\ref{prop:weakness of known factor instrument},
\begin{align*}
    \frac{1}{T} \sum_{t=1}^T \zeta_t u_{St} &\leq \left[ \frac{1}{T} \sum_{t=1}^T \zeta_t^2 \right]^{\frac{1}{2}} \left[ \frac{1}{T} \sum_{t=1}^T u_{St}^2 \right]^{\frac{1}{2}} = O_{\p}\left( \frac{1}{a_N \sqrt{N}} \right).
\end{align*}

Hence $\frac{1}{T} \sum_{t=1}^T \zeta_t p_{t} = O_{\p}\left( \frac{1}{\sqrt{NT}} \right) + O_{\p}\left( \frac{1}{a_N \sqrt{N}} \right)$, and the second term in the decomposition, which is $a_N S' \tilde{\Lambda}\, [\tilde{\Lambda}'\tilde{\Lambda}/N]^{-1}$ times $a_N \cdot \frac{1}{T} \sum_t \zeta_t p_t$, is $O_{\p}\left( \frac{1}{\sqrt{N}} \right)$ since $a_N S' \tilde{\Lambda} = O_{\p}(1)$. Combining with $I = O_{\p}(a_N/\sqrt{T})$ and the $O_{\p}(a_N/N)$ remainder,
\begin{equation*}
    \frac{a_N^2}{T} \sum_{t=1}^T S' \big( \hat{C}_t - \tilde{C}_t \big) p_{t}
    = O_{\p} \left( \frac{a_N}{\sqrt{T}} \right) + O_{\p} \left( \frac{1}{\sqrt{N}} \right) + O_{\p} \left( \frac{a_N}{N} \right) = o_{\p}(1),
\end{equation*}
where the last equality uses $a_N^2/T \to 0$ for the first term and $a_N \leq \sqrt{N}$ for the third.
\end{proof}

The two lemmas combine into Proposition~\ref{prop:first stage effect}, which is the form used by every proof in this appendix.

\subsection{Proof of Proposition \ref{prop:first stage effect}} \label{subsec:proof of first stage effect}

\begin{proof}
Since $\hat{z}_t = z_t - S'(\hat{C}_t - \tilde{C}_t)$, the numerator expands as
\begin{equation*}
    \frac{a_N}{\sqrt{T}} \sum_{t=1}^T \hat{z}_t Q_t = \frac{a_N}{\sqrt{T}} \sum_{t=1}^T z_t Q_t - \frac{a_N}{\sqrt{T}} \sum_{t=1}^T S' \big( \hat{C}_t - \tilde{C}_t \big) Q_t
\end{equation*}
and the denominator as
\begin{equation*}
    \frac{a_N^2}{T} \sum_{t=1}^T \hat{z}_t p_t = \frac{a_N^2}{T} \sum_{t=1}^T z_t p_t - \frac{a_N^2}{T} \sum_{t=1}^T S' \big( \hat{C}_t - \tilde{C}_t \big) p_t.
\end{equation*}
The leading term in each is the known-factor quantity of Section~\ref{section:known factors}. The two lemmas above bound the first-stage term.

The multiplier satisfies the hypotheses of Lemma~\ref{lemma:numerator}.(ii), so
\begin{align*}
    \frac{a_N}{\sqrt{T}} \sum_{t=1}^T S' \big( \hat{C}_t - \tilde{C}_t \big) Q_t
    &= \frac{1}{T} \sum_{t=1}^T \tilde{F}_t' Q_t \cdot \left[ \frac{\tilde{F}' \tilde{F}}{T} \right]^{-1} \frac{a_N}{\sqrt{T}} \sum_{m=1}^T \tilde{F}_m \Bar{u}_{Sm} \\
    & \enspace + O_{\p}\left( \frac{1}{\sqrt{N}} \right) + O_{\p} \left( \frac{\sqrt{T}}{N} \right) + O_{\p}\left( \frac{1}{\sqrt{T}} \right).
\end{align*}
The three remainders are $o_{\p}(1)$, the first two under $\sqrt{T}/N \to 0$ and the third under $T \to \infty$. The process $\{\tilde{F}_t' Q_t\}$ is strongly mixing under Assumption~\ref{ass:mixing for stationary time series}, and element-wise $\e|\tilde{F}_t^{(r)} Q_t|^{2 + \pi} \leq ( \e|\tilde{F}_t^{(r)}|^{4 + 2\pi} \e|Q_t|^{4 + 2\pi})^{1/2} < \infty$, so the law of large numbers in Corollary 3.48 of \citet{White2001AsymptoticEconometricians} gives $\frac{1}{T} \sum_t \tilde{F}_t' Q_t \xrightarrow{\p} \e[\tilde{F}_t' Q_t \mid \Lambda]$, and the same corollary applied to $\{\tilde{F}_t \tilde{F}_t'\}$ gives $[\tilde{F}'\tilde{F}/T]^{-1} \xrightarrow{\p} \Sigma_{\tilde{F}}^{-1}$, which is non-singular by Assumption~\ref{ass:LN_moments and CLT}. Lemma~\ref{lemma:numerator} also gives $\frac{a_N}{\sqrt{T}} \sum_m \tilde{F}_m \Bar{u}_{Sm} = O_{\p}(1)$, so the product of the two converging factors with this bounded factor may be replaced by its limit at the cost of an $o_{\p}(1)$ term. With $\Bar{u}_{Sm} = z_m$,
\begin{equation*}
    \frac{a_N}{\sqrt{T}} \sum_{t=1}^T S' \big( \hat{C}_t - \tilde{C}_t \big) Q_t
    = \frac{a_N}{\sqrt{T}} \sum_{t=1}^T \Delta_{\tilde{F} Q} \tilde{F}_t \, z_t + o_{\p}(1).
\end{equation*}
Subtracting this from $\frac{a_N}{\sqrt{T}} \sum_t z_t Q_t$ gives the first statement,
\begin{equation*}
    \frac{a_N}{\sqrt{T}} \sum_{t=1}^T \hat{z}_t Q_t = \frac{a_N}{\sqrt{T}} \sum_{t=1}^T z_t \big( Q_t - \Delta_{\tilde{F} Q} \tilde{F}_t \big) + o_{\p}(1) = \frac{a_N}{\sqrt{T}} \sum_{t=1}^T z_t \Bar{Q}_t + o_{\p}(1).
\end{equation*}
The second statement is Lemma~\ref{lemma:denominator} applied to the first-stage term of the denominator expansion.
\end{proof}

\subsection{Proofs of Theorems \ref{theorem:estimated_strong_consistency_aggregate} and \ref{theorem:estimated_weak_consistency_aggregate}} \label{subsec:proof of estimated_theorems}

Both theorems share a single argument once we adopt the unified scaling $a_N$ of \eqref{eq:a_N}. Specializing $a_N = 1$ recovers Theorem~\ref{theorem:estimated_strong_consistency_aggregate}. Specializing $a_N = N^{\delta}$ with $\delta = \min(1 - 1/\mu, 1/2)$ recovers Theorem~\ref{theorem:estimated_weak_consistency_aggregate}, which covers both $\mu \in (1,2)$ ($a_N = N^{1 - 1/\mu}$) and $\mu > 2$ ($a_N = \sqrt{N}$). The proof tracks only the additional terms introduced by the estimated-factor instrument. The remaining terms inherit their treatment from Theorems \ref{theorem:strong_consistency_aggregate} and \ref{theorem:weak_consistency_aggregate}.

\textbf{Step 1 (decomposition).} The difference between the estimator and the true value is
\begin{align*}
    \hat{\phi}_d - \phi_d = \frac{\sum_{t=1}^T \hat{z}_t \varepsilon_t}{ \sum_{t=1}^T \hat{z}_t p_{t}} = \frac{\sum_{t=1}^T z_t \varepsilon_t - \sum_{t=1}^T S'[\hat{C}_t - \tilde{C}_t] \varepsilon_t}{\sum_{t=1}^T z_t p_{t} - \sum_{t=1}^T S'[\hat{C}_t - \tilde{C}_t] p_{t}}
\end{align*}
Only one estimation error acts here, the estimated instrument, and it enters both the numerator and the denominator. Step 2 treats the denominator, Step 3 the numerator, and Step 4 delivers the limit.

\textbf{Step 2 (denominator).} By Proposition~\ref{prop:first stage effect}, $\frac{a_N^2}{T} \sum_t \hat{z}_t p_t = \frac{a_N^2}{T} \sum_t z_t p_t + o_{\p}(1)$. Proposition \ref{prop:denominator} gives $a_N^2 \Gamma_{zp} = O_{\p}(1)$ together with $[a_N^2 \Gamma_{zp}]^{-1} = O_{\p}(1)$ for the leading denominator, as in the proofs of Theorems \ref{theorem:strong_consistency_aggregate} to \ref{theorem:weak_consistency_aggregate}. The first-stage term is therefore $o_{\p}(1)$ relative to the leading denominator and disappears asymptotically.

\textbf{Step 3 (instrument error).} Assumption~\ref{ass:mixing for stationary time series} gives $\e|\varepsilon_t|^{4+2\pi} < \infty$, and $\e[\Bar{u}_{jt} \varepsilon_t] = 0$ for every $j$ by Assumption~\ref{ass:weak stationarity and idiosyncrasy}, so Proposition~\ref{prop:first stage effect} applies with $Q_t = \varepsilon_t$. Writing $\Delta_{\tilde{F} \varepsilon} = \e[\tilde{F}_t' \varepsilon_t] \Sigma_{\tilde{F}}^{-1}$ and $\bar{\varepsilon}_t = \varepsilon_t - \Delta_{\tilde{F} \varepsilon} \tilde{F}_t$, the numerator is
\begin{align*}
    \text{Num} = \frac{a_N}{\sqrt{T}} \sum_{t=1}^T \hat{z}_t \varepsilon_t
    = \frac{a_N}{\sqrt{T}} \sum_{t=1}^T \bar{\varepsilon}_t \sum_{j=1}^N S_j \Bar{u}_{jt} + o_{\p}(1).
\end{align*}

\textbf{Step 4 (central limit theorem).} Step 3 has already put the numerator in the form needed. For $\bar{\varepsilon}_t = \varepsilon_t - \Delta_{\tilde{F} \varepsilon} \tilde{F}_t$,  $\{(\bar{\varepsilon}_t, u_t')\}$ is a mixing sequence under Assumption \ref{ass:mixing for stationary time series}. Using $(|a + b|)^p \leq 2^{p-1}(|a|^p + |b|^p)$, we have $\e|\bar{\varepsilon}_t |^{8 + 2\pi} \leq 2^{7+2\pi} \big( \e|\varepsilon_t|^{8 + 2 \pi} + \|\Delta_{\tilde{F} \varepsilon} \|^{8 + 2 \pi} \e \|\tilde{F}_t \|^{8 + 2 \pi} \big) < \infty$ under Assumption \ref{ass:mixing for stationary time series}. Both $\varepsilon_t$ and $\tilde{F}_t$ are orthogonal to $u_t$ and independent of $(\s, \Lambda)$, so $\e[\bar{\varepsilon}_t \Bar{u}_{jt} \mid \s, \Lambda] = 0$ for every $j$ and the conditional moment bound coincides with the unconditional one.

The long-run variance $a_N^2 V_{z \bar{\varepsilon}}(S)$ is non-degenerate under Assumption \ref{ass:mixing for stationary time series}. Theorem \ref{theorem:central limit theorem} for $\mu > 1$ and Corollary \ref{corollary:central limit theorem} for $\mu \in (0,1)$ therefore apply to the scaled sum, $X_T = \frac{a_N}{\sqrt{T}} \sum_{t=1}^T \bar{\varepsilon}_t \sum_{j=1}^N S_j \Bar{u}_{jt}$, with $\bar{\varepsilon}_t$ replacing $Z_t$, and give
    \begin{equation*}
        \frac{X_T}{a_N \sqrt{V_{z \bar{\varepsilon}}(S)}} \xrightarrow{d} \n(0,1)
    \end{equation*}
Throughout the appendix $V_{z \bar{\varepsilon}}(S)$ is the un-rescaled long-run variance of the theorem statements, $V_{z \bar{\varepsilon}}(S) = \lim_{T \to \infty} \frac{1}{T} \sum_{s=1}^T \sum_{t=1}^T \e[z_t z_s \bar{\varepsilon}_t \bar{\varepsilon}_s \mid S]$, so that $a_N^2 V_{z \bar{\varepsilon}}(S) = O_{\p}(1)$ and the $a_N$ carried by $X_T$ is matched by the $a_N$ in the norming. The final result follows from the steps for the known-factor denominator in Theorems \ref{theorem:strong_consistency_aggregate} and \ref{theorem:weak_consistency_aggregate} and the application of Slutsky's theorem. Specializing $a_N = 1$ recovers the rate $\sqrt{T}$ of Theorem~\ref{theorem:estimated_strong_consistency_aggregate}. Specializing $a_N = N^{\delta}$ with $\delta = \min(1 - 1/\mu, 1/2)$ recovers the rate $\sqrt{T}/N^{\delta}$ of Theorem~\ref{theorem:estimated_weak_consistency_aggregate}. Under $\mu > 2$ and $N/T \to c > 0$ the failure clause follows as in Theorem \ref{theorem:weak_consistency_aggregate}, since Proposition \ref{prop:first stage effect} leaves the rate unchanged. There, $a_N = \sqrt{N}$ and $\hat{\phi}_d - \phi_d = O_{\p}(a_N/\sqrt{T}) = O_{\p}(\sqrt{N/T}) = O_{\p}(1)$, so the estimator is inconsistent.

\subsection{Proofs of Theorems \ref{theorem:estimated_strong_consistency_disaggregate} and \ref{theorem:estimated_weak_consistency_disaggregate}} \label{subsec:proof of estimated_disaggregate}

Both theorems share a single argument in the unified scaling $a_N$. Specializing $a_N = 1$ recovers Theorem~\ref{theorem:estimated_strong_consistency_disaggregate} and $a_N = N^{\delta}$ recovers Theorem~\ref{theorem:estimated_weak_consistency_disaggregate}. Two estimation errors act at once here, the estimated instrument and the equal-weights aggregation, and the proof separates them.

Write $a \defeq \frac{e_N}{N} - H$ for the weight error and $\lambda_a \defeq \tilde{\Lambda}' a$ for its loading contraction. Both aggregators sum to one, so $e_N' a = 0$ and the equal-loadings factor $F_{1t}$ cancels from the difference,
\begin{equation} \label{eq:weight error}
    v_{Et} - v_{Ht} = (\Lambda F_t + u_t)' a = \tilde{F}_t' \lambda_a + u_t' a.
\end{equation}
The vector $a$ is deterministic, since both aggregators are functions of $\Omega$ alone, and the bound $\|H\|^2 \leq K^2/(\underline\lambda^2 N)$ from the proof of Lemma \ref{lemma:v_H multiplier} gives
\begin{equation} \label{eq:weight error norms}
    \|a\| \leq \left( 1 + \frac{K}{\underline\lambda} \right) N^{-1/2}, \qquad \|\lambda_a\| = O_{\p}(N^{-1/2}).
\end{equation}
For the second bound, $\sum_i a_i = 0$ removes the common mean of the loadings, and the $\tilde{\Blambda}_i$ are i.i.d.\ across $i$ with $\e \|\tilde{\Blambda}_i\|^2 < \infty$ by Assumption \ref{ass:size_of_firm_power_law}.3--4, so the cross terms vanish and $\e \|\lambda_a\|^2 = \sum_i a_i^2 \, \e \|\tilde{\Blambda}_i - \e \tilde{\Blambda}_i\|^2 \leq C \|a\|^2$. Markov gives the rate.

\textbf{Step 1 (decomposition).} The equal-weights residual is exact, $y_{Et} - \phi_s p_t = v_{Et}$, so
\begin{equation*}
    \frac{\sqrt{T}}{a_N} [\hat{\phi}_s - \phi_s] = \frac{\frac{a_N}{\sqrt{T}} \sum_{t=1}^T \hat{z}_t v_{Et}}{\frac{a_N^2}{T} \sum_{t=1}^T \hat{z}_t p_t}.
\end{equation*}
Substituting $\hat{z}_t = z_t - S'[\hat{C}_t - \tilde{C}_t]$ and $v_{Et} = v_{Ht} + (v_{Et} - v_{Ht})$ in the numerator,
\begin{equation*}
    \frac{a_N}{\sqrt{T}} \sum_{t=1}^T \hat{z}_t v_{Et}
    = \underbrace{\frac{a_N}{\sqrt{T}} \sum_{t=1}^T z_t v_{Ht} - \frac{a_N}{\sqrt{T}} \sum_{t=1}^T S'[\hat{C}_t - \tilde{C}_t] v_{Ht}}_{\text{instrument error}}
    + \underbrace{A_{NT} - B_{NT}}_{\text{aggregation error}},
\end{equation*}
where $A_{NT} \defeq \frac{a_N}{\sqrt{T}} \sum_t z_t (v_{Et} - v_{Ht})$ is the weight substitution and $B_{NT} \defeq \frac{a_N}{\sqrt{T}} \sum_t S'[\hat{C}_t - \tilde{C}_t](v_{Et} - v_{Ht})$ is the interaction of the two errors. The first term enters the central limit theorem in Step 6.

\textbf{Step 2 (denominator).} Neither $\hat{z}_t$ nor $p_t$ depends on the aggregation, so the denominator is the one treated in Appendix \ref{subsec:proof of estimated_theorems}. Proposition~\ref{prop:first stage effect} makes the instrument correction $o_{\p}(1)$, and Proposition \ref{prop:denominator} gives $a_N^2 \Gamma_{zp} = O_{\p}(1)$ together with $[a_N^2 \Gamma_{zp}]^{-1} = O_{\p}(1)$.

\textbf{Step 3 (instrument error).} Lemma \ref{lemma:v_H multiplier} establishes that $Q_t = v_{Ht}$ is measurable with respect to $\sigma(F_t, u_t)$, has $\e|v_{Ht}|^{8 + 2\pi}$ bounded uniformly in $N$, and satisfies $\e[\Bar{u}_{jt} v_{Ht}] = 0$ for every $j$, this last being the defining property of $H$. Proposition~\ref{prop:first stage effect} therefore applies at $Q_t = v_{Ht}$, and with $\Delta_{\tilde{F} v} \defeq \e[\tilde{F}_t' v_{Ht} \mid \Lambda] \Sigma_{\tilde{F}}^{-1}$ and $\bar{v}_{Ht} = v_{Ht} - \Delta_{\tilde{F} v} \tilde{F}_t$ gives
\begin{equation*}
    \frac{a_N}{\sqrt{T}} \sum_{t=1}^T \hat{z}_t \, v_{Ht}
    = \frac{a_N}{\sqrt{T}} \sum_{t=1}^T z_t \, \bar{v}_{Ht} + o_{\p}(1).
\end{equation*}

\textbf{Step 4 (aggregation error).} By \eqref{eq:weight error}, $\e[F_t u_t'] = 0$ and $H' \Omega D_N = 0$,
\begin{equation*}
    \e[z_t (v_{Et} - v_{Ht}) \mid S, \Lambda] = \frac{S' D_N \Omega e_N}{N} - S' D_N \Omega H = b_N,
\end{equation*}
the moment violation of Section \ref{subsec:equal weights}. Centring, $A_{NT} = a_N \sqrt{T} \, b_N + \zeta_{NT}$ with $\zeta_{NT} = \frac{a_N}{\sqrt{T}} \sum_t \xi_t$ and $\xi_t \defeq z_t(v_{Et} - v_{Ht}) - b_N$.

The drift is Lemma \ref{lemma:hajek bound}, $a_N \sqrt{T} \, b_N = \frac{\sqrt{T}}{N} a_N \beta = O_{\p}(\sqrt{T}/N)$.

For the noise, condition on $(S, \Lambda)$ and set $q_0 = 2(2 + \pi)$. The Rosenthal bound of Assumption \ref{ass:LN_time and cross sectional dependence}.5, applied with the weights $\tilde{S}$ and the $\ell_q$ inequality of Lemma \ref{lemma:share inequalities}.1, gives $\|z_t\|_{q_0 \mid S} \leq C h_N^{1/2}$, and applied with the weights $a$ gives $\|u_t' a\|_{q_0} \leq C \|a\|$. With $\|\lambda_a\|$ and $\|a\|$ from \eqref{eq:weight error norms}, $\e\|\tilde{F}_t\|^{q_0} < \infty$ from Assumption \ref{ass:mixing for stationary time series}, and $|b_N| \leq \|\tilde{S}\| \|\tilde{r}\|/N \leq M h_N^{1/2} N^{-1/2}$ by Cauchy--Schwarz,
\begin{equation*}
    \|\xi_t\|_{2 + \pi \mid S, \Lambda} = O_{\p}\left( h_N^{1/2} N^{-1/2} \right) = O_{\p}\left( a_N^{-1} N^{-1/2} \right),
\end{equation*}
the last equality by the strength bound $a_N \sqrt{h_N} = O_{\p}(1)$ of Proposition \ref{prop:weakness of known factor instrument}. Conditionally on $(S, \Lambda)$ the $\xi_t$ inherit the mixing coefficients of $\{(F_t', u_t')\}$, so Davydov's inequality \citep{Rio1993CovarianceProcesses} and conditional Chebyshev give $\zeta_{NT} = O_{\p}(N^{-1/2})$. Hence
\begin{equation*}
    A_{NT} = O_{\p}\left( \frac{\sqrt{T}}{N} \right) + O_{\p}\left( \frac{1}{\sqrt{N}} \right) = o_{\p}(1).
\end{equation*}
Neither bound depends on the regime, and neither uses $N/T$. The moment violation scales with the strength of the instrument and $a_N$ is the exact inverse of that strength, so the two cancel.

\textbf{Step 5 (interaction).} The multiplier $Q_t = v_{Et} - v_{Ht}$ is measurable with respect to $\sigma(F_t, u_t)$ and has finite $4 + 2\pi$ moments by \eqref{eq:weight error}, so Lemma \ref{lemma:numerator}.(i) applies to $B_{NT}$. Part (ii) does not, since $\e[\Bar{u}_{jt} Q_t] = (D_N \Omega a)_j \neq 0$ in general. I take the three pieces of the expansion in turn.

The first piece carries $\frac{1}{T} \sum_t \tilde{F}_t Q_t$. Substituting \eqref{eq:weight error} and using $\e[\tilde{F}_t u_t'] = 0$,
\begin{equation*}
    \frac{1}{T} \sum_{t=1}^T \tilde{F}_t Q_t = \left[ \Sigma_{\tilde{F}} + O_{\p}(T^{-1/2}) \right] \lambda_a + O_{\p}(T^{-1/2}) \|a\| = O_{\p}(N^{-1/2}),
\end{equation*}
the fluctuation rates following from Davydov and Chebyshev as in Step 4. Multiplied by the two $O_{\p}(1)$ factors of Lemma \ref{lemma:numerator}.(i), this piece is $o_{\p}(1)$.

The second piece carries $\frac{1}{\sqrt{T}} \sum_t \zeta_t Q_t$ with $\zeta_t = \frac{1}{N} \sum_j \tilde{\Blambda}_j \Bar{u}_{jt}$. Split it into drift and fluctuation. The drift obeys
\begin{equation*}
    \left| \e[\zeta_t Q_t \mid \Lambda] \right| = \left| \frac{1}{N} \sum_j \tilde{\Blambda}_j (D_N \Omega a)_j \right| \leq \frac{\Bar{\Blambda}}{\sqrt{N}} \|D_N \Omega a\| \leq \Bar{\Blambda} K \|a\| N^{-1/2} = O(N^{-1}),
\end{equation*}
by Cauchy--Schwarz across $j$, then Assumption \ref{ass:bounded eigenvalues} and \eqref{eq:weight error norms}, so it contributes $O(\sqrt{T}/N)$. For the fluctuation, the Rosenthal bound with the weights $\tilde{\Blambda}_j/N$ gives $\|\zeta_t\|_{q_0 \mid \Lambda} = O_{\p}(N^{-1/2})$, and $\|Q_t\|_{q_0 \mid S, \Lambda} = O_{\p}(N^{-1/2})$ from Step 4, so the summand is $O_{\p}(N^{-1})$ in conditional $L^{2 + \pi}$ norm and Davydov gives the fluctuation the same order.

The third piece is the pair of remainders of Lemma \ref{lemma:numerator}.(i), $O_{\p}(\sqrt{T}/N) + O_{\p}(1/\sqrt{T})$. Every piece vanishes under $\sqrt{T}/N \to 0$ and $T \to \infty$, so $B_{NT} = o_{\p}(1)$.

\textbf{Step 6 (central limit theorem).} Insert Steps 3 to 5 into Step 1 and write $z_t = \Bar{u}_{St}$. The numerator is
\begin{equation*}
    \frac{a_N}{\sqrt{T}} \sum_{t=1}^T \hat{z}_t v_{Et}
    = \frac{a_N}{\sqrt{T}} \sum_{t=1}^T \left( v_{Ht} - \Delta_{\tilde{F} v} \tilde{F}_t \right) \Bar{u}_{St} + o_{\p}(1)
    = \frac{a_N}{\sqrt{T}} \sum_{t=1}^T \Bar{v}_{Ht} \Bar{u}_{St} + o_{\p}(1),
\end{equation*}
with $\Bar{v}_{Ht}$ the recentered residual of the theorems. Lemma \ref{lemma:v_H multiplier} gives the multiplier conditions for $v_{Ht}$, and the $c_r$ inequality extends them to $\Bar{v}_{Ht}$ as in Appendix \ref{subsec:proof of estimated_theorems}, while $\e[\Bar{v}_{Ht} \Bar{u}_{jt} \mid \s, \Lambda] = \e[v_{Ht} \Bar{u}_{jt} \mid \s, \Lambda] - \Delta_{\tilde{F} v} \e[\tilde{F}_t \Bar{u}_{jt} \mid \s, \Lambda] = 0$ for every $j$, the first term by Lemma \ref{lemma:v_H multiplier} and the second by the independence of $\tilde{F}_t$ from $(\s, \Lambda)$. The long-run variance $a_N^2 V_{z \bar{v}_H}(S, \Lambda)$ is non-degenerate under Assumption \ref{ass:mixing for stationary time series}. Theorem \ref{theorem:central limit theorem} for $\mu > 1$ and Corollary \ref{corollary:central limit theorem} for $\mu \in (0,1)$, applied conditionally on $\sigma(S, \Lambda)$, therefore give
\begin{equation*}
    \frac{1}{\sqrt{V_{z \bar{v}_H}(S, \Lambda)}} \cdot \frac{1}{\sqrt{T}} \sum_{t=1}^T \Bar{v}_{Ht} \Bar{u}_{St} \xrightarrow{d} \n(0,1).
\end{equation*}
The variance is the un-rescaled one of the theorem statements, so the $a_N$ of the display above cancels against the $a_N$ in $\sqrt{a_N^2 V_{z \bar{v}_H}(S, \Lambda)}$ and neither appears here.
Slutsky's theorem applied to the ratio, with Step 2 for the denominator, completes the proof. Specializing $a_N = 1$ recovers the rate $\sqrt{T}$ of Theorem \ref{theorem:estimated_strong_consistency_disaggregate}, and $a_N = N^{\delta}$ the rate $\sqrt{T}/N^{\delta}$ of Theorem \ref{theorem:estimated_weak_consistency_disaggregate}. The failure clause for $\mu > 2$ with $N/T \to c$ follows as in Theorem \ref{theorem:estimated_weak_consistency_aggregate}, since the aggregation error is $o_{\p}(1)$ whatever the trajectory of $N/T$. No choice of aggregation vector restores consistency. By Step 2 the aggregation enters only the numerator, and the denominator that drives the failure is the $H$-free object of Appendix \ref{app:proof of denominator}.

\subsection{Proof of Proposition \ref{prop:consistency of HAC_aggregate}} \label{subsec:proof or HAC_aggregate}

We use the unified scaling $a_N$ of \eqref{eq:a_N}. Specializing $a_N = 1$, $N^{\delta}$, and $\sqrt{N}$ covers $\mu \in (0,1)$, $\mu \in (1,2)$, and $\mu > 2$, respectively. The conditional asymptotic variance is the un-rescaled one of the theorem statements,
\[
    V_{z \bar{\varepsilon}}(S) = \text{plim}_{N,T \to \infty} \left[ \frac{1}{T} \sum_{s=1}^T \sum_{t=1}^T \e[z_t z_s \bar{\varepsilon}_t \bar{\varepsilon}_s | S] \right]
\]
which is $O_{\p}(a_N^{-2})$. The estimator below carries the rescaling $a_N^2$, so what is proved is $\hat{V}_{{z} \bar{\varepsilon}} - a_N^2 V_{z \bar{\varepsilon}}(S) \xrightarrow{p} 0$, with both terms $O_{\p}(1)$.
The HAC estimator is
\[
\hat{V}_{{z} \bar{\varepsilon}} = \frac{a_N^2}{T} \sum_{t=1}^T \hat{z}_t^2 \tilde{\varepsilon}_t^2 + \frac{2 a_N^2}{T} \sum_{s=1}^{b_T} w(\frac{s}{b_T}) \sum_{t = s+1}^T \hat{z}_t \hat{z}_s \tilde{\varepsilon}_t \tilde{\varepsilon}_s
\]
where $b_T$ is the bandwidth and $w(x)$ is a kernel function, $w: \mathbb{R}^+ \to [0,1]$ such that $w(x) = 0$ for $x>1$ and $w(0) = 1$. $\tilde{\varepsilon}_t = \hat{\varepsilon}_t - \frac{1}{T} \sum_t \hat{F}'_t \hat{H}' \hat{\varepsilon}_t \cdot \hat{\Sigma}_{\tilde{F}}^{-1} \cdot \hat{H} \hat{F}_t $, where $\hat{\varepsilon}_t = d_t - \hat{\phi}_d p_{t}$ and $\hat{\Sigma}_{\tilde{F}} = \left[ \hat{F}'_t \hat{H}' \hat{H} \hat{F}_t / {T} \right]$.

I apply Proposition 4.1 of \citet{Kojevnikov2021LimitVariables} to show the consistency of the HAC estimator. But before I do so, I need to replace the estimators in the expression for the HAC estimator with their true values. That is, replace $\hat{z}_t$ with $z_t$ and so on.
\begin{align*}
    \hat{\varepsilon}_t = d_t - \hat{\phi}_d p_{t} = \varepsilon_t + O_{\p}\left(\frac{a_N}{\sqrt{T}} \right) = \varepsilon_t + o_{\p}(1)
\end{align*}
The $o_{\p}(1)$ statement uses $a_N \leq \sqrt{N}$ together with $N/T \to 0$. The strong regime ($a_N = 1$) requires only $T \to \infty$.

Recall from Appendix \ref{app:estimation of common factors}
\begin{align*}
    \hat{F}_t - \hat{H}^{-1} \tilde{F}_t &= O_{\p} \left( \frac{1}{\sqrt{N}} \right) = o_{\p}(1) 
\end{align*}

Analyze the display below term by term
\begin{equation*}
    \frac{1}{T} \sum_t \hat{F}'_t  \hat{\varepsilon}_t \cdot \hat{\Sigma}_{\tilde{F}}^{-1} \cdot \hat{F}_t
\end{equation*}

Adding and subtracting $\hat{H}^{-1} \tilde{F}_t$ to $\hat{F}_t$,
\begin{align*}
    \frac{1}{T} \sum_t \hat{F}'_t  \hat{\varepsilon}_t &= \frac{1}{T} \sum_t [\hat{F}_t - \hat{H}^{-1} \tilde{F}_t + \hat{H}^{-1} \tilde{F}_t  ]' \hat{\varepsilon}_t =  \frac{1}{T} \sum_t \tilde{F}'_t  \hat{\varepsilon}_t \cdot \hat{H}^{-1'} + o_{\p}(1) \\
    \hat{\Sigma}_{\tilde{F}}^{-1} &= \left[ \frac{\hat{F}' \hat{F}}{T} \right]^{-1} = \left[ \frac{1}{T} \sum_{t=1}^T \hat{F}_t \hat{F}_t' \right]^{-1} + o_{\p}(1) \\
    &= \left[ \hat{H}^{-1} \frac{1}{T} \sum_{t=1}^T \tilde{F}_t \tilde{F}_t' \hat{H}^{-1'} \right]^{-1} + o_{\p}(1) \\
    &= \big[ \hat{H}^{-1} \Sigma_{\tilde{F}} \hat{H}^{-1'} \big]^{-1} + o_{\p}(1) = \hat{H} \Sigma_{\tilde{F}}^{-1} \hat{H}' + o_{\p}(1)
\end{align*}

Thus, combining all terms,
\begin{align*}
    \frac{1}{T} \sum_t \hat{F}'_t  \hat{\varepsilon}_t \cdot \hat{\Sigma}_{\tilde{F}}^{-1} \cdot \hat{F}_t &= \frac{1}{T} \sum_t \tilde{F}'_t  \hat{\varepsilon}_t  \hat{H}^{-1'} \cdot \hat{H} \Sigma_{\tilde{F}}^{-1} \hat{H}' \cdot \hat{H}^{-1} \tilde{F_t} +  o_{\p}(1) \\
    &= \frac{1}{T} \sum_t \tilde{F}'_t  {\varepsilon}_t \cdot {\Sigma}_{\tilde{F}}^{-1} \cdot  \tilde{F}_t + o_{\p}(1)
\end{align*}

Thus,
\begin{equation*}
    \tilde{\varepsilon}_t = \varepsilon_t - \frac{1}{T} \sum_t \tilde{F}'_t  {\varepsilon}_t \cdot {\Sigma}_{\tilde{F}}^{-1} \cdot  \tilde{F}_t + o_{\p}(1) = \bar{\varepsilon}_t + o_{\p}(1)
\end{equation*}

The relation $\hat{z}_t - z_t = -S'(\hat{C}_t - \tilde{C}_t)$, together with the influence-function expansion underlying Lemma~\ref{lemma:denominator}, gives $\hat{z}_t = z_t + o_{\p}(a_N^{-1})$. This is the rate needed for the substitution in $\hat{V}_{z \bar{\varepsilon}}$ at the rescaled $a_N^2$ scale. Thus we have
\begin{equation*}
    \hat{V}_{{z} \bar{\varepsilon}} = \frac{a_N^2}{T} \sum_{t=1}^T {z}_t^2 \bar{\varepsilon}_t^2 + \frac{2 a_N^2}{T} \sum_{s=1}^{b_T} w(\frac{s}{b_T}) \sum_{t = s+1}^T {z}_t {z}_s \bar{\varepsilon}_t \bar{\varepsilon}_s + o_{\p}(1)
\end{equation*}

Now, I am ready to apply Proposition 4.1 of \citet{Kojevnikov2021LimitVariables}. To make our setting comparable to the mentioned Proposition, define the following terms. For $X_t = a_N z_t \bar{\varepsilon}_t$, the rescaled summand whose long-run variance is bounded above and away from zero,
\begin{align*}
    \Omega_T(s) &= T^{-1} \sum_{t=1}^T \e[X_t X_{t-s}| \s] + \e[X_t X_{t+s}| \s] \\
    &= \frac{2}{T} \sum_{t=1}^T \e[X_t X_{t-s}| \s]
\end{align*}
for $s \neq 0$. The second equality in the display above comes from stationarity. For $s=0$, define
\begin{equation*}
    \Omega_T(0) = T^{-1} \sum_{t=1}^T \e[X_t^2| \s]
\end{equation*}

Using these new definitions,
\begin{equation*}
     a_N^2 V_{z \bar{\varepsilon}}(S) = \text{plim}_{N,T \to \infty} \left[\sum_{s=0}^{T-1} \Omega_T(s)  \right]
\end{equation*}

Now, define the sample version of the terms above. For $Y_t = a_N \hat{z}_t \tilde{\varepsilon}_t$, the sample counterpart of $X_t$,
\begin{align*}
    \tilde{\Omega}(s) &= \frac{2}{T} \sum_{t=1}^T Y_t Y_{t-s}, \quad s \neq 0 \\
    \tilde{\Omega}(0) &= \frac{1}{T} \sum_{t=1}^T Y_t^2
\end{align*}
Using these definitions,
\begin{equation*}
    \hat{V}_{{z} \bar{\varepsilon}} = \sum_{s=0}^T w(\frac{s}{b_T}) \tilde{\Omega}(s)
\end{equation*}
For convenience, call $w(\frac{s}{b_T}) = w_T(s)$.

Now I will verify Assumption 4.1 of \citet{Kojevnikov2021LimitVariables}. By replacing $Z_t$ with $\varepsilon_t$ in Theorem~\ref{theorem:central limit theorem} for $\mu > 1$ and Corollary~\ref{corollary:central limit theorem} for $\mu \in (0,1)$, we have verified Condition ND. Condition ND implies Assumption 4.1(i) and (iii). We just need to verify Assumption 4.1(ii). The idea of this assumption, like most others in that paper is that neighbors should not grow too fast as $T \to \infty$. Any bandwidth, $b_T$ and kernel function, $w(x)$ that satisfy Proposition 4.2 of \citet{Kojevnikov2021LimitVariables} satisfies Assumption 4.1(ii). Proposition 4.2 states that, there exist some constants, $C$ and $\eta$ such that
\begin{equation*}
    |w(x) - 1| \leq C |x|^{1 + \eta}
\end{equation*}
and $\frac{\log T}{b_T} = O_{\text{a.s}}(1)$. For the Newey-West estimator, 
$w(x) = 1 - \frac{x}{b_T+1}$ for all $x < b_T$ and $w(x) = 0$ for $ x \geq b_T$. For the optimal trade-off between size and power, \citet{Lazarus2018HARPractice} recommends $b_T = 1.3T^{\frac{1}{2}}$. For this choice of $b_T$, we have $\frac{\log T}{b_T} = o_{\text{a.s}}(1) =  O_{\text{a.s}}(1)$.

Thus Proposition 4.1 of \citet{Kojevnikov2021LimitVariables} applies and conditional on $\s$,
\begin{equation*}
        \hat{V}_{{z} \bar{\varepsilon}} - a_N^2 V_{z \bar{\varepsilon}}(S) \xrightarrow{p} 0.
    \end{equation*}

\subsection{Proof of Proposition \ref{prop:consistency of HAC_disaggregate}} \label{subsec:proof of HAC_disaggregate}

We use the same unified scaling $a_N$, with $a_N = 1$ for Theorem \ref{theorem:estimated_strong_consistency_disaggregate} and $a_N = N^{\delta}$ for Theorem \ref{theorem:estimated_weak_consistency_disaggregate}. The conditional asymptotic variance is again the un-rescaled one of the theorem statements,
\[
    V_{z \bar{v}_H}(S, \Lambda) = \text{plim}_{N,T \to \infty} \left[ \frac{1}{T} \sum_{s=1}^T \sum_{t=1}^T \e[z_t z_s \bar{v}_{Ht} \bar{v}_{Hs} | S, \Lambda] \right]
\]
The HAC estimator is
\[
    \hat{V}_{z v_H} = \frac{a_N^2}{T} \sum_{t=1}^T \hat{z}_t^2 \tilde{v}_{Et}^2 + \frac{2 a_N^2}{T} \sum_{s=1}^{b_T} w(\frac{s}{b_T}) \sum_{t = s+1}^T \hat{z}_t \hat{z}_s \tilde{v}_{Et} \tilde{v}_{Es}
\]
with the bandwidth $b_T$ and the kernel $w$ of Appendix \ref{subsec:proof or HAC_aggregate}. Here $\tilde{v}_{Et} = \hat{v}_{Et} - \frac{1}{T} \sum_t \hat{F}'_t \hat{v}_{Et} \cdot \hat{\Sigma}_{\tilde{F}}^{-1} \cdot \hat{F}_t$, where $\hat{v}_{Et} = y_{Et} - \hat{\phi}_s p_t$ and $\hat{\Sigma}_{\tilde{F}} = [\hat{F}' \hat{F}/T]$.

I apply Proposition 4.1 of \citet{Kojevnikov2021LimitVariables} again. But before I do so, I need to replace the estimators in the display above with their true values. That is, replace $\hat{z}_t$ with $z_t$ and $\tilde{v}_{Et}$ with $\bar{v}_{Ht}$. The instrument and the factors go as in the aggregate case. The multiplier carries one step more, since the estimator is built with equal weights but the variance it estimates is built by the optimal aggregator.

Start with $\hat{\phi}_s$, which enters the estimator through the residual $\hat{v}_{Et}$. Its rate comes from inverting the studentised limit of the theorems, so the bound $a_N^2 V_{z \bar{v}_H}(S, \Lambda) = O_{\p}(1)$ comes first. Conditionally on $(S, \Lambda)$, Cauchy--Schwarz with $q_0 = 2(2 + \pi)$ factorises the summand of $a_N^2 V_{z \bar{v}_H}(S, \Lambda)$ as
\begin{equation*}
    \|a_N z_t \bar{v}_{Ht}\|_{2 + \pi \mid S, \Lambda} \leq a_N \|z_t\|_{q_0 \mid S} \, \|\bar{v}_{Ht}\|_{q_0 \mid \Lambda}.
\end{equation*}
The first factor is $O_{\p}(1)$, by the Rosenthal bound $\|z_t\|_{q_0 \mid S} \leq C h_N^{1/2}$ of Appendix \ref{subsec:proof of estimated_disaggregate} and the strength bound $a_N \sqrt{h_N} = O_{\p}(1)$ of Proposition \ref{prop:weakness of known factor instrument}. The second is finite by Lemma \ref{lemma:v_H multiplier} and the $c_r$ inequality. Summing the Davydov covariance bounds over lags \citep{Rio1993CovarianceProcesses} gives $a_N^2 V_{z \bar{v}_H}(S, \Lambda) = O_{\p}(1)$. With $[a_N^2 \Gamma_{zp}]^{-1} = O_{\p}(1)$ from Proposition \ref{prop:denominator}, the studentised limit of Theorems \ref{theorem:estimated_strong_consistency_disaggregate} and \ref{theorem:estimated_weak_consistency_disaggregate} inverts to $\hat{\phi}_s - \phi_s = O_{\p}(a_N/\sqrt{T})$, and so
\begin{align*}
    \hat{v}_{Et} = y_{Et} - \hat{\phi}_s p_t = v_{Et} + O_{\p}\left( \frac{a_N}{\sqrt{T}} \right) = v_{Et} + o_{\p}(1).
\end{align*}
The $o_{\p}(1)$ statement uses $a_N \leq \sqrt{N}$ together with $N/T \to 0$. The strong regime ($a_N = 1$) requires only $T \to \infty$.

Take the factors and the instrument next. Every line of Appendix \ref{subsec:proof or HAC_aggregate} applies with $\varepsilon_t$ replaced by $v_{Et}$, since that argument uses only the mixing and moment conditions of Assumption \ref{ass:mixing for stationary time series}. Lemma \ref{lemma:v_H multiplier} supplies those for $v_{Ht}$, and \eqref{eq:weight error} with the $c_r$ inequality extends them to $v_{Et}$. Carrying through the factor reduction $\hat{F}_t = \hat{H}^{-1} \tilde{F}_t + O_{\p}(N^{-1/2})$ and the instrument reduction $\hat{z}_t = z_t + o_{\p}(a_N^{-1})$ under $\sqrt{T}/N \to 0$,
\begin{equation*}
    \tilde{v}_{Et} = \bar{v}_{Et} + o_{\p}(1), \qquad \bar{v}_{Et} \defeq v_{Et} - \Delta_{\tilde{F} v_E} \tilde{F}_t, \qquad \Delta_{\tilde{F} v_E} \defeq \e[\tilde{F}_t' v_{Et} \mid \Lambda] \Sigma_{\tilde{F}}^{-1}.
\end{equation*}
The law of large numbers runs at $\Lambda$ fixed, so the recentering coefficient is the conditional one, here and in $\Delta_{\tilde{F} v}$.

Only the multiplier is left, and it is the one term with no aggregate counterpart. Substituting \eqref{eq:weight error}, using $\e[\tilde{F}_t u_t'] = 0$ and the $\Lambda$-measurability of $\lambda_a$,
\begin{equation*}
    \Delta_{\tilde{F} v_E} - \Delta_{\tilde{F} v} = \e[\tilde{F}_t' (v_{Et} - v_{Ht}) \mid \Lambda] \Sigma_{\tilde{F}}^{-1} = \lambda_a' \Sigma_{\tilde{F}} \Sigma_{\tilde{F}}^{-1} = \lambda_a',
\end{equation*}
so the factor part of the weight error cancels against the recentering and
\begin{equation*}
    \bar{v}_{Et} - \bar{v}_{Ht} = u_t' a.
\end{equation*}
The Rosenthal bound of Appendix \ref{subsec:proof of estimated_disaggregate} with the weights $a$ gives $\|u_t' a\|_{q_0} \leq C \|a\| = O(N^{-1/2})$ by \eqref{eq:weight error norms}. Write $V_{z(\bar{v}_E - \bar{v}_H)}$ for the long-run variance of $a_N z_t u_t' a$, drift included. The conditional mean is $a_N b_N$, the moment violation of Section \ref{subsec:equal weights}, and $a_N |b_N| = O_{\p}(a_N^{-1}/N)$ by Lemma \ref{lemma:hajek bound}. The Davydov summation above, with $a_N \|z_t\|_{q_0 \mid S} = O_{\p}(1)$, gives $V_{z(\bar{v}_E - \bar{v}_H)} = O_{\p}(N^{-1})$. By Cauchy--Schwarz in the long-run-variance metric the perturbation of $\hat{V}_{z v_H}$ is bounded by
\begin{equation*}
    2 \left( V_{z(\bar{v}_E - \bar{v}_H)} \right)^{1/2} \left( a_N^2 V_{z \bar{v}_H} \right)^{1/2} + V_{z(\bar{v}_E - \bar{v}_H)} = O_{\p}(N^{-1/2}).
\end{equation*}
The bound is uniform in the bandwidth, since $|w| \leq 1$ and the mixing coefficients satisfy $\sum_h \alpha(h)^{\pi/(2 + \pi)} < \infty$ under Assumption \ref{ass:mixing for stationary time series}, so the lag sum converges rather than accumulating with $b_T$.

Thus we have
\begin{equation*}
    \hat{V}_{z v_H} = \frac{a_N^2}{T} \sum_{t=1}^T z_t^2 \bar{v}_{Ht}^2 + \frac{2 a_N^2}{T} \sum_{s=1}^{b_T} w(\tfrac{s}{b_T}) \sum_{t = s+1}^T z_t z_s \bar{v}_{Ht} \bar{v}_{Hs} + o_{\p}(1),
\end{equation*}
the known-factor HAC estimator of $a_N^2 V_{z \bar{v}_H}(S, \Lambda)$. Condition ND holds for the multiplier $\bar{v}_{Ht}$ by Lemma \ref{lemma:v_H multiplier} and the $c_r$ inequality, through Theorem \ref{theorem:central limit theorem} for $\mu > 1$ and Corollary \ref{corollary:central limit theorem} for $\mu \in (0,1)$. The bandwidth and kernel conditions are those verified in Appendix \ref{subsec:proof or HAC_aggregate}. Proposition 4.1 of \citet{Kojevnikov2021LimitVariables} therefore applies conditional on $(S, \Lambda)$ and gives $\hat{V}_{z v_H} - a_N^2 V_{z \bar{v}_H}(S, \Lambda) \xrightarrow{p} 0$.

\subsection{Proof of Proposition \ref{prop:weakness robust test_aggregate}} \label{subsec:proof of AR_aggregate}

The statistic is evaluated at the null, so it is a function of the data and $\phi_d^0$ only. The proof has two parts, the numerator and the variance estimator, combined by Slutsky's theorem.

\textbf{Numerator.} Under $H_0$, the inner sum is the estimated-factor numerator at the true value. With $a_N = \sqrt{N}$,
\begin{equation*}
    \sqrt{\tfrac{N}{T}} \sum_t \hat{z}_t (y_{St} - \phi_d^0 p_t) = \frac{a_N}{\sqrt{T}} \sum_t \hat{z}_t \varepsilon_t.
\end{equation*}
The proof of Theorem~\ref{theorem:estimated_weak_consistency_aggregate} shows this converges in distribution to $\n(0, N V_{z\bar{\varepsilon}}(S))$ conditional on $\s$, through the numerator central limit theorem (Theorem~\ref{theorem:central limit theorem}) and the first-stage reduction $\hat{z}_t = z_t + o_{\p}(a_N^{-1})$ under $\sqrt{T}/N \to 0$. Squaring,
\begin{equation*}
    \tfrac{N}{T} \Big( \sum_t \hat{z}_t \varepsilon_t \Big)^2 \xrightarrow{d} N V_{z\bar{\varepsilon}}(S) \cdot \chi^2_1.
\end{equation*}

\textbf{Variance.} The estimator $\hat{V}_{z\bar{\varepsilon}}(\phi_d^0)$ is the HAC estimator of Proposition~\ref{prop:consistency of HAC_aggregate}, built from the null-imposed residual $\tilde{\varepsilon}_t(\phi_d^0)$. Under $H_0$, $d_t - \phi_d^0 p_t = \varepsilon_t$ exactly, so the residual is the true structural error and the reduction $\hat{\varepsilon}_t \to \varepsilon_t$ in that proof holds immediately rather than up to $O_{\p}(a_N/\sqrt{T})$. The remaining steps are unchanged: the factor reduction $\hat{F}_t = \hat{H}^{-1} \tilde{F}_t + O_{\p}(1/\sqrt{N})$, the instrument reduction $\hat{z}_t = z_t + o_{\p}(a_N^{-1})$ under $\sqrt{T}/N \to 0$, and Proposition 4.1 of \citet{Kojevnikov2021LimitVariables} conditional on $\s$. Hence
\begin{equation*}
    N^{-1} \hat{V}_{z\bar{\varepsilon}}(\phi_d^0) \xrightarrow{p} V_{z\bar{\varepsilon}}(S).
\end{equation*}

\textbf{Combining.} The rate $a_N^2 = N$ cancels between the numerator and $\hat{V}_{z\bar{\varepsilon}}(\phi_d^0)$. By Slutsky's theorem, $\text{AR}(\phi_d^0) \xrightarrow{d} \chi^2_1$. The only growth condition the argument uses is $\sqrt{T}/N \to 0$.

\subsection{Proof of Proposition \ref{prop:weakness robust test_disaggregate}} \label{subsec:proof of AR_disaggregate}

The statistic is evaluated at the null, so it is a function of the data and $\phi_s^0$ only. The proof follows Appendix \ref{subsec:proof of AR_aggregate}, with the equal-weights aggregation as the one addition.

\textbf{Numerator.} Under $H_0$, $y_{Et} - \phi_s^0 p_t = v_{Et}$ exactly. With $a_N = \sqrt{N}$,
\begin{equation*}
    \sqrt{\tfrac{N}{T}} \sum_t \hat{z}_t (y_{Et} - \phi_s^0 p_t) = \frac{a_N}{\sqrt{T}} \sum_t \hat{z}_t v_{Et}.
\end{equation*}
The proof of Theorem \ref{theorem:estimated_weak_consistency_disaggregate} shows this converges in distribution to $\n(0, N V_{z \bar{v}_H}(S, \Lambda))$ conditional on $(S, \Lambda)$, through the instrument error, the aggregation error, the interaction and the central limit theorem of Appendix \ref{subsec:proof of estimated_disaggregate}. Only the denominator of that proof consumes $N/T$, and the statistic here has no denominator of that kind. The aggregation error is $O_{\p}(\sqrt{T}/N) + O_{\p}(N^{-1/2})$ and the interaction is $o_{\p}(1)$, both under $\sqrt{T}/N \to 0$ and $T \to \infty$ alone. Squaring,
\begin{equation*}
    \tfrac{N}{T} \Big( \sum_t \hat{z}_t v_{Et} \Big)^2 \xrightarrow{d} N V_{z \bar{v}_H}(S, \Lambda) \cdot \chi^2_1.
\end{equation*}

\textbf{Variance.} The estimator $\hat{V}_{z v_H}(\phi_s^0)$ is the one of Proposition \ref{prop:consistency of HAC_disaggregate}, built from the null-imposed residual. Under $H_0$ that residual is $v_{Et}$ exactly, so the reduction $\hat{v}_{Et} = v_{Et} + O_{\p}(a_N/\sqrt{T})$ in that proof holds with no error at all. This is the one step of Appendix \ref{subsec:proof of HAC_disaggregate} that consumed $N/T$, and with it goes the restriction. The rest is unchanged, namely the factor reduction $\hat{F}_t = \hat{H}^{-1} \tilde{F}_t + O_{\p}(N^{-1/2})$, the instrument reduction $\hat{z}_t = z_t + o_{\p}(a_N^{-1})$ under $\sqrt{T}/N \to 0$, the identity $\bar{v}_{Et} - \bar{v}_{Ht} = u_t' a$ with its $O_{\p}(N^{-1/2})$ perturbation of the long-run variance, and Proposition 4.1 of \citet{Kojevnikov2021LimitVariables} conditional on $(S, \Lambda)$. Hence
\begin{equation*}
    N^{-1} \hat{V}_{z v_H}(\phi_s^0) \xrightarrow{p} V_{z \bar{v}_H}(S, \Lambda).
\end{equation*}

\textbf{Combining.} The rate $a_N^2 = N$ cancels between the numerator and $\hat{V}_{z v_H}(\phi_s^0)$. By Slutsky's theorem, $\text{AR}_H(\phi_s^0) \xrightarrow{d} \chi^2_1$. The only growth condition the argument uses is $\sqrt{T}/N \to 0$.

\section{Notes on \citet{Banafti2022InferentialDimensions}} \label{appendix:banafti}

\citet{Banafti2022InferentialDimensions} conclude in their Theorems 2 and 4 that estimating the factor structure does not affect the asymptotic distribution. Theorem~\ref{theorem:estimated_strong_consistency_aggregate} of this paper concludes that it does.

The difference arises due to two reasons. The first is an assumption they impose. Their Assumption 4(iii) splits the cross section into a dominant block and a fringe. I show this Assumption is incompatible with the power law their own framework needs. This Assumption is what makes their Lemma 1(ii) and Term IV of their Lemma 2 go through. My Proposition~\ref{prop:behavior of S times lambda} corrects both results without their Assumption 4(iii). Thus, while the proofs survive with appropriate modifications, I establish the asymptotic distributional results without imposing the infeasible Assumption 4(iii).

The second is their conclusion about the variance of the limiting distribution. They show that Term II of their Lemma 2 disappears asymptotically and do not add to the variance of the estimator. I show that this is possible only under very strict assumptions, ones which \citet{Banafti2022InferentialDimensions} do not impose. Under the general setup, this Term II adds to the asymptotic variance of the GIV estimator.

I will proceed further as follows. First, I show that Assumption 4(iii) is incompatible with fat tails. Then I show how \citet{Banafti2022InferentialDimensions}'s Lemma 1(ii) and Term IV of Lemma 2 fails without Assumption 4(iii). I then rescue these proofs under general conditions using my Proposition \ref{prop:behavior of S times lambda}. Then I show that Term II of their Lemma 2 adds to the asymptotic variance of the GIV estimator\footnote{I write the rotation matrix as $\hat H$ to stress that it is a sample object, where \citet{Banafti2022InferentialDimensions} write $H$}.

\subsection{Assumption 4(iii) is incompatible with fat tails}

Granular instrumental variables draw their strength from a fat tailed size distribution. That is the source of the granularity. Assumption~\ref{ass:size_of_firm_power_law} states the tail as a power law. Under a power law the largest shares do not vanish as the cross section grows. A dominant--fringe partition of the kind Assumption 4(iii) imposes needs the leading shares to be $O_{\p}(1/N)$. The two cannot hold together. Before I formally state the result, I need a Lemma.

\begin{Lemma} \label{lemma:pitman1997}
Suppose $\s_{1},\s_{2},\dots$ are i.i.d., strictly positive, with regularly varying tail
$\p(\s>s)=s^{-\mu}L(s)$, where $\mu\in(0,1)$ and $L$ is slowly varying. For the ordered sequence, $\s_{(1),N}\ge\cdots\ge \s_{(N),N}$, define the ordered shares, $S_{(k),N}=\s_{(k),N}/\sum_j \s_j$. Then,
\begin{equation*}
    \bigl(S_{(1),N},S_{(2),N},\dots\bigr) \xrightarrow{d}(p_{1},p_{2},\dots)\sim\mathrm{PD}(\mu,0)
\end{equation*}
 in the product topology on $[0,1]^{\mathbb{N}}$, where $\mathrm{PD}(\alpha, \theta)$ is the two-parameter Poisson-Dirichlet distribution. Also, for every \textbf{finite} $k$, $S_{(k),N} \xrightarrow{d} p_{k}$ and $\p(p_k > 0) = 1$.
\end{Lemma}
\begin{proof}
    By the result in Section~1.2, p.~861, eq.~(20) of \citet{Pitman1997TheSubordinator},
    \begin{equation*}
    \bigl(S_{(1),N},S_{(2),N},\dots\bigr) \xrightarrow{d}(p_{1},p_{2},\dots)\sim\mathrm{PD}(\mu,0)
    \end{equation*}
    By their Proposition~10, pp.~862--863, there exist
$X_{1}<X_{2}<\cdots$, which are points of a unit-rate Poisson process on
$(0,\infty)$, namely $X_{n}=e_{1}+\cdots+e_{n}$ with the
$e_{i}$ i.i.d. exponential, such that
\[
p_{k}\;=\;\frac{X_{k}^{-1/\mu}}{\sum_{m\ge1}X_{m}^{-1/\mu}}\qquad\text{(their eq.~(29))}.
\]

Write $W \defeq \sum_{m\ge1} X_{m}^{-1/\mu}$ for the denominator in their
eq.~(29). For each fixed $k$, $X_{k}=e_{1}+\cdots+e_{k}$
is a finite sum of i.i.d.\ $\mathrm{Exp}(1)$'s, so $X_{k}\in(0,\infty)$ a.s.\
and hence $X_{k}^{-1/\mu}\in(0,\infty)$ a.s.

It remains to show $W\in(0,\infty)$ a.s., so that $p_k\in(0,1]$ a.s. Apply
Kolmogorov's strong law of large numbers to the i.i.d.\ sequence
$e_{1},e_{2},\dots$ with $\e[e_{1}]=1$:
\[
\frac{X_{m}}{m}\;=\;\frac{1}{m}\sum_{i=1}^{m}e_{i}
\;\xrightarrow[m\to\infty]{\textnormal{a.s.}}\;\e[e_{1}]\;=\;1.
\]
By continuity of $x\mapsto x^{-1/\mu}$ at $x=1$, this gives
\[
\lim_{m\to\infty}\frac{X_{m}^{-1/\mu}}{m^{-1/\mu}}
\;=\;\Bigl(\lim_{m\to\infty}\frac{X_{m}}{m}\Bigr)^{-1/\mu}
\;=\;1\quad\text{a.s.}
\]
Therefore the terms $X_{m}^{-1/\mu}$ behave, for large $m$, like the
deterministic terms $m^{-1/\mu}$. Because $\mu<1$ gives $1/\mu>1$, the series
$\sum_{m}m^{-1/\mu}$ converges, and hence $W=\sum_{m}X_{m}^{-1/\mu}<\infty$
a.s. Thus $W\in(0,\infty)$ a.s., and
$p_{k} = X_{k}^{-1/\mu}/W\in(0,1]$ a.s.\ for every finite $k$.
\end{proof}

With this Lemma, I can state the result formally.
\begin{Prop} \label{prop:BL_impossibility}
Under Assumption~\ref{ass:size_of_firm_power_law}, consider the ordered sequence of shares, $S_{(1),N}\ge\cdots\ge S_{(N),N}$, then
\begin{equation*}
    S_{(k),N} \neq O_{\p}(1/N) \enspace \text{for any finite $k\ge1$}
\end{equation*}
Equivalently, $NS_{(k),N}\to\infty$ in probability for any finite $k\ge1$
\end{Prop}
\begin{proof}
Under Assumption \ref{ass:size_of_firm_power_law}, the ordered sequence of shares, $S_{(1),N}\ge\cdots\ge S_{(N),N}$ satisfies Lemma~\ref{lemma:pitman1997}. Fix $k\ge1$. By Lemma~\ref{lemma:pitman1997}, $S_{(k),N}\xrightarrow{d} p_{k}$ and for any finite $k$, $\p(p_{k}>0)=1$. Hence, given any
$\delta>0$, we can choose a continuity point $\eta>0$ of the
distribution function of $p_{k}$ such that
\begin{equation}\label{eq:eta}
\p(p_{k}>\eta)\;>\;1-\delta.
\end{equation}
By the Portmanteau theorem applied to the open set $(\eta,\infty)$,
\begin{equation}\label{eq:port}
\liminf_{N\to\infty}\p\bigl(S_{(k),N}>\eta\bigr)\;\ge\;\p(p_{k}>\eta)\;>\;1-\delta.
\end{equation}
Fix $M>0$. For all $N$ with $N\eta>M$, $\{S_{(k),N}>\eta\}\subseteq\{NS_{(k),N}>M\}$, so~\eqref{eq:port} gives
\[
\liminf_{N\to\infty}\p\bigl(NS_{(k),N}>M\bigr)\;\ge\;1-\delta.
\]
Since $\delta>0$ was arbitrary, $\p(NS_{(k),N}>M)\to1$ for every
$M>0$, i.e.\ $NS_{(k),N}\to\infty$ in probability. Hence $S_{(k),N} \neq O_{\p}(1/N)$ for any finite $k\ge1$.
\end{proof}

The leading shares do not vanish at the rate the partition needs. So the dominant--fringe split of Assumption 4(iii) is not available under the power law. This directly affects two Lemmas in \citet{Banafti2022InferentialDimensions}. I show how the proofs fail and how I rescue them. Lets consider these issues one by one. 

\subsubsection{Their Lemma 1(ii) needs Assumption 4(iii)}

Their Lemma 1(ii) states that the variance of the price is bounded, $V(p_t) = \Theta(1)$. The bound runs through the dominant--fringe partition. It splits the share vector into a dominant block $S_d$ ($N_1$ fixed) and a fringe block $S_f$ ($N_2 \to \infty$), and controls each separately,
\begin{equation*}
    V(p_t) = \e[S'\Omega S] + \e[S_d'\tilde\Lambda_d\tilde\Lambda_d'S_d] + \e[S_f'\tilde\Lambda_f\tilde\Lambda_f'S_f] + V(\varepsilon_t).
\end{equation*}
The fringe term is $O(1)$ only because Assumption 4(iii) makes $\|S_f\|^2 = O(1/N)$. Once Proposition~\ref{prop:BL_impossibility} removes that assumption, the partition is gone and the bound has no proof.

This paper reaches $V(p_t) = O(1)$ without the partition. The price reduces to a few dominant terms. By Proposition~\ref{prop:weakness of known factor instrument}, $u_{St} = O_{\p}(a_N^{-1})$, and the intermediate result $\lambda_S' F_t = F_t^{(1)} + O_{\p}(a_N^{-1})$ holds, with $a_N S'\tilde\Lambda = O_{\p}(1)$ from Proposition~\ref{prop:behavior of S times lambda}. In the strong regime $a_N = 1$, so the dominant terms in $p_t$ are those in $\varepsilon_t$ and $F_t^{(1)}$. Both have finite second moments. Hence $V(p_t) = \Theta(1)$.

These orders hold uniformly across regimes. The same argument bounds the price in the weak regimes, where the partition was never available to begin with. So Proposition~\ref{prop:behavior of S times lambda} corrects the arguments of \citet{Banafti2022InferentialDimensions}, and then extends them to the weak regimes.

\subsubsection{Their Term IV of Lemma 2 needs Assumption 4(iii)}

Consider Term IV of their Lemma 2. They show that its order is $o_{\p}(1)$ after standardization. But again, the proof uses Assumption 4(iii), and once the assumption is gone their proof of negligibility goes with it.

To see this, consider Term IV as bounded by Cauchy--Schwarz in their
proof:
\begin{align*}
    IV &= O_{\p}(\delta_{NT}^{-1}) \left( \frac{1}{N} \sum_{i=1}^N S_i^2 \, \| \tilde\lambda_i' \hat H \|^2 \right)^{1/2} \\
       &\leq O_{\p}(\delta_{NT}^{-1}) \left( \frac{1}{N} \left( \sum_{i=1}^N S_i^4 \right)^{1/2} \left( \sum_{i=1}^N \| \tilde\lambda_i' \hat H \|^4 \right)^{1/2} \right)^{1/2} \\
       &= O_{\p}(\delta_{NT}^{-1}) \left( \frac{1}{N} O_{\p}( 1 ) \cdot O_{\p}( N^{1/2} ) \right)^{1/2}
       = O_{\p}(\delta_{NT}^{-1}) \cdot O_{\p}( N^{-1/4} ) = O_{\p}\!\left( N^{-3/4} \right).
\end{align*}

The intermediate orders use the assumptions of this paper. Under bounded
fourth moments of the loadings, $\sum_i \| \tilde\lambda_i' \hat H \|^4 = O_{\p}(N)$.
Under the power-law assumption in the strong regime, the maximum share is
$O_{\p}(1)$, and the Herfindahl is $O_{\p}(1)$, so
$\sum_i S_i^4 \leq (\max_i S_i)^2 \sum_i S_i^2 = O_{\p}(1)$. Substituting
back,
\begin{equation*}
    \mathrm{IV} = O_{\p}\!\left( N^{-3/4} \right),
\end{equation*}
so the standardized counterpart is
\begin{equation*}
    \sqrt{T}\,\mathrm{IV} = O_{\p}\!\left( \frac{\sqrt{T}}{N^{3/4}} \right).
\end{equation*}

The Cauchy--Schwarz bound their proof relies on is loose. It uses the maximum share and the Herfindahl, both of which Assumption 4(iii) supplies. Without the assumption the bound has no footing, and after standardization it does not go to zero. The standardized counterpart is $O_{\p}(\sqrt T/N^{3/4})$, and the first-stage requirement $\sqrt T/N\to0$ admits sequences where this diverges, for example $\sqrt T=N^{0.9}$ gives $\sqrt T/N^{3/4}=N^{0.15}\to\infty$. So Cauchy--Schwarz is not enough to keep Term IV negligible once Assumption 4(iii) is removed.

Proposition~\ref{prop:behavior of S times lambda} restores the negligibility without the assumption. It gives $a_N S'\tilde\Lambda=O_{\p}(1)$ uniformly across regimes, which is what the tight bound on the factor channel needs. The same Proposition that corrects Lemma 1(ii) corrects Term IV, and it carries the result to every regime of instrument strength, not just the strong one. 

I still need to show how what drives the difference between my results and \citet{Banafti2022InferentialDimensions}. This comes from Term II of their Lemma 2. This term disappears only under the stringent condition that the structural error $\varepsilon_t$ is uncorrelated with the common factors, $\tilde{F}_t$. That is a very stringent condition not present in most empirical applications and one that not even \citet{Banafti2022InferentialDimensions} impose.  

\subsubsection*{Their Term II of Lemma 2 drives the variance}

I work in the strong regime throughout, so $a_N = 1$ and the numerator is standardized by $\sqrt T$. \citet{Banafti2022InferentialDimensions} write the factor as $\eta_t$ and the demeaned panel as $\tilde y_t$. I use the notation of this paper, the factor $\tilde F_t$ and the demeaned panel $\bar y_t$, so the data equation reads $\bar y_{jt} = \tilde\lambda_j'\tilde F_t + \bar u_{jt}$, the entrywise form of $\bar Y = \tilde F\tilde\Lambda' + \bar u$. Their $\tilde y$ is my $\bar y$ in what follows. 

\citet{Banafti2022InferentialDimensions} reach the impugned conclusion through their Lemma 2. They write the first-stage error with the loading projection $\hat Q = I - \hat\Lambda(\hat\Lambda'\hat\Lambda)^{-1}\hat\Lambda'$ and expand $\hat Q - Q$ around the true loadings. With $\hat D = \hat\Lambda'\hat\Lambda/N$ and $D = \hat H'(\tilde\Lambda'\tilde\Lambda)\hat H/N$,
\begin{align*}
    \hat Q - Q = \frac1N\Big[ &(\hat\Lambda - \tilde\Lambda \hat H)\hat D^{-1}(\hat\Lambda - \tilde\Lambda \hat H)' + (\hat\Lambda - \tilde\Lambda \hat H)\hat D^{-1}\hat H'\tilde\Lambda' \\
    &+ \tilde\Lambda \hat H\hat D^{-1}(\hat\Lambda - \tilde\Lambda \hat H)' + \tilde\Lambda \hat H(\hat D^{-1} - D^{-1})\hat H'\tilde\Lambda' \Big].
\end{align*}
Sandwiching in $S'(\cdot)\bar y_{\cdot t}\varepsilon_t$ and averaging over $t$ gives their four terms, $\frac1T\sum_t S'(\hat Q - Q)\bar y_{\cdot t}\varepsilon_t = \mathrm{I}+\mathrm{II}+\mathrm{III}+\mathrm{IV}$. They bound terms II and III by Bai's Lemma B.1, used in the symmetric form. This paper explicitly states the symmetric form as Lemma~\ref{lemma:intermediate_1_factor_loadings}. It gives both terms as $O_{\p}(C_{NT}^{-2})$ with $C_{NT} = \min\{\sqrt N,\sqrt T\}$, so they conclude the first stage drops out. We will see below that the two terms have different orders, due to differences in the way they interact with the granular shares, $S$.

Take Term III first. The shares interact with the true loadings,
\begin{equation*}
    \mathrm{III} = \frac{1}{NT}\sum_t S'\tilde\Lambda \hat H\hat D^{-1}(\hat\Lambda - \tilde\Lambda \hat H)'\bar y_{\cdot t}\varepsilon_t
    = (S'\tilde\Lambda)\,\hat H\hat D^{-1}\Big[\frac1N(\hat\Lambda - \tilde\Lambda \hat H)'\frac1T\sum_t\bar y_{\cdot t}\varepsilon_t\Big].
\end{equation*}
Here $S'\tilde\Lambda = O_{\p}(1)$ and the loading error meets the data. With $\bar y_{it} = \tilde\lambda_i'\tilde F_t + \bar u_{it}$, $\frac1T\sum_t\bar y_{it}\varepsilon_t = \tilde\lambda_i'\e[\tilde F_t\varepsilon_t] + O_{\p}(1/\sqrt T)$, so the bracket splits into a covariance term and an idiosyncratic remainder,
\begin{equation*}
    \frac1N(\hat\Lambda - \tilde\Lambda \hat H)'\frac1T\sum_t\bar y_{\cdot t}\varepsilon_t
    = \Big[\frac1N(\hat\Lambda - \tilde\Lambda \hat H)'\tilde\Lambda\Big]\e[\tilde F_t\varepsilon_t] + O_{\p}(C_{NT}^{-2}).
\end{equation*}

$\e[\tilde F_t\varepsilon_t]$ is O(1). By the loading expansion \eqref{eq:asymptotic_linear_lambda_incomplete} and Lemma~\ref{lemma:order of a mixing sequence}, $\frac1N(\hat\Lambda - \tilde\Lambda \hat H)'\tilde\Lambda$ is $O_{\p}(C_{NT}^{-2})$. Hence $\mathrm{III} = O_{\p}(C_{NT}^{-2})$. Lemma~\ref{lemma:intermediate_1_factor_loadings} holds for Term III, and their bound on it is right.

Now Term II. The shares interact with the estimation error of the factor loadings,
\begin{equation*}
    \mathrm{II} = \frac{1}{NT}\sum_t S'(\hat\Lambda - \tilde\Lambda \hat H)\hat D^{-1}\hat H'\tilde\Lambda'\bar y_{\cdot t}\varepsilon_t
    = \big[S'(\hat\Lambda - \tilde\Lambda \hat H)\big]\,\hat D^{-1}\hat H'\Big[\frac1N\tilde\Lambda'\frac1T\sum_t\bar y_{\cdot t}\varepsilon_t\Big].
\end{equation*}
The two ends have traded roles. Now the data meets the true loadings, and the loading error meets the shares. Note that
\begin{equation*}
    \frac1N\tilde\Lambda'\frac1T\sum_t\bar y_{\cdot t}\varepsilon_t
    = \frac{\tilde\Lambda'\tilde\Lambda}{N}\cdot\frac1T\sum_t\tilde F_t\varepsilon_t
    + \frac1N\sum_j\tilde\lambda_j\frac1T\sum_t \bar u_{jt}\varepsilon_t
    = \Sigma_{\tilde\Lambda}\,\e[\tilde F_t\varepsilon_t] + O_{\p}(1/\sqrt{NT}).
\end{equation*}
On the share end the loading error is held out of the average,
\begin{equation*}
    S'(\hat\Lambda - \tilde\Lambda \hat H) = O_{\p}(1/\sqrt T),
\end{equation*}
by the loading expansion \eqref{eq:asymptotic_linear_lambda_incomplete}, whose leading part is $V_{NT}^{-1}(\tilde\Lambda'\hat\Lambda/N)\frac1T\sum_t\tilde F_t\bar u_{St}$ with the share weighted shock $\bar u_{St} = S'\bar u_t$. The granular shares do not average the loading error. They concentrate it. So
\begin{equation*}
    \mathrm{II} = O_{\p}\!\big(1/\sqrt T\big), \qquad \sqrt T\,\mathrm{II} = O_{\p}(1).
\end{equation*}

This is why the conclusions differ. Standardized by $\sqrt T$, Term III is $\sqrt T\,C_{NT}^{-2}=o_{\p}(1)$ under $\sqrt T/N\to0$, while Term II is $O_{\p}(1)$. This term is a consequence of estimation of the factor structure and adds to the variance, which my Theorems formalize.

\section{Edge Cases: $\mu = 1$ and $\mu = 2$} \label{app:edge_cases}

The proof of Proposition \ref{prop:weakness of known factor instrument} in Appendix \ref{app:behavior of herfindahl} gives the order of the Herfindahl $S'S$ for three ranges of $\mu$. These are $\mu \in (0,1)$, $\mu \in (1,2)$, and $\mu > 2$. It leaves out the two endpoints, $\mu = 1$ and $\mu = 2$. I cover them here.

Both endpoints are boundary cases of the stable law. At each one, exactly one of the two sums in $S'S = \frac{\sum_i \s_i^2}{[\sum_j \s_j]^2}$ has tail index one. That sum has an infinite moment. But the moment diverges only as a logarithm. So the norming picks up a $\ln N$ factor, and the order of $S'S$ differs from the neighboring rates by a power of $\ln N$.

As in Appendix \ref{app:behavior of herfindahl}, I set $c = 1$ and let $\s_i$ live on $[1, \infty)$, so that $\p(\s_i > s) = s^{-\mu}$ for $s \geq 1$.

\subsection{The case $\mu = 1$}

Start with the numerator. The squared size $\s_i^2$ has tail
\begin{equation*}
    \p(\s_i^2 > s) = \p(\s_i > \sqrt{s}) = s^{-\frac{\mu}{2}} = s^{-\frac{1}{2}}.
\end{equation*}
The tail index is $\frac{1}{2}$, which is below one. This is the same as the numerator in the $\mu \in (0,1)$ proof. Lévy's Theorem \ref{theorem:Levy's Theorem} applies with $\kappa_n = N^{2/\mu} = N^2$ and $b_n / \kappa_n = \frac{\mu}{2 - \mu} = 1$. The sum has the same order as $\kappa_n$,
\begin{equation*}
    \sum_{i=1}^N \s_i^2 = O_{\p}(N^2).
\end{equation*}

Now the denominator. The size $\s_i$ has tail index exactly $\mu = 1$. The mean is infinite,
\begin{equation*}
    \e[\s] = \int_1^{\infty} s \, \mu s^{-\mu - 1} \, ds = \int_1^{\infty} s^{-1} \, ds = \infty,
\end{equation*}
but it diverges only as a logarithm. Lévy's Theorem \ref{theorem:Levy's Theorem} applies with index $\alpha = 1$. The norming is $\kappa_n = \inf\{ x : \p(\s > x) \leq n^{-1} \} = N$. The centering is
\begin{equation*}
    b_n = N \, \e[\s \cdot 1_{(\s \leq \kappa_n)}] = N \int_1^N s \, \mu s^{-\mu - 1} \, ds = N \int_1^N s^{-1} \, ds = N \ln N.
\end{equation*}
Lévy's Theorem then gives $(\sum_j \s_j - b_n)/\kappa_n \xrightarrow{d} Y$ for a stable limit $Y$ with index one. Here $b_n / \kappa_n = \ln N \to \infty$, so the centering grows faster than the norming. The sum has the order of its centering,
\begin{equation*}
    \sum_{j=1}^N \s_j = N \ln N \, (1 + o_{\p}(1)) = O_{\p}(N \ln N).
\end{equation*}

Put the two together.
\begin{equation*}
    S'S = \frac{\sum_{i=1}^N \s_i^2}{\left[ \sum_{j=1}^N \s_j \right]^2} = \frac{O_{\p}(N^2)}{N^2 (\ln N)^2} = O_{\p}\left( \frac{1}{(\ln N)^2} \right).
\end{equation*}

Appendix \ref{app:behavior of herfindahl} shows $\var[z_t \mid S] = O_{\p}(S'S)$, so the instrument has standard deviation of order $\sqrt{S'S}$. For $\mu \in (0,1)$ we have $S'S = O_{\p}(1)$, so $z_t = O_{\p}(1)$ and the instrument does not dilute. For $\mu \in (1,2)$ we have $z_t = O_{\p}(N^{-\delta})$ with $\delta = 1 - 1/\mu$. At $\mu = 1$ the instrument dilutes as a logarithm instead, with $z_t = O_{\p}(1/\ln N)$. It dilutes at the slow rate $1/\ln N$, not at a power of $N$.

The dilution carries over to the GIV estimator. I formally state the result in Theorem \ref{theorem:estimated_log_weak} below, which is the analogue of Theorem~\ref{theorem:estimated_weak_consistency_aggregate} at $\mu = 1$. Theorem~\ref{theorem:estimated_log_weak_disaggregate} does the same for the supply elasticity.

But before I do that, I need to clarify the behavior of share-weighted loading $S'\tilde{\Lambda}$. The numerator $\sum_j \s_j \tilde{\Blambda}_j$ is the Breiman--Lévy object from the proof of Proposition~\ref{prop:behavior of S times lambda}. At $\mu = 1$ the product $\s_j X_j(g)$ has tail index one, so Lévy's Theorem~\ref{theorem:Levy's Theorem} applies with $\kappa_n = N$. The centering is $b_n = 0$ by the symmetry of $X_j(g)$ around zero. So the numerator carries no logarithm,
\begin{equation*}
    \sum_{j=1}^N \s_j \tilde{\Blambda}_j = O_{\p}(N).
\end{equation*}
The denominator is $\sum_j \s_j = O_{\p}(N \ln N)$ from above. Hence
\begin{equation*}
    S'\tilde{\Lambda} = \frac{O_{\p}(N)}{O_{\p}(N \ln N)} = O_{\p}\left( \frac{1}{\ln N} \right).
\end{equation*}

So the share-weighted loading and the instrument both scale as $1/\ln N$. Set $a_N = \ln N$. Then $a_N \Bar{u}_{St} = O_{\p}(1)$ and $a_N S'\tilde{\Lambda} = O_{\p}(1)$. These are the two inputs the estimated-factor proof needs.

\begin{Theorem}[GIV with estimated factors under logarithmic weakness] \label{theorem:estimated_log_weak}
    Suppose Assumptions \ref{ass:weak stationarity and idiosyncrasy} to \ref{ass:LN_moments and CLT} hold with $\mu = 1$, $(\ln N)^2/T \to 0$, and $\sqrt{T}/N \to 0$. Then, conditional on $S$, for almost every realization of the shares, the GIV estimator is consistent and asymptotically normal,
    \begin{equation*}
        \frac{\Gamma_{zp}}{\sqrt{V_{z \bar{\varepsilon}}(S)}} \cdot \sqrt{T}\, [\hat{\phi}_d - \phi_d] \xrightarrow{d} \n(0, 1),
    \end{equation*}
    where $\Gamma_{zp}$ and $V_{z \bar{\varepsilon}}(S)$ are as in Theorem~\ref{theorem:estimated_weak_consistency_aggregate}. The studentization satisfies $\Gamma_{zp}/\sqrt{V_{z\bar{\varepsilon}}(S)} = O_{\p}(1/\ln N)$, so the standardized statistic vanishes at the rate $\sqrt{T}/\ln N$.
\end{Theorem}
\begin{proof}
    The argument is the proof of Theorem~\ref{theorem:estimated_weak_consistency_aggregate} with the rescaling $a_N = \ln N$. The two inputs above, $a_N \Bar{u}_{St} = O_{\p}(1)$ and $a_N S'\tilde{\Lambda} = O_{\p}(1)$, let Lemmas~\ref{lemma:numerator} and~\ref{lemma:denominator} go through with this $a_N$.

    Consider the denominator. By Lemma~\ref{lemma:denominator}, the first-stage term is $O_{\p}(a_N/\sqrt{T}) + O_{\p}(1/\sqrt{N}) = O_{\p}(\ln N/\sqrt{T}) + O_{\p}(1/\sqrt{N})$. This is $o_{\p}(1)$ under $(\ln N)^2/T \to 0$. The leading denominator $\frac{a_N^2}{T} \sum_t z_t p_t$ converges to $a_N^2 \Gamma_{zp}$, which is bounded away from zero. The first-stage term therefore drops out.

    Consider the numerator. By Lemma~\ref{lemma:numerator}.(ii), the first-stage term is the factor-space term plus residuals $O_{\p}(1/\sqrt{N}) + O_{\p}(\sqrt{T}/N) + O_{\p}(1/\sqrt{T})$. The factor-space term is $O_{\p}(1)$ and enters the variance. The residuals are $o_{\p}(1)$ under $\sqrt{T}/N \to 0$ and $T \to \infty$. The central limit theorem for the leading numerator is Theorem~\ref{theorem:central limit theorem}, which holds conditional on $(S, \Lambda)$ and extends to the boundary index $\mu = 1$ with the normalization $a_N = \ln N$ by Remark~\ref{remark:CLT boundary indices}.

    Combining the two and applying Slutsky's theorem gives the stated limit. The rate is $\sqrt{T}/a_N = \sqrt{T}/\ln N$.
\end{proof}

Only the consistency condition changes from Theorem~\ref{theorem:estimated_weak_consistency_aggregate}. The requirement $N/T \to 0$ relaxes to $(\ln N)^2/T \to 0$, because logarithmic dilution barely penalizes a growing cross section. The first-stage requirement $\sqrt{T}/N \to 0$ is unchanged. So at $\mu = 1$ the estimator behaves almost like the strong regime. The one cost is the $\ln N$ inflation of the standard error.

The supply elasticity extends the same way. The estimator is \eqref{eq:estimated disaggregated estimator}, with the estimated instrument and equal weights. It shares $\hat{z}_t$ and $p_t$ with the demand estimator, so the two inputs above serve it unchanged. The only new ingredient is the equal-weights aggregation error.

That error does not depend on the regime. Lemma~\ref{lemma:hajek bound} holds for every $\mu > 0$ and gives $\beta = O_{\p}(\sqrt{h_N})$. Here $h_N = S'S = O_{\p}(1/(\ln N)^2)$ from above, so $\beta = O_{\p}(1/\ln N) = O_{\p}(a_N^{-1})$. The displacement argument of Section~\ref{subsec:equal weights} then goes through verbatim with $a_N = \ln N$, and the aggregation error is $O_{\p}(\sqrt{T}/N) + O_{\p}(1/\sqrt{N})$ as in Appendix~\ref{subsec:proof of estimated_disaggregate}.

\begin{Theorem}[Supply elasticity under logarithmic weakness] \label{theorem:estimated_log_weak_disaggregate}
    Suppose Assumptions \ref{ass:weak stationarity and idiosyncrasy} to \ref{ass:LN_moments and CLT} hold with $\mu = 1$, $(\ln N)^2/T \to 0$, and $\sqrt{T}/N \to 0$. Then, conditional on $(S, \Lambda)$, for almost every realization of the shares and loadings, the estimated-factor GIV estimator \eqref{eq:estimated disaggregated estimator} for the supply elasticity is consistent and asymptotically normal,
    \begin{equation*}
        \frac{\Gamma_{zp}}{\sqrt{V_{z \bar{v}_H}(S, \Lambda)}} \cdot \sqrt{T}\, [\hat{\phi}_s - \phi_s] \xrightarrow{d} \n(0, 1),
    \end{equation*}
    where $\Gamma_{zp}$ and $V_{z \bar{v}_H}(S, \Lambda)$ are as in Theorem~\ref{theorem:estimated_weak_consistency_disaggregate}. The studentization satisfies $\Gamma_{zp}/\sqrt{V_{z \bar{v}_H}(S, \Lambda)} = O_{\p}(1/\ln N)$, so the standardized statistic vanishes at the rate $\sqrt{T}/\ln N$.
\end{Theorem}
\begin{proof}
    The argument is the proof of Theorem~\ref{theorem:estimated_weak_consistency_disaggregate} with the rescaling $a_N = \ln N$. The instrument error is Proposition~\ref{prop:first stage effect} at $Q_t = v_{Ht}$, which needs $a_N^2/T = (\ln N)^2/T \to 0$. The aggregation error is Steps 4 and 5 of Appendix~\ref{subsec:proof of estimated_disaggregate}, whose bounds $O_{\p}(\sqrt{T}/N) + O_{\p}(1/\sqrt{N})$ do not involve the regime. The denominator and the central limit theorem are those of the proof of Theorem~\ref{theorem:estimated_log_weak}.
\end{proof}

Equal weights cost nothing at the boundary. The moment violation scales with $\sqrt{h_N}$, and at $\mu = 1$ that is exactly the strength $1/\ln N$ of the instrument, so the cancellation of Section~\ref{subsec:equal weights} is preserved.

\subsection{The case $\mu = 2$}

Start with the denominator. The mean is now finite,
\begin{equation*}
    \e[\s] = \frac{\mu}{\mu - 1} = 2.
\end{equation*}
The sizes $\s_i$ are independent with a finite mean. Kolmogorov's Law of Large Numbers gives
\begin{equation*}
    \frac{1}{N} \sum_{j=1}^N \s_j \xrightarrow{\text{a.s}} \e[\s],
\end{equation*}
so $\sum_j \s_j = N \e[\s] (1 + o_{\p}(1)) = O_{\p}(N)$.

Now the numerator. The squared size $\s_i^2$ has tail
\begin{equation*}
    \p(\s_i^2 > s) = \p(\s_i > \sqrt{s}) = s^{-\frac{\mu}{2}} = s^{-1}.
\end{equation*}
The tail index is exactly one. This is the boundary of the stable law, like the denominator in the $\mu = 1$ case. The second moment of $\s$ is infinite,
\begin{equation*}
    \e[\s^2] = \int_1^{\infty} s^2 \, \mu s^{-\mu - 1} \, ds = \int_1^{\infty} s^{-1} \, ds = \infty,
\end{equation*}
and again it diverges only as a logarithm. Lévy's Theorem \ref{theorem:Levy's Theorem} applies to $\{ \s_i^2 \}$ with index $\alpha = 1$ and $\kappa_n = N$. The cutoff $\s_i^2 \leq \kappa_n$ is the same as $\s_i \leq \sqrt{N}$, so the centering is
\begin{equation*}
    b_n = N \, \e[\s^2 \cdot 1_{(\s^2 \leq \kappa_n)}] = N \int_1^{\sqrt{N}} s^2 \, \mu s^{-\mu - 1} \, ds = N \int_1^{\sqrt{N}} 2 s^{-1} \, ds = N \ln N.
\end{equation*}
As in the $\mu = 1$ case, $b_n / \kappa_n = \ln N \to \infty$, so the centering dominates. The sum has the order of its centering,
\begin{equation*}
    \sum_{i=1}^N \s_i^2 = N \ln N \, (1 + o_{\p}(1)) = O_{\p}(N \ln N).
\end{equation*}

Put the two together.
\begin{equation*}
    S'S = \frac{\sum_{i=1}^N \s_i^2}{\left[ \sum_{j=1}^N \s_j \right]^2} = \frac{O_{\p}(N \ln N)}{N^2 \e[\s]^2 (1 + o_{\p}(1))} = O_{\p}\left( \frac{\ln N}{N} \right).
\end{equation*}

Thus, at $\mu =2$, the instrument, $z_t = O_{\p}(\sqrt{\ln N / N})$. It dilutes at a slightly slower rate than when $\mu > 2$.

This dilution carries over to the estimator. I extend both Theorem~\ref{theorem:estimated_weak_consistency_aggregate} and the Anderson--Rubin test of Proposition~\ref{prop:weakness robust test_aggregate} to $\mu = 2$. I do the same for their supply counterparts, Theorem~\ref{theorem:estimated_weak_consistency_disaggregate} and Proposition~\ref{prop:weakness robust test_disaggregate}.

First consider the share-weighted loading $S'\tilde{\Lambda}$. The numerator $\sum_j \s_j \tilde{\Blambda}_j$ is the Breiman--Lévy object from the proof of Proposition~\ref{prop:behavior of S times lambda}. At $\mu = 2$ the product $\s_j X_j(g)$ has tail index two and is symmetric around zero. Its second moment is infinite by a logarithm, so the sum lies in the domain of attraction of the normal law, with norming $\sqrt{N \ln N}$ rather than $\sqrt{N}$. So
\begin{equation*}
    \sum_{j=1}^N \s_j \tilde{\Blambda}_j = O_{\p}(\sqrt{N \ln N}).
\end{equation*}
The denominator is $\sum_j \s_j = O_{\p}(N)$ from above. Hence
\begin{equation*}
    S'\tilde{\Lambda} = \frac{O_{\p}(\sqrt{N \ln N})}{O_{\p}(N)} = O_{\p}\left( \sqrt{\frac{\ln N}{N}} \right).
\end{equation*}

So the share-weighted loading and the instrument both scale as $\sqrt{\ln N / N}$. Set $a_N = \sqrt{N/\ln N}$. Then $a_N \Bar{u}_{St} = O_{\p}(1)$ and $a_N S'\tilde{\Lambda} = O_{\p}(1)$.

\begin{Theorem}[GIV with estimated factors at the boundary $\mu = 2$] \label{theorem:estimated_mu2}
    Suppose Assumptions \ref{ass:weak stationarity and idiosyncrasy} to \ref{ass:LN_moments and CLT} hold with $\mu = 2$, $N/(T \ln N) \to 0$, and $\sqrt{T}/N \to 0$. Then, conditional on $S$, for almost every realization of the shares, the GIV estimator is consistent and asymptotically normal,
    \begin{equation*}
        \frac{\Gamma_{zp}}{\sqrt{V_{z \bar{\varepsilon}}(S)}} \cdot \sqrt{T}\, [\hat{\phi}_d - \phi_d] \xrightarrow{d} \n(0, 1),
    \end{equation*}
    where $\Gamma_{zp}$ and $V_{z \bar{\varepsilon}}(S)$ are as in Theorem~\ref{theorem:estimated_weak_consistency_aggregate}. The studentization satisfies $\Gamma_{zp}/\sqrt{V_{z\bar{\varepsilon}}(S)} = O_{\p}(\sqrt{\ln N / N})$, so the standardized statistic vanishes at the rate $\sqrt{T \ln N / N}$.
\end{Theorem}
\begin{proof}
    The argument is the proof of Theorem~\ref{theorem:estimated_weak_consistency_aggregate} with the rescaling $a_N = \sqrt{N/\ln N}$. The two inputs above let Lemmas~\ref{lemma:numerator} and~\ref{lemma:denominator} go through with this $a_N$.

    Consider the denominator. By Lemma~\ref{lemma:denominator}, the first-stage term is $O_{\p}(a_N/\sqrt{T}) + O_{\p}(1/\sqrt{N}) = O_{\p}(\sqrt{N/(T \ln N)}) + O_{\p}(1/\sqrt{N})$. This is $o_{\p}(1)$ under $N/(T \ln N) \to 0$. The leading denominator $\frac{a_N^2}{T} \sum_t z_t p_t$ converges to $a_N^2 \Gamma_{zp}$, which is bounded away from zero. The first-stage term therefore drops out.

    Consider the numerator. By Lemma~\ref{lemma:numerator}.(ii), the first-stage term is the factor-space term plus residuals $O_{\p}(1/\sqrt{N}) + O_{\p}(\sqrt{T}/N) + O_{\p}(1/\sqrt{T})$. The factor-space term is $O_{\p}(1)$ and enters the variance. The residuals are $o_{\p}(1)$ under $\sqrt{T}/N \to 0$ and $T \to \infty$. The central limit theorem for the leading numerator is Theorem~\ref{theorem:central limit theorem}, which holds conditional on $(S, \Lambda)$ and extends to the boundary index $\mu = 2$ with the normalization $a_N = \sqrt{N/\ln N}$ by Remark~\ref{remark:CLT boundary indices}.

    Combining the two and applying Slutsky's theorem gives the stated limit. The rate is $\sqrt{T}/a_N = \sqrt{T \ln N / N}$.
\end{proof}

The supply elasticity extends in the same way as at $\mu = 1$. The two inputs $a_N \Bar{u}_{St} = O_{\p}(1)$ and $a_N S'\tilde{\Lambda} = O_{\p}(1)$ serve both parameters. For the aggregation error, Lemma~\ref{lemma:hajek bound} gives $\beta = O_{\p}(\sqrt{h_N}) = O_{\p}(\sqrt{\ln N / N}) = O_{\p}(a_N^{-1})$, so the bounds $O_{\p}(\sqrt{T}/N) + O_{\p}(1/\sqrt{N})$ of Appendix~\ref{subsec:proof of estimated_disaggregate} hold at the boundary as well.

\begin{Theorem}[Supply elasticity at the boundary $\mu = 2$] \label{theorem:estimated_mu2_disaggregate}
    Suppose Assumptions \ref{ass:weak stationarity and idiosyncrasy} to \ref{ass:LN_moments and CLT} hold with $\mu = 2$, $N/(T \ln N) \to 0$, and $\sqrt{T}/N \to 0$. Then, conditional on $(S, \Lambda)$, for almost every realization of the shares and loadings, the estimated-factor GIV estimator \eqref{eq:estimated disaggregated estimator} for the supply elasticity is consistent and asymptotically normal,
    \begin{equation*}
        \frac{\Gamma_{zp}}{\sqrt{V_{z \bar{v}_H}(S, \Lambda)}} \cdot \sqrt{T}\, [\hat{\phi}_s - \phi_s] \xrightarrow{d} \n(0, 1),
    \end{equation*}
    where $\Gamma_{zp}$ and $V_{z \bar{v}_H}(S, \Lambda)$ are as in Theorem~\ref{theorem:estimated_weak_consistency_disaggregate}. The studentization satisfies $\Gamma_{zp}/\sqrt{V_{z \bar{v}_H}(S, \Lambda)} = O_{\p}(\sqrt{\ln N / N})$, so the standardized statistic vanishes at the rate $\sqrt{T \ln N / N}$.
\end{Theorem}
\begin{proof}
    The argument is the proof of Theorem~\ref{theorem:estimated_weak_consistency_disaggregate} with the rescaling $a_N = \sqrt{N/\ln N}$. The instrument error is Proposition~\ref{prop:first stage effect} at $Q_t = v_{Ht}$, which needs $a_N^2/T = N/(T \ln N) \to 0$. The aggregation error is bounded above. The denominator and the central limit theorem are those of the proof of Theorem~\ref{theorem:estimated_mu2}.
\end{proof}

At $\mu = 2$ the concentration parameter has only the slowly-divergent floor $\ln N$. As at $\mu > 2$, I use the Anderson--Rubin test for inference, evaluated at the null and built from the null-imposed residual $\varepsilon_t(\phi_d^0) = d_t - \phi_d^0 p_t$. The statistic is \eqref{eq:AR statistic_aggregate} specialized to $a_N^2 = N/\ln N$,
\begin{equation*}
    \text{AR}(\phi_d^0) = \frac{N}{T \ln N} \left( \sum_t \hat{z}_t (y_{St} - \phi_d^0 p_{t}) \right)^2 \frac{1}{\hat{V}_{z \bar{\varepsilon}}(\phi_d^0)}.
\end{equation*}
The $a_N^2 = N/\ln N$ normalization cancels between the numerator and $\hat{V}_{z \bar{\varepsilon}}(\phi_d^0)$, so the statistic is the same self-normalized object as at $\mu > 2$.

\begin{Prop}[Weakness-robust test at the boundary $\mu = 2$] \label{prop:AR_mu2}
    Suppose Assumptions \ref{ass:weak stationarity and idiosyncrasy} to \ref{ass:LN_moments and CLT} hold with $\mu = 2$ and $\sqrt{T}/N \to 0$. Then, for any trajectory of $N/T$, under the null $H_0: \phi_d = \phi_d^0$,
    \begin{equation*}
        \text{AR}(\phi_d^0) \xrightarrow{d} \chi^2_1.
    \end{equation*}
\end{Prop}
\begin{proof}
    The argument is the proof of Proposition~\ref{prop:weakness robust test_aggregate} with $a_N = \sqrt{N/\ln N}$. Under $H_0$ the numerator $\frac{a_N}{\sqrt{T}} \sum_t \hat{z}_t \varepsilon_t$ converges to $\n(0, a_N^2 V_{z\bar{\varepsilon}}(S))$ by the proof of Theorem~\ref{theorem:estimated_mu2}, needing only $\sqrt{T}/N \to 0$. The null-imposed residual is exact, so the variance estimator is consistent without the $N/(T \ln N) \to 0$ that Theorem~\ref{theorem:estimated_mu2} needs for its plug-in residual. The cancellation of $a_N$ and Slutsky's theorem give the limit.
\end{proof}

For the estimator, Theorem~\ref{theorem:estimated_mu2}, the requirement $N/T \to 0$ of Theorem~\ref{theorem:estimated_weak_consistency_aggregate} relaxes to $N/(T \ln N) \to 0$, since the instrument dilutes a factor $\sqrt{\ln N}$ more slowly than at $\mu > 2$. The first-stage requirement $\sqrt{T}/N \to 0$ is unchanged. The Anderson--Rubin test, Proposition~\ref{prop:AR_mu2}, places no condition on $N/T$ at all, exactly as at $\mu > 2$.

The supply elasticity admits the same test. Under $H_0: \phi_s = \phi_s^0$ the equal-weights residual $y_{Et} - \phi_s^0 p_t = v_{Et}$ is exact, so the null-imposed statistic again avoids the first stage. The statistic is \eqref{eq:AR statistic_disaggregate} specialized to $a_N^2 = N/\ln N$,
\begin{equation*}
    \text{AR}_H(\phi_s^0) = \frac{N}{T \ln N} \left( \sum_t \hat{z}_t (y_{Et} - \phi_s^0 p_{t}) \right)^2 \frac{1}{\hat{V}_{z v_H}(\phi_s^0)},
\end{equation*}
where $\hat{V}_{z v_H}(\phi_s^0)$ is the estimator \eqref{eq:HAC estimator_disaggregated_variance} built from $\tilde{v}_{Et}(\phi_s^0)$ with $a_N^2 = N/\ln N$. The rate cancels between the numerator and the variance estimator exactly as in \eqref{eq:AR statistic_aggregate}.

\begin{Prop}[Weakness-robust test for the supply elasticity at $\mu = 2$] \label{prop:AR_H_mu2}
    Suppose Assumptions \ref{ass:weak stationarity and idiosyncrasy} to \ref{ass:LN_moments and CLT} hold with $\mu = 2$ and $\sqrt{T}/N \to 0$. Then, for any trajectory of $N/T$, under the null $H_0: \phi_s = \phi_s^0$,
    \begin{equation*}
        \text{AR}_H(\phi_s^0) \xrightarrow{d} \chi^2_1.
    \end{equation*}
\end{Prop}
\begin{proof}
    The argument is the proof of Proposition~\ref{prop:weakness robust test_disaggregate} with $a_N = \sqrt{N/\ln N}$. The aggregation error is $O_{\p}(\sqrt{T}/N) + O_{\p}(1/\sqrt{N})$, which is $o_{\p}(1)$ for any trajectory of $N/T$ under $\sqrt{T}/N \to 0$. The cancellation of $a_N$ and Slutsky's theorem give the limit as in Proposition~\ref{prop:AR_mu2}.
\end{proof}

Proposition~\ref{prop:AR_H_mu2} places no condition on $N/T$, exactly as Proposition~\ref{prop:AR_mu2}.
\section{Controls} \label{appendix_sec_controls}

In the main text, I considered estimation without explicitly accounting for the presence of exogenous controls. In this section, I consider the presence of such controls. The modified structural equations are:
\begin{align*}
    d_t &= \phi_d p_{t} +  X_t^{d'}\beta_d  + \varepsilon_t \\
    y_{it} &= \phi_s p_t +  X_{it}^{y'} \beta_y + \lambda_i' F_t + u_{it}
\end{align*}
where $\beta_d$ is $k_d$ dimensional and $\beta_y$ is $k_y$ dimensional. We assume that the controls are exogenous, that is,
\begin{align*}
    \e[X_t^{d} \varepsilon_t] &= 0 \\
    \e[X_t^y F_t] &= \e[X_t^y u_t] = 0
\end{align*}

The treatment and the effect of controls on the inference of the GIV estimators depend critically on the nature of the controls, specifically whether the controls are time-series variables or panel variables. 

First, consider that the controls on the disaggregated side are just time-series variables, i.e., $X_{it}^{y} = X_{t}^{y}$. In this case, de-meaning the panel removes the controls as well. Hence, by the Frisch-Waugh-Lovell Theorem, we can apply our 2-stage sequential GIV estimation on the following equations.
\begin{align*}
    d_t &= \phi_d \tilde{p}_{t} + \varepsilon_t \\
    D_N y_t &= D_N \Lambda F_t + D_N u_t
\end{align*}

where $\tilde{p}_t = p_t - X_t^{d'} \beta_p$ and $\beta_p = \e[X_t^{d'} X_t^{d}]^{-1} \e[X_t^{d'} p_t]$. Thus, it is very clear that the presence of controls has no effect on the first-stage estimation of the factor structure and effects only the second-stage GIV estimation. 

But if the controls on the disaggregated side are panel variables, i.e., of the form $X_{it}^{y}$, then simple demeaning does not remove this part. Hence, we need to apply our 2-stage sequential GIV estimation on the following equations.
\begin{align*}
   d_t &= \phi_d \tilde{p}_{t} + \varepsilon_t \\
    D_N y_t - D_N X_t^{y} \beta_y &= D_N \Lambda F_t + D_N u_t
\end{align*}
where $\tilde{p}_t$ is defined above and $\beta_y = \e[X_t^{y'} D_N X_t^{y}]^{-1} \e[X_t^{y'} D_N y_{t}]$. As $\beta_y$ needs to be estimated, the presence of controls affect the first-stage estimation of the factor structure as well as the second-stage GIV estimation. 

Hence I consider these two cases separately. First, consider the case of only time-series controls. 

\subsection{Time-Series Controls}
Without loss of generality, assume the controls enter the aggregate and disaggregated sides identically, $X_{it}^{y} = X_{t}^{y} = X_{t}^{d} = X_t$. We impose only that they are exogenous to the demand shock. They may be correlated with the idiosyncratic supply shocks and with the common factors.

\begin{Ass} \label{ass:controls_general}
    The time-series controls satisfy $\e[X_t \varepsilon_t] = 0$. Their correlation with the supply-side shocks is unrestricted, so in general
    \begin{equation*}
        \e[X_t u_t] \neq 0 \quad \text{and} \quad \e[X_t \tilde F_t] \neq 0 .
    \end{equation*}
\end{Ass}

The controls are common across units, so demeaning the factor equation removes them,
\begin{equation} \label{eq:factor model}
    D_N y_t = D_N \Lambda F_t + D_N u_t ,
\end{equation}
and the first stage is unaffected. We estimate the common component $\hat C_t$ as in Section~\ref{section:estimated factors} and form the estimated-factor instrument $\hat z_t = S'[D_N y_t - \hat C_t] = z_t - S'(\hat C_t - \tilde C_t)$.

We estimate the control coefficients by OLS and run the control-augmented GIV as an instrumental-variables regression of $d_t$ on $(p_t, X_t)$ with instruments $(\hat z_t, X_t)$. By the Frisch--Waugh--Lovell theorem,
\begin{equation*}
    \hat\phi_d - \phi_d = \frac{\hat z' M_X \varepsilon}{\hat z' M_X p}, \qquad M_X = I_T - X(X'X)^{-1}X',
\end{equation*}
where $X$ stacks the $X_t'$ and $M_X$ partials the controls out of the time series.

The control coefficients are estimated, so $M_X \varepsilon \neq \varepsilon$ in finite samples. The estimation noise enters the numerator through this gap. Expanding,
\begin{equation*}
    \frac{1}{\sqrt T} \hat z' M_X \varepsilon = \frac{1}{\sqrt T} \sum_t \hat z_t \varepsilon_t - \left( \frac1T \sum_t \hat z_t X_t' \right) \left( \frac1T \sum_t X_t X_t' \right)^{-1} \frac{1}{\sqrt T} \sum_t X_t \varepsilon_t.
\end{equation*}
By the Law of Large Numbers in Theorem~3.1 of \citet{Kojevnikov2021LimitVariables}, applied element-wise to the process $\{z_t X_t'\}_t$ conditional on $\sigma(\s)$, $\frac1T \sum_t z_t X_t' \xrightarrow{\p} \Gamma_{zX} = \e[z_t X_t']$. Conditions 3.1 and 3.2 hold by the same argument as in the proof of Theorem~\ref{theorem:estimated_strong_consistency_aggregate}. Replacing $z_t$ by $\hat z_t$ changes this average at smaller order by Proposition~\ref{prop:first stage effect}. By the Law of Large Numbers in Corollary 3.48 of \citet{White2001AsymptoticEconometricians}, as applied to $\frac1T \sum_t \tilde F_t' \varepsilon_t$ in the proof of Theorem~\ref{theorem:estimated_strong_consistency_aggregate}, $\frac1T \sum_t X_t X_t' \xrightarrow{\p} \Sigma_{XX}$.

The first term, $\frac1{\sqrt T} \sum_t \hat z_t \varepsilon_t$, carries the error from estimating the factor structure. Proposition~\ref{prop:first stage effect} at $Q_t = \varepsilon_t$ and $a_N = 1$ replaces the structural error by its factor-residualised version,
\begin{equation*}
    \frac1{\sqrt T} \sum_t \hat z_t \varepsilon_t = \frac1{\sqrt T} \sum_t z_t \bar\varepsilon_t + o_{\p}(\cdot), \qquad \bar\varepsilon_t = \varepsilon_t - \Delta_{\tilde{F} \varepsilon} \tilde F_t .
\end{equation*}
Collecting the two terms,
\begin{equation*}
    \frac{1}{\sqrt T} \hat z' M_X \varepsilon = \frac{1}{\sqrt T} \sum_t \xi_t + o_{\p}(\cdot), \qquad \xi_t = z_t \bar\varepsilon_t - \Gamma_{zX} \Sigma_{XX}^{-1} X_t \varepsilon_t .
\end{equation*}

The influence $\xi_t$ has two channels. The factor estimation contributes $z_t \bar\varepsilon_t$. The control residualization contributes $\Gamma_{zX} \Sigma_{XX}^{-1} X_t \varepsilon_t$.

Partial the controls out of the instrument and write $z_t^{\perp} = z_t - \Gamma_{zX}\Sigma_{XX}^{-1}X_t$. The influence becomes
\begin{equation*}
    \xi_t = z_t^{\perp} \bar\varepsilon_t - g_t, \qquad g_t = \Gamma_{zX}\Sigma_{XX}^{-1} X_t\, \Delta_{\tilde{F} \varepsilon} \tilde F_t .
\end{equation*}
The term $g_t$ carries the product of the controls and the factors. Its mean is $\e[g_t] = \Gamma_{zX}\Sigma_{XX}^{-1} \e[X_t \tilde F_t'] \Delta_{\tilde{F} \varepsilon}'$, governed by the covariance between the controls and the factors. It is of the same order as $z_t^{\perp}\bar\varepsilon_t$, since $\Gamma_{zX} = O_{\p}(\|z\|)$, and it enters the asymptotic variance.

The denominator converges by the same laws of large numbers. The $S$-weighted average $\frac1T \sum_t \hat z_t p_t \xrightarrow{\p} \Gamma_{zp}$ is the granular denominator of Theorem~\ref{theorem:estimated_strong_consistency_aggregate}, and $\frac1T \sum_t X_t p_t \xrightarrow{\p} \Sigma_{Xp}$ by Corollary 3.48 of \citet{White2001AsymptoticEconometricians}. With $\Gamma_{zX}$ and $\Sigma_{XX}$ as above,
\begin{equation*}
    \frac1T \hat z' M_X p \xrightarrow{\p} \Gamma_{zp} - \Gamma_{zX}\Sigma_{XX}^{-1}\Sigma_{Xp} = \e[z_t \tilde p_t] \defeq \Gamma_{z\tilde p} .
\end{equation*}

\begin{Theorem} \label{theorem:estimated_controls_general}
    Suppose Assumptions~\ref{ass:weak stationarity and idiosyncrasy} to \ref{ass:LN_moments and CLT} and Assumption~\ref{ass:controls_general} hold with $\mu \in (0,1)$. Then, conditional on $S$, for almost every realization of the shares, the control-augmented GIV estimator is consistent and asymptotically normal,
    \begin{equation*}
        \frac{\Gamma_{z\tilde p}}{\sqrt{V_{\mathrm{ctrl}}(S)}}\, \sqrt T\, (\hat\phi_d - \phi_d) \xrightarrow{d} \n(0,1),
    \end{equation*}
    with influence function $\xi_t = z_t \bar\varepsilon_t - \Gamma_{zX}\Sigma_{XX}^{-1}X_t\varepsilon_t$ and asymptotic variance
    \begin{equation*}
        V_{\mathrm{ctrl}}(S) = \lim_{T\to\infty}\frac1T \sum_{s=1}^T \sum_{t=1}^T \e[\xi_t \xi_s \mid S] .
    \end{equation*}
    Decomposing the long-run variance,
    \begin{equation*}
        V_{\mathrm{ctrl}}(S) = V_{z\bar\varepsilon}(S) \;-\; 2\,\Gamma_{zX}\Sigma_{XX}^{-1} c(S) \;+\; \Gamma_{zX}\Sigma_{XX}^{-1}\Omega_{X\varepsilon}\Sigma_{XX}^{-1}\Gamma_{zX}' ,
    \end{equation*}
    where $\Omega_{X\varepsilon} = \lim_{T\to\infty}\frac1T \sum_{s}\sum_{t}\e[X_t\varepsilon_t \varepsilon_s X_s']$ is the long-run variance of the control score and $c(S) = \lim_{T\to\infty}\frac1T\sum_{s}\sum_{t}\e[X_t\varepsilon_t\, z_s\bar\varepsilon_s \mid S]$ is its long-run covariance with the GIV score.
\end{Theorem}
\begin{proof}
    The expansions of the numerator and denominator are derived above. Consistency and the limiting distribution follow the proof of Theorem~\ref{theorem:estimated_strong_consistency_aggregate} with the GIV score $z_t\bar\varepsilon_t$ replaced by $\xi_t$. The control score is a measurable function of $(X_t, \varepsilon_t)$ and inherits the mixing of Assumption~\ref{ass:weak stationarity and idiosyncrasy}, so the Law of Large Numbers in Theorem~3.1 of \citet{Kojevnikov2021LimitVariables} applies to the denominator and the Central Limit Theorem~\ref{theorem:central limit theorem} applies to the numerator, as in that proof.
\end{proof}

The cross term $c(S)$ carries the covariance between the control score and the GIV score. The covariance between the controls and the factors enters the variance here, through the $\Delta_{\tilde{F} \varepsilon} \tilde F_t$ inside $\bar\varepsilon_t$.

\begin{remark}
    When the controls are external, uncorrelated with both the factors and the idiosyncratic shocks, $\e[X_t u_t] = \e[X_t \tilde F_t] = 0$, the instrument is orthogonal to the controls, $\Gamma_{zX} = 0$. The last two terms of $V_{\mathrm{ctrl}}(S)$ vanish, $z_t^{\perp} = z_t$, $\xi_t = z_t \bar\varepsilon_t$, and
\begin{equation*}
    V_{\mathrm{ctrl}}(S) = V_{z\bar\varepsilon}(S) .
\end{equation*}
The controls are asymptotically free and inference is that of Theorem~\ref{theorem:estimated_strong_consistency_aggregate}.
\end{remark}

The order of $\xi_t$ matches the known-factor numerator in every regime. The instrument satisfies $\Gamma_{zX} = O_{\p}(\|z\|) = O_{\p}(N^{-\delta})$ and the control score satisfies $\frac{1}{\sqrt T}\sum_t X_t\varepsilon_t = O_{\p}(1)$, so the control channel is $O_{\p}(1)$ in the strong regime and $O_{\p}(N^{-\delta})$ in the nearly weak regime, with $\delta = \min(1 - 1/\mu, 1/2)$. The rate of convergence is $\sqrt T$ in the strong regime and $\sqrt T / N^{\delta}$ in the nearly weak regime. The controls affect only the asymptotic variance.

Now consider the supply elasticity. The control-augmented estimator is the instrumental-variables regression of $y_{Et}$ on $(p_t, X_t)$ with instruments $(\hat z_t, X_t)$, where $y_{Et}$ is the equally weighted aggregate of Section~\ref{subsec:equal weights}. The equation it estimates is
\begin{equation*}
    y_{Et} = \phi_s p_t + X_t'\beta_y + v_{Et}, \qquad v_{Et} = \frac{e_N'(\Lambda F_t + u_t)}{N} .
\end{equation*}

The controls are included instruments in this equation, and its residual is built from the factors and the idiosyncratic shocks. Assumption~\ref{ass:controls_general} leaves the correlation of the controls with both unrestricted. A control correlated with either is correlated with $v_{Et}$. The score $\frac{1}{\sqrt T}\sum_t X_t v_{Et}$ then drifts at the rate $\sqrt T$ and the estimator is inconsistent.

The latitude of Assumption~\ref{ass:controls_general} is therefore available only for the demand elasticity, whose residual $\varepsilon_t$ is orthogonal to the controls by assumption. For the supply elasticity I require the controls to be external.

\begin{Ass} \label{ass:controls_supply}
    The time-series controls satisfy $\e[X_t u_t] = 0$ and $\e[X_t \tilde F_t] = 0$, in addition to $\e[X_t \varepsilon_t] = 0$.
\end{Ass}

Under Assumption~\ref{ass:controls_supply} the instrument is orthogonal to the controls, $\Gamma_{zX} = \e[z_t X_t'] = 0$. By the Frisch--Waugh--Lovell theorem,
\begin{equation*}
    \hat\phi_s - \phi_s = \frac{\hat z' M_X v_E}{\hat z' M_X p},
\end{equation*}
with $M_X$ as above. Expanding the numerator in the rescaling $a_N$,
\begin{equation*}
    \frac{a_N}{\sqrt T} \hat z' M_X v_E = \frac{a_N}{\sqrt T} \sum_t \hat z_t v_{Et} - \left( \frac{a_N}{T} \sum_t \hat z_t X_t' \right) \left( \frac1T \sum_t X_t X_t' \right)^{-1} \frac{1}{\sqrt T} \sum_t X_t v_{Et}.
\end{equation*}
The first term is the numerator of the no-controls supply estimator, treated in Appendix~\ref{subsec:proof of estimated_disaggregate}. In the second term, the summand $a_N z_t X_t'$ has mean $a_N \Gamma_{zX} = 0$, so $\frac{a_N}{T} \sum_t z_t X_t' = O_{\p}(1/\sqrt T)$, and replacing $z_t$ by $\hat z_t$ changes the average at smaller order by Proposition~\ref{prop:first stage effect}. The control score satisfies $\frac{1}{\sqrt T}\sum_t X_t v_{Et} = O_{\p}(1)$, since $\e[X_t v_{Et}] = 0$ under Assumption~\ref{ass:controls_supply}. The control channel is therefore $O_{\p}(1/\sqrt T)$ and drops out of the limit.

The denominator carries the same cross term. With $\frac1T \sum_t X_t p_t \xrightarrow{\p} \Sigma_{Xp}$ as before,
\begin{equation*}
    \frac{a_N^2}{T} \hat z' M_X p = \frac{a_N^2}{T} \sum_t \hat z_t p_t + a_N \cdot O_{\p}(1/\sqrt T) = a_N^2 \Gamma_{zp} + o_{\p}(1),
\end{equation*}
under $a_N^2/T \to 0$. Partialling the controls out of $p_t$ costs nothing because the instrument never loaded on them.

\begin{Theorem} \label{theorem:estimated_controls_supply}
    Suppose Assumptions~\ref{ass:weak stationarity and idiosyncrasy} to \ref{ass:LN_moments and CLT} and Assumption~\ref{ass:controls_supply} hold with $\mu \in (0,1)$ and $\sqrt{T}/N \to 0$. Then, conditional on $(S, \Lambda)$, for almost every realization of the shares and loadings, the control-augmented GIV estimator for the supply elasticity is consistent and asymptotically normal,
    \begin{equation*}
        \frac{\Gamma_{zp}}{\sqrt{V_{z \bar{v}_H}(S, \Lambda)}}\, \sqrt T\, (\hat\phi_s - \phi_s) \xrightarrow{d} \n(0,1),
    \end{equation*}
    with $\Gamma_{zp}$ and $V_{z \bar{v}_H}(S, \Lambda)$ as in Theorem~\ref{theorem:estimated_strong_consistency_disaggregate}. The controls do not enter the limiting distribution.
\end{Theorem}
\begin{proof}
    The control channel in the numerator is $O_{\p}(1/\sqrt T)$ and the cross term in the denominator is $O_{\p}(a_N/\sqrt T)$, both derived above. What remains is the numerator and denominator of the no-controls estimator, treated in Appendix~\ref{subsec:proof of estimated_disaggregate}, so the limit is that of Theorem~\ref{theorem:estimated_strong_consistency_disaggregate}.
\end{proof}

The orthogonality that was a special case for the demand elasticity is the general case for the supply elasticity. In the nearly weak regime the same argument gives the rate $\sqrt T / N^{\delta}$ of Theorem~\ref{theorem:estimated_weak_consistency_disaggregate}, since the control channel is $O_{\p}(1/\sqrt T)$ in every regime.

\subsection{Disaggregated Controls}

Consider the case when the controls on the disaggregated side are panel variables. Demeaning no longer removes them, so we estimate their coefficient $\beta_y$ and recover the factor structure from the residualized panel
\begin{equation*}
    D_N (y_t - X_t^y \beta_y) = D_N \Lambda F_t + D_N u_t .
\end{equation*}
The estimate of $\beta_y$ enters the first stage, so it adds a new source of error on top of the factor estimation. This is distinct from the time-series case of the previous subsection, where the first stage was untouched. But we will see in this subsection that this error vanishes asymptotically. The coefficient $\beta_y$ is estimated from the full panel, so its error averages over both $N$ and $T$ and vanishes at the rate $1/\sqrt{NT}$. 

The aggregate controls $X_t^d$ enter the asymptotic variance through the same channel as in the previous subsection. Hence the limiting distribution is the one already derived in Theorem~\ref{theorem:estimated_controls_general}. 

\begin{Ass} \label{ass:controls_panel}
    The panel controls are exogenous to both the factors and the idiosyncratic shocks, conditionally on the shares and loadings,
    \begin{equation*}
        \e[X_{it}^y F_t \mid S, \Lambda] = 0, \qquad \e[X_{it}^y u_{jt} \mid S, \Lambda] = 0 \quad \text{for all } i,j,t,
    \end{equation*}
    and the aggregate controls satisfy $\e[X_t^d \varepsilon_t] = 0$. The within second moment $\frac1{NT}\sum_t X_t^{y\prime} D_N X_t^y$ converges to a positive definite limit $Q_{X^y}$. The granular average $\bar X_{St}^y = S' X_t^y$ inherits the weakness of the instrument, $a_N \bar X_{St}^y = O_{\p}(1)$ uniformly across regimes.
\end{Ass}

This assumption is stronger than Assumption~\ref{ass:controls_general}. The time-series controls were allowed to correlate with the supply shocks and the factors. The panel controls are not. They are partialled out of the factor equation by least squares, so they must be clean of the structure we are trying to recover.

The panel control survives demeaning, so we cannot remove it the way we removed the time-series control. We estimate its coefficient directly. Demeaning the disaggregated equation gives
\begin{equation*}
    D_N y_t = D_N X_t^y \beta_y + D_N \Lambda F_t + D_N u_t ,
\end{equation*}
and the within estimator is
\begin{equation*}
    \hat\beta_y = \Big(\sum_t X_t^{y\prime} D_N X_t^y\Big)^{-1}\sum_t X_t^{y\prime} D_N y_t .
\end{equation*}
Under Assumption~\ref{ass:controls_panel} the factors and shocks are orthogonal to the controls, so
\begin{equation*}
    \hat\beta_y - \beta_y = \Big(\tfrac1{NT}\sum_t X_t^{y\prime} D_N X_t^y\Big)^{-1}\tfrac1{NT}\sum_t X_t^{y\prime} D_N(\Lambda F_t + u_t) = O_{\p}\Big(\tfrac1{\sqrt{NT}}\Big).
\end{equation*}
The score averages over both $N$ and $T$, so the estimation error vanishes at the parametric panel rate. This rate is faster than any term that survives in the GIV moment. The panel control therefore leaves the limiting distribution untouched, as we now show.

The factor structure is estimated from the residualized panel $D_N(y_t - X_t^y \hat\beta_y)$. Let $\hat C_t^{(\beta)}$ denote the principal-component estimate of the common component obtained from this panel. It is the analogue of the estimate $\hat C_t$ of Section~\ref{section:estimated factors}, which runs on the raw $D_N y_t$. The superscript $(\beta)$ records that the panel was residualized by $\hat\beta_y$ before the principal components were taken. Writing the partialled model $D_N(y_t - X_t^y \beta_y) = \tilde C_t + D_N u_t$, the residualized data is
\begin{equation*}
    D_N(y_t - X_t^y \hat\beta_y) = \tilde C_t + D_N u_t - D_N X_t^y(\hat\beta_y - \beta_y) ,
\end{equation*}
and the estimated-factor instrument $\hat z_t = S'[D_N(y_t - X_t^y \hat\beta_y) - \hat C_t^{(\beta)}]$ decomposes as
\begin{equation*}
    \hat z_t = z_t - S'(\hat C_t^{(\beta)} - \tilde C_t) - \bar X_{St}^y(\hat\beta_y - \beta_y) .
\end{equation*}
The first correction is the first-stage factor error of Section~\ref{section:estimated factors}, now computed on residualized data. The second correction is the direct imprint of the control estimate on the instrument. Both are new relative to the time-series case.

\begin{Lemma} \label{lemma:panel_control_negligible}
    Suppose Assumptions~\ref{ass:weak stationarity and idiosyncrasy} to \ref{ass:LN_moments and CLT} and Assumption~\ref{ass:controls_panel} hold. Writing $\hat C_t^{(\beta)} - \tilde C_t = (\hat C_t - \tilde C_t) + (\hat C_t^{(\beta)} - \hat C_t)$,
    \begin{equation*}
        \frac{a_N}{\sqrt T}\sum_t \bar X_{St}^y(\hat\beta_y - \beta_y)\varepsilon_t = O_{\p}\Big(\tfrac1{\sqrt{NT}}\Big), \qquad
        \frac{a_N}{\sqrt T}\sum_t S'(\hat C_t^{(\beta)} - \hat C_t)\varepsilon_t = O_{\p}\Big(\tfrac1{\sqrt{NT}}\Big),
    \end{equation*}
    and the corresponding denominator terms are $o_{\p}(1)$.
\end{Lemma}
\begin{proof}
    For the direct term, factor out the control estimate,
    \begin{equation*}
        \frac{a_N}{\sqrt T}\sum_t \bar X_{St}^y(\hat\beta_y - \beta_y)\varepsilon_t = (\hat\beta_y - \beta_y)\,\frac{a_N}{\sqrt T}\sum_t \bar X_{St}^y \varepsilon_t .
    \end{equation*}
    By Assumption~\ref{ass:controls_panel}, $a_N \bar X_{St}^y = O_{\p}(1)$, and the Central Limit Theorem~\ref{theorem:central limit theorem} gives $\frac{a_N}{\sqrt T}\sum_t \bar X_{St}^y \varepsilon_t = O_{\p}(1)$. With $\hat\beta_y - \beta_y = O_{\p}(1/\sqrt{NT})$ the product is $O_{\p}(1/\sqrt{NT})$.

    For the factor term, replacing the data $D_N y_t$ by $D_N(y_t - X_t^y\hat\beta_y)$ shifts the input to the principal-component step by $-D_N X_t^y(\hat\beta_y - \beta_y)$. The influence function \eqref{eq:influence function of estimator} is linear in the demeaned residual, so $\hat C_t^{(\beta)} - \hat C_t$ equals \eqref{eq:influence function of estimator} with $\bar u_{jt}$ replaced by $-\bar X_{jt}^y(\hat\beta_y - \beta_y)$. Factoring out $(\hat\beta_y - \beta_y)$, this is the substitution of Remark~\ref{rem:shock substitution} at $\bar W = -\bar X^y$ and $Q_t = \varepsilon_t$. Assumption~\ref{ass:controls_panel} supplies its hypotheses, with $\e[\bar X_{jt}^y \varepsilon_t \mid S, \Lambda] = 0$ the orthogonality already used for the direct term. Part (ii) then bounds the factor-space term by $O_{\p}(1)\cdot(\hat\beta_y - \beta_y) = O_{\p}(1/\sqrt{NT})$ and the loading-space term by $O_{\p}(1/\sqrt N)\cdot(\hat\beta_y - \beta_y) = O_{\p}(1/(N\sqrt T))$. The denominator terms carry the same extra $(\hat\beta_y - \beta_y)$ factor relative to Lemma~\ref{lemma:denominator} and are therefore $o_{\p}(1)$.
\end{proof}

Both panel-control corrections are of smaller order than the GIV score. The numerator therefore reduces to the time-series expansion. The aggregate controls $X_t^d$ are residualized by $M_{X^d} = I_T - X^d(X^{d\prime}X^d)^{-1}X^{d\prime}$ exactly as before, so the influence is
\begin{equation*}
    \xi_t = z_t \bar\varepsilon_t - \Gamma_{zX^d}\Sigma_{X^d X^d}^{-1} X_t^d \varepsilon_t ,
\end{equation*}
identical to the time-series case. The denominator converges to $\Gamma_{z\tilde p}$.

\begin{Theorem} \label{theorem:estimated_controls_panel}
    Suppose Assumptions~\ref{ass:weak stationarity and idiosyncrasy} to \ref{ass:LN_moments and CLT} and Assumption~\ref{ass:controls_panel} hold with $\mu \in (0,1)$. Then, conditional on $S$, for almost every realization of the shares, the GIV estimator with panel and aggregate controls is consistent and asymptotically normal,
    \begin{equation*}
        \frac{\Gamma_{z\tilde p}}{\sqrt{V_{\mathrm{ctrl}}(S)}}\, \sqrt T\, (\hat\phi_d - \phi_d) \xrightarrow{d} \n(0,1) ,
    \end{equation*}
    with the same influence function $\xi_t$ and the same asymptotic variance $V_{\mathrm{ctrl}}(S)$ as in Theorem~\ref{theorem:estimated_controls_general}. The panel controls do not enter the limiting variance.
\end{Theorem}
\begin{proof}
    The numerator and denominator expansions are derived above. The panel-control corrections to the instrument are $O_{\p}(1/\sqrt{NT})$ by Lemma~\ref{lemma:panel_control_negligible}, dominated by the GIV score $z_t\bar\varepsilon_t$. The influence and the denominator limit therefore coincide with those of Theorem~\ref{theorem:estimated_controls_general}, and the limiting distribution follows from that theorem.
\end{proof}

The panel control is estimated from the full panel of $NT$ observations, so its error is $O_{\p}(1/\sqrt{NT})$. This shrinks faster than the contribution of the GIV score in every regime, so it leaves no trace in the limiting distribution. The rate of convergence is $\sqrt T$ in the strong regime and $\sqrt T / N^{\delta}$ in the nearly weak regime, with $\delta = \min(1 - 1/\mu, 1/2)$, exactly as without panel controls. Only the aggregate controls $X_t^d$ affect the asymptotic variance, through the channel of the previous subsection.

The supply elasticity uses the same first stage and the same estimate $\hat\beta_y$. Assumption~\ref{ass:controls_panel} already imposes the exogeneity. The estimator residualizes the equal-weights outcome,
\begin{equation*}
    \hat\phi_s = \frac{\sum_t \hat z_t \big( y_{Et} - \bar X_{Et}^{y\prime} \hat\beta_y \big)}{\sum_t \hat z_t p_t}, \qquad \bar X_{Et}^y = \frac{X_t^{y\prime} e_N}{N} .
\end{equation*}
Writing $y_{Et} - \bar X_{Et}^{y\prime}\hat\beta_y = \phi_s p_t + v_{Et} - \bar X_{Et}^{y\prime}(\hat\beta_y - \beta_y)$, the estimator adds one correction to those of Lemma~\ref{lemma:panel_control_negligible}, the outcome term $\frac{a_N}{\sqrt T}\sum_t \hat z_t \bar X_{Et}^{y\prime}(\hat\beta_y - \beta_y)$.

\begin{Theorem} \label{theorem:estimated_controls_panel_supply}
    Suppose Assumptions~\ref{ass:weak stationarity and idiosyncrasy} to \ref{ass:LN_moments and CLT} and Assumption~\ref{ass:controls_panel} hold with $\mu \in (0,1)$ and $\sqrt{T}/N \to 0$, and suppose $\e\|\bar X_{Et}^y\|^{4 + 2\pi}$ is bounded uniformly in $N$. Then, conditional on $(S, \Lambda)$, for almost every realization of the shares and loadings, the GIV estimator for the supply elasticity with panel controls is consistent and asymptotically normal,
    \begin{equation*}
        \frac{\Gamma_{zp}}{\sqrt{V_{z \bar{v}_H}(S, \Lambda)}}\, \sqrt T\, (\hat\phi_s - \phi_s) \xrightarrow{d} \n(0,1),
    \end{equation*}
    with $\Gamma_{zp}$ and $V_{z \bar{v}_H}(S, \Lambda)$ as in Theorem~\ref{theorem:estimated_strong_consistency_disaggregate}. The panel controls do not enter the limiting distribution.
\end{Theorem}
\begin{proof}
    The instrument corrections are those of Lemma~\ref{lemma:panel_control_negligible} and are $O_{\p}(1/\sqrt{NT})$. For the outcome term, factor out the control estimate,
    \begin{equation*}
        \frac{a_N}{\sqrt T}\sum_t \hat z_t \bar X_{Et}^{y\prime}(\hat\beta_y - \beta_y) = \left( \frac{a_N}{\sqrt T}\sum_t \hat z_t \bar X_{Et}^{y\prime} \right) (\hat\beta_y - \beta_y).
    \end{equation*}
    By Assumption~\ref{ass:controls_panel}, $\e[z_t \bar X_{Et}^y \mid S, \Lambda] = 0$, so the Central Limit Theorem~\ref{theorem:central limit theorem} gives $\frac{a_N}{\sqrt T}\sum_t z_t \bar X_{Et}^{y\prime} = O_{\p}(1)$, and replacing $z_t$ by $\hat z_t$ changes the sum at smaller order by Proposition~\ref{prop:first stage effect}. With $\hat\beta_y - \beta_y = O_{\p}(1/\sqrt{NT})$ the outcome term is $O_{\p}(1/\sqrt{NT})$. Every correction is dominated by the GIV score, so the limit is that of Theorem~\ref{theorem:estimated_strong_consistency_disaggregate}.
\end{proof}

The aggregate controls $X_t^d$ enter only the demand equation and play no role here. If time-series controls enter the supply equation as well, Assumption~\ref{ass:controls_supply} and the argument of the previous subsection apply unchanged. The rate is $\sqrt T$ in the strong regime and $\sqrt T / N^{\delta}$ in the nearly weak regime, exactly as without controls. 
\section{Data Construction Notes} \label{appendix_data}

\subsection{Copper}

Monthly country-level data from Bloomberg cover January~2009 to December~2025
($T = 204$ months).  The supply panel for the GIV instrument consists of refined
copper supply across $N = 29$ countries, including a rest-of-world residual.

The price series is the LME spot copper price (monthly average), deflated by
U.S.\ CPI rebased to 2015~$= 100$.  The transformation
in~\eqref{eq:yoygrowth} is applied to all series, leaving $T = 192$ estimation
periods.  The covariate matrix~$X_t$ includes an intercept, one lag of aggregate
refined demand growth, and the trade-weighted U.S.\ dollar index.

\subsection{Crude Oil}

Monthly EIA International Energy Statistics data on crude oil production cover
January~1973 to November~2025 ($T = 635$ months).  The USSR and its successor
states are treated as a single continuous unit, the Former USSR series pre-1992
and the sum of 15 successor states post-1991.  After dropping countries with any
zero or missing observation, the panel has $N = 21$ units, 20 countries plus
rest of world.  The year 2020 is excluded from estimation to avoid contamination
from the COVID-19 production collapse.

The primary price series is a splice of the FRED OILPRICE series (January
1946--August~2024) and WTI (September 2024--December 2025), deflated by U.S.\
CPI rebased to 2015~$= 100$.  The transformation in~\eqref{eq:yoygrowth} is
applied to all series, leaving $T = 611$ estimation periods.  Aggregate demand
is constructed from aggregate production growth adjusted for inventory changes.
The covariate matrix~$X_t$ includes an intercept and two lags of aggregate
production growth.

\subsection{Natural Gas}

Monthly country-level natural gas production data are obtained from the JODI
Gas Database, covering January~2010 to November~2025 ($T = 191$ months).
Production is measured in million standard cubic meters.  The sample begins
in January~2010 because most major producing countries---including the United
States, Canada, and Norway---are not reported in the JODI database before
that date.  It ends in November~2025 because China, a major producer,
is missing for December~2025.

Two major producers, Iran and Qatar, have substantial gaps in JODI
reporting.  Iran accounts for approximately 7\% of global production but
ceased reporting after mid-2018.  Qatar accounts for approximately 6.6\% but
has intermittent gaps between 2010 and 2018.  To address this, we supplement
the JODI data with annual dry natural gas production from the U.S.\ Energy
Information Administration (EIA), reported in billion cubic feet.  For years
in which both sources have complete monthly coverage, the ratio of JODI gross
production to EIA dry gas production is stable, approximately 1.25 for Iran
and 1.23 for Qatar.  For each year with missing or zero JODI months, we
scale the EIA annual total by this ratio, subtract the sum of any valid JODI
months, and distribute the residual uniformly across the missing months.  For
2025, where EIA data are not yet available, we use the 2024 EIA value as a
proxy.  Iran's series remains approximately 85\% interpolated and is
therefore absorbed into Rest of World rather than retained as a separate
cross-sectional unit.  Qatar has better actual coverage and is kept as an
individual country.

After this interpolation, countries with any remaining zero or missing
production value over the sample period are absorbed into Rest of World.  The
final panel contains $N = 27$ countries.  The year 2020 is excluded from
estimation to avoid contamination from the COVID-19 demand collapse.

The price series is the Henry Hub Natural Gas Spot Price (dollars per million
British thermal units), deflated by U.S.\ CPI rebased to
2015~$= 100$.  The transformation in~\eqref{eq:yoygrowth} is applied to all
series.  Combined with the lag burn-in, the effective estimation sample has
$T = 178$ periods.  The covariate matrix~$X_t$ includes an intercept, one lag
of aggregate production growth, the VIX, the year-on-year midpoint growth rate
of the real WTI crude oil price, and the trade-weighted U.S.\ dollar index.
The real oil price is included because natural gas and oil are partial
substitutes in power generation and industrial use, making oil price variation
a relevant demand shifter.

Figure~\ref{fig:prices} shows the year-on-year real price growth rates for
all three commodities, illustrating the sample periods and key episodes of
price variation exploited for identification.

\begin{figure}[H]
  \centering
  \includegraphics[width=\linewidth]{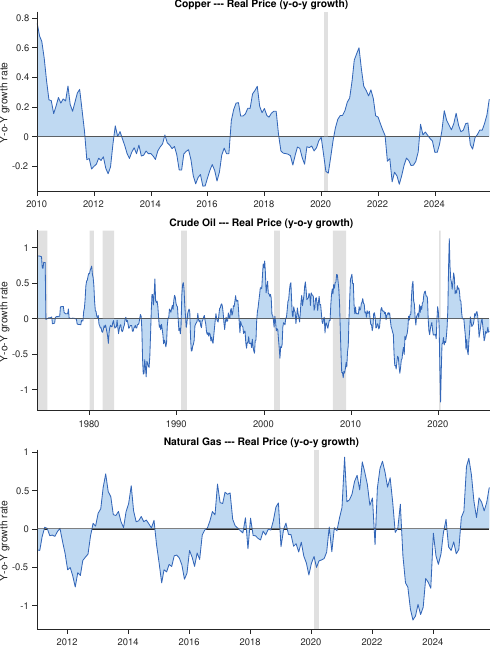}
  \caption{Year-on-year real price growth rates.
    Copper: January 2009--December 2025 (LME spot, CPI-deflated, 2015 base).
    Crude Oil: January 1973--November 2025 (spliced OILPRICE/WTI, CPI-deflated).
    Natural Gas: January 2011--November 2025 (Henry Hub spot, CPI-deflated, 2015 base).
    Gray bands indicate NBER recession dates.}
  \label{fig:prices}
\end{figure}
\section{Simulation Design}    \label{app:mc design}

\subsection{Data Generating Process}

The two structural equations of interest are the aggregate demand for a commodity,
\begin{equation}\label{eq:structural}
    d_t = y_{St} = \phi_d p_t + \varepsilon_t,
\end{equation}
The supply panel follows the factor model
\begin{equation}\label{eq:factor_model}
    y_t = \iota_N \phi_s p_t + \iota_N F_t^1 + \tilde{\Lambda} \tilde{F}_t + u_t,
\end{equation}
where $F_t^1$ is a common factor, $\tilde{\Lambda}$ is the $N \times r - 1$ matrix of demeaned factor loadings, $\tilde{F}_t$ is the $r-1 \times 1$ vector of factors, and $u_t$ is the $N \times 1$ vector of idiosyncratic shocks. Market clearing then gives the reduced form of the price,
\begin{equation}\label{eq:mc_price}
    p_t = \frac{F_t^1 + S'\tilde{\Lambda}\tilde{F}_t + S'u_t - \varepsilon_t}{\phi_d - \phi_s}.
\end{equation}

The supply elasticity enters through the loading $\iota_N$ on $p_t$, which demeaning annihilates. So $D_N y_t = \tilde{\Lambda}\tilde{F}_t + D_N u_t$ whatever the value of $\phi_s$, and the instrument block of the design does not depend on it. Setting $\phi_s = 0$ collapses \eqref{eq:mc_price} to $p_t = (y_{St} - \varepsilon_t)/\phi_d$ and recovers the design used for the aggregate results below. The disaggregated results set $\phi_s$ to its empirical estimate instead.

The demeaned factor loadings $\tilde{\Lambda}$ and the factors $\tilde{F}_t$ are extracted from each empirical estimation and held fixed across all Monte Carlo replications.

All remaining quantities are drawn independently each replication. The common factor and idiosyncratic shocks are drawn from normal distributions and scaled to be comparable to the empirical factor:
\begin{equation}\label{eq:draws}
    F_t^1 \sim \mathcal{N}(0, \sigma_F^2), \qquad
    u_{it} \sim \mathcal{N}(0, \sigma_u^2),
\end{equation}
where $\sigma_F = \text{std}(\hat{F})$ matches the standard deviation of the empirical factor series, and $\sigma_u = \text{std}(\hat{u}_t)$ matches the pooled standard deviation of the empirical PCA estimation. The structural error $\varepsilon_t$ is drawn with $\text{std}(\varepsilon_t) = \sigma_F$ and $\text{corr}(\varepsilon_t, \tilde{F}_t) = 0.8$.

\subsubsection{Calibration}

For the Copper DGP, the refined copper supply panel consists of year-on-year midpoint growth rates for $N = 29$ countries over $T = 192$ months. The ER criterion yields $r-1 = 1$ factor. The calibration scales are $\sigma_F = 0.109$ and $\sigma_u = 0.171$, and the true elasticity is set to $\phi_d = -0.148$, the control-robust GIV estimate of Table~\ref{tab:elasticity_controls}.

For the Crude Oil DGP, the crude oil production panel consists of year-on-year midpoint growth rates for $N = 21$ countries over $T = 611$ months (COVID year 2020 dropped). The ER criterion yields $r-1 = 2$ factors. The calibration scales are $\sigma_F = 0.090$ and $\sigma_u = 0.142$, and the true elasticity is set to $\phi_d = -0.104$, the control-robust GIV estimate of Table~\ref{tab:elasticity_controls}.

The supply elasticity is calibrated the same way, to the control-robust estimate of Table~\ref{tab:elasticity_controls} on the corresponding panel. This gives $\phi_s = 0.141$ for Copper and $\phi_s = 0.256$ for Crude Oil. The aggregate results of Section~\ref{subsec:mc results aggregate} are run at $\phi_s = 0$, which is the case \eqref{eq:factor_model} nests.

\subsubsection{Individual sizes and Granularity}

Individual sizes are drawn from the power distribution $\mathbb{P}(s_i > s) = c \, s^{-\mu}$ via the inverse CDF transform $s_i = (1 - U_i)^{-1/\mu}$ with $U_i \sim \mathcal{U}[0,1]$, and normalized so that $\sum_{i=1}^N S_i = 1$. Recall that the Pareto exponent $\mu$ determines the market concentration, and consequently the strength of the instrument. We consider $\mu \in \{0.3, 0.5, 0.8, 1.2, 1.4, 1.8, 2.5, 3.5, 6.0\}$, with $B = 5{,}000$ replications per value.

From this DGP, we construct the observational data using \eqref{eq:factor_model}, $y_{St} = S' y_t$, and the price \eqref{eq:mc_price}. The price is endogenous, and the sign of the resulting OLS bias differs across the two equations. For the demand equation, $\text{Cov}(p_t, \varepsilon_t) = -\sigma_\varepsilon^2 / \phi_d > 0$ since $\phi_d < 0$, so the bias is upward. For the supply equation, $\text{Cov}(p_t, v_{Et}) = [\sigma_F^2 + \sigma_u^2/N]/(\phi_d - \phi_s) < 0$ since $\phi_d - \phi_s < 0$, so the bias is downward. The two tables therefore misbehave in opposite directions once identification weakens, and neither is anomalous.

\section{Augmenting the Cross Section}    \label{app:augment N}

The Crude Oil panel has only $N = 21$ countries, which limits the scope for
studying the estimator's behavior as $N$ grows. To investigate the effect of a
larger cross section while preserving the empirical factor structure, we augment
$N$ by resampling rows of $\tilde{\Lambda}$ with replacement. For a target
$N_{\text{sim}} > N$, we draw $N_{\text{sim}}$ rows from the $21 \times r$
empirical loading matrix uniformly with replacement. The factors $\tilde{F}_t$,
the calibration scales $\sigma_F$ and $\sigma_u$, and the true elasticities
remain unchanged. Each replication still draws fresh $u_{it}$, $F_t^1$,
$\varepsilon_t$, and $S$, so the augmented countries are distinguished by their
idiosyncratic shocks and shares even when they share a loading vector.
Everything else follows the design of Appendix~\ref{app:mc design}. I report
results for $N \in \{50, 100\}$, to be read against the $N = 21$ panels of
Tables~\ref{tab:mc_table1} and~\ref{tab:dis_table1}.

\subsection{Demand Elasticity}

\begin{table}[htbp]
\centering
\caption{Monte Carlo Simulation Results: Augmented Cross-Section (Crude Oil)}
\label{tab:mc_table2}
\begin{tabular}{ccccccc}
\toprule
$\mu$ & Median $\hat{\phi}_d$ & Bias & RMSE & Coverage & CI Length & Median $F$ \\
\midrule
\multicolumn{7}{l}{\textit{Panel A: Crude Oil ($N=50$, $T=611$, $\phi_d=-0.104$)}} \\
\midrule
0.3 & -0.1044 & 0.0000 & 0.0027 & 0.940 & 0.0082 & 476.7 \\
0.5 & -0.1045 & -0.0001 & 0.0035 & 0.949 & 0.0107 & 285.3 \\
0.8 & -0.1045 & -0.0003 & 0.0057 & 0.945 & 0.0169 & 114.1 \\
1.2 & -0.1044 & -0.0007 & 0.0118 & 0.950 & 0.0307 & 37.3 \\
1.4 & -0.1043 & -0.0019 & 0.0427 & 0.954 & 0.0380 & 24.7 \\
1.8 & -0.1042 & 0.0377 & 2.9387 & 0.945 & 0.0543 & 12.1 \\
2.5 & -0.1046 & -0.0243 & 0.7354 & 0.953 & 0.0877 & 5.0 \\
3.5 & -0.1008 & -0.0733 & 4.8830 & 0.961 & 0.1380 & 2.1 \\
6.0 & -0.0910 & 0.0209 & 1.0870 & 0.974 & 0.2642 & 0.7 \\
\midrule
\multicolumn{7}{l}{\textit{Panel B: Crude Oil ($N=100$, $T=611$, $\phi_d=-0.104$)}} \\
\midrule
0.3 & -0.1044 & 0.0000 & 0.0023 & 0.947 & 0.0079 & 512.0 \\
0.5 & -0.1045 & -0.0001 & 0.0031 & 0.952 & 0.0099 & 326.8 \\
0.8 & -0.1045 & -0.0003 & 0.0056 & 0.948 & 0.0165 & 121.3 \\
1.2 & -0.1041 & -0.0015 & 0.0323 & 0.958 & 0.0321 & 33.3 \\
1.4 & -0.1047 & -0.0017 & 0.0570 & 0.950 & 0.0423 & 19.5 \\
1.8 & -0.1045 & -0.0021 & 0.2398 & 0.955 & 0.0641 & 8.6 \\
2.5 & -0.1024 & -0.0148 & 0.7250 & 0.956 & 0.1095 & 3.1 \\
3.5 & -0.0969 & 0.0157 & 1.3967 & 0.970 & 0.1753 & 1.3 \\
6.0 & -0.0856 & 0.0290 & 1.6255 & 0.980 & 0.3036 & 0.6 \\
\bottomrule
\end{tabular}

\vspace{4pt}
\raggedright\footnotesize
\textit{Notes:} Shares drawn from power law $\mathbb{P}(s_i > s) = c\,s^{-\mu}$. Coverage is the fraction of replications where the true value lies in the 95\% confidence interval. Standard errors: White heteroskedasticity-robust. $F$ is the first-stage $F$-statistic.
\end{table}

The augmented Crude Oil experiments separate the role of $N$ from the role of
$T$ and confirm the rate-based predictions of the theory regime by regime.
In the strong regime, performance is essentially invariant to $N$. RMSE at
$\mu = 0.3$ moves only from $0.0033$ at $N = 21$ to $0.0027$ at $N = 50$
and $0.0023$ at $N = 100$. In the nearly weak regime with
$\mu \in (1,2)$, consistency requires $N/T \to 0$. Increasing $N$ at
fixed $T$ therefore degrades performance, and indeed RMSE at $\mu = 1.4$
rises from $0.034$ at $N = 21$ to $0.043$ at $N = 50$ and $0.057$ at
$N = 100$. For $\mu > 2$ the rate is $\sqrt{T/N}$, so larger $N$ at
fixed $T$ is doubly costly. Raising $N$ at fixed $T$ also lifts $N/T$ away
from zero, moving the design out of the nearly weak regime and toward the
weak regime $N/T \to c > 0$, where the estimator is inconsistent. Consistent
with this, CI length grows uniformly with
$N$ and the median $F$ falls. At $\mu = 6.0$ it declines from $1.0$ at
$N = 21$ to $0.6$ at $N = 100$. RMSE remains large and
tail-dominated throughout, reaching $1.6$ at $N = 100$.

\subsection{Supply Elasticity}

\begin{table}[htbp]
\centering
\caption{Monte Carlo Simulation Results, Disaggregated Supply Elasticity: Augmented Cross-Section (Crude Oil)}
\label{tab:dis_table2}
\begin{tabular}{ccccccc}
\toprule
$\mu$ & Median $\hat{\phi}_s$ & Bias & RMSE & Coverage & CI Length & Median $F$ \\
\midrule
\multicolumn{7}{l}{\textit{Panel A: Crude Oil ($N=50$, $T=611$, $\phi_s=0.256$)}} \\
\midrule
0.3 & 0.2562 & 0.0005 & 0.0146 & 0.946 & 0.0446 & 476.7 \\
0.5 & 0.2563 & 0.0009 & 0.0187 & 0.952 & 0.0577 & 285.3 \\
0.8 & 0.2563 & 0.0023 & 0.0315 & 0.945 & 0.0911 & 114.1 \\
1.2 & 0.2566 & 0.0077 & 0.0674 & 0.942 & 0.1620 & 37.3 \\
1.4 & 0.2561 & 0.0178 & 0.1943 & 0.943 & 0.1998 & 24.7 \\
1.8 & 0.2563 & 0.0402 & 2.5172 & 0.929 & 0.2856 & 12.1 \\
2.5 & 0.2491 & 0.1360 & 4.2782 & 0.914 & 0.4506 & 5.0 \\
3.5 & 0.2257 & -1.5478 & 120.7673 & 0.889 & 0.6849 & 2.1 \\
6.0 & 0.1393 & -0.0447 & 3.6127 & 0.859 & 1.1612 & 0.7 \\
\midrule
\multicolumn{7}{l}{\textit{Panel B: Crude Oil ($N=100$, $T=611$, $\phi_s=0.256$)}} \\
\midrule
0.3 & 0.2565 & 0.0005 & 0.0130 & 0.944 & 0.0442 & 512.0 \\
0.5 & 0.2564 & 0.0009 & 0.0176 & 0.949 & 0.0552 & 326.8 \\
0.8 & 0.2558 & 0.0018 & 0.0331 & 0.945 & 0.0908 & 121.3 \\
1.2 & 0.2554 & 0.0144 & 0.3000 & 0.947 & 0.1738 & 33.3 \\
1.4 & 0.2566 & 0.0149 & 0.5178 & 0.939 & 0.2269 & 19.5 \\
1.8 & 0.2553 & 0.0268 & 2.1006 & 0.925 & 0.3388 & 8.6 \\
2.5 & 0.2490 & 0.0809 & 4.0606 & 0.906 & 0.5770 & 3.1 \\
3.5 & 0.1925 & -0.0376 & 5.9266 & 0.867 & 0.8647 & 1.3 \\
6.0 & 0.0917 & -0.0699 & 8.4745 & 0.828 & 1.2075 & 0.6 \\
\bottomrule
\end{tabular}

\vspace{4pt}
\raggedright\footnotesize
\textit{Notes:} Estimates of the disaggregated supply elasticity $\phi_s$ in $y_t = e_N\phi_s p_t + \Lambda F_t + u_t$, aggregated with equal weights $\hat\phi_s(I_N)$ as recommended in Theorems~\ref{theorem:estimated_strong_consistency_disaggregate} and~\ref{theorem:estimated_weak_consistency_disaggregate}. Shares drawn from power law $\mathbb{P}(s_i > s) = c\,s^{-\mu}$. Coverage is the fraction of replications where the true value lies in the 95\% confidence interval. Standard errors: White heteroskedasticity-robust, with the first-stage correction for estimation error in $\hat z$. $F$ is the first-stage $F$-statistic.
\end{table}

Table~\ref{tab:dis_table2} augments the cross section at fixed $T$ for Crude Oil, which pushes
$N/T$ up and so moves further into the failure case rather than out of it. Behavior degrades in
the direction the $N/T \to c$ theory predicts. At $\mu = 6.0$ coverage falls from $0.866$ to
$0.859$ to $0.828$ and median $F$ falls from $1.0$ to $0.7$ to $0.6$ as $N$ goes from $21$ to $50$
to $100$. 

\section{Estimation of Factor Loadings} \label{app:estimation of factor loadings}
The demeaned disaggregated panel has a factor structure. That is,
\begin{equation*}
    \Bar{y}_t = \tilde{\Lambda} \tilde{F}_t + \Bar{u}_t
\end{equation*}
where $\tilde{\Lambda}$ is a $N \times r - 1$ matrix. The columns of the estimator, $\hat{\Lambda}$ are the eigenvectors corresponding to the $r-1$ largest eigenvalues of $\frac{\Bar{Y}' \Bar{Y}}{NT}$.
\begin{equation} \label{eq:expansion of the estimator matrix}
    \frac{\Bar{Y}' \Bar{Y}}{NT} = \frac{\tilde{\Lambda} \tilde{F}' \tilde{F} \tilde{\Lambda}'}{NT} + \frac{\tilde{\Lambda} \tilde{F}' \Bar{u}}{NT} + \frac{\Bar{u}'\tilde{F} \tilde{\Lambda}'}{NT} +  \frac{\Bar{u}'\Bar{u}}{NT}
\end{equation}

That is, the estimator satisfies
\begin{equation*}
    \left[ \frac{\Bar{Y}' \Bar{Y}}{NT} \right] \hat{\Lambda} = \hat{\Lambda}V_{NT}
\end{equation*}
where $V_{NT}$ is the diagonal matrix consisting of the $r-1$ largest eigenvalues of $\frac{\Bar{Y}' \Bar{Y}}{NT}$ in decreasing order. Using \eqref{eq:expansion of the estimator matrix}, we can write this as:
\begin{equation}\label{eq:estimator_expanded}
    \left[ {\frac{\tilde{\Lambda} \tilde{F}' \tilde{F} \tilde{\Lambda}'}{NT} \cdot \hat{\Lambda} } + {\frac{\tilde{\Lambda} \tilde{F}' \Bar{u}}{NT} \cdot \hat{\Lambda} } + {\frac{\Bar{u}'\tilde{F} \tilde{\Lambda}'}{NT} \cdot \hat{\Lambda} } + {\frac{\Bar{u}'\Bar{u}}{NT} \cdot \hat{\Lambda}}  \right] = \hat{\Lambda}V_{NT}
\end{equation}

Identification of factor models requires restrictions. We impose $\hat{\Lambda}' \hat{\Lambda} = I_{r-1}$. In this case, the estimated factor loadings will be a consistent estimator of a rotation of the true loadings. That is, the estimated factor loadings, $\hat{\Lambda}$ will be a consistent estimator of $\tilde{\Lambda} \hat{H}$. Let the arbitrary rotation matrix be
\begin{equation*}
    \hat{H} = \left[ \frac{\tilde{F}' \tilde{F}}{NT} \right]\cdot \tilde{\Lambda}' \hat{\Lambda} \cdot V_{NT}^{-1}
\end{equation*}
Using this definition, we write \eqref{eq:estimator_expanded} as
\begin{equation}
    [\hat{\Lambda} - \tilde{\Lambda} \hat{H} ] V_{NT} = {\frac{\tilde{\Lambda} \tilde{F}' \Bar{u}}{NT} \cdot \hat{\Lambda} } + {\frac{\Bar{u}'\tilde{F} \tilde{\Lambda}'}{NT} \cdot \hat{\Lambda} } + {\frac{\Bar{u}'\Bar{u}}{NT} \cdot \hat{\Lambda}}
\end{equation}

Consider the $i$'th row of the above matrix equation
\begin{align*}
    V_{NT} [ \hat{\Blambda}_i - \hat{H}' \tilde{\Blambda}_i ] &= \frac{1}{N} \sum_{j=1}^N \hat{\Blambda}_j \tilde{\Blambda}_i' \tilde{F}' \Bar{u}_j/T + \frac{1}{N} \sum_{j=1}^N \hat{\Blambda}_j \tilde{\Blambda}_j' \tilde{F}' \Bar{u}_i/T + \\
    & \frac{1}{N} \sum_{j=1}^N \hat{\Blambda}_j \left( \Bar{u}_i' \Bar{u}_j/T - \e \left[ \Bar{u}_i' \Bar{u}_j/T \right]  \right) + \frac{1}{N} \sum_{j=1}^N \hat{\Blambda}_j \e \left[ \Bar{u}_i' \Bar{u}_j/T \right]
\end{align*}
Now we will define some terms for convenience
\begin{align*}
    \xi_{ij} &= \tilde{\Blambda}_i' \tilde{F}' \Bar{u}_j/T \\
    \eta_{ij}  &=  \tilde{\Blambda}_j' \tilde{F}' \Bar{u}_i/T \\
    \Bar{\gamma}(i,j) &= \e \left[ \Bar{u}_i' \Bar{u}_j/T \right] = \e \left[ \frac{1}{T} \sum_{t=1}^T \Bar{u}_{it} \Bar{u}_{jt} \right] \\
    \zeta_{ij} &= \Bar{u}_i' \Bar{u}_j/T - \Bar{\gamma}(i,j)
\end{align*}

Using these, we can rewrite the expression for the $i$'th row of the estimator as
\begin{equation} \label{eq:estimator for lambda expanded row}
 V_{NT} \big(  \hat{\Blambda}_i - \hat{H}' \tilde{\Blambda}_i \big) =  \frac{1}{N} \sum_{j=1}^N \hat{\Blambda}_j \Bar{\gamma}(i,j) + \frac{1}{N} \sum_{j=1}^N \hat{\Blambda}_j \zeta_{ij} + \frac{1}{N} \sum_{j=1}^N \hat{\Blambda}_j \eta_{ij} + \frac{1}{N} \sum_{j=1}^N \hat{\Blambda}_j \xi_{ij} 
\end{equation}

We will analyze each of the terms in the summand on the right hand side. But first, Recall that $\hat{\Lambda}' \hat{\Lambda} = I_r$ implies
\begin{align*}
      \Vert \hat{\Lambda} \Vert &= (\text{tr}[\hat{\Lambda}'\hat{\Lambda}])^{1/2} = \sqrt{r} = O_{\p}(1) \\
      \frac{1}{N} \sum_{i=1}^N \Vert {\lambda}_i \Vert ^2 &= \frac{1}{N} \text{tr} \left( {\Lambda} {\Lambda}' \right) = \frac{1}{N} \text{tr} \left( {\Lambda}' {\Lambda} \right) \xrightarrow{p} \text{tr}(\Sigma_{\Lambda}) \\
      \frac{1}{N} \sum_{i=1}^N \Vert {\hat{\lambda}}_i \Vert ^2 &= \frac{1}{N} \text{tr} \left( {\hat{\Lambda}} {\hat{\Lambda}}' \right) = \text{tr} \left( {\hat{\Lambda}}' {\hat{\Lambda}}/N \right) = r
\end{align*}

Note that the expressions involving our estimator contains the de-meaned idiosyncratic error. But all our assumptions are on the original error. Hence, it is useful to fix the relationship between the two.
\begin{align*}
    \Bar{u}_{it} \Bar{u}_{jt} &= (u_{it} - \Bar{u}_t)(u_{jt} - \Bar{u}_t) \\
    &= (u_{it} - O_{\p} ( \frac{1}{\sqrt{N}} \ ) )(u_{jt} - O_{\p} ( \frac{1}{\sqrt{N}}  ) ) \\
    &= u_{it} u_{jt} + O_{\p}\left( \frac{1}{\sqrt{N}} \right)
\end{align*}
where we use the fact that $u_t$ is mean zero, finite second moments, and weak cross-sectional dependence under Assumption \ref{ass:LN_time and cross sectional dependence}. Hence, for considerations of stochastic order, we can replace $\Bar{u}_{it} \Bar{u}_{jt}$ by $u_{it} u_{jt}$.

We will proceed further by stating the following Theorems and Lemmas. 
\begin{Lemma} \label{lemma:gamma_bounded}
    Suppose $|\gamma(i,j)|\leq M$ for all $1\leq i,j \leq N$ and $\frac{1}{N}\sum_{i=1}^N\sum_{j=1}^N|\gamma(i,j)|\leq M$. Then,
    \begin{equation*}
        \frac{1}{N}\sum_{i=1}^N\sum_{j=1}^N | \bar{\gamma}(i,j)|^2\leq M^2.
    \end{equation*}
\end{Lemma}
\begin{proof}
    Recall that
    \begin{equation*}
        \gamma(i,j) = \e \left[ \frac{1}{T} \sum_{t=1}^T \Bar{u}_{it} \Bar{u}_{jt} \right]
    \end{equation*}
    By the continuous application of the Cauchy-Scwartz inequality,
    \begin{align*}
        |\Bar{\gamma}(i,j) | & \leq \e \left[ \frac{1}{T} \sum_{t=1}^T |\Bar{u}_{it} \Bar{u}_{jt} | \right] \\
        & \leq \e \left[ \sqrt{\frac{1}{T} \sum_{t=1}^T |\Bar{u}_{it}|^2} \cdot \sqrt{\frac{1}{T} \sum_{t=1}^T |\Bar{u}_{jt}|^2} \right]  \\
        & \leq \sqrt{ \e \left[ \frac{1}{T} \sum_{t=1}^T |\Bar{u}_{it}^2| \right] \cdot  \e \left[ \frac{1}{T} \sum_{t=1}^T |\Bar{u}_{jt}^2| \right]} \\
        & = \sqrt{\Bar{\gamma}(i,i)} \cdot \sqrt{\Bar{\gamma}(j,j)}
    \end{align*}
Hence, $\frac{|\Bar{\gamma}(i,j) |}{\sqrt{\Bar{\gamma}(i,i) \Bar{\gamma}(j,j)}} \leq 1$ and
\begin{align*}
    \frac{1}{N}\sum_{i=1}^N \sum_{j=1}^N |\Bar{\gamma}(i,j)|^2 &= \frac{1}{N}\sum_{i=1}^N\sum_{j=1}^N \Bar{\gamma}(i,i) \Bar{\gamma}(j,j) \left( \frac{\Bar{\gamma}(i,j)}{\sqrt{\Bar{\gamma}(i,i) \Bar{\gamma}(j,j)}} \right)^2 \\
    & \leq M \frac{1}{N} \sum_{i=1}^N \sum_{j=1}^N  \sqrt{\Bar{\gamma}(i,i) \Bar{\gamma}(j,j)} \left( \frac{\Bar{\gamma}(i,j)}{\sqrt{\Bar{\gamma}(i,i) \Bar{\gamma}(j,j)}} \right)^2 \\
    & \leq M \frac{1}{N} \sum_{i=1}^N \sum_{j=1}^N \sqrt{\Bar{\gamma}(i,i) \Bar{\gamma}(j,j)} \frac{| \Bar{\gamma}(i,j) |}{| \sqrt{\Bar{\gamma}(i,i) \Bar{\gamma}(j,j)} |} \\
    & \leq M \frac{1}{N} \sum_{i=1}^N \sum_{j=1}^N |\Bar{\gamma}(i,j) | 
\end{align*}
Note that
\begin{align*}
    \Bar{\gamma}(i,j) &= \e \left[ \frac{1}{T} \sum_{t=1}^T \bar{u}_{it} \bar{u}_{jt}\right] \leq C \cdot \e \left[ \frac{1}{T} \sum_{t}  u_{it} u_{jt} \right] \\
    &= C \cdot \gamma(i,j)
\end{align*}
Thus
\begin{align*}
    \frac{1}{N}\sum_{i=1}^N \sum_{j=1}^N |\Bar{\gamma}(i,j)|^2 &\leq M \cdot C \cdot \frac{1}{N} \sum_{i=1}^N \sum_{j=1}^N |\gamma(i,j) | \leq M^2 \cdot C
\end{align*}
\end{proof}

\begin{Theorem} \label{th:norm squared of estimator}
    Under Assumptions \ref{ass:LN_strong factor structure} to \ref{ass:LN_time and cross sectional dependence},
    \begin{equation*}
        \delta_{NT}^2 \left( \frac{1}{N} \sum_{i=1}^N \Vert \hat{\Blambda}_i - \hat{H}' \Blambda_i \Vert^2 \right) = O_{\p}(1)
    \end{equation*}
    where $\delta_{NT} = \text{min}\{ \sqrt{N}, \sqrt{T} \}$.
\end{Theorem}

\begin{proof}
   \begin{equation*}
 V_{NT} [ \hat{\Blambda}_i - \hat{H}' \Blambda_i ] =  \frac{1}{N} \sum_{j=1}^N \hat{\Blambda}_j \Bar{\gamma}(i,j) + \frac{1}{N} \sum_{j=1}^N \hat{\Blambda}_j \zeta_{ij} + \frac{1}{N} \sum_{j=1}^N \hat{\Blambda}_j \eta_{ij} + \frac{1}{N} \sum_{j=1}^N \hat{\Blambda}_j \xi_{ij}
\end{equation*}

Using $(a+b+c+d)^2  \leq 4(a^2 + b^2 + c^2 + d^2)$ on the right hand side, and $\| V_{NT} x \|^2 \geq \sigma_{\min}(V_{NT})^2 \| x \|^2$ on the left, where $\sigma_{\min}(V_{NT})$ is the smallest diagonal entry of $V_{NT}$,
\begin{align} \label{eq:theorem norm_expansion}
    \frac{\sigma_{\min}(V_{NT})^2}{4} \frac{1}{N} \sum_{i=1}^N \|\hat{\Blambda}_i - \hat{H}' \Blambda_i \|^2 & \leq \frac{1}{N} \sum_{i=1}^N \left \| \frac{1}{N} \sum_{j=1}^N \hat{\Blambda}_j \Bar{\gamma}(i,j) \right \|^2 + \frac{1}{N} \sum_{i=1}^N \left \| \frac{1}{N} \sum_{j=1}^N \hat{\Blambda}_j \zeta_{ij} \right \|^2 + \nonumber \\
    & \frac{1}{N} \sum_{i=1}^N \left \| \frac{1}{N} \sum_{j=1}^N \hat{\Blambda}_j \eta_{ij} \right \|^2 + \frac{1}{N} \sum_{i=1}^N \left \| \frac{1}{N} \sum_{j=1}^N \hat{\Blambda}_j \xi_{ij} \right \|^2
\end{align}
We will analyze each of the terms on the right hand side one by one. Consider the \textbf{first term} of (\ref{eq:theorem norm_expansion}).

\begin{equation*}
    \left \| \sum_{j=1}^N \hat{\Blambda}_j \Bar{\gamma}(i,j) \right \|^2 \leq \left( \sum_{j=1}^N \| \hat{\Blambda}_j  \| \| \Bar{\gamma}(i,j)  \| \right)^2 \leq \sum_{j=1}^N \| \hat{\Blambda}_j \|^2 \sum_{j=1}^N \| \Bar{\gamma}(i,j) \|^2
\end{equation*}
where the first inequality comes from $\|a + b \|^2 \leq (\|a\| + \|b\| )^2$ and the second is the Cauchy-Schwartz inequality. Hence,
\begin{equation*}
    \sum_{i=1}^N \left \| \frac{1}{N} \sum_{j=1}^N \hat{\Blambda}_j \Bar{\gamma}(i,j) \right \|^2 \leq \frac{1}{N} \sum_{j=1}^N \| \hat{\Blambda}_j \|^2 \frac{1}{N} \sum_{i=1}^N \sum_{j=1}^N \| \Bar{\gamma}(i,j) \|^2 = O_{\p}(1)
\end{equation*}
as $\frac{1}{N} \sum_{j=1}^N \| \hat{\Blambda}_j \|^2 = r = O_{\p}(1)$ as already seen and $\frac{1}{N} \sum_{i=1}^N \sum_{j=1}^N \| \Bar{\gamma}(i,j) \|^2 = O_{\p}(1)$ as per Lemma \ref{lemma:gamma_bounded}. Hence
\begin{equation*}
    \textbf{First term} = O_{\p}(N^{-1})
\end{equation*}

Consider the \textbf{second term} of (\ref{eq:theorem norm_expansion}).
\begin{align*}
    \frac{1}{N} \sum_{i=1}^N \left \| \frac{1}{N} \sum_{j=1}^N \hat{\Blambda}_j \zeta_{ij} \right \|^2 &= \frac{1}{N^3} \sum_{i=1}^N \sum_{j=1}^N \sum_{k=1}^N \hat{\Blambda}_j' \hat{\Blambda}_k \zeta_{ij} \zeta_{ik} \\
    & \leq \frac{1}{N} \left[ \frac{1}{N^2} \sum_{j=1}^N \sum_{k=1}^N |\hat{\Blambda}_j' \hat{\Blambda}_k |^2 \right]^{\frac{1}{2}} \cdot \left[ \frac{1}{N^2} \sum_{j=1}^N \sum_{k=1}^N |\sum_{i=1}^N \zeta_{ij} \zeta_{ik} |^2 \right]^{\frac{1}{2}} 
\end{align*}

The first term in the bracket on the right hand side of the inequality is 
\begin{align*}
    \frac{1}{N^2} \sum_{j=1}^N \sum_{k=1}^N |\hat{\Blambda}_j' \hat{\Blambda}_k |^2 &= \frac{1}{N^2} \sum_{j=1}^N \sum_{k=1}^N \text{tr} \{ \hat{\Blambda}_k' \hat{\Blambda}_j \hat{\Blambda}_j' \hat{\Blambda}_k \} \\
    &= \frac{1}{N^2} \sum_{j=1}^N \sum_{k=1}^N \text{tr} \{ \hat{\Blambda}_k \hat{\Blambda}_k' \hat{\Blambda}_j \hat{\Blambda}_j' \} \\
    & \leq \frac{1}{N^2} \sum_{j=1}^N \sum_{k=1}^N \text{tr} \{\hat{\Blambda}_k \hat{\Blambda}_k'  \} \cdot \text{tr} \{\hat{\Blambda}_j \hat{\Blambda}_j'  \} \\
    &= \frac{1}{N^2} \cdot N^2 r^2 = r^2
\end{align*}

For the second term in the bracket on the right hand side of the inequality, to apply Markov inequality, consider
\begin{align*}
    \e | \sum_{i=1}^N \zeta_{ij} \zeta_{ik} |^2 &= \sum_{i=1}^N \sum_{l=1}^N \e[\zeta_{ij} \zeta_{ik}\zeta_{lj} \zeta_{lk}] \\
    & \leq \sum_{i=1}^N \sum_{l=1}^N \bigg[ \e[\zeta_{ij}^4] \e[\zeta_{ik}^4] \e[\zeta_{lj}^4] \e[\zeta_{lk}^4]  \bigg]^{\frac{1}{4}} \\
    & \leq N^2 \text{max} \enspace \e[\zeta_{ij}^4]
\end{align*}
where
\begin{align*}
    \e[\zeta_{ij}^4] &= \e \left[ \frac{1}{T} \sum_{t=1}^T \Bar{u}_{it} \Bar{u}_{jt} - \Bar{\gamma}(i,j)  \right]^4 \\
    &= \frac{1}{T^2} \cdot \e \left[ T^{-\frac{1}{2}} \sum_{t=1}^T \bigg( \Bar{u}_{it} \Bar{u}_{jt} - \e[\Bar{u}_{it} \Bar{u}_{jt}]  \bigg) \right]^4 \\
    & \leq \frac{C}{T^2} \cdot \e \left[ T^{-\frac{1}{2}} \sum_{t=1}^T \bigg( u_{it} u_{jt} - \e[u_{it} u_{jt}]  \bigg) \right]^4 \\
    & \leq \frac{CM^4}{T^2}
\end{align*}

where the last inequality comes from Assumption \ref{ass:LN_time and cross sectional dependence}.2. Thus the second term in the bracket on the right hand side of the inequality is $O_{\p}(\frac{N}{T})$. Putting all together the \textbf{second term} of (\ref{eq:theorem norm_expansion}):
\begin{equation*}
    \textbf{Second term} = \frac{1}{N} O_{\p}(1) O_{\p}(\frac{N}{T}) = O_{\p}(T^{-1})
\end{equation*}

Now consider the \textbf{third term} of (\ref{eq:theorem norm_expansion}):
\begin{align*}
    \frac{1}{N} \sum_{i=1}^N \left \| \frac{1}{N} \sum_{j=1}^N \hat{\Blambda}_j \eta_{ij} \right \|^2 &= \frac{1}{N} \sum_{i=1}^N \left \| \frac{1}{N} \sum_{j=1}^N \hat{\Blambda}_j {\Blambda}_j' \cdot \frac{1}{T} \sum_{t=1}^T \tilde{F}_t \Bar{u}_{it} \right \|^2 \\
    & \leq \frac{1}{N} \sum_{i=1}^N \left \| \frac{1}{N} \sum_{j=1}^N \hat{\Blambda}_j {\Blambda}_j'  \right \|^2 \cdot \left \| \frac{1}{T} \sum_{t=1}^T \tilde{F}_t \Bar{u}_{it} \right \|^2 \\
    & \leq \left[ \frac{1}{N} \sum_{j=1}^N \| \hat{\Blambda}_j {\Blambda}_j' \| \right]^2 \cdot \frac{1}{N} \sum_{i=1}^N \left \| \frac{1}{T} \sum_{t=1}^T \tilde{F}_t \Bar{u}_{it} \right \|^2
\end{align*}
Consider the first term on the right hand side of the inequality
\begin{align*}
    \left[ \frac{1}{N} \sum_{j=1}^N \| \hat{\Blambda}_j {\Blambda}_j' \| \right]^2 & \leq \left[ \frac{1}{N} \sum_{j=1}^N \| \hat{\Blambda}_j \| \| {\Blambda}_j' \| \right]^2 \\
    & \leq \frac{1}{N} \sum_{j=1}^N \| \hat{\Blambda}_j \|^2 \cdot \frac{1}{N} \sum_{j=1}^N \| {\Blambda}_j \|^2 = O_{\p}(1)
\end{align*}

Consider the second term on the right hand side of the inequality
\begin{align*}
    \frac{1}{N} \sum_{i=1}^N \left \| \frac{1}{T} \sum_{t=1}^T \tilde{F}_t \Bar{u}_{it} \right \|^2 &= \frac{1}{N} \sum_{i=1}^N \left \| \frac{1}{T} \sum_{t=1}^T \tilde{F}_t \left( u_{it} - \frac{1}{N} \sum_{k=1}^N u_{kt} \right) \right \|^2 \\
    &= \frac{1}{N} \sum_{i=1}^N \left \| \frac{1}{T} \sum_{t=1}^T \tilde{F}_t u_{it} - \frac{1}{NT} \sum_{t=1}^T \tilde{F}_t u_{1t} - \dots - \frac{1}{NT} \sum_{t=1}^T \tilde{F}_t u_{Nt}  \right \|^2 \\
    &\leq \frac{1}{N} \sum_{i=1}^N \left \| \frac{1}{T} \sum_{t=1}^T \tilde{F}_t u_{it} \right \|^2 + \frac{1}{N} \sum_{i=1}^N \frac{1}{N} \sum_{k=1}^N \left \| \frac{1}{T} \sum_{t=1}^T \tilde{F}_t u_{kt} \right \|^2 \\
    &= \frac{1}{N} \sum_{i=1}^N \left \| \frac{1}{T} \sum_{t=1}^T \tilde{F}_t u_{it} \right \|^2 +  \frac{1}{N} \sum_{k=1}^N \left \| \frac{1}{T} \sum_{t=1}^T \tilde{F}_t u_{kt} \right \|^2 \\
    &= 2 \cdot \frac{1}{T} \cdot \frac{1}{N} \sum_{i=1}^N \left \| \frac{1}{\sqrt{T}} \sum_{t=1}^T \tilde{F}_t u_{it} \right \|^2 \\
    &= 2 \cdot O_{\p}(\frac{1}{T}) O_{\p}(1) = O_{\p}(\frac{1}{T})
\end{align*}
where the first inequality comes from the repeated use of the Cauchy-Schwartz inequality, and the $O_{\p}(1)$ term in the last line comes from Assumption \ref{ass:weak dep between agg shocks and idio errors}. Thus
\begin{equation*}
    \textbf{Third term} = O_{\p}(T^{-1})
\end{equation*}

Now consider the \textbf{fourth term} of (\ref{eq:theorem norm_expansion}):
\begin{align*}
    \frac{1}{N} \sum_{i=1}^N \left \| \frac{1}{N} \sum_{j=1}^N \hat{\Blambda}_j \xi_{ij} \right \|^2 &= \frac{4}{N} \sum_{i=1}^N \left \| \frac{1}{N} \sum_{j=1}^N \hat{\Blambda}_j \cdot \frac{1}{T} \sum_{t=1}^T {\Blambda}_i' \tilde{F}_t \Bar{u}_{jt}  \right \|^2 \\
    &= \frac{4}{N} \sum_{i=1}^N \left \| \frac{1}{N} \sum_{j=1}^N \hat{\Blambda}_j \cdot \frac{1}{T} \sum_{t=1}^T  \tilde{F}_t \Bar{u}_{jt} {\Blambda}_i  \right \|^2 \\
    & \leq \frac{4}{N} \sum_{i=1}^N \left \| \frac{1}{N} \sum_{j=1}^N \hat{\Blambda}_j \cdot \frac{1}{T} \sum_{t=1}^T  \tilde{F}_t \Bar{u}_{jt} \right \|^2 \Bigg \| {\Blambda}_i  \Bigg \|^2 \\
    & \leq \frac{4}{N} \sum_{i=1}^N \| {\Blambda}_i \|^2 \cdot \frac{1}{N} \sum_{j=1}^N \| \hat{\Blambda}_j' \cdot \frac{1}{T} \sum_{t=1}^T \tilde{F}_t \Bar{u}_{jt}  \|^2 \\
    &\leq \frac{4}{N} \sum_{i=1}^N \| {\Blambda}_i \|^2 \cdot \frac{1}{N} \sum_{j=1}^N \| \hat{\Blambda}_j' \|^2 \cdot \frac{1}{N} \sum_{j=1}^N \| \frac{1}{T} \sum_{t=1}^T \tilde{F}_t \Bar{u}_{jt}  \|^2 \\
    &= O_{\p}(1) \cdot \frac{1}{N} \sum_{j=1}^N \left \| \frac{1}{T} \sum_{t=1}^T \tilde{F}_t \Bar{u}_{jt} \right \|^2 = O_{\p}(T^{-1})
\end{align*}
where the inequality in the second last line uses $\| \hat{\Blambda}_j \|^2 = 1$ and $O_{\p}(T^{-1})$ in the last line uses Assumption \ref{ass:weak dep between agg shocks and idio errors} applied to the term as in the case of the third term. Thus 
\begin{equation*}
    \sigma_{\min}(V_{NT})^2 \frac{1}{N} \sum_{i=1}^N \|\hat{\Blambda}_i - \hat{H}' \Blambda_i \|^2 = O_{\p}(N^{-1}) + O_{\p}(T^{-1}) + O_{\p}(T^{-1}) + O_{\p}(T^{-1})
\end{equation*}
Theorem \ref{th:convergence of lambda_hat times lambda} shows that $V_{NT}$ converges in probability to an invertible diagonal matrix with distinct non-zero entries, so $\sigma_{\min}(V_{NT})^{-1} = O_{\p}(1)$. Hence the result follows.
\end{proof}

\begin{Lemma}\label{lemma:intermediate lemmas}
    Under Assumptions \ref{ass:LN_strong factor structure} to \ref{ass:LN_moments and CLT}, for all $i$:
    \begin{enumerate}[label=(\alph*)]
        \item $\frac{1}{N} \sum_{j=1}^N \hat{\Blambda}_j \Bar{\gamma}(i,j) = O_{\p} \left(\frac{1}{\sqrt{N} \delta_{NT}} \right)$
        \item $\frac{1}{N} \sum_{j=1}^N \hat{\Blambda}_j \zeta_{ij} = O_{\p} \left(\frac{1}{\sqrt{T} \delta_{NT}} \right)$
        \item $\frac{1}{N} \sum_{j=1}^N \hat{\Blambda}_j \eta_{ij} = O_{\p} \left(\frac{1}{\sqrt{T}} \right)$
        \item $\frac{1}{N} \sum_{j=1}^N \hat{\Blambda}_j \xi_{ij} = O_{\p} \left(\frac{1}{\sqrt{T} \delta_{NT}} \right)$
    \end{enumerate}
\end{Lemma}
\begin{proof}
    For \textbf{(a)},
    \begin{align*}
        \frac{1}{N} \sum_{j=1}^N \hat{\Blambda}_j \Bar{\gamma}(i,j) &= \frac{1}{N} \sum_{j=1}^N [ \hat{\Blambda}_j - \hat{H}' \Blambda_j + \hat{H}' \Blambda_j ]\Bar{\gamma}(i,j)  \\
        &= \frac{1}{N} \sum_{j=1}^N [ \hat{\Blambda}_j - \hat{H}' \Blambda_j ] \Bar{\gamma}(i,j) + \hat{H}' \frac{1}{N} \sum_{j=1}^N \Blambda_j \Bar{\gamma}(i,j)
    \end{align*}
Consider the first term on the right hand side on the last equality
\begin{align*}
    \left| \frac{1}{N} \sum_{j=1}^N [ \hat{\Blambda}_j - \hat{H}' \Blambda_j ] \Bar{\gamma}(i,j)  \right | &\leq \left[  \frac{1}{N} \sum_{j=1}^N \| \hat{\Blambda}_j - \hat{H}' \Blambda_j  \|^2 \right]^{\frac{1}{2}} \cdot \left[ \frac{1}{N} \sum_{j=1}^N | \Bar{\gamma}(i,j)  |^2\right]^{\frac{1}{2}} \\
    &= O_{\p}(\frac{1}{\delta_{NT}}) \frac{1}{\sqrt{N}} O_{\p}(1)
\end{align*}
by Theorem \ref{th:norm squared of estimator} and Lemma \ref{lemma:gamma_bounded}.

The second term is a vector. Consider one of its element.
\begin{align*}
    \left | \frac{1}{N} \sum_{j=1}^N \lambda_j^r \Bar{\gamma}(i,j)  \right | & \leq N^{-1} \left| \text{max}(\lambda_j^r) \sum_{j=1}^N \Bar{\gamma}(i,j) \right| \\
    & \leq N^{-1} \text{max}\| \lambda_j^r \| \cdot \sum_{j=1}^N | \Bar{\gamma}(i,j) | \\
    &= O_{\p}(N^{-1})
\end{align*}

Thus, the entire second term is $O_{\p}(N^{-1})$, and \textbf{(a)} $= O_{\p}\left( \frac{1}{\sqrt{N} \delta_{NT}} \right)$

For \textbf{(b)},
\begin{align*}
    \frac{1}{N} \sum_{j=1}^N \hat{\Blambda}_j \zeta_{ij} = \frac{1}{N} \sum_{j=1}^N [ \hat{\Blambda}_j - \hat{H}' \Blambda_j ] \zeta_{ij} + \hat{H}' \frac{1}{N} \sum_{j=1}^N \Blambda_j \zeta_{ij}
\end{align*}

For the first term on the right hand side of the above equality
\begin{align*}
    \left| \frac{1}{N} \sum_{j=1}^N [ \hat{\Blambda}_j - \hat{H}' \Blambda_j ] \zeta_{ij}  \right | &\leq \left[  \frac{1}{N} \sum_{j=1}^N \| \hat{\Blambda}_j - \hat{H}' \Blambda_j  \|^2 \right]^{\frac{1}{2}} \cdot \left[ \frac{1}{N} \sum_{j=1}^N | \zeta_{ij}  |^2\right]^{\frac{1}{2}}
\end{align*}
where
\begin{align*}
    \frac{1}{N} \sum_{j=1}^N  \zeta_{ij}^2  &= \frac{1}{N} \sum_{j=1}^N \left[ \frac{1}{T} \sum_{t=1}^T \Bar{u}_{it} \Bar{u}_{jt} - \Bar{\gamma}(i,j) \right]^2 \\
    &= \frac{1}{T} \cdot \frac{1}{N} \sum_{j=1}^N \left[ T^{-\frac{1}{2}} \sum_{t=1}^T \bigg( \Bar{u}_{it} \Bar{u}_{jt} - \e[\Bar{u}_{it} \Bar{u}_{jt}]  \bigg) \right]^2 \\
    & \leq \frac{C^2}{T} \cdot \frac{1}{N} \sum_{j=1}^N \left[ T^{-\frac{1}{2}} \sum_{t=1}^T \bigg( u_{it} u_{jt} - \e[u_{it} u_{jt}]  \bigg) \right]^2 \\
    & \leq \frac{C^2 \cdot M^2}{T}
\end{align*}
Hence the first term is $O_{\p} \left(\frac{1}{\sqrt{T} \delta_{NT}} \right)$. For the second term
\begin{align*}
    \frac{1}{N} \sum_{j=1}^N \Blambda_j \zeta_{ij} &= \frac{1}{N} \sum_{j=1}^N  \Blambda_j \left( \frac{1}{T} \sum_{t=1}^T \Bar{u}_{it} \Bar{u}_{jt} - \Bar{\gamma}(i,j) \right) \\
    & \leq \frac{ C}{N} \sum_{j=1}^N \Blambda_j \left( \frac{1}{T} \sum_{t=1}^T \bigg( u_{it} u_{jt} - \e[u_{it} u_{jt}] \bigg) \right) \\
    &= \frac{C}{NT} \sum_{j=1}^N \sum_{t=1}^T \Blambda_j  \bigg( u_{it} u_{jt} - \e[u_{it} u_{jt}] \bigg) \\
    &= O_{\p}\left( \frac{1}{\sqrt{NT}} \right)
\end{align*}
by Assumption \ref{ass:LN_moments and CLT}.1 and Markov's inequality. Thus \textbf{(b)} $= O_{\p}\left( \frac{1}{\sqrt{T} \delta_{NT}} \right)$

For \textbf{(c)},
\begin{align*}
    \frac{1}{N} \sum_{j=1}^N \hat{\Blambda}_j \eta_{ij} = \frac{1}{N} \sum_{j=1}^N [ \hat{\Blambda}_j - \hat{H}' \Blambda_j ] \eta_{ij} + \hat{H}' \frac{1}{N} \sum_{j=1}^N \Blambda_j \eta_{ij}
\end{align*}
For the first term on the right hand side of the above equality
\begin{align*}
    \left| \frac{1}{N} \sum_{j=1}^N [ \hat{\Blambda}_j - \hat{H}' \Blambda_j ] \eta_{ij}  \right | &\leq \left[  \frac{1}{N} \sum_{j=1}^N \| \hat{\Blambda}_j - \hat{H}' \Blambda_j  \|^2 \right]^{\frac{1}{2}} \cdot \left[ \frac{1}{N} \sum_{j=1}^N | \eta_{ij}  |^2\right]^{\frac{1}{2}}
\end{align*}
where
\begin{align*}
    \frac{1}{N} \sum_{j=1}^N  \eta_{ij}^2  &= \frac{1}{N} \frac{1}{T^2} \sum_{j=1}^N \left( \Blambda_j' \sum_{t=1}^T \tilde{F}_t \Bar{u}_{it}  \right)^2 \\
    & \leq \frac{1}{T^2} \cdot \frac{1}{N} \sum_{j=1}^N \| \Blambda_j \|^2 \cdot \| \sum_{t=1}^T \tilde{F}_t \Bar{u}_{it} \|^2 \\
    &= \frac{1}{T^2} \cdot \| \sum_{t=1}^T \tilde{F}_t \Bar{u}_{it} \|^2 \cdot O_{\p}(1) \\
    &\leq \frac{4}{T} \cdot \| \frac{1}{\sqrt{T}} \sum_{t=1}^T \tilde{F}_t u_{it} \|^2 \cdot O_{\p}(1) \\
    &= O_{\p} \left(\frac{1}{T} \right)
\end{align*}
where the replacement of $\Bar{u}_{it}$ by $u_{it}$ follows the same steps as in the proof of Lemma \ref{th:norm squared of estimator}. The last equality is from Assumption \ref{ass:weak dep between agg shocks and idio errors}. Thus the first term is $O_{\p} \left( \frac{1}{\delta_{NT} \sqrt{T}} \right)$. For the second term,
\begin{align*}
    \frac{1}{N} \sum_{j=1}^N \Blambda_j \eta_{ij} &= \frac{1}{NT} \sum_{j=1}^N \Blambda_j \Blambda_j' \sum_{t=1}^T \tilde{F}_t \Bar{u}_{it} \\
    &= \left( \frac{1}{N}  \sum_{j=1}^N \Blambda_j \Blambda_j' \right) \cdot \frac{1}{T} \sum_{t=1}^T \tilde{F}_t \Bar{u}_{it} \\
    &\leq  \left( \frac{2}{N}  \sum_{j=1}^N \Blambda_j \Blambda_j' \right) \cdot \frac{1}{T} \sum_{t=1}^T \tilde{F}_t u_{it} \\
    &= O_{\p}(1) O_{\p} \left(\frac{1}{\sqrt{T}} \right)
\end{align*}

Thus \textbf{(c)} $= O_{\p} \left( \frac{1}{\delta_{NT} \sqrt{T}} \right) +  O_{\p} \left(\frac{1}{\sqrt{T}} \right) = O_{\p} \left(\frac{1}{\sqrt{T}} \right)$

For \textbf{(d)},
\begin{align*}
    \frac{1}{N} \sum_{j=1}^N \hat{\Blambda}_j \xi_{ij} = \frac{1}{N} \sum_{j=1}^N [ \hat{\Blambda}_j - \hat{H}' \Blambda_j ] \xi_{ij} + \hat{H}' \frac{1}{N} \sum_{j=1}^N \Blambda_j \xi_{ij}
\end{align*}
For the first term on the right hand side of the above equality
\begin{align*}
    \left| \frac{1}{N} \sum_{j=1}^N [ \hat{\Blambda}_j - \hat{H}' \Blambda_j ] \xi_{ij}  \right | &\leq \left[  \frac{1}{N} \sum_{j=1}^N \| \hat{\Blambda}_j - \hat{H}' \Blambda_j  \|^2 \right]^{\frac{1}{2}} \cdot \left[ \frac{1}{N} \sum_{j=1}^N | \xi_{ij}  |^2\right]^{\frac{1}{2}}
\end{align*}
where
\begin{align*}
    \frac{1}{N} \sum_{j=1}^N  \xi_{ij} ^2  &= \frac{1}{N} \frac{1}{T^2} \sum_{j=1}^N \left( \Blambda_i' \sum_{t=1}^T \tilde{F}_t \Bar{u}_{jt}  \right)^2  \leq \frac{1}{T^2} \cdot \frac{1}{N} \sum_{j=1}^N \| \Blambda_i \|^2 \cdot \| \sum_{t=1}^T \tilde{F}_t \Bar{u}_{jt} \|^2 \\
    &= \| \Blambda_i \|^2 \cdot \frac{1}{T} \frac{1}{N} \sum_{j=1}^N  \left \| \frac{1}{\sqrt{T}} \sum_{t=1}^T \tilde{F}_t \Bar{u}_{it} \right \|^2  \leq \| \Blambda_i \|^2 \cdot \frac{4}{TN} \sum_{j=1}^N  \left \| \frac{1}{\sqrt{T}} \sum_{t=1}^T \tilde{F}_t u_{it} \right \|^2 \\
    &= O_{\p} \left(\frac{1}{T} \right)
\end{align*}
Thus the first term is $O_{\p} \left( \frac{1}{\delta_{NT} \sqrt{T}} \right)$. For the second term,
\begin{align*}
    \frac{1}{N} \| \sum_{j=1}^N \Blambda_j \xi_{ij} \| &= \frac{1}{NT} \left \| \sum_{j=1}^N \sum_{t=1}^T \Blambda_j \tilde{F}_t \Bar{u}_{jt} \Blambda_i \right \| \leq \frac{1}{NT} \left \| \sum_{j=1}^N  \sum_{t=1}^T \Blambda_j \tilde{F}_t \Bar{u}_{jt} \right \| \| \Blambda_i \| \\
    & \leq \frac{2}{NT} \left \| \sum_{j=1}^N  \sum_{t=1}^T \Blambda_j \tilde{F}_t u_{jt} \right \| \| \Blambda_i \| = O_{\p} \left(\frac{1}{\sqrt{NT}} \right)
\end{align*}
the last equality from Assumption \ref{ass:LN_moments and CLT}. Thus \textbf{(d)} $= O_{\p} \left( \frac{1}{\delta_{NT} \sqrt{T}} \right) +  O_{\p} \left(\frac{1}{\sqrt{NT}} \right) = O_{\p} \left( \frac{1}{\delta_{NT} \sqrt{T}} \right)$

\end{proof}

Thus, based on Lemma \ref{lemma:intermediate lemmas}, we have 
\begin{equation*} 
 V_{NT} \big(  \hat{\Blambda}_i - \hat{H}' \Blambda_i \big) =  \frac{1}{N} \sum_{j=1}^N \hat{\Blambda}_j \eta_{ij} + O_{\p} \left(\frac{1}{\sqrt{N} \delta_{NT}} \right) + O_{\p} \left(\frac{1}{\sqrt{T} \delta_{NT}} \right) + O_{\p} \left(\frac{1}{\sqrt{T} \delta_{NT}} \right)
\end{equation*}
Substituting the expression for $\eta_{ij}$, we have
\begin{align}
    V_{NT} \big(  \hat{\Blambda}_i - \hat{H}' \Blambda_i \big) &= \frac{1}{N} \sum_{j=1}^N \hat{\Blambda}_j \tilde{\Blambda}_j' \tilde{F}' \Bar{u}_i/T + O_{\p} \left(\frac{1}{\delta_{NT}^2} \right) \nonumber \\
    &= \frac{\tilde{\Lambda}' \hat{\Lambda}}{N} \cdot \frac{1}{T} \sum_{t=1}^T \tilde{F}_t \Bar{u}_{it} + O_{\p} \left(\frac{1}{\delta_{NT}^2} \right) \label{eq:asymptotic_linear_lambda_incomplete}
\end{align}

To fully characterize the asymptotic linear form of our estimator, we need to show that $V_{NT}$ is invertible and characterize the asymptotic behavior of $\frac{\tilde{\Lambda}' \hat{\Lambda}}{N}$. I do this in Theorem \ref{th:convergence of lambda_hat times lambda}. But the Theorem requires applying some results from \citet{Stock2002ForecastingPredictors} to our setting in the following Lemmas. First, we start by defining some terms

Let $a$ be a $N \times 1$ vector and
\begin{align*}
    A &= \{a | a'a/N = 1 \} \\
    R(a) &= N^{-1} a' \frac{\Bar{Y}' \Bar{Y}}{NT} a \\
    R^*(a) &= N^{-2} a' \tilde{\Lambda} \hat{\Sigma}_{\Tilde{F}} \tilde{\Lambda}' a
\end{align*}
where $\hat{\Sigma}_{\tilde{F}} = \frac{\tilde{F}' \tilde{F}}{T}$.

Our first result is about the uniform convergence of $R(a)$ and $R^*(a)$ in $A$.
\begin{Lemma} \label{lemma:uniform convergence of R and R^*}
    Under Assumptions \ref{ass:LN_strong factor structure} to \ref{ass:LN_moments and CLT},
    \begin{equation*}
        \sup_{a \in A} | R(a) - R^*(a) | \xrightarrow{p} 0
    \end{equation*}
\end{Lemma}
\begin{proof}
From the definition,
\begin{align*}
    R(a) - R^*(a) &= (N^2 T)^{-1} a' \Bar{u}' \Bar{u} a + 2(N^2 T)^{-1} a' \tilde{\Lambda} \tilde{F}' \Bar{u} a \\
    \sup_{a \in A} | R(a) - R^*(a) | &\leq \sup_{a \in A} (N^2 T)^{-1} a' \Bar{u}' \Bar{u} a + \sup_{a \in A} 2(N^2 T)^{-1} a' \tilde{\Lambda} \tilde{F}' \Bar{u} a 
\end{align*}
We will analyze each of the terms in the right hand side. The first term is
\begin{align*}
    I &= \frac{1}{N^2} \sum_{i=1}^N a_i \sum_{j=1}^N a_j \cdot \Bar{u}_i' \Bar{u}_j/T = \frac{1}{N^2} \sum_{i=1}^N \sum_{j=1}^N a_i  a_j \cdot \Bar{u}_i' \Bar{u}_j/T \\
    & \leq  \left[ \frac{1}{N^2} \sum_{i=1}^N \sum_{j=1}^N (a_i a_j)^2 \right]^{\frac{1}{2}} \left[ \frac{1}{N^2} \sum_{i=1}^N \sum_{j=1}^N (\Bar{u}_i' \Bar{u}_j/T)^2 \right]^{\frac{1}{2}} \\
    &= \left[ \frac{1}{N^2} \sum_{i=1}^N \sum_{j=1}^N (\Bar{u}_i' \Bar{u}_j/T)^2 \right]^{\frac{1}{2}} = O_{\p}(\frac{1}{\sqrt{N}})
\end{align*}
where the stochastic order comes from Lemma \ref{lemma:gamma_bounded}. As the expression above is independent of $a$, we have
\begin{equation*}
    \sup_{a \in A} (N^2 T)^{-1} a' \Bar{u}' \Bar{u} a = O_{\p}(\frac{1}{\sqrt{N}})
\end{equation*}

Consider the second term:
\begin{align*}
    II &= \frac{1}{N^2} \sum_{i=1}^N a_i \sum_{j=1}^N a_j \cdot \Blambda_j' \frac{2}{T} \sum_{t=1}^T \tilde{F}_t \Bar{u}_{it} \\
    &= \frac{1}{N^2} \sum_{i=1}^N a_i \sum_{j=1}^N a_j \cdot \eta_{ij} = \frac{1}{N^2} \sum_{i=1}^N \sum_{j=1}^N a_i  a_j \cdot \eta_{ij} \\
    & \leq  \left[ \frac{1}{N^2} \sum_{i=1}^N \sum_{j=1}^N (a_i a_j)^2 \right]^{\frac{1}{2}} \left[ \frac{1}{N^2} \sum_{i=1}^N \sum_{j=1}^N \eta_{ij}^2 \right]^{\frac{1}{2}} \\
    &= \left[ \frac{1}{N} \sum_{i=1}^N \frac{1}{N} \sum_{j=1}^N \eta_{ij}^2 \right]^{\frac{1}{2}} \leq \left[ \frac{1}{N} \sum_{i=1}^N \sup_{i} \frac{1}{N} \sum_{j=1}^N \eta_{ij}^2 \right]^{\frac{1}{2}} \\
    &= \left[ \sup_{i} \frac{1}{N} \sum_{j=1}^N \eta_{ij}^2 \right]^2 = O_{\p}(\frac{1}{\sqrt{T}})
\end{align*}
where the stochastic order comes from an intermediate step in Lemma \ref{lemma:intermediate lemmas}. Like the first term, as the expression above is independent of $a$, we have
\begin{equation*}
    \sup_{a \in A} 2(N^2 T)^{-1} a' \tilde{\Lambda} \tilde{F}' \Bar{u} a = O_{\p}(\frac{1}{\sqrt{T}})
\end{equation*}

Hence we have showed that
\begin{align*}
    \sup_{a \in A} | R(a) - R^*(a) | &\leq O_{\p}(\frac{1}{\sqrt{N}}) + O_{\p}(\frac{1}{\sqrt{T}}) \\
    \sup_{a \in A} | R(a) - R^*(a) | &= o_{\p}(1) \xrightarrow{p} 0
\end{align*}
\end{proof}

\begin{Lemma} \label{lemma:behavior of R_star}
Under Assumptions \ref{ass:LN_strong factor structure} to \ref{ass:LN_moments and CLT},
\begin{equation*}
    \sup_{a \in A} R^*(a) \xrightarrow{p} \delta_1
\end{equation*}
where $\delta_1$ is the largest eigenvalue of $\Sigma_{\tilde{\Lambda}}^{\frac{1}{2}} \Sigma_{\Tilde{F}} \Sigma_{\tilde{\Lambda}}^{\frac{1}{2}}$.
\end{Lemma}
\begin{proof}
    Decompose $\frac{\tilde{\Lambda}' \tilde{\Lambda}}{N} = \left( \frac{\tilde{\Lambda}' \tilde{\Lambda}}{N}  \right)^{\frac{1}{2}} \left( \frac{\tilde{\Lambda}' \tilde{\Lambda}}{N}  \right)^{\frac{1}{2}} $ where $\left( \frac{\tilde{\Lambda}' \tilde{\Lambda}}{N}  \right)^{\frac{1}{2}}$ is unique, symmetric, and positive definite by Theorem 7.2.6 of \citet{Horn2012MatrixAnalysis}. For some $r$ dimensional vector, $a_1$ and $N$ dimensional vector $a_2$, we can write
    \begin{equation*}
        a = \tilde{\Lambda} \left( \frac{\tilde{\Lambda}' \tilde{\Lambda}}{N} \right)^{- \frac{1}{2}} a_1 + a_2
    \end{equation*}
    where $a_2' \tilde{\Lambda} = 0$. $a_1$ and $a_2$ exists uniquely because we have $N+r$ unknowns and $N + r$ restrictions. 
    \begin{equation*}
        \frac{a'a}{N} = a_1' a_1 + \frac{a_2' a_2}{N} =1
    \end{equation*}
    This representation implies $a_1' a_1 \leq 1$
    \begin{align*}
        R^*(a) &= N^{-2} a' \Lambda \hat{\Sigma}_{\Tilde{F}} \Lambda' a \\
                &= a_1' \left( \frac{\tilde{\Lambda}' \tilde{\Lambda}}{N} \right)^{\frac{1}{2}} \hat{\Sigma}_{\Tilde{F}} \left( \frac{\tilde{\Lambda}' \tilde{\Lambda}}{N} \right)^{\frac{1}{2}} a_1 \\
                & \leq d_1 a_1' a_1 \\
        \therefore \sup_{a \in A} R^*(a) &= \sup_{a_1, a_1'a_1 \leq 1}  d_1 a_1' a_1 = d_1
    \end{align*}

    where $d_1$ is the largest eigenvalue of $\left( \frac{\tilde{\Lambda}' \tilde{\Lambda}}{N} \right)^{\frac{1}{2}} \hat{\Sigma}_{\Tilde{F}} \left( \frac{\tilde{\Lambda}' \tilde{\Lambda}}{N} \right)^{\frac{1}{2}}$ which $\xrightarrow{p} \Sigma_{\Lambda}^{\frac{1}{2}} \Sigma_{\Tilde{F}} \Sigma_{\Lambda}^{\frac{1}{2}} $. Therefore, by Weyl's Theorem (Theorem 2 of Sec 6.7 in \citet{Franklin2012MatrixTheory}),
    \begin{equation*}
        d_1 \xrightarrow{p} \delta_1
    \end{equation*}
\end{proof}

\begin{Lemma} \label{lemma:SW_intermediate results}
    Under Assumptions \ref{ass:LN_strong factor structure} to \ref{ass:LN_moments and CLT},
    \begin{enumerate}
        \item $\sup_{a \in A} R(a) \xrightarrow{p} \delta_1$, where $\delta_1$ is the largest eigenvalue of $\Sigma_{\tilde{\Lambda}}^{\frac{1}{2}} \Sigma_{\Tilde{F}} \Sigma_{\tilde{\Lambda}}^{\frac{1}{2}}$.
        \item Let $ \hat{\Blambda}_1$ be the eigenvector of $ \frac{\Bar{Y}' \Bar{Y}}{NT}$, corresponding to its largest eigenvalue and hence the first column of our estimator, $\hat{\Lambda}$, then
        \begin{align*}
            & \hat{\Blambda}_1 = \text{arg} \sup_{a \in A} R(a), \quad \text{and}\\
            & R^*(\hat{\Blambda}_1)  \xrightarrow{p} \delta_1
        \end{align*}
    \end{enumerate}
\end{Lemma}
\begin{proof}
    For the \textbf{first} result, recall by Lemma \ref{lemma:uniform convergence of R and R^*}
    \begin{equation*}
        \sup_{a \in A} | R(a) - R^*(a) | \xrightarrow{p} 0
    \end{equation*}
    By the definition of supremum,
    \begin{equation*}
        | \sup_{a \in A} R(a) - \sup_{a \in A} R^*(a)  | \leq \sup_{a \in A} | R(a) - R^*(a) | \xrightarrow{p} 0
    \end{equation*}
    Thus $\sup_{a \in A} R(a) - \sup_{a \in A} R^*(a) \xrightarrow{p} 0$. By Lemma \ref{lemma:behavior of R_star}, $\sup_{a \in A} R^*(a) \xrightarrow{p} \delta_1$ and the result follows.

    Consider the \textbf{second} result. $\hat{\Blambda}_1$ is the eigenvector of $\frac{\Bar{Y}' \Bar{Y}}{NT}$, corresponding to its largest eigenvalue and hence the first column of our estimator, $\hat{\Lambda}$. Call the eigenvalues of $\frac{\Bar{Y}' \Bar{Y}}{NT}$ as $c_1 \geq c_2 \geq \dots \geq c_N$. Therefore,
    \begin{align*}
    R(\hat{\Blambda}_1) &= N^{-1} \hat{\Blambda}_1' \frac{\Bar{Y}' \Bar{Y}}{NT} \hat{\Blambda}_1 \\
                        &= N^{-1} \hat{\Blambda}_1' c_1 \hat{\Blambda}_1 \\
                        &= c_1
\end{align*}
Now note that
    \begin{align*}
        R(a) &= N^{-1} a' \frac{\Bar{Y}' \Bar{Y}}{NT} a \\
            & \leq c_1 \frac{a'a}{N} = c_1 \\
        \therefore \sup_{a \in A} R(a) &= c_1
    \end{align*}
Thus $\hat{\Blambda}_1 = \text{arg} \sup_{a \in A} R(a)$. By the first result of this Lemma, $R(\hat{\Blambda}_1) \xrightarrow{p} \delta_1$. In the proof of Lemma \ref{lemma:uniform convergence of R and R^*}, we saw that for all $ a \in A$,
\begin{align*}
    R(a) - R^*(a) \xrightarrow{p} 0
\end{align*}
Thus $R(\hat{\Blambda}_1) \xrightarrow{p} \delta_1 \implies R^*(\hat{\Blambda}_1) \xrightarrow{p} \delta_1$
\end{proof}

\begin{Theorem} \label{th:convergence of lambda_hat times lambda}
    Suppose Assumptions \ref{ass:LN_strong factor structure} to \ref{ass:LN_moments and CLT} hold. Let $\hat{\Lambda}$ be the $r-1$ normalized and ordered eigenvectors of $\frac{\Bar{Y}' \Bar{Y}}{NT}$, corresponding to its $r$ largest eigenvalue. Then 
    \begin{enumerate}
        \item The $r-1 \times r-1$ matrix $\frac{\hat{\Lambda}' \Lambda}{N}$ is invertible and
    \begin{equation*}
        \frac{\hat{\Lambda}' \Lambda}{N} \xrightarrow{p} \Upsilon' \Sigma_{\Lambda}^{\frac{1}{2}}
    \end{equation*}
    where $\Upsilon$ holds all the orthonormal eigenvectors of $\Sigma_{{\Lambda}}^{1/2} \Sigma_{\tilde{F}} \Sigma_{{\Lambda}}^{1/2}$

        \item $V_{NT}$ is invertible and $V_{NT} \xlongrightarrow{p} V$, where $V$ is a diagonal matrix containing the eigenvalues of $\Sigma_{\tilde{\Lambda}}\Sigma_{\Tilde{F}}$
    \end{enumerate}

\end{Theorem}
\begin{proof}
    Consider the first column of $\hat{\Lambda}$, which we call $\hat{\Blambda}_1$. For some $r$ dimensional vector $\hat{a}_1$ and $N$ dimensional vector $\hat{a}_2$, we can write
    \begin{equation*}
        \hat{\Blambda}_1 = \tilde{\Lambda} \left( \frac{\tilde{\Lambda}' \tilde{\Lambda}}{N} \right)^{- \frac{1}{2}} \hat{a}_1 + \hat{a}_2
    \end{equation*}
    where $\hat{a}_2' \tilde{\Lambda} = 0$. $\hat{a}_1$ and $\hat{a}_2$ exists uniquely because we have $N + r$ unknowns and $N + r$ restrictions. 
    \begin{equation*}
        \frac{\hat{\Blambda}_1' \hat{\Blambda}_1 }{N} = \hat{a}_1' \hat{a}_1 + \frac{\hat{a}_2' \hat{a}_2}{N} =1
    \end{equation*}
    This representation implies $\hat{a}_1' \hat{a}_1 \leq 1$. Call $C_{NT} = \left( \frac{\tilde{\Lambda}' \tilde{\Lambda}}{N} \right)^{\frac{1}{2}} \hat{\Sigma}_{\Tilde{F}} \left( \frac{\tilde{\Lambda}' \tilde{\Lambda}}{N} \right)^{\frac{1}{2}} $
    \begin{align*}
        R^*(\hat{\Blambda}_1 ) &= N^{-2} \hat{\Blambda}_1 ' \tilde{\Lambda} \hat{\Sigma}_{\Tilde{F}} \tilde{\Lambda}' \hat{\Blambda}_1  \\
                &= \hat{a}_1' \left( \frac{\tilde{\Lambda}' \tilde{\Lambda}}{N} \right)^{\frac{1}{2}} \hat{\Sigma}_{\Tilde{F}} \left( \frac{\tilde{\Lambda}' \tilde{\Lambda}}{N} \right)^{\frac{1}{2}} \hat{a}_1 \\
                &=  \hat{a}_1' C_{NT} \hat{a}_1
    \end{align*}
    where, in one of the intermediate steps we use the decomposition of  $\frac{\tilde{\Lambda}' \tilde{\Lambda}}{N} = \left( \frac{\tilde{\Lambda}' \tilde{\Lambda}}{N}  \right)^{\frac{1}{2}} \left( \frac{\tilde{\Lambda}' \tilde{\Lambda}}{N}  \right)^{\frac{1}{2}} $ where $\left( \frac{\tilde{\Lambda}' \tilde{\Lambda}}{N}  \right)^{\frac{1}{2}}$ is unique, symmetric, and positive definite by Theorem 7.2.6 of \citet{Horn2012MatrixAnalysis}.

   Now observe that
    \begin{align*}
        R^*(\hat{\Blambda}_1 ) - \delta_1 &= \hat{a}_1' C_{NT} \hat{a}_1 - \delta_1 \\
        &= \hat{a}_1'[C_{NT} - \Sigma_{{\Lambda}}^{1/2} \Sigma_{\tilde{F}} \Sigma_{{\Lambda}}^{1/2}]\hat{a}_1 + \hat{a}_1'[\Sigma_{{\Lambda}}^{1/2} \Sigma_{\tilde{F}} \Sigma_{{\Lambda}}^{1/2}]\hat{a}_1 - \delta_1
    \end{align*}
     Note that $C_{NT} \xrightarrow{p} \Sigma_{{\Lambda}}^{1/2} \Sigma_{\tilde{F}} \Sigma_{{\Lambda}}^{1/2}$. By Lemma \ref{lemma:SW_intermediate results}, $R^*(\hat{\Blambda}_1 ) \xrightarrow{p} \delta_1$. Hence it has to be the case that
     \begin{equation*}
         \hat{a}_1'[\Sigma_{{\Lambda}}^{1/2} \Sigma_{\tilde{F}} \Sigma_{{\Lambda}}^{1/2}]\hat{a}_1 - \delta_1 \xrightarrow{p} 0
     \end{equation*}

     We can decompose the symmetric $\Sigma_{{\Lambda}}^{1/2} \Sigma_{\tilde{F}} \Sigma_{{\Lambda}}^{1/2} = \Upsilon \Delta \Upsilon'$, where $\Upsilon$ holds all the normalized eigenvectors such that $\Upsilon'\Upsilon = I$ and the diagonal $\Delta$ contains all the eigenvalues in order. The above display becomes
     \begin{align*}
          \hat{a}_1'\Upsilon \cdot \Delta \cdot  \Upsilon'\hat{a}_1 - \delta_1 & \xrightarrow{p} 0 \\
          [(\hat{a}_1' \Upsilon_1) - 1] \delta_1 + \sum_{j=2}^r (\hat{a}_1' \Upsilon_j)^2 \delta_j \xrightarrow{p} 0
     \end{align*}
     where $\Upsilon_j$ is a column of $\Upsilon$. As $\delta_i > 0$ for all $i$, the above display implies
     \begin{align*}
         \hat{a}_1' \Upsilon_1 &\xrightarrow{p} 1 \\
         \hat{a}_1' \Upsilon_j &\xrightarrow{p} 0 \enspace \forall \enspace j = 2,3,\dots,r
     \end{align*}
     In matrix form we can write the above display together as
     \begin{equation*}
         \Upsilon' \hat{a}_1 \xrightarrow{p} e_1
     \end{equation*}
     where the $r$ dimensional $e_1 = (1,0,\dots,0)'$. By continuous mapping theorem, left multiplying by $\Upsilon$ gives
     \begin{equation*}
         \hat{a}_1 \xrightarrow{p} \Upsilon e_1 = \Upsilon_1
     \end{equation*}
    Therefore, the first row of the matrix of interest
    \begin{align*}
        \frac{\hat{\Blambda}_1' \tilde{\Lambda}}{N} &= \hat{a}_1' \left( \frac{\tilde{\Lambda}' \tilde{\Lambda}}{N} \right)^{\frac{1}{2}} \\
            & \xrightarrow{p} \Upsilon_1' \Sigma_{\tilde{\Lambda}}^{\frac{1}{2}}
    \end{align*}

    For the second column of $\hat{\BLambda}$, define a new set
    \begin{equation*}
        A_1 = A \setminus \text{range}(\hat{\Blambda}_1). 
    \end{equation*}
    We can decompose the symmetric $\frac{\Bar{Y}' \Bar{Y}}{NT}$ as $\frac{\Bar{Y}' \Bar{Y}}{NT} = PCP'$. where $P$ is orthonormal and $C$ is a diagonal matrix that holds the eigenvalues of $\frac{\Bar{Y}' \Bar{Y}}{NT}$. Call the eigenvalues $c_1 \geq c_2 \geq \dots \geq c_N$. Using already defined terms, we can write $P$ as
    \begin{equation*}
        P = N^{-\frac{1}{2}} [\hat{\Lambda} \enspace \hat{\BB}]
    \end{equation*}
    where $\hat{\Lambda}$ is the $N \times r-1$ collection of eigenvectors such that $\frac{\hat{\Lambda}' \hat{\Lambda}}{N} = I_{r-1}$. $\hat{\BB}$ are the remaining eigenvectors, scaled such that $\frac{\hat{\BB}' \hat{\BB}}{N} = I_{N-r+1}$. The corresponding column vectors, $\hat{\Blambda}_1,\dots,\hat{\Blambda}_{r-1},\hat{\Bb}_1,\dots,\hat{\Bb}_{N-r+1}$ form a basis of A. 

    We can write any $a \in A_1$ as
    \begin{align*}
        &a = k_2 \hat{\Blambda}_2+ \dots + k_r \hat{\Blambda}_{r-1} + k_{r+1} \hat{\Bb}_1 +\dots + k_N \hat{\Bb}_{N-r+1} \\
        &\text{where} \enspace \frac{a'a}{N} = 1 \implies \sum_{j=2}^N k_j^2 = 1
    \end{align*} 
    Now note that
    \begin{align*}
        R(a) &= N^{-1} a' \frac{\Bar{Y}' \Bar{Y}}{NT} a \\
        &= N^{-2} a' [\hat{\Lambda} \enspace \hat{\BB}] C [\hat{\Lambda} \enspace \hat{\BB}]' a \\
        &= [0 \enspace k_2 \enspace k_3 \enspace \dots \enspace k_N] C [0 \enspace k_2 \enspace k_3 \enspace \dots \enspace k_N]' \\
        &= \sum_{j=2}^N c_j k_j^2 \leq \sum_{j=2}^N c_2 k_j^2 = c_2 \sum_{j=2}^N k_j^2 = c_2 \\
        \therefore \sup_{a \in A_1} R(a) &= c_2
    \end{align*}

    Applying the previous Lemmas for the subspace $A_1$, we can easily see that $\hat{\Blambda}_2 = \text{arg} \sup_{a \in A_1} R(a)$ and $ R^*(\hat{\Blambda}_2)  \xrightarrow{p} \delta_2$. Repeating the arguments in the proof above, we can conclude that
    \begin{equation*}
        \frac{\hat{\Blambda}_2' \tilde{\Lambda}}{N} \xrightarrow{p} \Upsilon_2' \Sigma_{\tilde{\Lambda}}^{\frac{1}{2}}
    \end{equation*}
    Doing this sequentially for $r-1$ orthonormal subspaces of $A$ and collecting all the results, we have
    \begin{equation*}
        \frac{\hat{\Lambda}' \tilde{\Lambda}}{N} \xrightarrow{p} \Upsilon' \Sigma_{\tilde{\Lambda}}^{\frac{1}{2}}
    \end{equation*}
    $\frac{\hat{\Lambda}' \tilde{\Lambda}}{N}$ is invertible in the limit as both the the $r-1 \times r-1$ matrices $\Upsilon$ and $\Sigma_{\tilde{\Lambda}}^{\frac{1}{2}}$ are invertible. Hence, for large values of $N$ and $T$,  $\frac{\hat{\Lambda}' \tilde{\Lambda}}{N}$ is invertible.

    For the second result, Lemma \ref{lemma:SW_intermediate results} showed that $c_1 \xrightarrow{p} \delta_1$. Applying the result to the sub-spaces $A_1,\dots,A_{r-1}$ as above shows that $c_i \xrightarrow{p} \delta_i$, for all $1 \leq i \leq r-1$. $V_{NT}$ is a diagonal matrix that holds these c's, and $V$ holds the corresponding $\delta$'s. Thus, we have
    \begin{equation}
        V_{NT} \xrightarrow{p} V
    \end{equation}
    The $\delta$'s which are the eigenvalues of $\Sigma_{\tilde{\Lambda}} \Sigma_{\tilde{F}}$, are distinct and non-zero by Assumption \ref{ass:LN_strong factor structure}.3. By Weyl's inequality, $V_{NT}$ has non-zero and distinct entries for $N$ and $T$ large. Hence $V_{NT}$ is invertible.
\end{proof}

\paragraph{Asymptotic Linear Form} Thus, we get the asymptotic linear for of our estimator of $\tilde{\Lambda}$. Each column of the estimator is such that

\begin{equation} \label{eq:asymptotic_linear_lambda}
\hat{\Blambda}_i - \hat{H}' \Blambda_i = V_{NT}^{-1} \cdot \frac{\tilde{\Lambda}' \hat{\Lambda}}{N} \cdot \frac{1}{T} \sum_{t=1}^T \tilde{F}_t \Bar{u}_{it} + O_{\p} \left(\frac{1}{\delta_{NT}^2} \right)
\end{equation}
\section{Estimation of Common Factors} \label{app:estimation of common factors}
The estimation of the common factors arises from the first order conditions. The $T \times r-1$ matrix of the estimated common factors, $\hat{F}$ is given by
\begin{equation*}
    \hat{F} = \frac{\Bar{Y} \hat{\Lambda}}{N}
\end{equation*}
Using $\Bar{Y} = \tilde{F} \tilde{\Lambda}' + \Bar{u}$, we can expand the expression for the estimator as
\begin{align*}
    \hat{F} = \tilde{F} \left[ \frac{\tilde{\Lambda}' \hat{\Lambda}}{N} \right] + \frac{\Bar{u} \hat{\Lambda}}{N}
\end{align*}
Each column of the estimator is such that
\begin{equation*}
    \hat{F}_t = \left[ \frac{\hat{\Lambda}' \tilde{\Lambda} }{N} \right] \tilde{F}_t + \frac{\hat{\Lambda}' \Bar{u}_t}{N}
\end{equation*}
We write $\tilde{\Lambda} = \tilde{\Lambda} - \hat{\Lambda} \hat{H}^{-1} + \hat{\Lambda} \hat{H}^{-1}$ and $\hat{\Lambda} = \hat{\Lambda} - \tilde{\Lambda}H + \tilde{\Lambda}H $. $\hat{H}^{-1}$ exists due to the following reasons. Recall the definition of $\hat{H}$
\begin{equation*}
    \hat{H} = \left[ \frac{\tilde{F}' \tilde{F}}{NT} \right]\cdot \tilde{\Lambda}' \hat{\Lambda} \cdot V_{NT}^{-1}
\end{equation*}
It is a product of three $r-1 \times r-1$ matrices, each of which are invertible under Assumption \ref{ass:LN_strong factor structure}.1 and Theorem \ref{th:convergence of lambda_hat times lambda}. The estimator becomes
\begin{align*}
    \hat{F}_t &= N^{-1} [\hat{\Lambda}' ( \tilde{\Lambda} - \hat{\Lambda} \hat{H}^{-1} + \hat{\Lambda} \hat{H}^{-1}) ] \tilde{F}_t + N^{-1} \hat{\Lambda}' \Bar{u}_t \\
    &= N^{-1} \hat{\Lambda}' (\tilde{\Lambda} - \hat{\Lambda} \hat{H}^{-1}) \tilde{F}_t + N^{-1} \hat{\Lambda}'\hat{\Lambda} \hat{H}^{-1} \tilde{F}_t + N^{-1} ( \hat{\Lambda} - \tilde{\Lambda}H + \tilde{\Lambda}H)' \Bar{u}_t \\
    \hat{F}_t - \hat{H}^{-1} \tilde{F}_t &= H' \cdot \frac{1}{N} \sum_{i = 1}^N \tilde{\Blambda}_i \Bar{u}_{it} + N^{-1} \hat{\Lambda}' (\tilde{\Lambda} - \hat{\Lambda} \hat{H}^{-1}) \tilde{F}_t + N^{-1} ( \hat{\Lambda} - \tilde{\Lambda}H)' \Bar{u}_t
\end{align*}
where for the third equality, we use $N^{-1} \hat{\Lambda}' \hat{\Lambda} = I_{r-1}$. For the expression above, only the first term will be dominant. To see that, we need the following Lemmas

\begin{Lemma} \label{lemma:intermediate_1_factor_loadings}
    Under Assumptions \ref{ass:LN_strong factor structure} to \ref{ass:LN_moments and CLT},
    \begin{equation*}
        N^{-1} ( \hat{\Lambda} - \tilde{\Lambda}H)' \Bar{u}_t = O_{\p}\left(\frac{1}{\delta_{NT}^2} \right)
    \end{equation*}
\end{Lemma}
\begin{proof}
    We know from the previous section,
    \begin{equation*}
         \hat{\Blambda}_i - \hat{H}' \tilde{\Blambda}_i  = V_{NT}^{-1} \left(  \frac{1}{N} \sum_{j=1}^N \hat{\Blambda}_j \Bar{\gamma}(i,j) + \frac{1}{N} \sum_{j=1}^N \hat{\Blambda}_j \zeta_{ij} + \frac{1}{N} \sum_{j=1}^N \hat{\Blambda}_j \eta_{ij} + \frac{1}{N} \sum_{j=1}^N \hat{\Blambda}_j \xi_{ij} \right )
    \end{equation*}
    Note that $N^{-1} ( \hat{\Lambda} - \tilde{\Lambda}H)' \Bar{u}_t = \frac{1}{N} \sum_{i=1}^N (\hat{\Blambda}_i - \hat{H}' \Blambda_i) \Bar{u}_{it}$. Substituting the expression above to this, we have
    \begin{align*}
        N^{-1} ( \hat{\Lambda} - \tilde{\Lambda}H)' \Bar{u}_t &= V_{NT}^{-1} \Bigg(  \frac{1}{N^2} \sum_{i=1}^N \sum_{j=1}^N \hat{\Blambda}_j \Bar{\gamma}(i,j) \Bar{u}_{it} + \frac{1}{N^2} \sum_{i=1}^N \sum_{j=1}^N \hat{\Blambda}_j \zeta_{ij} \Bar{u}_{it} \\
        &+ \frac{1}{N^2} \sum_{i=1}^N \sum_{j=1}^N \hat{\Blambda}_j \eta_{ij} \Bar{u}_{it} + \frac{1}{N^2} \sum_{i=1}^N \sum_{j=1}^N \hat{\Blambda}_j \xi_{ij} \Bar{u}_{it}  \Bigg) \\
        &= V_{NT}^{-1}[I + II + III + IV]
    \end{align*}
    By Theorem \ref{th:convergence of lambda_hat times lambda}, we know that $V_{NT} = O_{\p}(1)$. Thus, we will analyze each of the terms one by one.

    \paragraph{First term} Write the first term as
    \begin{equation*}
        I = N^{-2} \sum_{i=1}^N \sum_{j=1}^N (\hat{\Blambda}_j - \hat{H}' \tilde{\Blambda}_j) \Bar{\gamma}(i,j) \Bar{u}_{it} + N^{-2} \hat{H}'  \sum_{i=1}^N \sum_{j=1}^N \tilde{\Blambda}_j \Bar{\gamma}(i,j) \Bar{u}_{it}
    \end{equation*}
    Consider the norm of the first term above
    \begin{align*}
        I_1 &= \left \| N^{-2} \sum_{i=1}^N \sum_{j=1}^N (\hat{\Blambda}_j - \hat{H}' \tilde{\Blambda}_j) \Bar{\gamma}(i,j) \Bar{u}_{it}   \right \|  \leq  N^{-2} \sum_{i=1}^N \sum_{j=1}^N \left \| (\hat{\Blambda}_j - \hat{H}' \tilde{\Blambda}_j) \Bar{\gamma}(i,j) \Bar{u}_{it}   \right \| 
    \end{align*}
    Apply Cauchy-Schwartz on the inner-sum
    \begin{equation*}
        \sum_{j=1}^N \left \| (\hat{\Blambda}_j - \hat{H}' \tilde{\Blambda}_j) \Bar{\gamma}(i,j) \Bar{u}_{it}   \right \| \leq \left[ \sum_j \| \hat{\Blambda}_j - \hat{H}' \tilde{\Blambda}_j \|^2 \right]^{\frac{1}{2}} \cdot \left[ \sum_j \| \Bar{\gamma}(i,j) \Bar{u}_{it} \|^2 \right]^{\frac{1}{2}}
    \end{equation*}
    Inserting this into the display above,
    \begin{align*}
        I_1 & \leq N^{-1} \left[ \sum_j \| \hat{\Blambda}_j - \hat{H}' \tilde{\Blambda}_j \|^2 \right]^{\frac{1}{2}} \cdot N^{-1} \sum_{i=1}^N \left[ \sum_j \| \Bar{\gamma}(i,j) \Bar{u}_{it} \|^2 \right]^{\frac{1}{2}} \\
        &= N^{-1/2} \left[ \frac{1}{N} \sum_j \| \hat{\Blambda}_j - \hat{H}' \tilde{\Blambda}_j \|^2 \right]^{\frac{1}{2}} \cdot N^{-1} \sum_{i=1}^N \left[ \sum_j  | \Bar{\gamma}(i,j)|^2 \Bar{u}_{it}^2  \right]^{\frac{1}{2}} \\
        &= N^{-1/2} \cdot O_{\p}(\delta_{NT}^{-1}) \cdot N^{-1} \sum_{i=1}^N \left[ \sum_j  | \Bar{\gamma}(i,j)|^2 \Bar{u}_{it}^2  \right]^{\frac{1}{2}}
    \end{align*}
    where the stochastic order in the last equality comes from Theorem \ref{th:norm squared of estimator}. For the second term
    \begin{align*}
        N^{-1} \sum_{i=1}^N \left[ \sum_j  | \Bar{\gamma}(i,j)|^2 \Bar{u}_{it}^2  \right]^{\frac{1}{2}} &= N^{-1} \sum_{i=1}^N \Bar{u}_{it}  \left[ \sum_j  | \Bar{\gamma}(i,j)|^2  \right]^{\frac{1}{2}} \\
        & \leq  N^{-1} \left[ \sum_{i=1}^N \Bar{u}_{it}^2 \cdot \sum_i \sum_j | \Bar{\gamma}(i,j)|^2 \right]^{\frac{1}{2}} \\
        &= \left[ \frac{1}{N} \sum_{i=1}^N \Bar{u}_{it}^2 \cdot \frac{1}{N} \sum_i \sum_j | \Bar{\gamma}(i,j)|^2 \right]^{\frac{1}{2}} = O_{\p}(1)
    \end{align*}
    where the stochastic order comes from $\Bar{u}_{it}^2 \leq C u_{it}^2$ and Assumption \ref{ass:LN_time and cross sectional dependence}.1 and Lemma \ref{lemma:gamma_bounded}.

Thus we have
\begin{equation*}
    I_1 = O_{\p}\left( N^{-\frac{1}{2}} \delta_{NT}^{-1} \right)
\end{equation*}
Now consider the second part of the first term
\begin{equation*}
    I_2 = N^{-2} \hat{H}'  \sum_{i=1}^N \sum_{j=1}^N \tilde{\Blambda}_j \Bar{\gamma}(i,j) \Bar{u}_{it}
\end{equation*}
Recall that 
\begin{equation*}
     \hat{H} = \left[ \frac{\tilde{F}' \tilde{F}}{T} \right]\cdot \left[  \frac{\tilde{\Lambda}' \hat{\Lambda}}{N} \right] \cdot V_{NT}^{-1} = O_{\p}(1)
\end{equation*}
where the stochastic order comes from Assumption \ref{ass:LN_strong factor structure} and Theorem \ref{th:convergence of lambda_hat times lambda}. Hence we consider only $N^{-2} \sum_{i=1}^N \sum_{j=1}^N \tilde{\Blambda}_j \Bar{\gamma}(i,j) \Bar{u}_{it}$. 
As $\tilde{\Blambda}_j$ is independent of all $u_t$, the expectation of the above is zero. So we consider the variance of the term. Before that, write the term conveniently as
\begin{equation*}
    N^{-2} \sum_{i=1}^N \sum_{j=1}^N \tilde{\Blambda}_j \Bar{\gamma}(i,j) \Bar{u}_{it} = \frac{1}{N^2} \sum_{j=1}^N \tilde{\Blambda}_j \sum_{i=1}^N \Bar{\gamma}(i,j) \Bar{u}_{it} \defeq \frac{1}{N^2} \sum_{j=1}^N \tilde{\Blambda}_j U_{jt}
\end{equation*}
Consider the variance of the term above
\begin{align*}
    \var \left[ \frac{1}{N^2} \sum_{j=1}^N \tilde{\Blambda}_j U_{jt} \right] &= \frac{1}{N^4} \sum_{j=1}^N \sum_{k=1}^N \e[\tilde{\Blambda}_j \tilde{\Blambda}_k'] \e[U_{jt} U_{kt}] = \frac{1}{N^4} \sum_{j=1}^N \e[\tilde{\Blambda}_j \tilde{\Blambda}_j'] \e[U_{jt}^2] \\
    &= \frac{1}{N^4} \sum_{j=1}^N \e \left[ \left( \sum_{i=1}^N \Bar{\gamma}(i,j) \Bar{u}_{it} \right)^2 \right] \Sigma_{\tilde{\Lambda}} \\
    &= \frac{1}{N^4} \sum_{j=1}^N \sum_{i=1}^N \sum_{k=1}^N \Bar{\gamma}(i,j) \Bar{\gamma}(k,j) \e[\Bar{u}_{it} \Bar{u}_{kt}] \cdot O(1) \\
    & \leq \frac{1}{N^3} \sum_{j=1}^N \sum_{i=1}^N \Bar{\gamma}(i,j) \left[ \frac{1}{N} \sum_{k=1}^N |\Bar{\gamma}(k,j)|^2 \right]^{\frac{1}{2}} \left[ \frac{1}{N} \sum_{k=1}^N  |\e[\Bar{u}_{it} \Bar{u}_{kt}]|^2 \right]^{\frac{1}{2}} O(1)
\end{align*}
where the second equality uses independence of $\tilde{\Blambda}_j$ across $j$, and the last inequality is Cauchy-Schwartz over $k$. In the proof of Lemma \ref{lemma:gamma_bounded}, we have seen that
\begin{align*}
    |\Bar{\gamma}(k,j)|^2 &\leq |\Bar{\gamma}(j,j)| \cdot |\Bar{\gamma}(k,k)| \\
  \therefore  \frac{1}{N} \sum_{k=1}^N |\Bar{\gamma}(k,j)|^2 & \leq |\Bar{\gamma}(j,j)| \cdot \frac{1}{N} \sum_{k=1}^N |\Bar{\gamma}(k,k)| \leq M
\end{align*}
As the fourth moment of $u_{it}$ is bounded, $\frac{1}{N} \sum_{k=1}^N  |\e[\Bar{u}_{it} \Bar{u}_{kt}]|^2 \leq M$. Thus,
\begin{align*}
    \var \left[ \frac{1}{N^2} \sum_{j=1}^N \tilde{\Blambda}_j U_{jt} \right] &= \frac{1}{N^2} \frac{1}{N} \sum_{j=1}^N \sum_{i=1}^N \Bar{\gamma}(i,j) \cdot M \cdot O(1) = O \left( \frac{1}{N^2} \right) 
\end{align*}
where the last equality uses Assumption \ref{ass:LN_time and cross sectional dependence}.2.b. Hence
\begin{equation*}
    N^{-2} \sum_{i=1}^N \sum_{j=1}^N \tilde{\Blambda}_j \Bar{\gamma}(i,j) \Bar{u}_{it} = O_{\p}(N^{-1})
\end{equation*}

Thus, $I_2 = O_{\p}(N^{-1})$, and we can conclude
\begin{equation*}
    I = O_{\p}\left( N^{-\frac{1}{2}} \delta_{NT}^{-1} \right) + O_{\p}(N^{-1}) = O_{\p}\left( N^{-\frac{1}{2}} \delta_{NT}^{-1} \right) 
\end{equation*}

\paragraph{Second Term} Similarly expand the second term as
\begin{equation*}
    II = N^{-2} \sum_{i=1}^N \sum_{j=1}^N (\hat{\Blambda}_j - \hat{H}' \tilde{\Blambda}_j) \zeta_{ij} \Bar{u}_{it} + N^{-2} \hat{H}'  \sum_{i=1}^N \sum_{j=1}^N \tilde{\Blambda}_j \zeta_{ij} \Bar{u}_{it}
\end{equation*}
We will analyze the second term first. Recall that
\begin{align*}
    \zeta_{ij} &= \Bar{u}_i' \Bar{u}_j/T - \Bar{\gamma}(i,j) \\
        &= \frac{1}{T} \sum_{t=1}^T \Bar{u}_{it} \Bar{u}_{jt} - \frac{1}{T} \sum_{t=1}^T \e[\Bar{u}_{it} \Bar{u}_{jt}]
\end{align*}
Similar to the last time, we ignore $\hat{H}$
\begin{align*}
    II_2 &= N^{-2} \sum_i \sum_j \tilde{\Blambda}_j \left[ \frac{1}{T} \sum_{t=1}^T \Bar{u}_{it} \Bar{u}_{jt} - \frac{1}{T} \sum_{t=1}^T \e[\Bar{u}_{it} \Bar{u}_{jt}] \right] \Bar{u}_{it} \\
    & \leq C N^{-2} \sum_i \sum_j \tilde{\Blambda}_j \left[ \frac{1}{T} \sum_{t=1}^T u_{it} u_{jt} - \frac{1}{T} \sum_{t=1}^T \e[u_{it} u_{jt}] \right] \Bar{u}_{it} \\
    &= \frac{C}{\sqrt{NT}} \frac{1}{N} \sum_i \left[ \frac{1}{\sqrt{NT}} \sum_j \sum_t \tilde{\Blambda}_j \left( u_{it} u_{jt} - \e[u_{it} u_{jt}]  \right) \right] \Bar{u}_{it} \\
    & \defeq \frac{C}{\sqrt{NT}} \frac{1}{N} \sum_i L_i \Bar{u}_{it}
\end{align*}
Note that $\e \| L_i \Bar{u}_{it} \| \leq \big( \e \|L_i \|^2 \e \| \Bar{u}_{it} \|^2 \big)^{\frac{1}{2}} \leq M$ from Assumption \ref{ass:LN_moments and CLT}.1 and boundedness of $\e[u_{it} u_{jt}]$. Thus,
\begin{equation*}
    II_2 = O_{\p}\left( \frac{1}{\sqrt{NT}} \right)
\end{equation*}

Consider the norm of the first term
\begin{align*}
    \left \| II_1  \right \| &= N^{-2} \left \| \sum_{j=1}^N (\hat{\Blambda}_j - \hat{H}' \tilde{\Blambda}_j) \sum_{i=1}^N \zeta_{ij} \Bar{u}_{it} \right \| \leq N^{-2} \sum_{j=1}^N \left \|  (\hat{\Blambda}_j - \hat{H}'\tilde{\Blambda}_j) \sum_{i=1}^N \zeta_{ij} \Bar{u}_{it} \right \| \\
    & \leq N^{-2} \left[ \sum_j \| \hat{\Blambda}_j - \hat{H}'\tilde{\Blambda}_j \|^2 \right]^{\frac{1}{2}} \left[\sum_j \left( \sum_i \zeta_{ij} \Bar{u}_{it} \right)^2 \right]^{\frac{1}{2}} \\
    &= \left[ \frac{1}{N} \sum_j \| \hat{\Blambda}_j - \hat{H}'\tilde{\Blambda}_j \|^2 \right]^{\frac{1}{2}} \left[\frac{1}{N} \sum_j \left( \frac{1}{N} \sum_i \zeta_{ij} \Bar{u}_{it} \right)^2 \right]^{\frac{1}{2}} \\
    &= O_{\p}(\delta_{NT}^{-1}) \left[\frac{1}{N} \sum_j \left( \frac{1}{N} \sum_i \zeta_{ij} \Bar{u}_{it} \right)^2 \right]^{\frac{1}{2}}
\end{align*}
Now consider
\begin{align*}
    \frac{1}{N} \sum_i \zeta_{ij} \Bar{u}_{it} &= \frac{1}{N} \sum_i \left[ \frac{1}{T} \sum_{t=1}^T \Bar{u}_{it} \Bar{u}_{jt} - \frac{1}{T} \sum_{t=1}^T \e[\Bar{u}_{it} \Bar{u}_{jt}]  \right]\Bar{u}_{it} \\
    &= \frac{1}{\sqrt{T}} \cdot \frac{C}{N} \sum_i \left[ \frac{1}{\sqrt{T}} \sum_{t=1}^T u_{it} u_{jt} - \frac{1}{T} \sum_{t=1}^T \e[u_{it} u_{jt}]  \right]\Bar{u}_{it} \\ 
    & \defeq \frac{1}{\sqrt{T}} \cdot \frac{C}{N} \sum_i L_{ijt} \Bar{u}_{it}
\end{align*}
$\e \| L_{ijt} \Bar{u}_{it} \| \leq \left[\e \| L_{ijt} \|^2 \cdot \e \| \Bar{u}_{it} \|^2  \right]^{\frac{1}{2}} \leq M$ under Assumption \ref{ass:LN_time and cross sectional dependence}.1 and \ref{ass:LN_time and cross sectional dependence}.3. Thus $\frac{1}{N} \sum_i \zeta_{ij} \Bar{u}_{it} = O_{\p}(\frac{1}{\sqrt{T}})$. Thus, $II_1 = O_{\p} \left(\frac{1}{\delta_{NT} \sqrt{T}} \right)$. The second term is
\begin{equation*}
    II = O_{\p} \left(\frac{1}{\delta_{NT} \sqrt{T}} \right) + O_{\p}\left( \frac{1}{\sqrt{NT}} \right) = O_{\p} \left(\frac{1}{\delta_{NT} \sqrt{T}} \right)
\end{equation*}

\paragraph{Third Term} Expand the third term as
\begin{equation*}
    III = N^{-2} \sum_{i=1}^N \sum_{j=1}^N (\hat{\Blambda}_j - \hat{H}' \tilde{\Blambda}_j) \eta_{ij} \Bar{u}_{it} + N^{-2} \hat{H}'  \sum_{i=1}^N \sum_{j=1}^N \tilde{\Blambda}_j \eta_{ij} \Bar{u}_{it}
\end{equation*}
By the steps seen above, the first part of the third term is
\begin{align*}
    III_1 &\leq \left[ \frac{1}{N} \sum_j \| \hat{\Blambda}_j - \hat{H}'\tilde{\Blambda}_j \|^2 \right]^{\frac{1}{2}} \left[\frac{1}{N} \sum_j \left( \frac{1}{N} \sum_i \eta_{ij} \Bar{u}_{it} \right)^2 \right]^{\frac{1}{2}} \\
    &= O_{\p}(\delta_{NT}^{-1}) \left[\frac{1}{N} \sum_j \left( \frac{1}{N} \sum_i \eta_{ij} \Bar{u}_{it} \right)^2 \right]^{\frac{1}{2}}
\end{align*}
Now consider
\begin{align*}
    \frac{1}{N} \sum_i \eta_{ij} \Bar{u}_{it} &= \frac{1}{N} \tilde{\Blambda}_j'  \sum_i \left( \frac{1}{T} \sum_{s=1}^T \tilde{F}_s \Bar{u}_{is} \right) \Bar{u}_{it}  \\
    & \leq \frac{1}{\sqrt{T}} \tilde{\Blambda}_j' \left[ \frac{1}{N} \sum_i \left \| \frac{1}{\sqrt{T}} \sum_{s=1}^T \tilde{F}_s \Bar{u}_{is} \right \|^2 \right]^{\frac{1}{2}} \left[ \frac{1}{N} \sum_i \Bar{u}_{it}^2 \right]^{\frac{1}{2}} \\
    &= O_{\p}\left( \frac{1}{\sqrt{T}} \right) O_{\p}(1) O_{\p}(1) = O_{\p}\left( \frac{1}{\sqrt{T}} \right)
\end{align*}
Thus,
\begin{equation*}
    III_1 = O_{\p}(\delta_{NT}^{-1}) O_{\p}\left( \frac{1}{\sqrt{T}} \right) = O_{\p}\left( \frac{1}{\delta_{NT} \sqrt{T}} \right)
\end{equation*}

The second term, disregarding $\hat{H}$
\begin{align*}
    III_2 &= N^{-2}  \sum_{i=1}^N \sum_{j=1}^N \tilde{\Blambda}_j \tilde{\Blambda}_j' \left( \frac{1}{T} \sum_{s=1}^T \tilde{F}_s \Bar{u}_{is} \right) \Bar{u}_{it} \\
    &= \left[ \frac{1}{N} \sum_{j=1}^N \tilde{\Blambda}_j \tilde{\Blambda}_j' \right] \left[ \frac{1}{NT} \sum_{i=1}^N \sum_{s=1}^T \tilde{F}_s \Bar{u}_{is} \Bar{u}_{it}  \right] \\
    &= O_{\p}(1) \frac{1}{NT} \sum_{i=1}^N \sum_{s=1}^T \tilde{F}_s \Bar{u}_{is} \Bar{u}_{it}
\end{align*}
Consider the expectation of the remaining term
\begin{align*}
    \e\left[ \frac{1}{NT} \sum_{i=1}^N \sum_{s=1}^T \tilde{F}_s \Bar{u}_{is} \Bar{u}_{it} \right] & \leq \frac{1}{NT} \sum_i \sum_s \big( \e[\| \tilde{F}_s \|^2] \big)^{\frac{1}{2}} \big( \e[\|\Bar{u}_{is} \Bar{u}_{it} \|^2] \big)^{\frac{1}{2}} \\
    & \leq \frac{C^2}{NT} \sum_i \sum_s \big( \e[\| \tilde{F}_s \|^2] \big)^{\frac{1}{2}} \big( \e[u_{is}^2 u_{it}^2] \big)^{\frac{1}{2}} = O_{\p}(\frac{1}{T})
\end{align*}
where the stochastic order in the final step comes from Assumption \ref{ass:LN_time and cross sectional dependence}.2 as $\sum_s \big( \e[u_{is}^2 u_{it}^2] \big)^{\frac{1}{2}} = \sum_s \tau_{st}^{\frac{1}{2}} = O_{\p}(1)$. Thus we have
\begin{equation*}
    III = O_{\p}\left( \frac{1}{\delta_{NT} \sqrt{T}} \right) + O_{\p}\left( \frac{1}{T} \right) = O_{\p}\left( \frac{1}{\delta_{NT} \sqrt{T}} \right)
\end{equation*}

\paragraph{Fourth Term} Expand the fourth term as
\begin{equation*}
    IV = N^{-2} \sum_{i=1}^N \sum_{j=1}^N (\hat{\Blambda}_j - \hat{H}' \tilde{\Blambda}_j) \xi_{ij} \Bar{u}_{it} + N^{-2} \hat{H}'  \sum_{i=1}^N \sum_{j=1}^N \tilde{\Blambda}_j \xi_{ij} \Bar{u}_{it}
\end{equation*}
By the steps seen above, the first part of the fourth term is
\begin{align*}
    IV_1 &\leq \left[ \frac{1}{N} \sum_j \| \hat{\Blambda}_j - \hat{H}'\tilde{\Blambda}_j \|^2 \right]^{\frac{1}{2}} \left[\frac{1}{N} \sum_j \left( \frac{1}{N} \sum_i \xi_{ij} \Bar{u}_{it} \right)^2 \right]^{\frac{1}{2}} \\
    &= O_{\p}(\delta_{NT}^{-1}) \left[\frac{1}{N} \sum_j \left( \frac{1}{N} \sum_i \xi_{ij} \Bar{u}_{it} \right)^2 \right]^{\frac{1}{2}}
\end{align*}
Now consider
\begin{align*}
    \frac{1}{N} \sum_i \xi_{ij} \Bar{u}_{it} &= \frac{1}{N}   \sum_i \tilde{\Blambda}_i' \left( \frac{1}{T} \sum_{s=1}^T \tilde{F}_s \Bar{u}_{js} \right) \Bar{u}_{it}  \\
    &\leq \left[ \frac{1}{N} \sum_i \|\tilde{\Blambda}_i \|^2 \right]^{\frac{1}{2}} \left[ \frac{1}{N} \sum_i \left \| \frac{1}{T} \sum_{s=1}^T \tilde{F}_s \Bar{u}_{js}  \right \|^2 \Bar{u}_{it}^2 \right]^{\frac{1}{2}} = O_{\p}\left( \frac{1}{\sqrt{T}} \right)
\end{align*}
Thus,
\begin{equation*}
    IV_1 = O_{\p}\left( \frac{1}{\delta_{NT} \sqrt{T}} \right)
\end{equation*}

The second term, disregarding $\hat{H}$
\begin{align*}
    IV_2 &= N^{-2}  \sum_{i=1}^N \sum_{j=1}^N \tilde{\Blambda}_j \tilde{\Blambda}_i' \left( \frac{1}{T} \sum_{s=1}^T \tilde{F}_s \Bar{u}_{js} \right) \Bar{u}_{it} \\
    & \leq \frac{1}{N} \sum_j \left[ \frac{1}{N} \sum_i \| \tilde{\Blambda}_j \tilde{\Blambda}_i' \|^2 \right]^{\frac{1}{2}} \left[ \frac{1}{NT} \sum_i \sum_{s=1}^T \left \| \tilde{F}_s \Bar{u}_{js} \Bar{u}_{it} \right \|^2  \right]^{\frac{1}{2}} = O_{\p} \left( \frac{1}{T} \right)
\end{align*}
where $\frac{1}{N} \sum_i \| \tilde{\Blambda}_j \tilde{\Blambda}_i' \|^2$ is $O_{\p}(1)$ as $\| \tilde{\Blambda}_j \|^4 < \infty$ for every $j$ under Assumption \ref{ass:LN_strong factor structure}.2. The stochastic order of the second term comes from the same steps as the ones for $III_2$. Thus
\begin{equation*}
    IV = O_{\p}\left( \frac{1}{\delta_{NT} \sqrt{T}} \right) + O_{\p}\left( \frac{1}{T} \right) = O_{\p}\left( \frac{1}{\delta_{NT} \sqrt{T}} \right)
\end{equation*}

Combining all the four terms
\begin{align*}
    N^{-1} ( \hat{\Lambda} - \tilde{\Lambda}H)' \Bar{u}_t &= O_{\p}\left( \frac{1}{\delta_{NT} \sqrt{N}} \right) + O_{\p} \left(\frac{1}{\delta_{NT} \sqrt{T}} \right) + O_{\p}\left( \frac{1}{\delta_{NT} \sqrt{T}} \right) + O_{\p}\left( \frac{1}{\delta_{NT} \sqrt{T}} \right) \\
    & = O_{\p}\left( \frac{1}{\delta_{NT}^2 } \right) 
\end{align*}
\end{proof}

To proceed further, we need to analyze $ N^{-1} \hat{\Lambda}' (\tilde{\Lambda} - \hat{\Lambda} \hat{H}^{-1}) \tilde{F}_t = N^{-1} \hat{\Lambda}'(\hat{\Lambda} - \tilde{\Lambda} \hat{H}) (-\hat{H}^{-1}) \tilde{F}_t$. We will start by analyzing $ N^{-1} (\hat{\Lambda} - \tilde{\Lambda} \hat{H})' \tilde{\Lambda}$. 

\begin{Lemma} \label{lemma:intermediate_2_factor_loadings}
    Under Assumptions \ref{ass:LN_strong factor structure} to \ref{ass:LN_moments and CLT},
    \begin{equation*}
       N^{-1} (\hat{\Lambda} - \tilde{\Lambda} \hat{H})' \tilde{\Lambda} = O_{\p}\left(\frac{1}{\delta_{NT}^2} \right)
    \end{equation*}
\end{Lemma}
\begin{proof}
    \begin{align*}
        N^{-1} ( \hat{\Lambda} - \tilde{\Lambda}H)' \tilde{\Lambda} &= V_{NT}^{-1} \Bigg(  \frac{1}{N^2} \sum_{i=1}^N \sum_{j=1}^N \hat{\Blambda}_j \tilde{\Blambda}_i' \Bar{\gamma}(i,j) + \frac{1}{N^2} \sum_{i=1}^N \sum_{j=1}^N \hat{\Blambda}_j \tilde{\Blambda}_i' \zeta_{ij} \\
        &+ \frac{1}{N^2} \sum_{i=1}^N \sum_{j=1}^N \hat{\Blambda}_j \tilde{\Blambda}_i' \eta_{ij} + \frac{1}{N^2} \sum_{i=1}^N \sum_{j=1}^N \hat{\Blambda}_j \tilde{\Blambda}_i' \xi_{ij}  \Bigg) \\
        &= V_{NT}^{-1}[I + II + III + IV]
    \end{align*}
The proof for $I$ and $II$ follows the proof of the similar terms in Lemma \ref{lemma:intermediate_2_factor_loadings} with $\Bar{u}_{it}$ replaced with $\tilde{\Blambda}_i$. As both $\frac{1}{N} \sum_{i=1}^N \Bar{u}_{it}^2$ and $\frac{1}{N} \sum_{i=1}^N \|\tilde{\Blambda}_i \|^2$ are $O_{\p}(1)$, the proofs go through identically and we have
\begin{align*}
    I = O_{\p}\left( \frac{1}{\delta_{NT} \sqrt{N}} \right) \quad II = O_{\p} \left(\frac{1}{\delta_{NT} \sqrt{T}} \right)
\end{align*}
The third and fourth terms uses limited dependence between $\Bar{u}_{it}^2$ and $\Bar{u}_{jt}^2$ which is not available to us in this context. Hence we consider those terms in detail.

\paragraph{Third Term} Expand the third term as
\begin{equation*}
    III = N^{-2} \sum_{i=1}^N \sum_{j=1}^N (\hat{\Blambda}_j - \hat{H}' \tilde{\Blambda}_j) \tilde{\Blambda}_i' \eta_{ij} + N^{-2} \hat{H}'  \sum_{i=1}^N \sum_{j=1}^N \tilde{\Blambda}_j \tilde{\Blambda}_i' \eta_{ij} 
\end{equation*}
By the steps in the previous proof, the first part of the third term is
\begin{align*}
    III_1 &\leq \left[ \frac{1}{N} \sum_j \| \hat{\Blambda}_j - \hat{H}'\tilde{\Blambda}_j \|^2 \right]^{\frac{1}{2}} \left[\frac{1}{N} \sum_j \left \| \frac{1}{N} \sum_i \tilde{\Blambda}_i'  \eta_{ij}  \right \|^2 \right]^{\frac{1}{2}} \\
    &= O_{\p}(\delta_{NT}^{-1}) \left[\frac{1}{N} \sum_j \left \| \frac{1}{N} \sum_i \tilde{\Blambda}_i'  \eta_{ij}  \right \|^2 \right]^{\frac{1}{2}}
\end{align*}
The remaining term
\begin{align*}
    \left[\frac{1}{N} \sum_j \left \| \frac{1}{N} \sum_i \tilde{\Blambda}_i'  \eta_{ij}  \right \|^2 \right]^{\frac{1}{2}} &= \left[\frac{1}{N} \sum_j \left \| \frac{1}{N} \sum_i \tilde{\Blambda}_i' \tilde{\Blambda}_j'  \frac{1}{T} \sum_{s=1}^T \tilde{F}_s \Bar{u}_{is}  \right \|^2 \right]^{\frac{1}{2}} \\
    &= \frac{1}{\sqrt{NT}} \left[\frac{1}{N} \sum_j \left \| \tilde{\Blambda}_j' \frac{1}{\sqrt{NT}}  \sum_i \sum_{s=1}^T \tilde{F}_s \tilde{\Blambda}_i'  \Bar{u}_{is}  \right \|^2 \right]^{\frac{1}{2}} = O_{\p} \left( \frac{1}{\sqrt{NT}} \right)
\end{align*}
under Assumption \ref{ass:LN_moments and CLT}.3. We use the fact that $\tilde{\Blambda}_j' \tilde{F}_s$ is a scalar. Thus,
\begin{equation*}
    III_1 = O_{\p} \left( \frac{1}{\delta_{NT} \sqrt{NT}} \right)
\end{equation*}
Now consider the second term, ignoring $\hat{H}$,
\begin{align*}
    N^{-2} \sum_{i=1}^N \sum_{j=1}^N \tilde{\Blambda}_j \tilde{\Blambda}_i' \tilde{\Blambda}_j'  \frac{1}{T} \sum_{s=1}^T \tilde{F}_s \Bar{u}_{is} &= N^{-2} \sum_{i=1}^N \sum_{j=1}^N  \frac{1}{T} \sum_{s=1}^T \tilde{\Blambda}_j \tilde{\Blambda}_j' \tilde{F}_s  \tilde{\Blambda}_i'   \Bar{u}_{is} \\
    &= \frac{1}{N} \sum_{j=1}^N \tilde{\Blambda}_j \tilde{\Blambda}_j' \cdot \frac{1}{NT} \sum_{i=1}^N \sum_{s=1}^T \tilde{F}_s  \tilde{\Blambda}_i'   \Bar{u}_{is} = O_{\p} \left( \frac{1}{\sqrt{NT}} \right)
\end{align*}
under Assumption \ref{ass:LN_moments and CLT}.3. Thus, we conclude
\begin{equation*}
    III = O_{\p} \left( \frac{1}{\delta_{NT} \sqrt{NT}} \right) + O_{\p} \left( \frac{1}{\sqrt{NT}} \right) = O_{\p} \left( \frac{1}{\sqrt{NT}} \right)
\end{equation*}

\paragraph{Fourth Term} Expand the fourth term as
\begin{equation*}
    IV = N^{-2} \sum_{i=1}^N \sum_{j=1}^N (\hat{\Blambda}_j - \hat{H}' \tilde{\Blambda}_j) \tilde{\Blambda}_i' \xi_{ij}  + N^{-2} \hat{H}'  \sum_{i=1}^N \sum_{j=1}^N \tilde{\Blambda}_j \tilde{\Blambda}_i' \xi_{ij} 
\end{equation*}
By the steps in the previous proof, the first part of the third term is
\begin{align*}
    IV_1 &\leq \left[ \frac{1}{N} \sum_j \| \hat{\Blambda}_j - \hat{H}'\tilde{\Blambda}_j \|^2 \right]^{\frac{1}{2}} \left[\frac{1}{N} \sum_j \left \| \frac{1}{N} \sum_i \tilde{\Blambda}_i'  \xi_{ij}  \right \|^2 \right]^{\frac{1}{2}} \\
    &= O_{\p}(\delta_{NT}^{-1}) \left[\frac{1}{N} \sum_j \left \| \frac{1}{N} \sum_i \tilde{\Blambda}_i'  \xi_{ij}  \right \|^2 \right]^{\frac{1}{2}}
\end{align*}
The remaining term
\begin{align*}
   \frac{1}{N} \sum_j \left \| \frac{1}{N} \sum_i \tilde{\Blambda}_i'  \xi_{ij}  \right \|^2  &= \frac{1}{N} \sum_j \left \| \frac{1}{N} \sum_i \tilde{\Blambda}_i' \tilde{\Blambda}_i'  \frac{1}{T} \sum_{s=1}^T \tilde{F}_s \Bar{u}_{js}  \right \|^2 \\
   & \leq \frac{1}{N} \sum_j \left( \frac{1}{N} \sum_i \| \tilde{\Blambda}_i \|^2 \cdot \frac{1}{N} \sum_i \| \frac{1}{T} \sum_{s=1}^T \tilde{\Blambda}_i' \tilde{F}_s \Bar{u}_{js} \|^2 \right)
\end{align*}
For every $i$, $\frac{1}{\sqrt{T}} \sum_{s=1}^T  \tilde{F}_s \Bar{u}_{js}$ is fixed and is $O_{\p}(1)$. As $r$ is fixed and $\tilde{\Blambda}_i$ is finite, $\frac{1}{\sqrt{T}} \sum_{s=1}^T  \tilde{\Blambda}_i'  \tilde{F}_s \Bar{u}_{js}$ is a finite linear combination of $O_{\p}(1)$ terms and hence, is also $O_{\p}(1)$. Thus
\begin{equation*}
    IV_1 = O_{\p} \left( \frac{1}{\delta_{NT} \sqrt{T}} \right)
\end{equation*}
Consider the second term, ignoring $\hat{H}$,
\begin{align*}
    N^{-2} \sum_{i=1}^N \sum_{j=1}^N \tilde{\Blambda}_j \tilde{\Blambda}_i' \tilde{\Blambda}_i'  \frac{1}{T} \sum_{s=1}^T \tilde{F}_s \Bar{u}_{js} &= \frac{1}{N} \sum_{i=1}^N  \left( \frac{1}{NT} \sum_{j=1}^N \sum_{s=1}^T \tilde{\Blambda}_i' \tilde{F}_s    \tilde{\Blambda}_j   \Bar{u}_{js}  \right) \tilde{\Blambda}_i' 
\end{align*}
The $r \times r$ matrix, $\frac{1}{\sqrt{NT}} \sum_{j=1}^N \sum_{s=1}^T  \tilde{\Blambda}_j \tilde{F}_s '  \Bar{u}_{js} $ is bounded. Which means every element of the matrix is bounded. For each value of $i$, $\tilde{\Blambda}_i$ is bounded. Thus, $\frac{1}{\sqrt{NT}} \sum_{j=1}^N \sum_{s=1}^T \tilde{\Blambda}_i' \tilde{F}_s \tilde{\Blambda}_j   \Bar{u}_{js}$ is just a vector, from the bounded matrix, which is scaled by a bounded constant. Hence $IV_2 = O_{\p} \left ( \frac{1}{\sqrt{NT}} \right)$, and
\begin{equation*}
    IV = O_{\p} \left( \frac{1}{\delta_{NT} \sqrt{T}} \right) + O_{\p} \left ( \frac{1}{\sqrt{NT}} \right) = O_{\p} \left( \frac{1}{\delta_{NT} \sqrt{T}} \right)
\end{equation*}
Combining all the four terms
\begin{align*}
    N^{-1} ( \hat{\Lambda} - \tilde{\Lambda}H)' \tilde{\Lambda} &= O_{\p}\left( \frac{1}{\delta_{NT} \sqrt{N}} \right) + O_{\p} \left(\frac{1}{\delta_{NT} \sqrt{T}} \right) + O_{\p}\left( \frac{1}{ \sqrt{NT}} \right) + O_{\p}\left( \frac{1}{\delta_{NT} \sqrt{T}} \right) \\
    & = O_{\p}\left( \frac{1}{\delta_{NT}^2 } \right) 
\end{align*}
\end{proof}

Now we have all the components to analyze, $ N^{-1} (\hat{\Lambda} - \tilde{\Lambda} \hat{H})' \hat{\Lambda}$
\begin{Lemma} \label{lemma:intermediate_3_factor_loadings}
    Under Assumptions \ref{ass:LN_strong factor structure} to \ref{ass:LN_moments and CLT},
    \begin{equation*}
       N^{-1} (\hat{\Lambda} - \tilde{\Lambda} \hat{H})' \hat{\Lambda} = O_{\p}\left(\frac{1}{\delta_{NT}^2} \right)
    \end{equation*}
\end{Lemma}
\begin{proof}
    \begin{align*}
        N^{-1} (\hat{\Lambda} - \tilde{\Lambda} \hat{H})' \hat{\Lambda} &= N^{-1} (\hat{\Lambda} - \tilde{\Lambda} \hat{H})' ( \hat{\Lambda} - \tilde{\Lambda} \hat{H}) + N^{-1} (\hat{\Lambda} - \tilde{\Lambda} \hat{H})' \tilde{\Lambda} \hat{H} \\
        &= \frac{1}{N} \sum_j \| \hat{\Blambda}_j - \hat{H}'\tilde{\Blambda}_j \|^2 + N^{-1} (\hat{\Lambda} - \tilde{\Lambda} \hat{H})' \tilde{\Lambda} \hat{H}
    \end{align*}
where the first term in the display is $O_{\p}(\delta_{NT}^{-1})$ by Theorem \ref{th:norm squared of estimator} and the second term is $O_{\p}(\delta_{NT}^{-1})$ by Lemma \ref{lemma:intermediate_2_factor_loadings}.
\end{proof}

Now, we can fully characterize the asymptotic linear for the estimator of the common factors. 
\begin{align*}
    \hat{F}_t - \hat{H}^{-1} \tilde{F}_t &= \hat{H}' \cdot \frac{1}{N} \sum_{i = 1}^N \tilde{\Blambda}_i \Bar{u}_{it} + N^{-1} \hat{\Lambda}' (\tilde{\Lambda} - \hat{\Lambda} \hat{H}^{-1}) \tilde{F}_t + N^{-1} ( \hat{\Lambda} - \tilde{\Lambda} \hat{H})' \Bar{u}_t  \\
    &= \hat{H}' \cdot \frac{1}{N} \sum_{i = 1}^N \tilde{\Blambda}_i \Bar{u}_{it} + O_{\p} \left( \frac{1}{\delta_{NT}^2} \right)
\end{align*}

\section{Estimation of the Common Component} \label{app:estimation of common components}
In this section, I will bring the estimation of the factor loadings and common factors together to characterize the asymptotic linear form of the common component,$\hat{C}_{it}$, where
\begin{equation*}
    \hat{C}_{it} = \hat{\Blambda}_i' \hat{F}_t
\end{equation*}
The difference between the estimated component and the true component is
\begin{align*}
    \hat{C}_{it} - \tilde{C}_{it} &= \hat{\Blambda}_i' \hat{F}_t - \tilde{\Blambda}_i' \tilde{F}_t \\
    &= \hat{\Blambda}_i' \hat{F}_t - \tilde{\Blambda}_i' \hat{H} \hat{H}^{-1} \tilde{F}_t + \hat{\Blambda}_i' \hat{H}^{-1} \tilde{F}_t - \hat{\Blambda}_i' \hat{H}^{-1} \tilde{F}_t \\
    &= \big[ \hat{\Blambda}_i - \hat{H}' \tilde{\Blambda}_i \big]' \hat{H}^{-1} \tilde{F}_t + \hat{\Blambda}_i' \big[ \hat{F}_t - \hat{H}^{-1} \tilde{F}_t \big]
\end{align*}
That is, the difference is the sum of two terms. The first term in the sum is easier to analyze based on what we have derived in Appendix \ref{app:estimation of factor loadings}. But the second term is slightly more involved due to multiplication of the asymptotic linear term with $\hat{\Blambda}_i$. So we will consider this term first.
\begin{align*}
    \hat{\Blambda}_i' \big[ \hat{F}_t - \hat{H}^{-1} \tilde{F}_t \big] &= \tilde{\Blambda}_i' \hat{H} \big[ \hat{F}_t - \hat{H}^{-1} \tilde{F}_t \big] + \big[ \hat{\Blambda}_i' -  \tilde{\Blambda}_i' \hat{H} \big] \big[ \hat{F}_t - \hat{H}^{-1} \tilde{F}_t \big] \\
    &= \tilde{\Blambda}_i' \hat{H} \big[ \hat{F}_t - \hat{H}^{-1} \tilde{F}_t \big] + O_{\p} \left( \frac{1}{\sqrt{T}} \right)O_{\p} \left( \frac{1}{\sqrt{N}} \right)
\end{align*}
Thus, the difference between the estimated component and the true component simplifies to
\begin{align*}
    \hat{C}_{it} - \tilde{C}_{it} &= \big[ \hat{\Blambda}_i - \hat{H}' \tilde{\Blambda}_i \big]' \hat{H}^{-1} \tilde{F}_t + \tilde{\Blambda}_i' \hat{H} \big[ \hat{F}_t - \hat{H}^{-1} \tilde{F}_t \big] + O_{\p} \left( \frac{1}{\sqrt{NT}} \right)
\end{align*}
Using the asymptotic linear forms that we derived in the previous two sections, we can write the above display as
\begin{align*}
    \hat{C}_{it} - \tilde{C}_{it} &= \tilde{F}_t' \hat{H'}^{-1} V_{NT}^{-1} \cdot \frac{\tilde{\Lambda}' \hat{\Lambda}}{N} \cdot \frac{1}{T} \sum_{s=1}^T \tilde{F}_s \Bar{u}_{is} +  \tilde{\Blambda}_i' \hat{H} \hat{H}' \cdot \frac{1}{N} \sum_{j = 1}^N \tilde{\Blambda}_j \Bar{u}_{jt} + O_{\p} \left( \frac{1}{\delta_{NT}^2} \right)
\end{align*}
Recall the definition of $\hat{H}$
\begin{equation*}
    \hat{H} = \left[ \frac{\tilde{F}' \tilde{F}}{NT} \right]\cdot \tilde{\Lambda}' \hat{\Lambda} \cdot V_{NT}^{-1}
\end{equation*}
Thus,
\begin{equation*}
    \hat{H'}^{-1} V_{NT}^{-1} \frac{\tilde{\Lambda}' \hat{\Lambda}}{N} = \left[ \frac{\tilde{F}' \tilde{F}}{T} \right]^{-1}
\end{equation*}
To analyze $\hat{H} \hat{H}'$, we have the following Lemma
\begin{Lemma} \label{lemma:H_prime_H}
    Under Assumptions \ref{ass:LN_strong factor structure} to \ref{ass:LN_moments and CLT},
    \begin{equation*}
       \hat{H} \hat{H}' = \left[ \frac{\tilde{\Lambda}' \tilde{\Lambda}}{N} \right]^{-1} + O_{\p}(\delta_{NT}^{-2})
    \end{equation*}
\end{Lemma}
\begin{proof}
As per Lemma \ref{lemma:intermediate_2_factor_loadings},
\begin{align*}
    N^{-1} (\hat{\Lambda} - \tilde{\Lambda} \hat{H})' \tilde{\Lambda} &= O_{\p}\left(\frac{1}{\delta_{NT}^2} \right) \\
    \left[ \frac{\hat{\Lambda}' \tilde{\Lambda}}{N} \right] - \hat{H}' \left[ \frac{\tilde{\Lambda}' \tilde{\Lambda}}{N} \right] &= O_{\p}(\delta_{NT}^{-2})\\
    \hat{H}' \left[ \frac{\tilde{\Lambda}' \tilde{\Lambda}}{N} \right] &= \left[ \frac{\hat{\Lambda}' \tilde{\Lambda}}{N} \right] + O_{\p}(\delta_{NT}^{-2}) \\
    \hat{H} &= \left[ \frac{\tilde{\Lambda}' \tilde{\Lambda}}{N} \right]^{-1} \left[ \frac{\tilde{\Lambda}' \hat{\Lambda}}{N} \right] + O_{\p}(\delta_{NT}^{-2})
\end{align*}
As per Lemma \ref{lemma:intermediate_3_factor_loadings}
\begin{align*}
    N^{-1} (\hat{\Lambda} - \tilde{\Lambda} \hat{H})' \hat{\Lambda} &= O_{\p}\left(\frac{1}{\delta_{NT}^2} \right) \\
    \left[ \frac{\hat{\Lambda}' \hat{\Lambda}}{N} \right] - \hat{H}' \left[ \frac{\tilde{\Lambda}' \hat{\Lambda}}{N} \right] &= O_{\p}(\delta_{NT}^{-2})\\
    \hat{H}' \left[ \frac{\tilde{\Lambda}' \hat{\Lambda}}{N} \right] &= I_{r-1} + O_{\p}(\delta_{NT}^{-2}) \\
    \left[ \frac{\tilde{\Lambda}' \hat{\Lambda}}{N} \right] &= \hat{H'}^{-1} + \hat{H'}^{-1} O_{\p}(\delta_{NT}^{-2}) 
\end{align*}
Plugging this expression for $\left[ \frac{\tilde{\Lambda}' \hat{\Lambda}}{N} \right]$ in the equation for $\hat{H}$, we have
\begin{align*}
    \hat{H} &= \left[ \frac{\tilde{\Lambda}' \tilde{\Lambda}}{N} \right]^{-1} \left[ \hat{H'}^{-1} + \hat{H'}^{-1} O_{\p}(\delta_{NT}^{-2})  \right] + O_{\p}(\delta_{NT}^{-2})
\end{align*}

Multiplying on the right by $\hat{H'}^{-1}$
\begin{align*}
    \hat{H} \hat{H'}^{-1} &= \left[ \frac{\tilde{\Lambda}' \tilde{\Lambda}}{N} \right]^{-1} + \left[ \frac{\tilde{\Lambda}' \tilde{\Lambda}}{N} \right]^{-1} \hat{H'}^{-1} O_{\p}(\delta_{NT}^{-2}) \hat{H'}^{-1} + O_{\p}(\delta_{NT}^{-2}) \hat{H'}^{-1} \\
    &= \left[ \frac{\tilde{\Lambda}' \tilde{\Lambda}}{N} \right]^{-1} + O_{\p}(\delta_{NT}^{-2})
\end{align*}
where the stochastic order in the last equality comes from Assumption \ref{ass:LN_strong factor structure}.2 and the boundedness of $\hat{H}$, which we have already seem in the proof of Lemma \ref{lemma:intermediate_1_factor_loadings}.

\end{proof}

Thus, the difference between the estimated component and the true component simplifies to
\begin{align*}
    \hat{C}_{it} - \tilde{C}_{it} &= \tilde{F}_t' \left[ \frac{\tilde{F}' \tilde{F}}{T} \right]^{-1} \frac{1}{T} \sum_{s=1}^T \tilde{F}_s \Bar{u}_{is} +  \tilde{\Blambda}_i' \left[ \frac{\tilde{\Lambda}' \tilde{\Lambda}}{N} \right]^{-1} \frac{1}{N} \sum_{j = 1}^N \tilde{\Blambda}_j \Bar{u}_{jt} + O_{\p} \left( \frac{1}{\delta_{NT}^2 \sqrt{N}} \right) + O_{\p} \left( \frac{1}{\delta_{NT}^2} \right)
\end{align*}
In vector form,
\begin{equation} \label{eq:app_influence function of estimator}
    \hat{C}_t - \tilde{C}_t = \tilde{F}_t' \left[ \frac{\tilde{F}' \tilde{F}}{T} \right]^{-1} \frac{1}{T} \sum_{s=1}^T \tilde{F}_s \Bar{u}_{s} + \tilde{\Lambda} \left[ \frac{\tilde{\Lambda}' \tilde{\Lambda}}{N} \right]^{-1} \frac{1}{N} \sum_{j = 1}^N \tilde{\Blambda}_j \Bar{u}_{jt} + O_{\p} \left( \frac{1}{\delta_{NT}^2} \right)
\end{equation}

\end{document}